\documentclass[a4paper,reqno,11pt]{amsart}
\usepackage[hmargin=2.5cm,vmargin=2.5cm]{geometry}

\usepackage[dvipsnames]{xcolor}
\usepackage{amssymb,amsthm,mathtools,cancel,mathrsfs,mathdots,graphicx,framed,enumitem,tikz,subcaption,bm,eucal,thmtools}

\usetikzlibrary{patterns, decorations.markings,arrows, calc}

\usepackage[colorlinks=true, pdfstartview=FitV, urlcolor=blue, citecolor=red, linkcolor=blue, unicode, backref=page]{hyperref}
\usepackage{stmaryrd}
\SetSymbolFont{stmry}{bold}{U}{stmry}{m}{n} 

\usepackage{cleveref}

\usepackage[T1]{fontenc}
\usepackage[landscape=true]{pdflscape}

\renewcommand{\Re}{\operatorname{Re}}
\renewcommand{\Im}{\operatorname{Im}}

\def\wt{\widetilde}
\def\wh{\widehat}

\def\d{\mathrm{d}}
\def\e{\mathrm{e}}
\def\i{\mathrm{i}}
\def\sfC{\mathsf{C}}
\def\sfq{\mathsf{q}}

\def\res{\mathop{\mathrm {res}}\limits_}

\newtheorem{theorem}{Theorem}[section]

\newtheorem{definition}[theorem]{Definition}

\newtheorem{lemma}[theorem]{Lemma}
\newtheorem{proposition}[theorem]{Proposition} 
\newtheorem{corollary}[theorem]{Corollary}

\newtheorem{dRHp}[theorem]{Discrete Riemann--Hilbert problem}
\newtheorem{cRHp}[theorem]{Continuous Riemann--Hilbert problem}

\theoremstyle{remark}
\newtheorem{remark}[theorem]{Remark}

\allowdisplaybreaks

\begin{document}

\numberwithin{equation}{section}

\title[
Multiplicative Averages of Plancherel Random Partitions
]{
Multiplicative Averages of Plancherel Random Partitions: Elliptic Functions, Phase Transitions, and Applications
}

\author{Mattia Cafasso}
\address[M.~Cafasso]{Univ Angers, CNRS, LAREMA, SFR MATHSTIC, F-49000 Angers, France}
\email{mattia.cafasso@univ-angers.fr}

\author{Matteo Mucciconi}
\address[M.~Mucciconi]{ 
 Department of Mathematics, National University of Singapore,
 S17, 10 Lower Kent Ridge Road, 119076, Singapore.}
\email{matteomucciconi@gmail.com}

\author{Giulio Ruzza}
\address[G.~Ruzza]{CEMS.UL, Departamento de Ci\^encias Matem\'aticas, Faculdade de Ci\^encias da Universidade de Lisboa, Campo Grande Edif\'{i}cio C6, 1749-016, Lisboa, Portugal}
\email{gruzza@ciencias.ulisboa.pt}

\date{}

\begin{abstract}
    We consider random integer partitions~$\lambda$ that follow the Poissonized Plancherel measure of parameter~$t^2$.
    Using Riemann--Hilbert techniques, we establish the asymptotics of the multiplicative averages 
    \[
        Q(t,s)=\mathbb{E} \left[ \prod_{i\ge 1} \left(1+\e^{\eta(\lambda_i-i+\frac 12-s)}\right)^{-1} \right]
    \]
    for fixed $\eta>0$ in the regime $t\to+\infty$ and $s/t=O(1)$.
    We compute the large-$t$ expansion of $\log Q(t,xt)$ expressing the rate function $\mathcal{F}(x) = -\lim_{t \to \infty}  t^{-2}\log Q(t,xt)$ and the subsequent divergent and oscillatory contributions explicitly in terms of elliptic theta functions. 
    The associated equilibrium measure presents, in general, nontrivial saturated regions and it undergoes two third-order phase transitions of different nature which we describe.
    Applications of our results include an explicit characterization of tail probabilities of the height function of the $q$-deformed polynuclear growth model and of the edge of the positive-temperature discrete Bessel process and asymptotics of radially symmetric solutions to the 2D~Toda equation with step-like initial data.
\end{abstract}


\maketitle

\setcounter{tocdepth}{1}

{
\hypersetup{linkcolor=black}
\tableofcontents
}

\section{Introduction and results}

\subsection{Overview}

The Plancherel measure is a probability measure over the set of integer partitions $\lambda$ of a natural number $n$, which arises naturally in representation-theoretic, combinatorial, and probabilistic contexts \cite{okounkov2006uses}. It assigns to a partition $\lambda$ a probability mass proportional to the square of the dimension of the irreducible representation of the symmetric group $S_n$ indexed by $\lambda$. It is also famously related to the distribution of the longest increasing subsequence of a uniformly distributed random permutation of $n$ elements \cite{Schensted1961,logan_shepp1977variational,VershikKerov_LimShape1077,baik1999distribution,romik2015surprising}; see also \cite{GREENE1979,Baik1999second,Borodin2000b,Johansso1999Plancherel}.
The \emph{Poissonized Plancherel measure} (which we denote by $\mathcal{P}_{t^2}$) occurs when the natural number $n$ is also random and follows a Poisson distribution with a parameter that in this paper we denote by $t^2$.
For an integer partition $\lambda = (\lambda_1,\lambda_2,\ldots)$ it is given explicitly by
\begin{equation}\label{eq:poissonized plancherel}
    \mathcal{P}_{t^2}\bigl(\lbrace\lambda\rbrace\bigr) \,=\, \e^{-t^2} \, \frac{t^{2|\lambda|}}{ \prod_{i,j} (\lambda_i-i + \lambda_j'-j+1)^2},
\end{equation}
where $\lambda'=(\lambda_1',\lambda_2',\dots)$ with $\lambda_j'=\#\{ k: \lambda_k \ge j \}$ is the \emph{transposed} partition. 
The product in the denominator ranges over $i,j\in\mathbb{N}$ such that $1\leq j\leq \lambda_i$ and it is commonly called \emph{hook product} because $\lambda_i-i + \lambda_j'-j+1$ is the \emph{hook length} of the cell $(i,j)$, as depicted in \Cref{fig:partition_hook_maya}.

\begin{figure}
    \centering
    \includegraphics[width=0.4\linewidth]{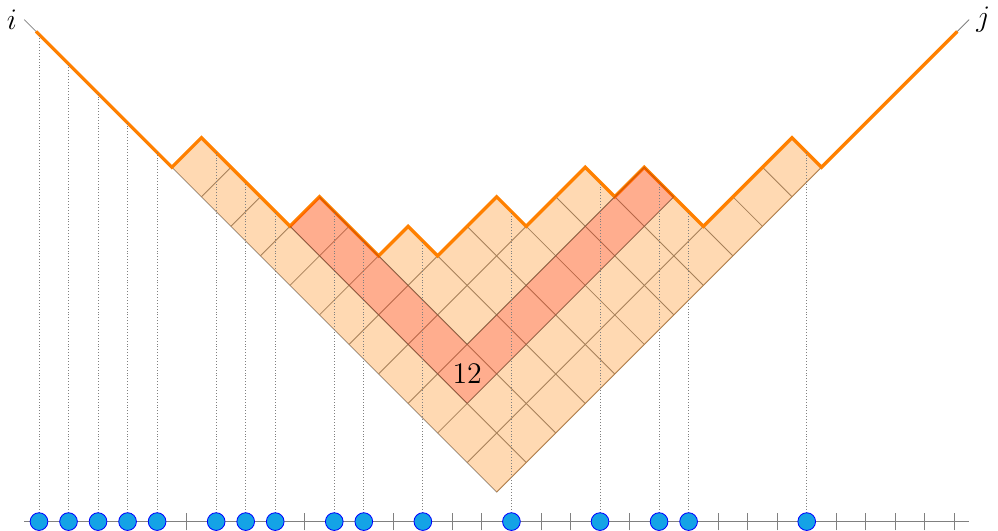}
    \caption{In \textcolor{brown}{orange}, the Young diagram (in Russian notation) of the partition $\lambda = (11,8,8,7,5,3,2,2,1,1,1)$. In \textcolor{Cerulean}{blue}, the Maya diagram $\mathscr{D}(\lambda)$.
    The darker shaded cells represent the hook of the cell $(3,2)$, whose length is $12$.}
    \label{fig:partition_hook_maya}
\end{figure}

The Poissonized Plancherel measure possesses a rich structure.
For instance, defining the \textit{Maya diagram} of an integer partition~$\lambda$ as the subset
\begin{equation}\label{eq:point process D}
    \mathscr{D}(\lambda) \,= \,\left \lbrace \lambda_i - i +\frac 12 \, : \, i\in \mathbb{N} \right\rbrace 
\end{equation}
of the half-integers $\mathbb{Z}'=\mathbb{Z}+\frac 12$, the point process $\mathscr{D}(\lambda)$ on~$\mathbb{Z}'$ (where $\lambda$ follows the Poissonized Plancherel measure) is \emph{determinantal}, with correlation kernel~\cite{Borodin1999RSK,Borodin2000a,Johansso1999Plancherel,Borodin2000b} 
\begin{equation}
\label{eq:dBKernel}
   \mathsf{K}(i,j;t)\, = \,t \,\frac{ \mathrm{J}_{i-\frac 12}(2t) \,\mathrm{J}_{j+\frac 12}(2t) \,-\, \mathrm{J}_{i+\frac 12}(2t) \,\mathrm{J}_{j-\frac 12}(2t) }{i\,-\,j},
   \qquad i,j\in\mathbb{Z}'.
\end{equation}
Here, $\mathrm{J}_k$ denotes the Bessel function of the first kind of order~$k$ and the diagonal entries of the kernel are defined by the limit $\mathsf{K}(i,i;t) = \lim_{j \to i} \mathsf{K}(i,j;t)$.
This means that $\mathbb{P}\left[ S \subset \mathscr{D}(\lambda) \right] = \det\left[ \mathsf{K}(i,j;t) \right]_{i,j \in S}$ for any finite set~$S\subset\mathbb{Z}'$. 

\medskip

This paper is devoted to the study of the following multiplicative averages of the Poissonized Plancherel measure, defined for $\eta>0$, $t>0$, $s\in\mathbb{Z}$, by
\begin{equation}
\label{eq:defQ}
    Q(t,s) \,=\, \mathbb{E}_{t^2}  \left[ \prod_{\xi\in\mathscr D(\lambda)} \frac{1}{1+\exp\bigl(\eta(\xi-s)\bigr)}\right] =\, \mathbb{E}_{t^2} \left[ \prod_{i \ge 1} \frac{1}{1+\exp\bigl(\eta(\lambda_i-i+\frac 12-s)\bigr)} \right],
\end{equation}
and, in particular, to its asymptotics when $t\to+\infty$ with $s/t=O(1)$.
For brevity, we will often omit the dependence on~$\eta$ in the notation.

\medskip

Multiplicative averages of the form \eqref{eq:defQ} have recently attracted considerable interest due to their connection with integrable systems~\cite{okounkov2001infinite,krajenbrink2020painleve,cafasso2021airy,quastel_remenik_2022_KP,ghosal2023universality,CafassoRuzza23,ruzza2025bessel,matetski2025polynuclear}, their use in \emph{number rigidity} and \emph{thinning} of the underlying point process \cite{ghosh2015determinantal,bufetov2016rigidity,claeys2023determinantal,bufetov2024expectation,claeys2024janossy},
and, especially, their application in the study of solvable growth processes in the Kardar--Parisi--Zhang universality class \cite{borodin2016moments,borodin2016stochastic_MM,BO2016_ASEP,corwin2018coulomb,lwtail,tsai_lower_tail,cafasso_claeys_KPZ,bothner2022momenta,charlier2022uniform,claeys2024deformations,claeys2024large,claeys2024integrable,das2025large,dlm24,zhong2024large,Ghosal_Silva_6VM}.
Hence, the asymptotic analysis of multiplicative averages such as~$Q(t,s)$ has immediate applications to the description of the asymptotic behavior of integrable systems and large deviations of solvable models of one-dimensional growing interfaces, as we will see below.

\medskip

Global asymptotics of the (Poissonized) Plancherel measure are often carried out through its relation with log-gases, originally discovered independently by B.~Logan and L.~Shepp~\cite{logan_shepp1977variational} and by A.~Vershik and S.~Kerov~\cite{VershikKerov_LimShape1077}.
This connection allows one to express the probability mass of a large partition $\lambda$ as
\begin{equation}
\label{eq:poissonized_Plancherel_log_ergy}
    \mathcal{P}_{t^2}\bigl(\lbrace\lambda\rbrace\bigr) \,=\, \exp\bigl( -t^2 ( 1+ \mathcal{E}[\mathfrak{h}] ) + o(t^2) \bigr).
\end{equation}
Here, $\mathcal{E}$ is the logarithmic energy
\begin{equation}
\label{eq:VKLSenergy}
    \mathcal{E}[\mathfrak{h}] = \iint \log \frac{1}{|\mu-\nu|} \mathfrak{h}(\mu) \mathfrak{h}(\nu) \d \mu \, \d \nu + \int 2 \left( \mu \log |\mu| - \mu  \right) \mathfrak{h}(\mu) \d \mu
\end{equation}
and the function $\mathfrak{h}$ is the \emph{half-complemented rescaled empirical measure} of the point process $\mathscr{D}(\lambda)$, defined by
\begin{equation}
\label{eq:complementation}
    \mathfrak{h}(\mu) = \rho(\mu)- \mathbf{1}_{(-\infty,0]}(\mu),
\end{equation}
where $\rho$ is the rescaled \emph{empirical measure}
\begin{equation}
\label{eq:empirical measure}
    \rho(\mu) = \mathbf 1_{\mathscr D(\lambda)}\left(\lfloor t\mu \rfloor +\tfrac{1}{2}\right) .
\end{equation}
(Throughout this paper, we denote by $\mathbf 1_X$ the characteristic function of a set~$X$.)
The domain of~$\mathfrak{h}$ is the convex set
\begin{equation}
\label{eq:Hprimespace}
    \mathcal{H} = \left\{ \mathfrak{h} \in L^1(\mathbb{R}): \mathbf{1}_{(-\infty,0]}(\mu)+ \mathfrak{h}(\mu) \in [0,1] \, \, \text{and} \, \, \int_{\mathbb{R}} \mathfrak{h}(\mu) \d \mu=0 \right\}.
\end{equation}
Using the expansion analogous to \eqref{eq:poissonized_Plancherel_log_ergy} for the (non Poissonized) Plancherel measure, A. Vershik, S. Kerov, B. Logan, and L. Shepp proved that the random function~$\rho$ converges (weakly almost surely) to the \emph{arccosine density}
\begin{equation}
\label{eq:arccosine law}
    \rho(\mu) \xlongrightarrow[t \to \infty]{} \rho_{\mathrm{VKLS}}(\mu) = \mathbf 1_{(-\infty,-2)}(\mu)+\frac{\mathbf 1_{(-2,2)}(\mu)}{\pi}\arccos\left(\frac{\mu}{2} \right).
\end{equation}
In the same article~\cite{logan_shepp1977variational}, B.~Logan and L.~Shepp also described the limiting behavior of a large partition~$\lambda$ subject to the condition of having small first row~$\lambda_1$ and/or small first column~$\lambda_1'$, by solving constrained optimization problems associated with the energy~$\mathcal{E}$.
These results have later found use in more probabilistic context to describe the lower-tail large deviation rate function for the longest increasing subsequence of a uniform random permutations or a Poisson planar environment~\cite{deuschel_zeitouni_1999,seppalainen_98_increasing}.
We recover these limiting behaviors of constrained large partitions from the results of this paper in the limit $\eta\to+\infty$; see~\Cref{rem:q0}.

\subsection{Results}

The connection between the Poissonized Plancherel measure and log-gases is crucial also in the asymptotic analysis of the multiplicative average~$Q(t,s)$.
A basic Laplace-style argument, using the Hardy--Ramanujan bound to control the number of partitions of size~$O(t^2)$, shows that
\begin{equation}
\label{eq:Q by Laplace argument} 
	Q(t,xt) \,=\, \exp\bigl( - t^2 \mathcal{F}(x) + o(t^2) \bigr)\,,\qquad \text{as }t\to+\infty,
\end{equation}
(see \Cref{thm:Laplace_argument}), where
\begin{equation}
	\mathcal{F}(x) \,=\, 1\, +\, \frac{\eta x^2}{2}  \mathbf{1}_{(-\infty,0)}(x) \,+\, \min_{\mathfrak{h} \in \mathcal{H}} \mathcal{E}_{\eta,x} \left[ \mathfrak{h} \right]\,, 
\end{equation}
and~$\mathcal{E}_{\eta,x}$ is the logarithmic energy
\begin{equation}
\label{eq:log-energy}
    \mathcal{E}_{\eta,x}\left[\mathfrak{h} \right] = \iint \log \frac{1}{|\mu-\nu|} \mathfrak{h}(\mu) \mathfrak{h}(\nu) \d \mu \, \d \nu + \int  \bigl( 2\mu (\log |\mu| - 1) + V_{\eta,x}(\mu)  \bigr) \mathfrak{h}(\mu) \d \mu,
\end{equation}
with
\begin{equation}
\label{eq:potentialVqx}
V_{\eta,x}(\mu)=\frac{\eta}{2}[ \mu - x ]_+\,.
\end{equation}
(Throughout this paper, we denote~$[f]_+=\max\lbrace f,0\rbrace$.) Then, by~\eqref{eq:Q by Laplace argument}, the leading term in the asymptotic form of~$\log Q(t,xt)$ can be obtained by an optimization problem associated with the (quadratic and positive definite) energy~$\mathcal{E}_{\eta,x}$.
Despite differing from the energy~$\mathcal{E}$ simply by a piecewise linear potential term, constrained minimization problems of~$\mathcal{E}_{\eta,x}$ over the set of (half-complemented) densities~$\mathcal{H}$ turn out to be substantially richer, with optimal profiles exhibiting new behaviors depending on the parameter~$x$.
The first main result of this paper is the explicit formula for the optimal density~$\rho_{\eta,x}$ associated with the energy~$\mathcal{E}_{\eta,x}$, which is presented in the following subsection.

\begin{figure}[t]
\centering
\begin{subfigure}{.45\textwidth}
\frame{\includegraphics[scale=.28]{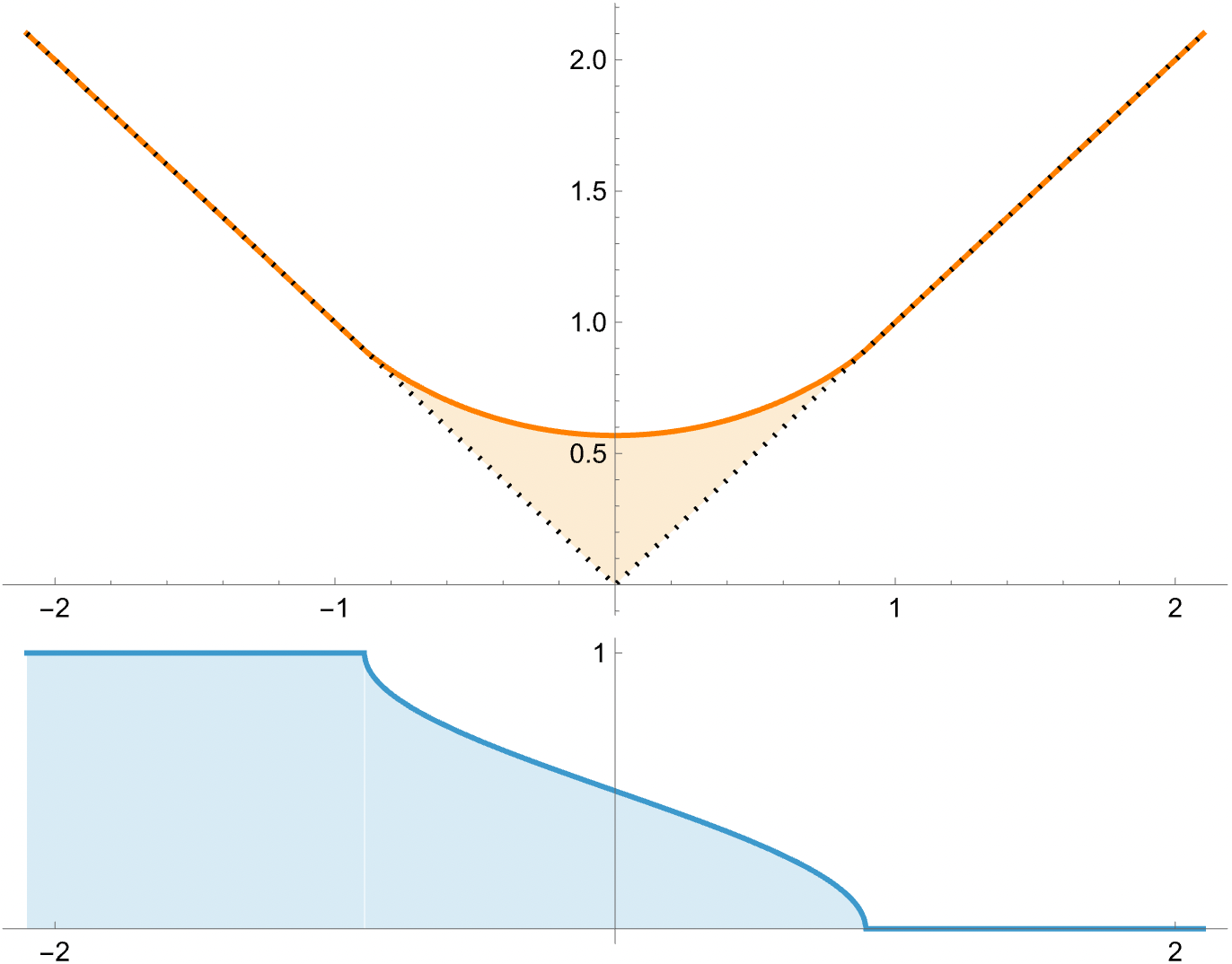}}
\caption*{$x \le x_* = -0.994136\ldots$}
\end{subfigure}
\quad
\begin{subfigure}{.45\textwidth}
\frame{\includegraphics[scale=.28]{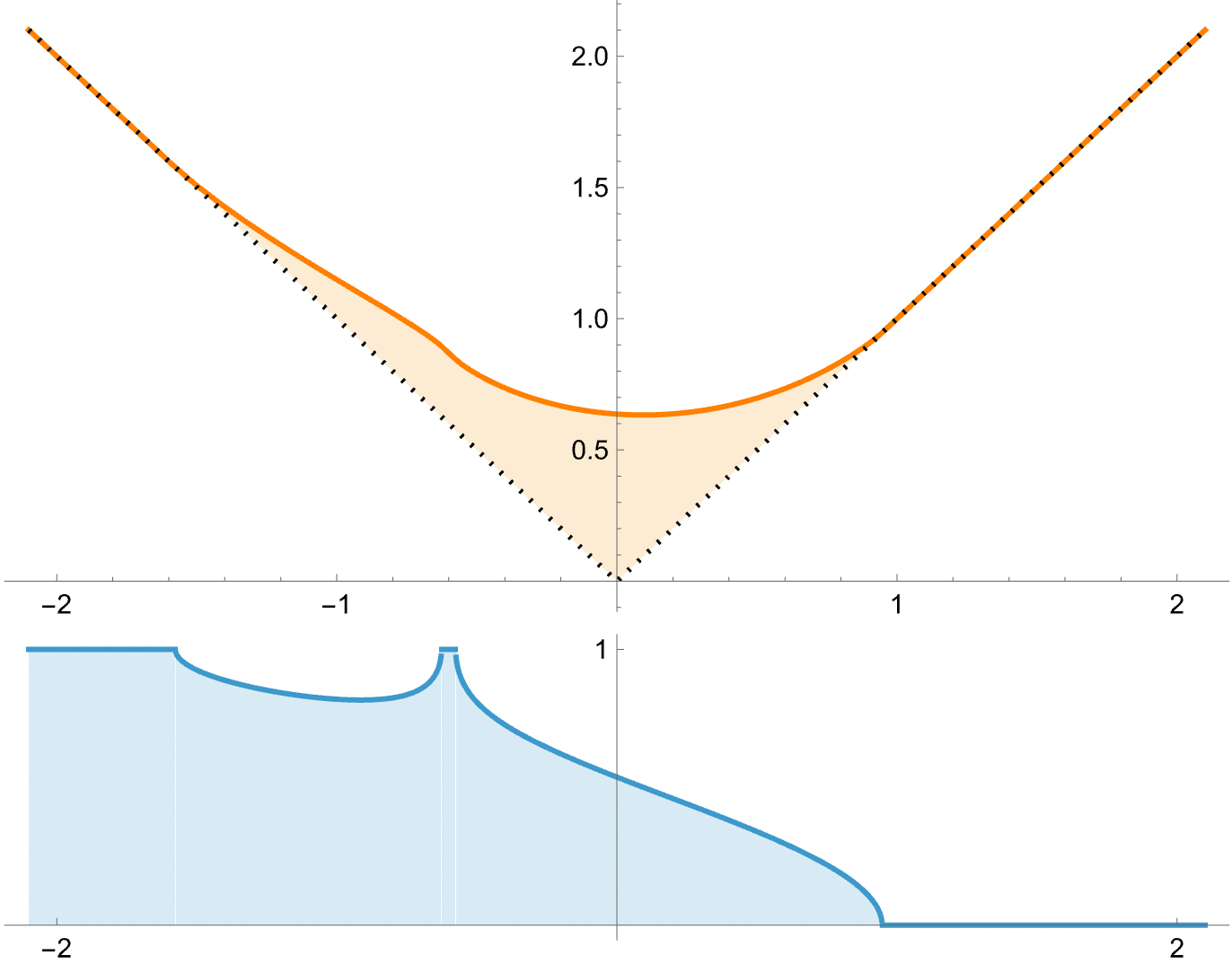}}
\caption*{$x=-0.6$}
\end{subfigure}
\begin{subfigure}{.45\textwidth}
\frame{\includegraphics[scale=.28]{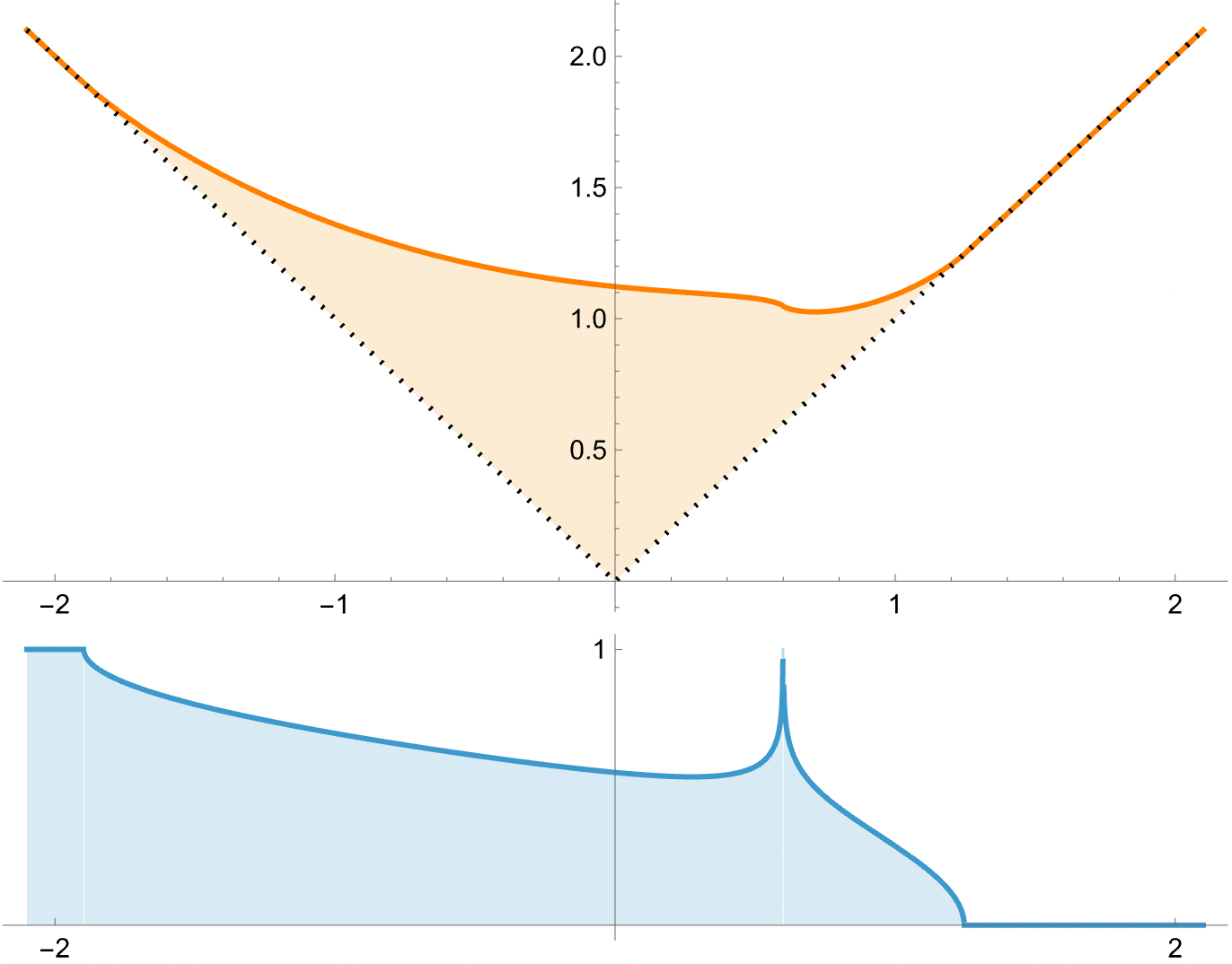}}
\caption*{$x=0.6$}
\end{subfigure}
\quad
\begin{subfigure}{.45\textwidth}
\frame{\includegraphics[scale=.28]{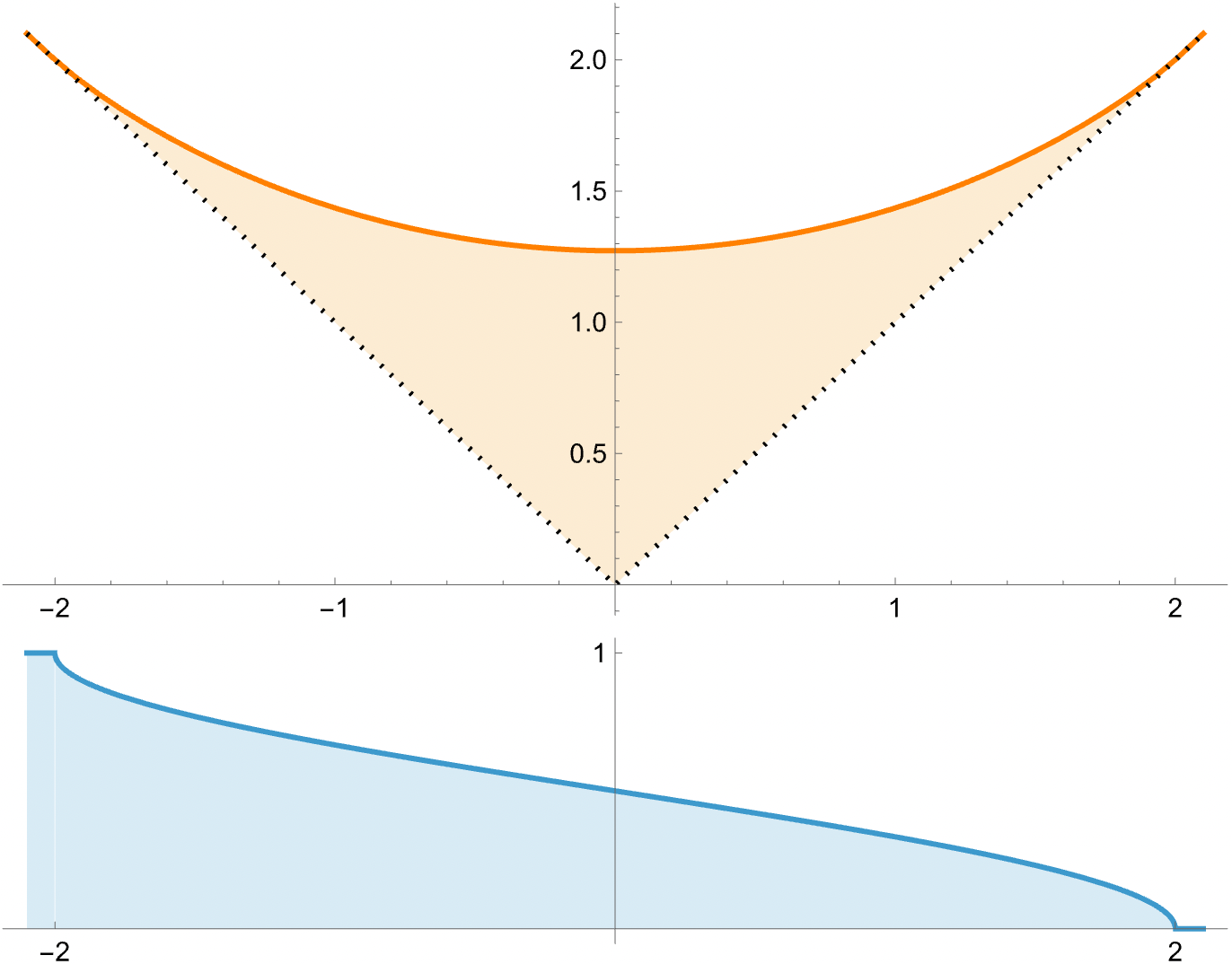}}
\caption*{$x \ge 2$}
\end{subfigure}
\caption{The various phases of the equilibrium measure $\rho_{\eta,x}$ (in \textcolor{Cerulean}{blue}) and the corresponding limiting partition (in \textcolor{brown}{orange}) in Russian notation. Here $\eta=\log 5$. In the top-left panel, when $x \le x_*$, the density $\rho_{\eta,x}$ is a rescaling of the Vershik--Kerov--Logan--Shepp density by a factor~$\e^{-\eta/2}$.
In the top-right and bottom-left panels the cases when $\rho_{\eta,x}$ possesses two saturated regions and its explicit form is given in \eqref{eq:maintheorem:densitymiddle}. In the bottom-right panel, where $x \ge 2$, $\rho_{\eta,x}$ coincides with the Vershik--Kerov--Logan--Shepp density.}
\label{fig:equilibrium measures}
\end{figure}

\subsubsection{Minimization of logarithmic energy}

We will need the following elliptic theta functions (see also Appendix~\ref{app:elliptic}):
\begin{equation}
\label{eq:theta}
\begin{aligned}
\vartheta(z|\tau) &= \sum_{n\in\mathbb{Z}}\e^{2\pi\i n z+\i\pi n^2\tau},
&\vartheta_{10}(z|\tau) &= \vartheta\bigl(z+\tfrac{1}{2}\big|\tau\bigr),
\\
\vartheta_{01}(z|\tau) &= \e^{\i\pi(z+\frac \tau 4)}\,\vartheta\bigl(z+\tfrac{\tau}{2}\big|\tau\bigr),
&
\qquad
\vartheta_{11}(z|\tau) 
&= \i\,\e^{\i\pi(z+\frac \tau 4)}\,\vartheta\bigl(z+\tfrac{1+\tau}{2}\big|\tau\bigr),
\end{aligned}
\end{equation}
for $z\in\mathbb{C}$ and $\tau\in\mathbb{C}$ satisfying~$\Im\tau>0$. We will denote by $\vartheta'(z|\tau)$ the derivative in the argument $z$ of $\vartheta(z|\tau)$, and similarly for $\vartheta_{ij}$.
We also need to define some special functions.

\begin{definition}[Implicit half-period] \label{def:K(x)}
    For all $\eta>0$ we denote
    \begin{equation}
    \label{eq:xq}
    x_*\,=\,-2\,(1-\e^{-\eta})\,\eta^{-1}.
    \end{equation}
    We define the function $\mathcal{U}:(\eta/2,+\infty)\to \mathbb{R}$ by
    \begin{equation}
    \label{eq:f1}
        \mathcal{U}(K) \,=\, 2K\e^{\frac {\eta}2\bigl(\frac{\eta}{2K}-1\bigr)}\frac{\vartheta_{11}(\frac{\eta}{2K}|\frac{\i\pi}{K})}{\vartheta_{11}'(0|\frac{\i\pi}{K})}
    \end{equation}
    and the function $\mathcal{K}:(x_*,2)\to(\eta/2,+\infty)$ by letting, for any $x\in(x_*,2)$, $\mathcal{K}(x)=K$ where $K$ is the unique solution in $(\eta/2,+\infty)$ of
    \begin{equation}
    \label{eq:x}
        \left(1-K\frac{\partial}{\partial K}\right)\mathcal{U}(K) \,=\, -\frac{\eta x}2\,.
    \end{equation}
\end{definition}
We will show (see Proposition~\ref{prop:f1f2body}) that $\left(1-K\frac{\partial}{\partial K}\right)\mathcal{U}(K)$ is an increasing function of $K\in (\eta/2,+\infty)$ with range~$(-\eta,1-\e^{-\eta})$.
Hence, the function $\mathcal{K}$ is well defined. 

\medskip

We will denote by~$\rho_{\mathrm{VKLS}}$ the Vershik--Kerov--Logan--Shepp density; see~\eqref{eq:arccosine law}.

\begin{theorem}[Equilibrium measure]\label{thm:minimizer}
    Let $\eta>0$ and $x\in \mathbb{R}$.
    The minimizer $\rho_{\eta,x}$ of the logarithmic energy $\mathcal{E}_{\eta,x}$ is given explicitly as follows.
    \begin{enumerate}[leftmargin=*]
    \item If $x \le x_*$, $\rho_{\eta,x}(\mu)=\rho_{\mathrm{VKLS}}\bigl(\e^{\eta/2}\mu\bigr)$.
    
    \item If $x_*<x<2$,
    \begin{equation}
    \label{eq:maintheorem:densitymiddle}
    \begin{aligned}
        \rho_{\eta,x}(\mu) ={}& \mathbf 1_{(-\infty,a)\cup(b,c)}(\mu)
        \\
        &+\mathbf 1_{(a,b)}(\mu)\biggl[1+\frac{R(\mu)}{\pi}\biggl(\int_{d}^{+\infty}\frac{\d\nu}{R(\nu)(\nu-\mu)}-\frac {\eta}{2\pi}\int_c^d\frac{\d\nu}{R(\nu)(\nu-\mu)}\biggr)\biggr]
        \\&+\mathbf 1_{(c,d)}(\mu)\biggl[ 1-\frac{R(\mu)}{\pi}
        \biggl(\int_d^{+\infty}\frac{\d\nu}{R(\nu)(\nu-\mu)}-\frac {\eta}{2\pi}\,\mathrm{p.v.}\int_c^d\frac{\d\nu}{R(\nu)(\nu-\mu)}\biggl) \biggr]
    \end{aligned}
    \end{equation}
    where
    \begin{equation}
    R(\mu)=\bigl|(\mu-a)(\mu-b)(\mu-c)(\mu-d)\bigr|^{1/2}.
    \end{equation}
    The endpoints $a=a(\eta,x)$, $b=b(\eta,x)$, $c=c(\eta,x)$, and $d=d(\eta,x)$ are given by
    \begin{equation}
    \label{eq:maintheorem:endpoints}
    a=\mathcal T\biggl(\frac{\vartheta'_{01}(\frac{\eta}{4\mathcal{K}}| \frac{ \i\pi}{ \mathcal{K}})}{\vartheta_{01}(\frac{\eta}{4\mathcal{K}}| \frac{ \i\pi}{ \mathcal{K}})}\biggr),
    \
    b=\mathcal T\biggl(\frac{\vartheta'(\frac{\eta}{4\mathcal{K}}| \frac{ \i\pi}{ \mathcal{K}})}{\vartheta(\frac{\eta}{4\mathcal{K}}| \frac{ \i\pi}{ \mathcal{K}})}\biggr),
    \
    c=\mathcal T\biggl(\frac{\vartheta'_{10}(\frac{\eta}{4\mathcal{K}}| \frac{ \i\pi}{ \mathcal{K}})}{\vartheta_{10}(\frac{\eta}{4\mathcal{K}}| \frac{ \i\pi}{ \mathcal{K}})}\biggr),
    \
    d=\mathcal T\biggl(\frac{\vartheta'_{11}(\frac{\eta}{4\mathcal{K}}| \frac{ \i\pi}{ \mathcal{K}})}{\vartheta_{11}(\frac{\eta}{4\mathcal{K}}| \frac{ \i\pi}{ \mathcal{K}})}\biggr),
    \end{equation}
    where $\mathcal{K}=\mathcal{K}(x)$ is given in Definition~\ref{def:K(x)}, the elliptic theta functions are given in~\eqref{eq:theta}, and $\mathcal{T}$ is the affine transformation
    \begin{equation}
    \mathcal{T}(z)=\frac{\mathcal{U}(\mathcal{K})}{\mathcal{K}}\biggl(z - \frac{\eta}2 - \frac{\vartheta'_{11}(\frac{\eta}{2\mathcal{K}}| \frac{ \i\pi}{ \mathcal{K}})}{\vartheta_{11}(\frac{\eta}{2\mathcal{K}}|\frac{ \i\pi}{ \mathcal{K}})} \biggr).
    \end{equation}
    
    \item If $x\ge 2$, $\rho_{\eta,x}(\mu)=\rho_{\mathrm{VKLS}}(\mu)$.
    \end{enumerate}
\end{theorem}

\begin{figure}
    \centering
    \includegraphics[width=0.5\linewidth]{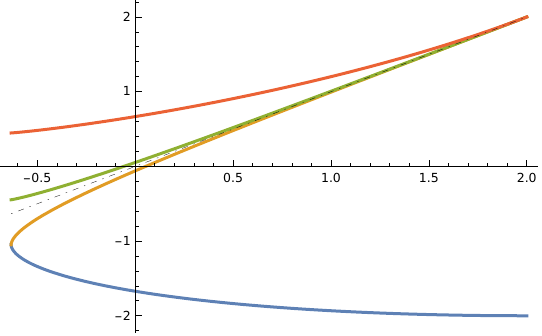}
    \caption{A plot of the endpoints $a$, $b$, $c$, and $d$ (in blue, orange, green, and red, respectively) as functions of $x\in(x_*,2)$. The dot-dashed thin curve is the graph of $x$ and lies between~$b$ and~$c$.
    We take $\eta=\log 20$, such that $x_*=-\frac{19}{10\log 20}=-0.634236\ldots$}
    \label{fig:endpoints}
\end{figure}

The proof is given in Section~\ref{sec:variational}.
In~\Cref{fig:equilibrium measures} we plot~$\rho_{\eta,x}$ for some values of~$x$ and~$\eta$. We observe two phase transitions modulated by the parameter~$x$. When~$x > 2$, the optimal profile is given by the arccosine law $\rho_{\mathrm{VKLS}}$; see \Cref{fig:equilibrium measures}, bottom-right panel. When $x\in(x_*,2)$, with explicit critical value $x_*$ as in \eqref{eq:xq}, the optimal density $\rho_{\eta,x}$ is characterized by a nontrivial saturated region (namely, a finite interval where $\rho_{\eta,x}$ is identically equal to~$1$) in a neighborhood of the macroscopic location $x$; see \Cref{fig:equilibrium measures}, top-right and bottom-left panels. When $x< x_*$, the two saturated regions coalesce around the point $-1-\e^{-\eta}$ (see~\Cref{sec:propertiesendpoints}), which is macroscopically far away from $x_*$ and the optimizer $\rho_{\eta,x}$ becomes a rescaling of the arccosine law; see \Cref{fig:equilibrium measures}, top-left panel. We also give a plot of the endpoints $a,b,c,d$ as functions of $x \in (x_*,2)$ for a fixed $\eta$ in~\Cref{fig:endpoints}.

\subsubsection{Large-$t$ asymptotic expansion} \label{subs:intro asymptotic expansion}

The rich behavior of the optimizer $\rho_{\eta,x}$ is reflected in the large-$t$ asymptotics of~$Q(t,xt)$ which we now describe.

\begin{definition}[Rate function]\label{def:Fq} We define the function
$\mathcal{F}:\mathbb{R}\to\mathbb{R}_{\ge 0}$ by
\begin{equation}\label{eq:Fq}
    \mathcal{F}(x) = 
        \begin{dcases}
            \frac{\eta x^2}2 +1-\e^{-\eta} &\text{if } x \leq x_*,
            \\
            \frac{\eta x^2}2 +1-\e^{-\eta} +\eta\int_{x_*}^x\mathcal{L}(y)\d y&\text{if } x_*<x<2,
            \\
            0  &\text{if } x\geq 2,
        \end{dcases}
    \end{equation}
where, recalling $\mathcal{U},\mathcal{K}$ from \Cref{def:K(x)}, we set
\begin{equation}
    \label{eq:defL}
    \mathcal{L}(x)  \,=\, -\frac{\frac{\eta x}{2}+\mathcal{U}\bigl(\mathcal{K}(x)\bigr)}{\mathcal{K}(x)}\,=\, -\left.\frac{\partial \mathcal{U}(K)}{\partial K}\right|_{K=\mathcal{K}(x)}\,.
\end{equation}
\end{definition}

\begin{remark}
When $x\in(x_*,2)$, we have the equivalent expression
\begin{equation}
\label{eq:explicitF}
\mathcal{F}(x) = 1+\frac{\eta}{2}\biggl(1-\frac{\eta}{2\mathcal{K}}\biggr)x^2-\frac{3\eta\mathcal{U}(\mathcal{K})}{4\mathcal{K}}x+\mathcal{U}(\mathcal{K})^2\frac{\d^2}{\d\eta^2}\log\vartheta_{11}\bigl(\frac{\eta}{2\mathcal{K}}\big|\frac{\i\pi}{\mathcal{K}}\bigr),
\end{equation}
where $\mathcal{K}=\mathcal{K}(x)$ is given in Definition~\ref{def:K(x)}; see Remark~\ref{remark:explicitFWeierstrass}.
\end{remark}

Plots of $\mathcal{F}$, $\mathcal{K}$, and $\mathcal{L}$ can be found in Figure~\ref{fig:rate function F}. 
We are now ready to state the second main result of this paper.

\begin{theorem}[Asymptotic expansion]
\label{thm:main}
Fix $\eta>0$.
For all $x<2$ there exists $\mathcal C(x)\in\mathbb{R}$ such that, when $t\to+\infty$,
\begin{equation}\label{eq:main}
Q(t,xt) =
\begin{cases}
\mathcal C(x) \,\exp\bigl(-t^2\mathcal{F}(x)\bigr)\left(1+O(t^{-1})\right)
&\mbox{if }x<x_*,
\\[1em]
\mathcal C(x) \,\vartheta\bigl(t \mathcal{L}(x) \,\big|\, \frac{\i\pi }{\mathcal{K}(x)}\bigr) \,
\exp\bigl(-t^2\mathcal{F}(x) +\mathcal{A}\log t \bigr)
\left(1+O(t^{-1/2})\right)
&\mbox{if }x_*<x<2,
\end{cases}
\end{equation}
where $\mathcal{K}$ is given in Definition~\ref{def:K(x)}, $\mathcal{F}, \mathcal{L}$ are given in Definition~\ref{def:Fq}, and  
$\mathcal{A}$ is a constant depending only on $\eta$, cf.~\eqref{eq:Afinal}.
These asymptotics are uniform for $x\in(-\infty,x_*-\delta]\cup[x_*+\delta,2-\delta]$ for all $\delta>0$.
\end{theorem}

\begin{figure}
    \centering
    \frame{\includegraphics[width=0.45\linewidth]{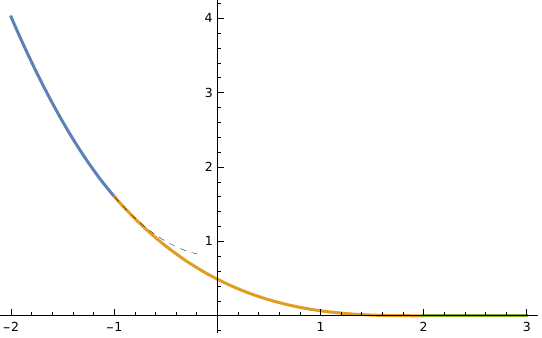}}\\[1em]
    \frame{\includegraphics[width=0.4\linewidth]{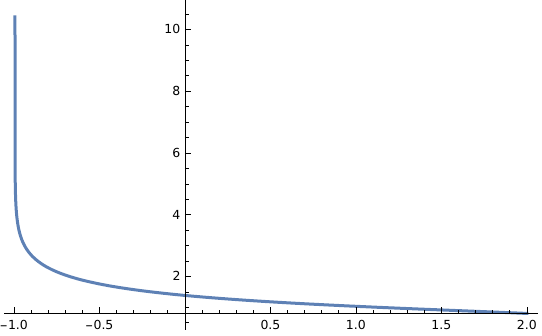}}
    \quad
    \frame{\includegraphics[width=0.4\linewidth]{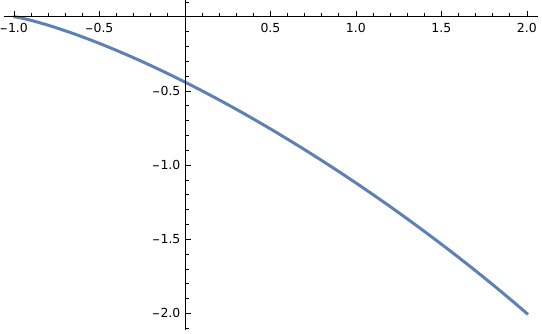}}
    \caption{
    Top: plot of the rate function $\mathcal{F}(x)$ (the different colors correspond to $x\leq x_*$, $x_*<x<2$, and $x\geq 2$; the thin black dashed part is the analytic continuation of the parabola defining $\mathcal{F}(x)$ for $x\leq x_*$).
    Bottom: plots of the functions $\mathcal{K}(x)$ and $\mathcal{L}(x)$, respectively, for $x\in(x_*,2)$.
    We take $\eta=\log 5$, such that $x_* = -\frac{8}{5\, \log 5} = -0.994136\ldots$}
    \label{fig:rate function F}
\end{figure}

The region $x>2$, which is not described in this theorem, is covered by the results contained in \cite{das2025large}; see Remark \ref{rem:dlm} below.
The proof of Theorem~\ref{thm:main} is given in Section~\ref{sec:proof of main thm} and is based on the Riemann--Hilbert asymptotic analyses carried out in Sections~\ref{sec:DeiftZhou1cut} and~\ref{sec:DeiftZhou2cut}.
This theorem gives all divergent terms of the asymptotic expansion of $\log Q(t,xt)$, as well as an explicit periodic contribution in the $O(1)$~term; see Remarks~\ref{rem:constant order} and~\ref{rem:uniformly} for comments on the constant-order term.
The coefficient~$\mathcal{A}$ of the logarithmic term is given by a rather involved expression, namely, as the average of a periodic function expressed as a rational combination of elliptic theta functions and their first derivatives.
It is not clear whether this expression can be simplified nor whether the logarithmic term possesses a physical explanation; it would be interesting to clarify the nature of this logarithmic contribution.
Nevertheless, the explicit expression which we provide is still suitable for practical computations; see Figure~\ref{fig:A} for a plot of $\mathcal{A}$ as a function of $\eta$ and Remark~\ref{rem:A}.

The rate function~$\mathcal{F}(x)$ undergoes two phase transitions of the third order, as stated next.

\begin{theorem}[Phase transitions] \label{thm:phase transition}
    For all $\eta>0$, $\mathcal{F}\in C^2(\mathbb{R}) \setminus C^3(\mathbb{R})$ and the following properties hold.
    \begin{enumerate}[leftmargin=*]
        \item \label{item:TW phase transition} As $x\uparrow 2$, we have $\mathcal{F}(x) = \frac{1}{12} (2-x)^3 +O(2-x)^4$.
        \item \label{item:BOAC phase transition} The function $\mathcal{F}''$ is not H{\"older} continuous at $x_*$ for any H{\"older} exponent.
    \end{enumerate}
 \end{theorem}

The proof is in Section~\ref{sec:propertiesrateproof}.
At~$x=2$, $\mathcal{F}(x)$ exhibits a third-order phase transition of Tracy--Widom type, related to the asymptotic fluctuation result $Q(t,2t-st^{1/3}) \to F_{\mathrm{GUE}}(-s)$, $F_{\mathrm{GUE}}$ being the cumulative distribution function of the GUE Tracy--Widom distribution; see~\cite{aggarwal_borodin_wheeler_tPNG}.
At $x=x_*$, $\mathcal{F}(x)$ exhibits a third-order phase transition due to the coalescence of the two saturated regions discussed above and which shows similarities with the ``birth of a cut'' in random matrix Theory~\cite{eynard2006universal,claeys2008birth,mo2008riemann,bertola2009mesoscopic}.

\medskip

Several comments are in order.

\begin{remark}\label{rem:dlm}
    The asymptotic relation as $t\to+\infty$
    \begin{equation}
    \label{eq:DLM}
        Q(t,xt) \sim 1 - \exp\bigl(-t\, \Psi_+(x)\bigr), \qquad \text{for }x>2,
    \end{equation}
    where $\Psi_+$ is an explicit (strictly positive) function, has been established in~\cite{das2025large}, see~Equation~(1.16) there.
    This result covers the region $x>2$, which is left out in Theorem~\ref{thm:main} and implies that $\lim_{t\to+\infty} t^{-2}\log Q(t,xt) =- \mathcal{F}(x)$ for all~$x\in\mathbb{R}$.
    Indeed, in \cite{das2025large}, the authors already proved that the rate function $\mathcal{F}(x)$ exists, that it is $C^1(\mathbb{R})$, that it vanishes identically when~$x \ge 2$, and that it equals $\frac 12\eta x^2+1-\e^{-\eta}$ when~$x$ is sufficiently large and negative (without identifying the threshold value~$x_*$).
    Furthermore, building upon a result of~\cite{CafassoRuzza23} (see also~\cite{matetski2025polynuclear}), they speculated that~$\mathcal{F}$ solves a second-order nonlinear ODE obtained as a scaling limit of the cylindrical Toda equation (see Equation~5.3 in \emph{op. cit.}) and gave an explicit formula for the solution to this ODE that was conjecturally the relevant one, see~\cite[Section 5.1]{das2025large}.
    That formula does not match the expression for~$\mathcal{F}$ given in~\eqref{eq:Fq} and therefore it is incorrect.
    Moreover, it can be checked that the rate function~$\mathcal{F}$ in~\eqref{eq:Fq} does not satisfy the Equation~5.3 in \emph{op.~cit.} on~$(x_*,2)$.
\end{remark}

\begin{remark} \label{rem:constant order}
    The expression of the constant order term $\mathcal C(x) \,\vartheta\bigl(t \mathcal{L}(x) \big| \frac{\i\pi }{\mathcal{K}(x)}\bigr)$ in the expansion \eqref{eq:main}, for $x_*<x<2$, deserves a brief explanation. From the point of view of the asymptotic expansion this form might appear unconventional, since the term $\mathcal{C}(x)$ is not given explicitly. Nevertheless, such decomposition is motivated by the fact that the oscillatory term $\vartheta\bigl(t \mathcal{L}(x) \,\big|\, \frac{\i\pi }{\mathcal{K}(x)}\bigr)$ fully captures the constant-order contribution to the asymptotic expansion of the discrete logarithmic derivative $\log \frac{Q(t,s)}{Q(t,s-1)}$, in the transition regime $x_* t < s < 2t$. This is natural (and relevant) in view of the fact that $\log \frac{Q(t,s)}{Q(t,s-1)}$ solves the cylindrical Toda equation (a reduction of the 2D Toda equation)~\cite{CafassoRuzza23,matetski2025polynuclear}; see Section~\ref{subs:asymptotic Toda}.
\end{remark}

\begin{remark} \label{rem:uniformly}
The asymptotic relations in~\eqref{eq:main} hold uniformly for~$x$ at a bounded distance from~$x_*$ and~$2$. The uniform control of the expansion of $\log Q(t,xt)$ for $x \in \mathbb{R}$ would require a finer Riemann--Hilbert analysis at the transition points $x=x_*$ and $x=2$.
This analysis will likely be instrumental in addressing the \emph{constant problem}, namely, determining the constant~$\mathcal C(x)$ appearing in Theorem~\ref{thm:main}. 
Indeed, this is a recurring delicate problem in the asymptotic analysis of Fredholm determinants (and, more generally, of tau functions of integrable systems) \cite{tracy1991asymptotics,ehrhardt2006dyson,ehrhardt2010asymptotics,Deift_Its_Krasovsky_Airy,deift2007widom,bothner2018large,basor1991fisher,baik2008asymptotics,krasovsky2009large,Lisovyy_Dyson,BleherBothnerConstant,bothner2019analysis} 
and we therefore leave the determination of this constant to future work.
\end{remark}

\begin{remark}
\label{rem:q0}
    In the limit $\eta\to+\infty$ we have $x_*\uparrow 0$ and the rate function $\mathcal{F}(x)$ converges, for $x>0$, to the lower-tail large deviation rate function of the last passage percolation through a Poisson planar environment derived in \cite{seppalainen_98_increasing}; see~\Cref{subs: limit q0} for details.
\end{remark}

\begin{remark}
As we mentioned above, the point~$x=2$ is a third-order phase transition of Tracy--Widom type.
The third-order phase transition at~$x_*$ is of a different nature.
Namely, it is induced by the discrete character of the process which forces the formation of nontrivial, macroscopically spaced, saturated regions.
To the best of our knowledge, this is the first time such a discreteness-induced birth of a cut is studied.
The closest analog appears to be the phase transition of the symmetric six-vertex model with domain wall boundary conditions between the disordered and anti-ferroelectric phases first studied in detail by P.~Zinn-Justin~\cite{zinn2000sixvertex}.
Nevertheless, in that case, the parameter governing the phase transition of the free energy is analogous to our~$\eta$ in the regime when~$\eta\to+\infty$ and the phase transition is of infinite order. A double third-order phase transition also occurs in the analysis of the partition function of the five-vertex model with certain boundary condition as found in the recent articles \cite{burenev2024thermodynamics,colomo2025five}.
Another related discreteness-induced phase transition is the Douglas--Kazakov one~\cite{douglas1993large,forrester2011non} occurring in the context of Euclidean $U(N)$~Yang--Mills theory on the two-dimensional sphere (see also \cite{LevyMaida,LiechtyWang2016}).
Nevertheless, it is more analogous to the phase transition of the rate function $\mathcal{F}(x)$ at $x=2$, the discontinuity in the third derivative of the Yang--Mills free energy being linked to Tracy--Widom statistics in a related matrix model~\cite{forrester2011non}.    
\end{remark}

\subsection{Methods}

We achieve the asymptotic result contained in Theorem~\ref{thm:main} starting from a \emph{discrete} Riemann--Hilbert characterization of~$Q(t,s)$ developed in~\cite{CafassoRuzza23} following the work of A.~Borodin on discrete integrable kernels~\cite{borodin2000riemann,borodin2003discrete} (see~\Cref{sec:notationsDeltaNabla}).
After a dressing procedure (which relies on a novel construction involving Bessel and Hankel functions, whose analytic properties require special care) we reformulate the characterization of~$Q(t,s)$ in terms of \emph{continuous} Riemann--Hilbert problems, such that we are able to apply the Deift--Zhou nonlinear steepest descent method~\cite{deift1993steepest}.
The main ingredient of the method, the so-called $g$-function, is closely related to the equilibrium measure minimizing the logarithmic energy~$\mathcal{E}_{\eta,x}$, and so our Riemann--Hilbert approach builds upon the study of this variational problem (which we carry out in Section~\ref{sec:variational}). Depending on the value of $x$, the equilibrium problem exhibits two qualitatively different regimes. Accordingly, we perform separate steepest descent analyses for the cases $x < x_*$ and $x_* < x < 2$, corresponding to one-cut and two-cut equilibrium measures respectively (see Sections \ref{sec:DeiftZhou1cut} and \ref{sec:DeiftZhou2cut}).  

This approach is inspired by the asymptotic analysis of \emph{discrete} orthogonal polynomials in the regime of large degree; see \cite{BKMM2003,bleher2011uniform}.
In this setting, the discreteness of the orthogonality measure gives rise to \emph{saturated} regions in the associated equilibrium measure, which describes the limiting density of roots of the orthogonal polynomials.
The presence of saturated regions leads to the appearance of hyperelliptic theta functions in the asymptotics, in a similar way as for \emph{continuous} orthogonal polynomials associated with non-convex potentials \cite{DeiftEtAl_multicut}.

A central technical challenge in our approach is the explicit determination of the endpoints of the equilibrium measure in the two-cut phase.
In general, these endpoints enter the Riemann-Hilbert steepest-descent analysis in a critical way and ultimately determine the form of the asymptotic expansion, yet they are rarely available in closed form in terms of the physical parameters of the model (which are $x$ and~$\eta$ in our case).
The endpoints are characterized only implicitly, though a system of coupled and typically transcendental equations arising from the Euler--Lagrange variational conditions for the equilibrium problem; see~\cite[Chapter~6]{DeiftBook}

Apart from a small number of exceptional cases in the literature (most notably the analysis by P.~Zinn-Justin of the six-vertex model with domain-wall boundary conditions in the anti-ferroelectric phase \cite{zinn2000sixvertex}, which inspired subsequent works by P.~Bleher and K.~Liechty \cite{bleher2010exact,bleher2013random} and by V.~Gorin and K.~Liechty \cite{gorin2025boundary}) such equations for the endpoints generally resist explicit solutions; see also \cite{deift1994collisionless,bothner2015asymptotic,girotti2021rigorous,girotti2023soliton}.
By constrast, in the present work we show that the equilibrium conditions can be brought into a substantially more tractable form through a sequence of nontrivial transformations. In particular, via an elliptic uniformization, the original system of transcendental equations is reduced to three linear relations together with a single remaining transcendental constraint; see \Cref{sec:minimizationxq<x<2}. This constraint implicitly defines the half-period $\mathcal{K}(x)$ of the associated elliptic curve (see \Cref{def:K(x)}) which enters our formulas as a crucial parameter. We stress that, unlike in the analysis of the six-vertex model by P.~Zinn-Justin, where remarkable cancellations lead to a purely linear system (see also~\cite[Section~2]{bleher2010exact}), no such simplifications occur in our situation.
The presence of the implicit function $\mathcal{K}$ leads to further involved calculations, such as those of \Cref{app:monotonic}. 

Our Riemann--Hilbert asymptotic analysis also enables us to compute subleading contributions beyond the rate function.
In particular, we obtain an explicit formula for the divergent term of order $\log t$ in the expansion of $\log Q(t,xt)$. The presence of such a logarithmic term is not universal in model of this type and its origin is not fully clear to us. For example, in the anti-ferroelectric six-vertex model studied by P.~Bleher and K.~Liechty \cite{bleher2010exact} the free energy expansion lacks an $O(\log t)$ term, but it does arise in the six-vertex model’s disordered phase \cite{bleher2006exact}, or in the asymptotic expansion of lower tail probabilities in geometric last passage percolation recently analyzed by S.-S.~Byun, C.~Charlier, P.~Moreillon, and N.~Simm~\cite{byun2025precise}.

Strictly speaking, the Riemann--Hilbert analysis only provides precise asymptotics for (discrete) partial derivatives of $\log Q(t,s)$, namely, for
\begin{equation}
    \alpha(t,s)=-\frac 12\,t\,\partial_t\log Q(t,s) \qquad
    \text{and}
    \qquad 
    \beta(t,s)=\frac{Q(t,s-1)}{Q(t,s)}-1,
\end{equation}
reported in Propositions~\ref{prop:finalqminus} and~\ref{prop:finalqplus}. As a result, the expansion of~$\log Q(t,xt)$ is obtained by an integration of these asymptotics, which we perform in Section~\ref{sec:proof of main thm}. This poses an additional technical challenge as the integration of the discrete logarithmic derivative $\beta$ involves sums with highly oscillatory terms, whose cancellation properties must be understood to extract the correct asymptotic behavior. We resolve this problem using tools from analytic number theory; specifically, we apply van der Corput-type estimates for oscillatory sums to show that the remainder from these sums is sharply bounded by $O(t^{-1/2})$, leading to the error term in~\eqref{eq:main}.

\subsection{Open directions} \label{subs:open_directions}


The results of this paper open several natural directions for investigation, besides the constant problem already mentioned in \Cref{rem:uniformly}. We briefly discuss some of them here.

\medskip

First, let us point out that the analysis presented in this paper should apply to multiplicative averages $\mathbb{E}\left[\prod_{i\ge 1} \varsigma(p_i -s) \right]$ of more general discrete determinantal point processes $(p_i)$ with log-gas structure.
Natural and interesting first candidates would be other instances of the Schur measure~\cite{okounkov2001infinite} such as the Meixner ensembles.

In \Cref{sec:applications} we describe three applications of our main asymptotic result. The first concerns the explicit characterization of the lower-tail large deviation rate function of the $q$-deformed polynuclear growth model, a stochastic growth process interpolating between the polynuclear growth and the KPZ equation. While Theorem~\ref{thm:lower-tail} describes the leading-order asymptotics, a natural open problem is to extend the description to the full tail distribution uniformly in $x,t$. In particular, in the limit $x\uparrow 2$ one should observe a transition from a large deviation regime to moderate deviation and finally to fluctuation regime, known to be governed by the Tracy--Widom distribution.

In view of the fact that multiplicative averages of the Meixner ensembles are known to describe the integrated current of the stochastic six-vertex model and, in an appropriate scaling, of the asymmetric simple exclusion process \cite{BO2016_ASEP,borodin2016stochastic_MM}, extending our asymptotic analysis to such special instances of the Schur measure would provide a complete and explicit description of the tails of these stochastic particle systems, a problem which has remained elusive despite recent progresses in the area; for example, see~\cite{tracy2009asymptotics,aggarwal2023asep,landon2023tail,aggarwal2024scaling,dlm24,Ghosal_Silva_6VM}. This approach was employed in the recent work \cite{Ghosal_Silva_6VM}, where authors  presented a Riemann--Hilbert analysis of the Meixner ensemble in the one-cut regime to describe moderate deviations (and not large deviations) of the stochastic six-vertex model.

The second application concerns the positive-temperature discrete Bessel process, which may be interpreted as a system of discrete free-fermions in a linear potential at positive temperature. In this case the rate function $\mathcal{F}$ itself serves as the lower tail large deviation rate function for the rightmost occupied state. At $x_*$, the rate function develops a singularity corresponding to a third-order phase transition, as identified in \Cref{thm:phase transition}. In \Cref{rem:condensation} we observe that such phase transition is evidence of a ``condensation'' phenomenon of the positive temperature discrete Bessel process (see \Cref{fig:condensation}) the nature of which is at this point unclear and warrants further studies.

As already pointed out at the end of \Cref{subs:intro asymptotic expansion}, the same phase transition corresponds, in the Poissonized Plancherel picture, to the appearance of a discreteness induced ``birth-of-a-cut'' phenomenon in the equilibrium measure around the point~$x_*$. As such the analysis of the ``condensation'' of the positive temperature discrete Bessel process should relate to the analysis of microscopic fluctuations of the point process $\mathscr{D}(\lambda)$ in the band region around the macroscopic location $-1-\e^{-\eta}$, in the critical regime $x\downarrow x_*$ when the band vanishes. We hope to clarify these aspects in future works.

As shown in~\cite{CafassoRuzza23}, $Q(t,s)$ solves the (bilinear form of the) cylindrical Toda equation, which is the radial reduction of the 2D Toda equation.
Hence, our results yield the asymptotic behavior of a distinguished class of solutions with step-like initial conditions, as we discuss in Section~\ref{subs:asymptotic Toda}.
This type of solutions has been the subject of an intense line of research in integrable systems (starting from the seminal work~\cite{gurevich1973decay}) and it was studied in great depth for many integrable equations such as the Korteweg--de~Vries (KdV) equation, the modified KdV equation, the nonlinear Schr\"odinger equation, the Toda equation, and others; see Section~\ref{subs:asymptotic Toda} for a discussion.
Extending our results to general step-like initial conditions, as well as studying the behavior of these solutions around the two transition points (such as $x_*$ and $2$ for our class of solutions) are interesting open problems that we will address in the future.

\subsection*{Acknowledgments}

MC acknowledges the support of the Centre Henri Lebesgue, program ANR-11-LABX0020-0, and the International Research Project PIICQ, funded by CNRS Math\'ematiques. 
MC is grateful to the Group of Mathematical Physics (Grupo de Física Matemática, FCT – Portuguese national funding, UID/00208/2025) for supporting a one-week visit to Instituto Superior T\'ecnico in Lisbon. MC thanks Manuela Girotti for fruitful discussions.

MM is grateful to Alexei~Borodin, Percy~Deift, Vadim~Gorin, and Ken~McLaughlin for useful discussions during the Simons Symposium ``Solvable Lattice Models \& Interacting Particle Systems~(2025)'' and to Sung-Soo~Byun for useful comments.   

The work of GR is funded by FCT - Fundação para a Ciência e a Tecnologia, I.P., through national funds, under the project UID/04561/2025.
GR is grateful to Giordano Cotti, Gabriele Degano, and Davide Masoero for useful discussions.
GR acknowledges the support provided by the CNRS Mathématiques through a \emph{poste rouge} for a three-month stay at the UMR LAREMA (Angers) and by the National University of Singapore for a three-week visit, during which a large part of this work was carried out.

We are grateful to Tamara Grava for insightful discussions on solutions of integrable equations with step-like initial conditions.

\section{Applications} \label{sec:applications}

The study of the multiplicative average~\eqref{eq:defQ} is especially interesting due to the fact that the function~$Q(t,s)$ describes, thanks to relations descending from symmetric functions~\cite{borodin2016stochastic_MM,aggarwal_borodin_wheeler_tPNG,IMS_matching}, at the same time the law of the height function of the $q$-deformed polynuclear growth model and the edge distribution of the positive-temperature discrete Bessel process, which is a model of free fermions at positive temperature.
In this section we apply~\Cref{thm:main} to these situations.

\subsection{$q$-deformed polynuclear growth}

The $q$-deformed polynuclear growth ($q$-PNG) model, where $q$ is a parameter in $(0,1)$, is a stochastic growth process introduced by A.~Aggarwal, A.~Borodin, and M.~Wheeler~\cite{aggarwal_borodin_wheeler_tPNG} as a solvable deformation of the famous polynuclear growth (PNG) model (recovered when $q=0$); see~\cite{Praehofer2002,imamura2005polynuclear,johansson_rahman_multitime,johansson_rahman_inhomogeneous,matetski2025polynuclear} and references therein. 

The $q$-PNG model is a stochastic evolution of height profiles $h(x,t)$ over time $t\in\mathbb{R}_{>0}$.
The height profiles are integer-valued piecewise constant functions of~$x\in\mathbb{R}$ with unit jump discontinuities; we also allow the height to take value $-\infty$.
We denote the set of such functions by
\begin{equation}\label{eq:piecewise constants}
    \mathsf{PW}(\mathbb{R}) \,=\, \bigl\lbrace h:\mathbb{R} \to \mathbb{Z} \cup \{-\infty\} \,  : \, h(x) - h(x_\pm)  \in \lbrace 0, 1, + \infty \rbrace \, \forall x \in \mathbb{R}\bigr\rbrace\,,
\end{equation}
where $f(x_\pm) = \lim_{\epsilon \downarrow 0} f(x \pm \epsilon)$ and we agree that~$-\infty -(- \infty)=0$.
We can view such height profiles as collections of islands of unit thickness stacked on top of each other.
As time $t$ increases, islands expand laterally with constant speed which we assume to be equal to one without loss of generality.
Two islands merge upon collision and, with probability~$q$, a new island of infinitesimal width is created on top of the collision point.
Creation of a new island of infinitesimal width is referred to as ``nucleation''.
Nucleations also occur at random space-time points $(x,t)$ sampled according to a space-time Poisson point process with intensity~$\Lambda>0$, in addition to the collision mechanism just explained.
A snapshot the $q$-PNG model dynamics is shown in~\Cref{fig:qPNG}.
When $q=0$, colliding islands merge without triggering nucleations and we recover the PNG model.

\begin{figure}
    \centering
    \includegraphics[width=0.7\linewidth]{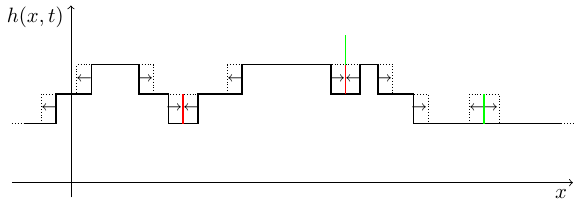}
    \caption{A depiction of the $q$-PNG dynamics.
    Red segments represent the collision interface between two islands.
    Green segments correspond to nucleations of new islands of infinitesimal width.
    }
    \label{fig:qPNG}
\end{figure}

In this paper, we will consider the special case of \emph{droplet initial conditions} for the $q$-PNG model.
In this case, at time~$t=0$ the height function takes the degenerate value
\begin{equation}
    h(x,t=0) \,=\, 
    \begin{cases}
        0 \qquad &\text{if } x=0,
        \\
        -\infty &\text{if } x \neq 0,
    \end{cases}
\end{equation}
and, at any time $t>0$, the height function~$h(x,t)$ takes finite values only in the forward light-cone of the origin~$|x| \leq t$.
In the original work~\cite{aggarwal_borodin_wheeler_tPNG}, the following central limit theorem was proven:
\begin{equation}\label{eq:qPNG convergence}
    \frac{h(x,t) - v_{\Lambda,q} \left(t^2-x^2 \right)^{1/2}}{\sigma_{\Lambda,q} \left(t^2-x^2 \right)^{1/6}} \xlongrightarrow[t \to \infty]{ \quad \text{in distribution} \quad } \chi_{\mathrm{GUE}},
\end{equation}
where $v_{\Lambda,q} = \Lambda/(1-q)$, $\sigma_{\Lambda,q} = \bigl( \Lambda/(2-2q) \bigr)^{1/3}$ and $\chi_{\mathrm{GUE}}$ is a random variable obeying the GUE Tracy--Widom distribution~\cite{tracy1993level}.
In~\cite{drillick2023strong}, H.~Drillick and Y.~Lin proved that the law of large number for the height function $h$ with droplet initial conditions holds in the strong sense.
The central limit theorem~\eqref{eq:qPNG convergence} relies on the aforementioned relationship between the Poissonized Plancherel measure and the height function $h$.
Indeed, according to~\cite{aggarwal_borodin_wheeler_tPNG,IMS_matching,das2025large}, we have
\begin{equation}\label{eq:relationqPNG multiplicative average}
    \mathbb{P} \left[ h(0,t) + \chi + S \le s \right] = \mathbb{E}_{t^2}\left[\prod_{i \ge 1} \frac{1}{1+q^{s+i-\lambda_i}} \right],
\end{equation}
where $\chi$ and $S$ are independent random variables with laws
\begin{equation}
\begin{aligned}
    \mathbb{P}\left[\chi = k\right] &= q^k \prod_{i\ge 1}\bigl(1-q^{k+i}\bigr),
    &&\text{for }  k \in \lbrace 0,1,2, \dots \rbrace,
\\
    \mathbb{P}\left[S = k\right] &= \frac{q^{k^2/2}}{\sum_{n\in\mathbb{Z}}q^{n^2/2}},
    &&\text{for }k \in \mathbb{Z}.
\end{aligned}
\end{equation}

The following large deviation principles for the one point distribution of the height function $h(x,t)$ have been established in~\cite{das2025large}: 
\begin{equation}
\begin{aligned}
    \lim_{t\to \infty} t^{-1}\, \log \mathbb{P}\left[ h(0,t) \ge t \mu \right] &\,=\, - \Phi_+(\mu),
\\
    \lim_{t\to \infty} t^{-2}\,\log \mathbb{P}\left[ h(0,t) \le t \mu \right] &\,=\, - \Phi_-(\mu).
\end{aligned}
\end{equation}
The upper-tail large deviation rate function~$\Phi_+$ was derived explicitly from the relation~\eqref{eq:relationqPNG multiplicative average}.
The derivation of the lower-tail rate function~$\Phi_-$ is more subtle because, in the lower-tail regime, in the left-hand side of~\eqref{eq:relationqPNG multiplicative average} there is a nontrivial competition between the tails of the random variable~$h$ and of the discrete Gaussian~$S$.
As a result the rate function~$\Phi_-$ ends up being related to the scaling limit of the right-hand side of~\eqref{eq:relationqPNG multiplicative average} through an infimal deconvolution operation, which itself requires establishing \emph{a~priori} convexity property of~$\Phi_-$, see \cite{das2025large,dlm24}.
The next theorem completes the explicit description of the rate function~$\Phi_-$.

\begin{theorem}[Lower-tail large deviation rate function] \label{thm:lower-tail}
    Fix $q\in (0,1)$ and let $\eta=-\log q$.
    Let $h(x,t)$ be the height function of the $q$-PNG with intensity $\Lambda=2(1-q)$ and droplet initial conditions.
    Then, for $\mu \in [0,2]$ we have
    \begin{equation}
        \lim_{t\to +\infty} t^{-2} \, \log \mathbb{P} \left[ h(0,t) \le \mu t \right] = -\Phi_-(\mu),
    \end{equation}
    where
    \begin{equation} \label{eq:Phi_minus}
        \Phi_-(\mu) = \max_{y \in \mathbb{R}} \left\{ \mathcal{F} (y) - \frac{\eta}{2} (\mu-y)^2 \right\}.
    \end{equation} 
\end{theorem}
\begin{proof}
    This is a straightforward corollary of~\cite[Theorem 1.3]{das2025large}, in which the same expression \eqref{eq:Phi_minus} was given for the rate function $\Phi_-$, but without explicit description of the function~$\mathcal{F}$.
\end{proof}

Plots of the rate function~$\Phi_-$ (for some values of~$q$) are given in~\Cref{fig:rate functions Phi}. 

\begin{figure}
    \centering
    \includegraphics[width=0.5\linewidth]{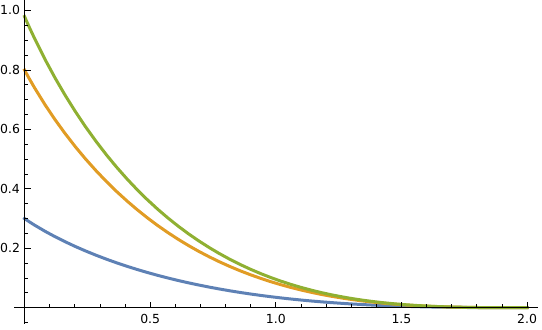}
    \caption{A plot of the rate functions $\Phi_-(x)$ for values of $q=7/10$ (\textcolor{Cerulean}{blue}), $q=1/5$ (\textcolor{brown}{orange}) and $q=1/50$ (\textcolor{OliveGreen}{green}).}
    \label{fig:rate functions Phi}
\end{figure}

\subsection{Positive-temperature discrete Bessel process} \label{subs:positive-temp discrete Bessel}

The positive-temperature discrete Bessel process is a determinantal point process $\mathsf{B} = (\mathsf{b}_n)_{n \in \mathbb{N}} \subset \mathbb{Z}'$ (with $\mathsf{b}_{n}>\mathsf{b}_{n+1}$ for all $n\in\mathbb{N}$), which depends on parameters~$\eta>0$ and~$t>0$, characterized by the correlation kernel
\begin{equation}
\label{eq:PosTempdBKernel}
\mathsf{K}_\eta(i,j;t) = \sum_{\ell \in \mathbb{Z}} \frac{ \mathrm{J}_{i+\ell-\frac 12}(2t) \mathrm{J}_{j+\ell-\frac 12}(2t) }{1+\e^{-\ell\eta}} \qquad (i,j\in\mathbb{Z}').
\end{equation}
It arises in several related contexts, which include
\begin{itemize}[leftmargin=*]
\item the grand canonical ensemble of \emph{free fermions} in one dimension with one-particle Hamiltonian
\begin{equation}
    \mathscr{H} (\varphi)_j = - t \left( \varphi_{j+1} + \varphi_{j-1} \right) + j \varphi_j,
\end{equation}
where the parameter $\eta$ is the inverse temperature
(see~\cite{betea_bouttier_periodic,dean2015universal,dean2019noninteracting} for properties of ground states of more general systems of free fermions), and
    \item the \emph{periodic Schur measure} with Plancherel specialization~\cite{borodin2007periodic,betea_bouttier_periodic}.
\end{itemize}
A sampling procedure for the point process~$\mathsf{B}$ was described in~\cite[Section 2.3]{das2025large} through the periodic Robinson--Schensted correspondence~\cite{sagan1990robinson} and we now recall it.
Introducing
\begin{equation}
    \mathscr{C}_r = \mathbb{R}/r \mathbb{Z}
    \qquad \text{with} \qquad r=\sqrt{2} (1-\e^{-\eta}) t
\end{equation}
we define the set of periodic height functions
\begin{equation}
    \mathsf{PW}(\mathscr{C}_r) = \bigl\lbrace h: \mathscr{C}_r \to \mathbb{Z} \, : \, h(x) - h(x_\pm) \in \lbrace 0, 1 \rbrace\, \forall x \in \mathscr{C}_r \bigr\rbrace\,,
\end{equation}
where $f(x_\pm) = \lim_{\epsilon \downarrow 0} f(x \pm \epsilon)$.
This is a periodic variant of~\eqref{eq:piecewise constants}.
The \emph{multi-layer PNG} model on~$\mathscr{C}_r$ is a dynamics on a family of height functions $(h_{n}(\cdot,s))_{n \in \mathbb{N},\, s \in \mathbb{R}}$ such that
\begin{equation}
    h_{n}(x,s) > h_{n+1}(x,s) \quad \text{for all }n\in\mathbb{N},\, x\in \mathscr{C}_t,\,\text{and }s\in \mathbb{R}.
\end{equation}
The top height function $h_1(x,s)$ evolves analogously to a ($q=0$) PNG model.
Namely, island expand at unit speed and merge upon collision and nucleations occur at rates given by a (potentially inhomogeneous) Poisson point process $\mathcal{P}$ on $\mathscr{C}_r \times \mathbb{R}$.
For $n \ge 2$, the $n$-th layer height function $h_n$ also evolves similarly to a PNG model, with islands spreading laterally at speed 1 and nucleations triggered by merging of islands at layer $n-1$, i.e. if at time $t$ two islands merge at layer $n-1$ at location $x$, then at time $t$ a nucleation occurs at layer $n$ at location $x$.

To relate the positive-temperature discrete Bessel point process to the periodic multi-layer PNG model, we consider the following procedure~\cite{das2025large}
\begin{enumerate}
    \item Sample a sequence $(\kappa_j)_{j\in\mathbb{Z}}$ of independent Bernoulli random variables
    with law
    \begin{equation}
        \mathbb{P} \left[ \kappa_j =1 \right] = \frac{1}{1+\e^{j\eta}}.
    \end{equation}
    and initialize at $s=-\infty$ the functions $h_i$ to take the constant (random) values
    \begin{equation}
        h_n(x,-\infty) = \max\biggl\lbrace j \in \mathbb{Z} \,: \,\sum_{r \ge j} \eta_r = n \biggr\rbrace, \qquad \text{for all } n \in \mathbb{N}.  
    \end{equation}
    Namely $h_n(\cdot ,-\infty)$ is a constant function with value equal to the location of the $n$-th rightmost point of the sequence $(\kappa_j)_{j \in \mathbb{Z}}$; see \Cref{fig:cylindrical PNG}, central panel.
    \item Let $\mathfrak{P}$ be the Poisson point process on $\mathscr{C}_r \times \mathbb{R}$ with inhomogeneous intensity
    \begin{equation}
        \lambda(p) = \sum_{k \ge 0} \e^{-k\eta} \mathbf{1}_{\mathfrak{R}_k}(p),
    \end{equation}
    where the subsets $\mathfrak{R}_k$ are\footnote{Here $\| \cdot \|_1$ is the norm on $\mathscr{C}_r \times \mathbb{R}$ induced by the $\ell_1$ norm $\| (x_1,x_2) \|_1 = |x_1|+|x_2|$ on $\mathbb{R}^2$.}
    \begin{equation}\label{eq:regions of cylinder}
        \mathfrak{R}_k = \left\{ p \in \mathscr{C}_r \times \mathbb{R} : \| p-p_k \|_{1} \le \frac{r}{2} \right\},
        \qquad
        \text{where}
        \qquad
        p_k = \left( \frac{kr}{2},-\frac{(k+1)r}{2} \right),
    \end{equation}
    for $k \ge 0$, as depicted in \Cref{fig:cylindrical PNG}, left panel. 
    In other words, $\mathfrak{P}$ is the disjoint union of independent Poisson point processes supported in $\mathfrak{R}_k$ and with rate $\e^{-k\eta}$ for $k\ge 0$.
    \item For time $s > -\infty$ evolve the family of height functions $\bigl(h_n(x,s)\bigr)_{n\in\mathbb{N}}$ following a multi-layer PNG model with nucleation rates given by the point process~$\mathfrak{P}$.
\end{enumerate}
\begin{proposition}
    Let $\bigl(h_n(x,s)\bigr)_{n\in\mathbb{N},\,s\in\mathbb{R}}$ be the family of height functions sampled as above.
    Then, the point process
    \begin{equation}
    \frac{1}{2}+h_n(0,0) \qquad \text{for } n\in\mathbb{N},
    \end{equation}
    is equal in law to a positive-temperature discrete Bessel process.
\end{proposition}
\begin{proof}
    This is a combination of \cite[Proposition 2.7]{das2025large} and \cite[Proposition 2.9]{das2025large}.
\end{proof}

\begin{figure}
    \centering
    \includegraphics[width=\linewidth]{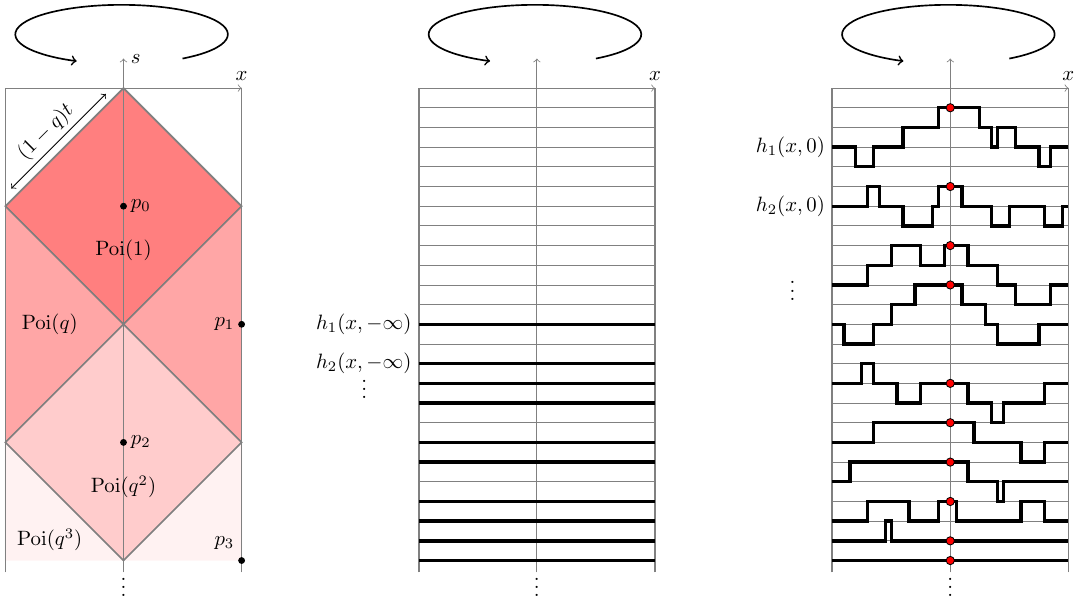}
    \caption{
    In the left panel, the space-time regions $\left( \mathfrak{R}_k \right)_{k\ge 0}$ described in \eqref{eq:regions of cylinder}.
    In the central panel, the initialization of the height functions $h_n$ at time $s=-\infty$.
    In the right panel a depiction of the multi-level cylindrical PNG model dynamics under the inhomogeneous nucleation rates pictured in the left panel.
    The red dots represent the values of the heights $\bigl(h_n(0,0)\bigr)_{n\in\mathbb{N}}$.}
    \label{fig:cylindrical PNG}
\end{figure}

The following theorem, which is an immediate consequence of~\Cref{thm:main}, describes the large deviations (in the parameter~$t$) of the marginal $\mathsf{b}_1$ of a positive-temperature discrete Bessel point process.

\begin{theorem} \label{thm:lower tail positive temp discrete bessel}
Let $\mathsf{B}=(\mathsf{b}_n)_{n\in\mathbb{N}}$ be the positive-temperature discrete Bessel process with parameters~$\eta>0$ and~$t>0$.
Then, we have
\begin{equation}\label{eq:ldp positive-temperature bessel}
    \mathbb{P} (\mathsf{b}_1 \le x t) = 
    \begin{cases}
        \mathcal C(x) \,\exp\bigl(-t^2\mathcal{F}(x)\bigr)\left(1+O(t^{-1})\right)
        &\mbox{if }x<x_*,
        \\[1em]
        \mathcal C(x) \,\vartheta\bigl(t \mathcal{L}(x) \,\big|\, \frac{\i\pi}{\mathcal{K}(x)}\bigr)\,\exp\bigl(-t^2\mathcal{F}(x)+\mathcal{A}\log t\bigr)\left(1+O(t^{-1/2})\right)
        &\mbox{if }x_*<x<2,
        \\[1em]
        1+O(t^{-\infty}) &\mbox{if } x>2,
    \end{cases}
\end{equation}
with $\mathcal{F}$, $\mathcal{K}$, $\mathcal{L}$ given in Definitions~\ref{def:K(x)} and~\ref{def:Fq}, and $\mathcal{A}$ given in~\eqref{eq:Afinal}.
\end{theorem}
\begin{proof}
    Let us denote by $\mathscr K(t)$ and $\mathscr K_{\eta}(t)$ the operators on $\ell^2(\mathbb{Z}')$ with kernels given, respectively, by $\mathsf K(i,j;t)$ and $\mathsf{K}_\eta(i,j;t)$, see~\eqref{eq:dBKernel} and~\eqref{eq:PosTempdBKernel}.
    For any $s\in\mathbb{Z}$ we have
    \begin{equation}
        \mathbb{P}( \mathsf{b}_1 \le s) 
        = \det_{\ell^2(\mathbb{Z'})}  \bigl( 1-\mathscr K_\eta(t) \bigr)
        = \det _{\ell^2(\mathbb{Z'})} \bigl(1 - \sqrt{\varsigma}(\cdot -s)\mathscr{K}(t) \sqrt{\varsigma}(\cdot -s)\bigr) = Q(t,s),
    \end{equation}
    where $\varsigma(n) = (1+\e^{-n\eta})^{-1}$.
    In this chain of equalities, the middle one follows from Sylvester's determinantal identity $\det (1-AB) = \det (1 + BA)$ while the first and last ones follow from the well-known expression of multiplicative expectations of a determinantal point process in terms of Fredholm determinants (see \cite{CafassoRuzza23,IMS_matching} for more details).
    The statement is then a corollary of Theorem~\ref{thm:main}.
\end{proof}

\begin{remark} \label{rem:prob empty positive temperature bessel}
    By the sampling argument described above it is immediate to verify that for any $x<0$, we have
    \begin{equation}
        \begin{split}
            &\mathbb{P}\left[ \mathsf{b}_i = \lfloor t x \rfloor -i + \frac{1}{2}~ \forall i \in \mathbb{N} \right]
            \ge \mathbb{P}\left[ \mathfrak{P} = \varnothing \text{ and } \kappa_i = \mathbf{1}_{(-\infty,\lfloor t x \rfloor)}\left(- i \right) \, \, \forall i \in \mathbb{Z} \right]
            \\
            &= \left(\prod_{ j \le \lfloor t x \rfloor } \frac{1}{1+\e^{j\eta}} \right) \left(\prod_{ j > \lfloor t x \rfloor  } \frac{1}{1+\e^{-j\eta}} \right) 
            = \exp \left(-t^2\left( (1-\e^{-\eta}) + \eta\frac{x^2}{2} \right) + O(t) \right).
        \end{split}
    \end{equation}
    Namely, the probability that a sample of the positive temperature discrete Bessel process $\mathsf{B}$ consists of all half integer points to the left of $\lfloor t x \rfloor+\frac{1}{2}$, when $t$ gets large is bounded from below by $\exp\left( - t^2 F(x) +O(t) \right)$, where $F(x) = (1-\e^{-\eta}) + \eta\frac{x^2}{2}$. Notice that $F(x) = \mathcal{F}(x)$ when $x<x_*$.
\end{remark}

\begin{remark} \label{rem:condensation}
    \Cref{thm:lower tail positive temp discrete bessel} is especially interesting in view of the phase transition of~$\mathcal{F}$ at the point $x_*$; see \Cref{thm:phase transition}.
    In light of \Cref{rem:prob empty positive temperature bessel}, this suggests that, under the condition that the rightmost point~$\mathsf{b}_1$ lies to the left of $t x_*$, the positive-temperature discrete Bessel point process will ``condensate'' and become strongly concentrated around the configuration $\mathsf{b}_i = \lfloor t x_* \rfloor - i +\frac 12$, for $i\ge 1$; see \Cref{fig:condensation}. 
    It would be interesting to make the above prediction mathematically precise.
\end{remark}

\begin{figure}
    \centering
    \includegraphics[width=0.8\linewidth]{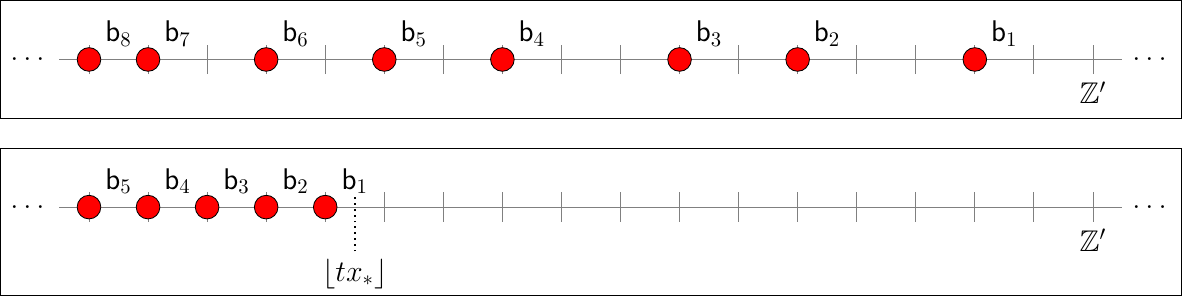}
    \caption{In the top panel, a depiction of the positive temperature discrete Bessel process. In the bottom panel, a depiction of the condensation phenomenon described in \Cref{rem:condensation}.}
    \label{fig:condensation}
\end{figure}

\subsection{Solutions to the cylindrical Toda equation with step-like initial conditions} \label{subs:asymptotic Toda}

The multiplicative averages~\eqref{eq:defQ} (and more general ones) were studied from the standpoint of integrable systems in~\cite{CafassoRuzza23}.
Indeed, Theorems~I and~II in \emph{op.~cit.} imply that
\begin{equation}
y(t,s) \,=\, \log \frac{Q(t,s)}{Q(t,s-1)}\qquad (t>0,\,s\in\mathbb{Z})
\end{equation}
satisfies the \emph{cylindrical Toda equation}
\begin{equation}
\label{cToda}
    \partial_t^2y(t,s)\,+\,t^{-1}\partial_t y(t,s) \,=\, 4\bigl(\e^{y(t,s+1) - y(t,s)} \,-\, \e^{y(t,s) - y(t,s-1)}\bigr),
\end{equation}
with initial conditions
\begin{equation}
\label{eq:cTodaIC}
    y(0,s) \,=\, \log\bigl(1 + \e^{-\eta s}\bigr),\qquad
\partial_ty(t,s)\big|_{t=0}=0.
\end{equation}
The cylindrical Toda equation~\eqref{cToda} corresponds to the radial reduction $Y(T,\overline{T},s)=y(T\overline{T},s)$ of the celebrated \textit{2D Toda equation}
\begin{equation}
    \partial_{T}\partial_{\overline{T}}Y(T,\overline{T},s) = \e^{Y(T,\overline{T},s+1) - Y(T,\overline{T},s)} \,-\, \e^{Y(T,\overline{T},s) - Y(T,\overline{T},s-1)}.
\end{equation}
It is illustrative to look at the asymptotic results of this paper from the point of view of cylindrical Toda dynamics.
Indeed, any linear configuration $y(t,s)=\nu s$ (for some~$\nu\in\mathbb{R}$) is an equilibrium of~\eqref{cToda} and so the initial condition~\eqref{eq:cTodaIC} is close to two different equilibria when $s$ is large positive or negative (with $\nu=0$ and $\nu=-\eta$, respectively).
It is reasonable to expect that, when $|s|$ is large compared to $t$, the cylindrical Toda solution~$y(t,s)$ remains close to the same equilibria while exhibiting a nontrivial transition regime interpolating between the two.
The asymptotic analysis of this paper makes this prediction explicit and proves it rigorously.
Namely, we show that
\begin{equation}
\label{eq:ylarget}
    y(t,xt) \,=\, t\,{y}_1(x) \,+\, y_0(t,x) \,+\, O(t^{-1}),\quad\mbox{ as }t\to+\infty,
\end{equation}
uniformly for $x$ away from $x_*$ and $2$, where (with the same notations as in Theorem~\ref{thm:main})
\begin{equation}\label{eq:ylargetbis}
\begin{aligned}
    {y}_1(x) &\,=\, 
    \begin{cases}
        -\eta x & \text{if }x \leq x_*, \\
        -\eta\bigl(x + \mathcal L(x)\bigr) & \text{if }x_* < x < 2,\\
        0 & \text{if }x\geq2,
    \end{cases}
\\
    {y}_0(t,x) &\,=\, 
    \begin{cases}
    \frac{\eta}{2} & \text{if }  x \leq x_*, \\
    \frac{\eta}2-\frac{\eta^2}{4\mathcal{K}(x)}  - \log \frac{\vartheta\bigl(t \mathcal{L}(x) + \frac{\eta}{2\mathcal{K}(x)} \,\big|\,\frac{\i \pi}{ \mathcal{K}(x)}\bigr)}{\vartheta \bigl(t \mathcal{L}(x)\,\big|\, \frac{\i \pi}{\mathcal{K}(x)} \bigr)} & \text{if } x_* < x < 2,\\
    0 & \text{if } x \geq 2.
    \end{cases}
\end{aligned}
\end{equation}
Noting that $\wh\beta(t,x)=-y(t,xt)$ by~\eqref{eq:wtalphabetagamma}, the asymptotic expansion~\eqref{eq:ylarget} follows from Propositions~\ref{prop:finalqminus} and~\ref{prop:finalqplus} for $x<x_*$ and $x_*<x<2$, respectively, and from~\eqref{eq:DLM} for $x>2$.
Plots of $y_1(x)$ and ${y}_0(t,x)$ are given in Figure~\ref{fig:TodaAsymp}.

The oscillating behavior of the solution, which appears in the subleading term $y_0$ of the asymptotic, is best seen
by introducing the following coordinates, which are a direct analog of the classical \emph{Flaschka variables} for the one-dimensional Toda lattice:
\begin{equation}
\begin{aligned}
\mathfrak a(t,s) &= \exp\biggl(\frac{y(t,s+1) - y(t,s)}{2}\biggr) = \frac{\sqrt{Q(t,s+1)Q(t,s-1)}}{Q(t,s)}, 
\\ 
\mathfrak b(t,s) &= \frac{1}2 \frac{\partial}{\partial t}y(t,s) = \frac{1}2\frac{\partial}{\partial t}\log\frac{Q(t,s)}{Q(t,s-1)}.
\end{aligned}
\end{equation}
Note that~$\mathfrak{a}(s,t)$ represents the relative displacement of neighboring particles in the lattice, rather than their position~$y(s,t)$.
The cylindrical Toda equation~\eqref{cToda} is equivalent to the system
\begin{equation}
\label{eq:TodaFlaschka}
\begin{aligned}
    \frac{\partial}{\partial t}\mathfrak{a}(t,s) &= \mathfrak{a}(t,s)\bigl(\mathfrak{b}(t,s+1) - \mathfrak{b}(t,s) \bigr),
    \\
    \frac{\partial}{\partial t}\mathfrak{b}(t,s) &= 2\bigl(\mathfrak{a}(t,s)^2 - \mathfrak{a}(t,s-1)^2\bigr) - \frac{\mathfrak{b}(t,s)}t,
\end{aligned}
\end{equation}
with initial conditions~\eqref{eq:cTodaIC} corresponding to
\begin{equation}
\label{eq:ICTodaFlaschka}
\mathfrak{a}(0,s) = \left(\frac{1 + \e^{-\eta(s+1)}}{1 + {\rm e}^{-\eta s}}\right)^{1/2}, \qquad \mathfrak{b}(0,s) = 0.
\end{equation}
Note that $\lim_{s \to +\infty} \mathfrak{a}(0,s) = 1$ and $\lim_{s \to -\infty} \mathfrak{a}(0,s) = \e^{-\eta/2}$.
Therefore, the initial condition~\eqref{eq:ICTodaFlaschka} belongs to the class of \emph{step-like initial conditions}, which have been extensively studied in the case of the one-dimensional Toda lattice, see the review article \cite{Mich16} and Remark~\ref{remshock} below.
Incidentally, let us recall that the spectrum of a constant-coefficient Jacobi operator (with entries equal to~$\mathsf{b}$ along the diagonal and equal to~$\mathsf{a}$ immediately above and below the diagonal) is $[\mathsf{b}-2\mathsf{a},\mathsf{b}+2\mathsf{a}]$, such that, in our case, the \emph{left} and \emph{right background spectra} (employing standard terminology for the Toda lattice, see for example \textit{op. cit.}) are $ [-2{\rm e}^{-\eta/2}, 2{\rm e}^{-\eta/2}]$ and $[-2, 2]$, respectively.
Hence, we are in the case of \emph{embedded background spectra}, which is a mixed case combining features of the Toda shock and rarefaction problems.

The asymptotic results of this paper imply that, as $t\to+\infty$ with $x=s/t=O(1)$, the variables $\mathfrak{a}$ and $\mathfrak{b}$ remain bounded, with the following asymptotic form:
\begin{equation}
\mathfrak{a}(t,xt) = \mathfrak{a}_0(t,x) + O(t^{-1}), \qquad \mathfrak{b}(t,xt) = \mathfrak{b}_0(t,x) + O(t^{-1}),
\end{equation}
uniformly for $x$ away from $x_*$ and $2$, where
\begin{equation}
\label{eq:asympab}
\begin{aligned}
    \mathfrak{a}_0(t,x) &= 
    \begin{cases}
    \exp(-\frac{\eta}2) & \text{if } x \leq x_*,
    \\
    \exp\bigl(-\frac{\eta}2 + \frac{\eta^2}{4\mathcal{K}(x)}\bigr)\,\frac{\sqrt{\vartheta\bigl(t\mathcal{L}(x) + \frac{\eta}{2\mathcal{K}(x)}\big|\frac{\i \pi}{\mathcal{K}(x)}\bigr)\vartheta\bigl(t\mathcal{L}(x) - \frac{\eta}{2\mathcal{K}(x)}\big|\frac{\i \pi}{\mathcal{K}(x)}\bigr)}}{\vartheta\bigl(t\mathcal{L}(x)\big| \frac{\i \pi}{\mathcal{K}(x)}\bigr)}  & \text{if } x_* < x < 2,\\
    1 & \text{if } x \geq 2,
    \end{cases}
\\
\mathfrak{b}_0(t,x) &= 
\begin{cases}
    0 & \text{if }x \leq x_*\text{ or } x \geq 2,
    \\
    \frac{\mathcal{U}\left(\mathcal{K}(x)\right)}{2\mathcal{K}(x)}
    \biggl(\eta +\frac{\vartheta'\bigl(t\mathcal{L}(x) + \frac{\eta}{2\mathcal{K}(x)}\big| \frac{\i \pi}{\mathcal{K}(x)}\bigr)}{\vartheta\bigl(t\mathcal{L}(x) + \frac{\eta}{2\mathcal{K}(x)}\big| \frac{\i \pi }{\mathcal{K}(x)}\bigr)} - \frac{\vartheta'\bigl(t\mathcal{L}(x)\big| \frac{\i \pi}{\mathcal{K}(x)}\bigr)}{\vartheta\bigl(t\mathcal{L}(x)\big| \frac{\i \pi }{\mathcal{K}(x)}\bigr)}\biggr) & \text{if }x_* < x < 2.
\end{cases}
\end{aligned}
\end{equation}

With the notation in~\eqref{eq:wtalphabetagamma}, we have
\begin{equation}
\mathfrak{a}(t,xt)= \exp\biggl(\frac 12\log\frac{\widehat{\beta}(t,x)}{\widehat{\beta}(t,x+\frac 1t)}\biggr),\quad
\mathfrak{b}(t,xt)= \widehat{\alpha}(t,x-\tfrac 1t)-\widehat{\alpha}(t,x) ,
\end{equation}
and so the claimed asymptotic relations follow from the asymptotic relations for $\widehat{\alpha}(t,x)$ and for $\widehat{\beta}(t,x)$ provided in Propositions~\ref{prop:finalqminus} and~\ref{prop:finalqplus} (for $x<x_*$ and $x_*<x<2$, respectively, in the latter case also using the identities of Lemma~\ref{lemma:finalidentities}) and from~\eqref{eq:DLM} for $x>2$.

\begin{figure}
    \centering
    \includegraphics[width=0.4\linewidth]{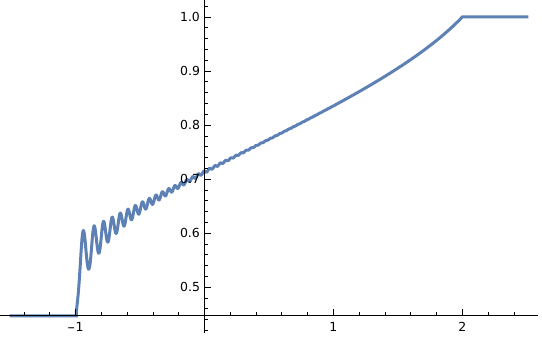}
    \qquad
    \includegraphics[width=0.4\linewidth]{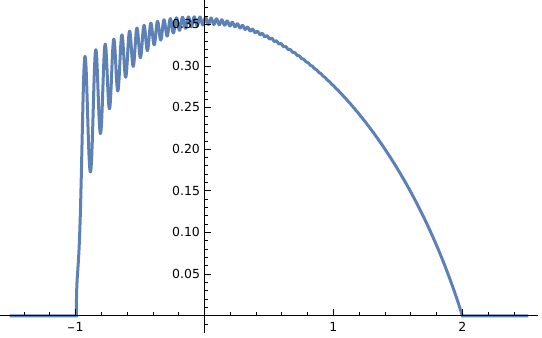}\\
    \caption{
    Plots of $\mathfrak{a}_0(t,x)$ and $\mathfrak{b}_0(t,x)$ (left and right, respectively) with $\eta=\log 5$ (such that $x_* = -\frac{8}{5\, \log 5} = -0.994136\ldots$) and $t=40$.
    }
    \label{fig:TodaAsymp}
\end{figure}

\begin{remark}
It is straightforward to check that $w(T,s) = y \bigl( \e^{T/2},s \bigr)-Ts$ is a solution to the well-known one-dimensional Toda lattice, namely,
\begin{equation}
    \partial_T^2 w(T,s) = \e^{w(T,s+1) - w(T,s)} - \e^{w(T,s) - w(T,s-1)}.
\end{equation}
Nevertheless, the initial conditions for $w(T,s)$ as $T\to-\infty$ are ill-defined, see~\eqref{eq:cTodaIC}.
\end{remark}

\begin{remark}\label{remshock}
Starting with the classical work of A.~Gurevich and L.~Pitaevskii \cite{gurevich1973decay} on the KdV equation and Whitham theory~\cite{Whitham}, step-like initial conditions have played a central role in nonlinear wave theory. 
Indeed, this line of research was subsequently implemented rigorously and extended to many more integrable models.
While not attempting to give an exhaustive account of the very vast literature, we refer to \cite{gurevich1973decay,gurevich1974nonstationary,Khruslov76,ablowitz1977asymptotic,deift1994collisionless,kamchatnov2021gurevich, EPTKdV} for the KdV equation, to
\cite{KhruslovKotlyarov89,KotlyarovMinakov2010,BertolaMinakov,GravaMinakov} for the modified KdV equation, to \cite{BuckinghamVenakides07,BoutetKotlyarovShepelski11} for the nonlinear Schr\"odinger equation, and to \cite{BlochKodama,venakides1991toda,kamvissis1993toda,deift1996toda,egorova2014long,egorova2020long} for the (one-dimensional) Toda equation.
The prominent feature of this theory is that solutions develop rapid oscillations (described by elliptic functions) in a transition regime connecting the limiting values of the step-like initial condition.

The results of this paper open the study of step-like initial conditions for the cylindrical Toda equation, which (to the best of our knowledge) have not been considered before. (Nevertheless, cylindrical Toda periodic solutions have been studied intensively, see for instance ~\cite{widomToda,tracywidomcyltoda1,tracywidomcyltoda2, GIK+25} and references therein).
Moreover, our approach (based on the probabilistic content of the solutions we consider in this paper) differs from the standard ones (based on tools such as Whitham modulation theory and spectral theory of Lax operators) employed in the vast integrable systems literature.
The interplay between the probabilistic approach and more traditional ones is an interesting topic that certainly deserves further investigations.
\end{remark}

\section{Minimization of logarithmic energy}\label{sec:variational}

\subsection{Poissonized Plancherel measure and logarithmic energy}

The relation between the (Poissonized) Plancherel measure and log-gases has been understood since the work of B.~Logan and L.~Shepp \cite{logan_shepp1977variational} and of A.~Vershik and S.~Kerov \cite{VershikKerov_LimShape1077}, who used it to determine the asymptotic shape which bears their names. In this subsection we elaborate on these rather established results to formulate the log-gas problem associated to the asymptotics of the multiplicative average~\eqref{eq:defQ}.

\begin{theorem} \label{thm:Laplace_argument}
    Fix $\eta>0$ and $x\in \mathbb{R}$.
    Recall the energy $\mathcal{E}_{\eta,x}$ from \eqref{eq:log-energy} and the set $\mathcal{H}$ from \eqref{eq:Hprimespace}.
    Then, we have
    \begin{equation}
        \lim_{t \to \infty} - \frac{1}{t^2} \log Q(t,xt) = 1 + \frac{\eta x^2}{2}  \mathbf{1}_{(-\infty,0)}(x) + \inf_{\mathfrak{h} \in \mathcal{H}} \left\{ \mathcal{E}_{\eta,x} \left[ \mathfrak{h} \right] \right\}.
    \end{equation}
\end{theorem}

The proof of this theorem is a collage of arguments that can be found in \cite{logan_shepp1977variational,VershikKerov_LimShape1077,romik2015surprising,das2025large}.
We report it for completeness, as it constitutes the starting point of our analysis.

\begin{proof}
We have the trivial lower bound
\begin{equation}
\label{eq:Q lower bound}
    \mathcal S=\sup_{\lambda} \biggl\lbrace \mathcal{P}_{t^2}\bigl(\lbrace\lambda\rbrace\bigr) \prod_{i \ge 1} \frac{1}{1+\e^{\eta(\lambda_i-i+\frac 12-xt)}} \biggr\rbrace \leq Q(t,xt),
\end{equation}
where $\mathcal P_{t^2}$ is the Poissonized Plancherel measure, see~\eqref{eq:poissonized plancherel}.
The Hardy--Ramanujan approximation of the number of partitions of an integer~$n$ implies that there exists $\mathsf{c}>0$ such that the bound $\#\bigl\lbrace \lambda \,:\, |\lambda| = n \bigr\rbrace \,\leq\, \e^{ \mathsf{c} \sqrt{n}}$ holds for all $n\geq 0$.
Then, for all $M>1$, we also have the upper bound
\begin{equation}
    \begin{aligned}
        Q(t,xt) &= 
        \sum_{\lambda}
        \mathcal{P}_{t^2}\bigl(\lbrace\lambda\rbrace\bigr)\prod_{i \ge 1}
        \frac{1}{1+\e^{\eta(\lambda_i-i+\frac 12-xt)}}
        \\
        & \leq 
        \sum_{\lambda: |\lambda|\leq Mt^2}
        \mathcal{P}_{t^2}\bigl(\lbrace\lambda\rbrace\bigr)\prod_{i \ge 1}
        \frac{1}{1+\e^{\eta(\lambda_i-i+\frac 12-xt)}}
        + \mathcal{P}_{t^2} \biggl(\bigcup_{\lambda:|\lambda|>Mt^2}\lbrace\lambda\rbrace\biggr)
        \prod_{i \ge 1}\frac{1}{1+\e^{\eta(-i+\frac 12-xt)}}
        \\
        & \leq \biggl(\sum_{m=0}^{\lfloor Mt^2\rfloor} \e^{ \mathsf{c} \sqrt{m}} \biggr)\, \mathcal{S}  
        + \e^{-t^2(1+M \log M -M)} \prod_{i \geq 1}
        \frac{1}{1+\e^{\eta(-i+\frac 12-xt)}},
    \end{aligned}
\end{equation}
where in the last inequality we use the Hardy--Ramanujan bound and the basic Chernoff bound $\mathbb P(|\lambda|>Mt^2)<\e^{-t^2(1+M\log M-M)}$ (valid for all $M>1$) to estimate the tail of $|\lambda|$ (which is a Poisson random variable of mean $t^2$).
By the obvious bounds $\sum_{m=0}^{X} \e^{  \mathsf{c} \sqrt{m}}\leq (X+1)\e^{ \mathsf{c} \sqrt{X}}$ and (by looking at the empty partition)
\begin{equation}
\mathcal{S}  \geq \e^{-t^2}\prod_{i\geq 1} \frac{1}{1+\e^{\eta(-i+\frac 12-xt)}}
\end{equation}
we obtain
\begin{equation}
\label{eq:Q upper bound}
Q(t,xt)\leq \mathcal{S}\,\biggl((Mt^2+1)\e^{  \mathsf{c} t\sqrt{M}}+\e^{-t^2(M\log M-M)}\biggr),\quad\text{for all }M>1.
\end{equation}
Combining \eqref{eq:Q lower bound} and \eqref{eq:Q upper bound} shows
that 
\begin{equation}\label{eq:ldp}
    \lim_{t \to \infty} \frac{1}{t^2} \log Q(t,xt) = \lim_{t \to \infty} \frac{1}{t^2} \log \mathcal{S}.
\end{equation}
Let now $\lambda$ be a large integer partition and consider the scaling
\begin{equation}
    \phi(y) = t^{-1} \lambda_{\lfloor yt \rfloor+1}.
\end{equation}
Plugging the above scaling in the definition of the Poissonized Plancherel measure \eqref{eq:poissonized plancherel}, we have
\begin{equation}
    \begin{split}
        &\mathcal{P}_{t^2}\bigl(\lbrace\lambda\rbrace\bigr) \prod_{i \ge 1} \frac{1}{1+\e^{\eta(\lambda_i-i+\frac 12-xt)}}
        \\
        &= \exp \biggl\lbrace -t^2 - 2 \sum_{(i,j) \in \lambda}  \log \left( \frac{\lambda_i -i + \lambda_j'-j+1}{t} \right) - \sum_{i\ge 1} \log \left( 1+ \e^{\eta(\lambda_i-i+\frac 12-xt) } \right)  \biggr\rbrace
        \\
        & = \exp \biggl\lbrace -t^2 \left[1 - 2  \iint \log \biggl( \phi (y)-z+ \phi^{-1}(z) -y \right) \d y \, \d z 
        \\
        & \qquad \qquad \qquad \qquad \qquad \qquad + \eta \int \bigl[ \phi(y)-y-x \bigr]_+  \d y +O\bigl( t^{-1}\log t\bigr) \biggr] \biggr\rbrace.
    \end{split}
\end{equation}
The error term $O\bigl( t^{-1}\log t\bigr)$ comes from the discrepancy between the summand and the integrand, which is of order $\log t$ for all cells $(i,j)$ such that $\lambda_i-i+\lambda_j'-j$ is of order~1. These cells are only those in the proximity of \emph{corner cells} of $\lambda$, namely cells $(i,j)$ such that $\lambda_i-i+\lambda_j'-j=0$.
The number of corner cells is in bijection with rows $\lambda_i$ such that $\lambda_i>\lambda_{i+1}$ and so if a partition has $r$ corner cells then it must contain a \emph{staircase partition} $(r,r-1,\dots,2,1)$ and so its size is greater than $\frac{1}{2}r(r-1)$.
As a result, for a partition of size~$O(t^2)$ the number of corner cells is~$O(t)$.

It is convenient to work with the change of coordinates
\begin{equation}
    u = \frac{y-z}{\sqrt{2}}, \qquad v=\frac{y+z}{\sqrt{2}},
    \qquad 
    v = h(u) + |u|, \qquad \frac{v-u}{\sqrt{2}} = \phi\left( \frac{u+v}{\sqrt{2}} \right).
\end{equation}
In the $u,v$ variables, and denoting $h'$ the derivative of $h$, the above integrals become (see \cite[Section 1.14]{romik2015surprising})
\begin{equation}\label{eq:hook integral manipulations}
    \begin{split}
        &- 2  \int \log \left( \phi (y)-z+ \phi^{-1}(z) -y \right) \d y \, \d z
         + \eta \int \left[ \phi(y)-y-x \right]_+  \d y
        \\
        & = \frac{1}{2} \int \log \frac{1}{|u-v|} h'(u) h'(v) \d u \, \d v -2 \int \left[ u \log |u| -u + u \frac{\log 2}{2}  \right] h'(u) \d u
        \\
        & \qquad \qquad \qquad \qquad +\eta \sqrt{2} \int\left[ -\sqrt{2} u - x \right]_+ \left( \mathbf{1}_{[0,\infty)}(u) + \frac{1}{2} h'(u) \right) \d u
        \\
        & = \frac{\eta}{2} x^2 \mathbf{1}_{(-\infty,0)}(x) + \frac{1}{2} \int \log \frac{1}{|u-v|} h'(u) h'(v) \d u \, \d v 
        \\
        &
        \qquad \qquad \qquad \qquad -2 \int \left[ u \log |u| -u + u \frac{\log 2}{2} - \frac{\eta}{\sqrt{8}}\left[-\sqrt{2} u -x \right]_+  \right] h'(u) \d u,
    \end{split}
\end{equation}
where in the last equality we used the identity
\begin{equation}
    \eta \sqrt{2} \int \mathbf{1}_{(0,\infty)}(u) \left[ -\sqrt{2} u - x \right]_+ \d u = \frac{\eta}{2} x^2 \mathbf{1}_{(-\infty,0)}(x).
\end{equation}
We can finally operate the change of variable
\begin{equation}
    u = - \frac{\mu}{\sqrt{2}},
    \qquad
    v = - \frac{\nu}{\sqrt{2}},
    \qquad
    \mathfrak{h}(\mu) = \frac{1}{2} h' \left( - \frac{\mu}{\sqrt{2}}\right),
\end{equation}
to transform the right-hand side of \eqref{eq:hook integral manipulations} into
\begin{equation}
    \frac{\eta}{2} x^2 \mathbf{1}_{(-\infty,0)}(x) + \iint \log \frac{1}{|\mu-\nu|} \mathfrak{h}(\mu) \mathfrak{h}(\nu) \d \mu \, \d \nu  + \int 2 \left( \mu \log |\mu| - \mu + \frac{\eta}{2}\left[ \mu -x \right]_+  \right) \mathfrak{h}(\mu) \d \mu.
\end{equation}
We necessarily have $\mathfrak{h} \in \mathcal{H}$ and combining the above approximation with \eqref{eq:ldp} we complete the proof.
\end{proof}

Tracing the various transformations of the original large partition $\lambda$ performed in this proof, it is straightforward to check that~$\mathfrak{h}$ is related to the rescaled empirical measure~$\rho(\mu)$ of the point process $\mathscr{D}(\lambda)$, defined in \eqref{eq:empirical measure}, by the relation~\eqref{eq:complementation}.

\subsection{Variational problem}

The functional~$\mathcal{E}_{\eta,x}$ is a strictly convex functional on the convex set $\mathcal H$.
Thus, if a minimizer exists, it is unique.
Sufficient conditions for the existence of a minimizer are given in the following proposition.

\begin{proposition}
\label{prop:mingeneral}
    Let $\eta>0$ and $x\in\mathbb{R}$.
    Assume that $\mathfrak h_*\in\mathcal H$ satisfies, for some $\ell\in\mathbb{R}$,
    \begin{equation}
    \label{eq:Robin}
     2 \int_\mathbb{R} \log \frac{1}{|\mu - \nu|}  \mathfrak{h}_*(\nu)  \d \nu + 2\mu (\log |\mu| - 1) + V_{\eta,x}(\mu) 
    \,\,
    \begin{cases}
        \ge \ell& \text{if } \mu \in {I}_0,
        \\
        = \ell  & \text{if } \mu \in {I},
        \\
        \le \ell & \text{if } \mu \in {I}_1,
    \end{cases}
    \end{equation}
where $V_{\eta,x}(\mu)=\frac {\eta}2[\mu-x]_+$, see~\eqref{eq:potentialVqx}, and
    \begin{equation}
    \begin{aligned}
    {I}_0    &= \{ \mu \in \mathbb{R} \, : \, \mathbf{1}_{(-\infty,0]}(\mu) + \mathfrak{h}_*(\mu)=0 \},
    \\
    {I}    &= \{ \mu \in \mathbb{R} \, : \, \mathbf{1}_{(-\infty,0]}(\mu) + \mathfrak{h}_*(\mu) \in (0,1) \},
    \\
    {I}_1    &= \{ \mu \in \mathbb{R} \, : \, \mathbf{1}_{(-\infty,0]}(\mu) + \mathfrak{h}_*(\mu)=1 \}.
\end{aligned}
\end{equation}
Then, $\mathfrak h_*$ is the unique minimizer of $\mathcal{E}_{\eta,x}$ on $\mathcal H$.
\end{proposition}

\begin{proof}
Let $p(\mu)$ be the left-hand side of~\eqref{eq:Robin}.
The condition~\eqref{eq:Robin} and the definition of $\mathcal H$ imply that $p(\mu)\bigl(\mathfrak h(\mu)-\mathfrak h_*(\mu)\bigr)\geq 0$ for all $\mu\in\mathbb{R}$ and all $\mathfrak h\in\mathcal H$.
Integrating over $\mu\in\mathbb{R}$ yields $\mathcal{E}_{\eta,x}[\mathfrak h]\geq\mathcal{E}_{\eta,x}[\mathfrak h_*]$ for all $\mathfrak h\in\mathcal H$.
\end{proof}

To explicitly determine the minimizer, the standard approach is to rewrite the variational conditions from Proposition~\ref{prop:mingeneral} in terms of the boundary values of an antiderivative ${g}(z)$ of the (modified) \emph{Cauchy transform}
\begin{equation}
 g'(z) = \int_\mathbb{R} \frac{\mathfrak{h}_*(\mu)}{\mu - z}\d\mu + \log z,
\end{equation}
which is an analytic function of $z \in \mathbb{C} \setminus \mathbb{R}$.
(The additional logarithmic term accounts for the infinite support of $\mathfrak{h}_* + \mathbf{1}_{(-\infty,0)}$.)
We formalize this approach in the next proposition.

\begin{proposition}
\label{prop:minfromg}
    Let $\eta>0$ and $x\in\mathbb{R}$.
    Assume that $ g(z)$ is a function analytic for $z\in\mathbb{C}\setminus\mathbb{R}$ such that the following conditions are fulfilled.
    \begin{enumerate}[leftmargin=*]
        \item We have $ g(z) = z(\log z-1)+ g_\infty+O(z^{-1})$ (for some $ g_\infty\in\mathbb{C}$) as $z\to\infty$ uniformly in $\mathbb{C}\setminus\mathbb{R}$, where $\log z$ denotes the principal branch, analytic for $z\in\mathbb{C}\setminus(-\infty,0]$ and real-valued on $(0,+\infty)$.
        \item The boundary values $ g_\pm(\mu) = \lim_{\varepsilon\downarrow 0} g(\mu\pm\i\varepsilon)$ and $ g_\pm'(\mu) = \lim_{\varepsilon\downarrow 0} g'(\mu\pm\i\varepsilon)$ exist and are continuous for all $\mu\in\mathbb{R}$.
        \item There exist a nonnegative integer $N$ and $ p_0< p_1<\dots< p_{2N+1}$ in $\mathbb{R}$ such that, denoting $ I=\bigcup_{j=0}^{N}\left( p_{2j}, p_{2j+1}\right)$, we have $\mathbb{R}\setminus I=I_0\cup I_1$ where $I_0$ and $I_1$ are finite unions of closed intervals and 
    \begin{equation}
     g_+'(\mu)- g_-'(\mu)=0\ \text{  if  }\ \mu\in I_0,\qquad
     g_+'(\mu)- g_-'(\mu)=2\pi\i\ \text{  if  }\ \mu\in I_1.
    \end{equation}
    \item For some $\ell\in\mathbb{R}$ we have
    \begin{equation}
    g_+(\mu)+g_-(\mu)+ V_{\eta,x}(\mu) \,\,\begin{cases}
        \ge \ell & \text{if } \mu \in {I}_0,
        \\
        = \ell  & \text{if } \mu \in {I},
        \\
        \le \ell & \text{if } \mu \in {I}_1.
    \end{cases}
    \end{equation}
    \item The function
    \begin{equation}
    \label{eq:minimizerfromboundaryvalues}
    \mathfrak h_*(\mu)=\frac 1{2\pi\i}\bigl(g'_+(\mu)-g'_-(\mu)\bigr)-\mathbf 1_{(-\infty,0)}(\mu),\qquad \mu\in\mathbb{R},
    \end{equation}
    is in $\mathcal H$.
    \end{enumerate}
    Then, $\mathfrak h_*(\mu)$ is the unique minimizer of $\mathcal{E}_{\eta,x}$ on $\mathcal H$.
\end{proposition}

The complex-analytic arguments guaranteeing that $\mathfrak{h}_*$ defined by~\eqref{eq:minimizerfromboundaryvalues} satisfies the conditions of Proposition~\ref{prop:mingeneral} are standard.
Therefore, we omit the proof of Proposition~\ref{prop:minfromg}.

We observe that this method of solving the minimization problem involves making an ansatz for the sets $I=\bigcup_{j=0}^N( p_{2j}, p_{2j+1})$, $I_0$, and $I_1$.
As we will show, this minimization problem undergoes two phase transitions: when $x<x_*$ or $x>2$, ${I}$ consists of a single interval, whereas when~$x_*<x<2$, it consists of two intervals. 
We separate our analysis accordingly.

We anticipate that, when $x>2$, the minimization problem reduces to the classical Vershik--Kerov--Logan--Shepp one (see~Section~\ref{sec:minimizationx>2}) and, when $x<x_*$, the minimizer turns out to be a rescaling of the Vershik--Kerov--Logan--Shepp density.

\subsection{Case $x\leq x_*$}\label{sec:minimizationx<xq}
We start by making a ``one-cut'' assumption which will be justified \emph{a posteriori}.
This means that we assume $I=(u,v)$, for some $u<v$ to be determined.
We also make the assumption (which we will also justify below) that $x<u$.
We want to construct a function $g(z)$ satisfying the conditions of Proposition~\ref{prop:minfromg}.
We first introduce its derivative $ g'(z)$ by
\begin{equation}
\label{gfrakprimeSP}
\begin{aligned}
 g'(z)&=
 r(z)\left(\int_{-\infty}^u \frac{\d\nu}{ r(\nu)(\nu-z)}-\frac{\eta}{2\pi\i}\int_u^v \frac{\d\nu}{ r_+(\nu)(\nu-z)}\right)
\\
&=\pm \i\pi- r(z)\left(\int_{v}^{+\infty}\frac{\d\nu}{ r(\nu)(\nu-z)}+\frac{\eta}{2\pi\i}\int_u^v \frac{\d\nu}{ r_+(\nu)(\nu-z)}\right)
\end{aligned}
\end{equation}
where
\begin{equation}
 r(z)=\sqrt{(z-u)(z-v)}
\end{equation}
(analytic for $z\in\mathbb{C}\setminus[u,v]$ and $\sim z$ as $z\to\infty$) and $ r_+(\nu)=\lim_{\varepsilon\downarrow 0} r(\nu+\i\varepsilon)=\i\sqrt{|(\nu-u)(\nu-v)|}$.
In the second line of~\eqref{gfrakprimeSP} the sign $\pm$ is chosen according to $\pm\Im z>0$ and
the equality of the two lines follows from Cauchy's theorem.

The function $g' (z)$ is analytic for $z\in\mathbb{C}\setminus(-\infty,v]$ and, by the Sokhotski--Plemelj formulas, the boundary values of $ g'$ satisfy
\begin{equation}
\label{eq:gprimefrakleftjump}
\begin{aligned}
 g_+'(\mu)- g_-'(\mu)&= 2\pi\i,&&\mu\in(-\infty,u),
\\
 g_+'(\lambda)+ g_-'(\mu)&=-\eta,&&\mu\in(u,v).
\end{aligned}
\end{equation}
Using $ r(z)\int_u^v\frac{\d\nu}{ r_+(\nu)(\nu-z)} = \i\pi$, it is elementary to show that
\begin{equation}
\label{eq:frakgprimeleft}
 g'(z) =-\frac{\eta}2 + \log\frac{1+\sqrt{\frac{z-v}{z-u}}}{1-\sqrt{\frac{z-v}{z-u}}}
\end{equation}
where $\sqrt{\frac{z-v}{z-u}}$ is analytic for $z\not\in [u,v]$ and $\sim 1$ as $z\to\infty$ and we take the principal branch of the logarithm.
Hence, when $z\to\infty$,
\begin{equation}
\label{eq:frakgleftexpansioninfinity}
 g'(z)=\log z+ g_{-1}+ g_0z^{-1}+ g_1 z^{-2}+O(z^{-3})
\end{equation}
with
\begin{equation}
 g_{-1}=\log\frac{4\e^{-\eta/2}}{v-u},\quad  g_0=-\frac{u+v}{2},\quad  g_1=-\frac{3u^2+2uv+3v^2}{16}.
\end{equation}
The endpoints are determined by enforcing the asymptotic condition $ g'(z) = \log z+O(z^{-1})$.
Indeed, the unique solution $u < v$ to the system ${g}_{-1} = {g}_0 = 0$ is given by $(u = -2\e^{-\eta/2},\, v = 2\e^{-\eta/2})$, and for the remainder of this paragraph we assume that $u$ and $v$ are fixed accordingly.

Since $x_*<-2\e^{-\eta/2}$ for all $\eta>0$, cf~\eqref{eq:xq}, we can now check that our initial assumption $x<u$ is justified.
Let us also record the value 
\begin{equation}
\label{eq:frakg1left}
 g_1=-\e^\eta
\end{equation}
for later convenience and note that
\begin{equation}
\label{eq:periodgfrakleft}
\frac 1 {2\pi\i}\int_{-2\e^{-\eta/2}}^{2\e^{-\eta/2}}\bigl({g}'_+(\mu) - {g}'_-(\mu)\bigr)\d\mu =-2\e^{-\eta/2},
\end{equation}
which follows from Cauchy's theorem.

Next, we introduce
\begin{equation}
\label{eq:frakgleft}
 g(z)=\int_{2\e^{-\eta/2}}^z  g'(y)\,\d y=- r(z)-\frac{\eta}2(z-2\e^{-\eta/2})+z\log\frac{1+\sqrt{\frac{\e^{\eta/2}z-2}{\e^{\eta/2}z+2}}}{1-\sqrt{\frac{\e^{\eta/2}z-2}{\e^{\eta/2}z+2}}},
\end{equation}
which is analytic for $z\in\mathbb{C}\setminus(-\infty,2\e^{-\eta/2}]$.
Its boundary values satisfy (for some $\ell\in\mathbb{C}$)
\begin{equation}
\label{eq:gfrakleftjump}
\begin{aligned}
 g_+(\mu)- g_-(\mu)&= 2\pi\i\mu,&&\mu\in(-\infty,- 2\e^{-\eta/2}),
\\
 g_+(\lambda)+ g_-(\mu)&=-\eta(\mu-x)+\ell,&&\mu\in(- 2\e^{-\eta/2}, 2\e^{-\eta/2}).
\end{aligned}
\end{equation}
To derive the first relation we use~\eqref{eq:periodgfrakleft}.
From the second relation, since $g_\pm(\mu)\to 0$ as $\mu\to 2\e^{-\eta/2}$, we obtain
\begin{equation}
\label{eq:fraklleft}
\ell=\eta(2\e^{-\eta/2}-x).
\end{equation}
Moreover, from the explicit expression~\eqref{eq:frakgleft} we see that, as $z\to\infty$,
\begin{equation}
\label{eq:frakginftyleft}
 g(z)=z(\log z-1)+ g_\infty- g_1 z^{-1}+O(z^{-2}),\quad g_\infty=\eta\e^{-\eta/2}.
\end{equation}
(The explicit values of $g_\infty$ and $g_1$ will be useful later on.)

\begin{proposition}
Let $x\leq x_*$.
The inequalities
\begin{equation}
 g_+(\mu)+ g_-(\mu)+ V_{\eta,x}(\mu)
\,\,
\begin{cases}
\geq \ell& \text{if } \mu \in [2\e^{-\eta/2},+\infty),
\\
\leq \ell & \text{if } \mu \in (-\infty,-2\e^{-\eta/2}],
\end{cases}
\end{equation}
where $\ell=\eta(2\e^{-\eta/2}-x)$ as in~\eqref{eq:fraklleft}, are satisfied.
\end{proposition}
\begin{proof}
The first inequality is equivalent to $2 g(\mu)+\eta(\mu-2\e^{-\eta/2})\geq 0$ for all $\mu\geq 2\e^{-\eta/2}$.
This reduces to an equality in the limit $\mu\downarrow 2 \e^{-\eta/2}$ hence it is enough to show that $2 g'(\mu)+\eta>0$ for all $\mu>2\e^{-\eta/2}$. This is trivial by~\eqref{eq:frakgprimeleft}.
The second inequality is equivalent to the following pair of inequalities:
\begin{equation}
\label{eq:onecutinequalityx}
\begin{aligned}
 g_+(\mu)+ g_-(\mu)+ \eta(\mu -2\e^{-\eta/2})&\leq 0,&&\mu\in [x,-2\e^{-\eta/2}],\\
 g_+(\mu)+ g_-(\mu)+ \eta( x-2\e^{-\eta/2})&\leq 0,&&\mu\in (-\infty,x].
\end{aligned}
\end{equation}
The first one reduces to an equality in the limit $\mu\uparrow -2\e^{-\eta/2}$, see~\eqref{eq:gfrakleftjump}, hence it is enough to show that $ g_+'(\mu)+ g_-'(\mu)+ \eta>0$ for all $\mu<-2\e^{-\eta/2}$.
By~\eqref{eq:frakgprimeleft}, this is equivalent to
\begin{equation}
\log\left|\frac{1+\sqrt{\frac{\e^{\eta/2}\mu-2}{\e^{\eta/2}\mu+2}}}{1-\sqrt{\frac{\e^{\eta/2}\mu-2}{\e^{\eta/2}\mu+2}}}\right|>0,\qquad \mu<-2\e^{-\eta/2},
\end{equation}
which can be easily verified.
Finally, to show the second inequality in~\eqref{eq:onecutinequalityx} we first study the derivative in $\mu$ of the left-hand side, which is, again by~\eqref{eq:frakgprimeleft},
\begin{equation}
 g_+'(\mu)+ g_-'(\mu)=2\log\left|\frac{1+\sqrt{\frac{\e^{\eta/2}\mu-2}{\e^{\eta/2}\mu+2}}}{1-\sqrt{\frac{\e^{\eta/2}\mu-2}{\e^{\eta/2}\mu+2}}}\right|-\eta.
\end{equation}
It is easily checked that this function is decreasing for $\mu<-2\e^{-\eta/2}$ and has a (unique) zero at $\mu=-1-\e^{-\eta}$.
This means that the function $ g_+(\mu)+ g_-(\mu)+ \eta( x-2\e^{-\eta/2})$ for $\mu<-2\e^{-\eta/2}$ is concave and has a global maximum at $\mu=-1-\e^{-\eta}$.
The value of this function at the maximum is, by~\eqref{eq:frakgleft},
\begin{equation}
\left.\left(2\sqrt{\mu^2-4\e^{-\eta}}-\eta(\mu-x)+2\mu\log\left|\frac{1+\sqrt{\frac{\e^{\eta/2}\mu-2}{\e^{\eta/2}\mu+2}}}{1-\sqrt{\frac{\e^{\eta/2}\mu-2}{\e^{\eta/2}\mu+2}}}\right| \right)\right|_{\mu=-1-\e^{-\eta}}= 2(1-\e^{-\eta})+x\eta.
\end{equation}
Hence, for all $x<x_*=-2(1-\e^{-\eta})\eta^{-1}$, see~\eqref{eq:xq}, the maximum is negative and so also the second inequality in~\eqref{eq:onecutinequalityx} is proved.
\end{proof}

\begin{remark}
    \label{remark:strictineqqminus}
    In the proof of the last proposition, we also showed that for any $\epsilon,\delta>0$ there exists $ k>0$ such that for all $x\leq x_*-\delta$ we have
\begin{equation}
 g_+(\mu)+ g_-(\mu)+ V_{\eta,x}(\mu) 
\,\,
\begin{cases}
\geq \ell+ k& \text{if } \mu \in [2\e^{-\eta/2}+\epsilon,+\infty),
\\
\leq \ell-k & \text{if } \mu \in (-\infty,-2\e^{-\eta/2}-\epsilon].
\end{cases}
\end{equation}
\end{remark}

\begin{remark}
\label{remark:BirthOfACut}
The proof of the last proposition shows the necessity of the condition $x \leq x_*$, even though we can find candidate endpoints for the support of the equilibrium measure in the wider region $x < -2\e^{-\eta/2}$.
It is illustrative to plot the function  
\begin{equation}
\label{eq:functiontoplot}
\mu\,\mapsto\, g_+(\mu)+ g_-(\mu)+ \eta\left( \max\lbrace x,\mu\rbrace -2\e^{-\eta/2}\right)
\end{equation}
for $\mu < -2\e^{-\eta/2}$ and for various values of $x$, as in Figure~\ref{fig:BirthOfACut}.  
The graph of this function, which must be negative for all $\mu < -2\e^{-\eta/2}$ in order for Proposition~\ref{prop:minfromg} to apply, actually crosses the horizontal axis when $x$ increases past $x_*$.  
In the critical case $x = x_*$, the graph is tangent to the horizontal axis at $\mu = -1 - \e^{-\eta}$, suggesting the emergence of a new ``cut'' (i.e., of a new connected component of $I$) at this location when $x$ increases past $x_*$, a fact that we will prove rigorously in the next paragraph (see~\Cref{sec:propertiesendpoints}).
\end{remark}
\begin{figure}[t]
\centering
\begin{subfigure}{.32\textwidth}
\includegraphics[scale=.5]{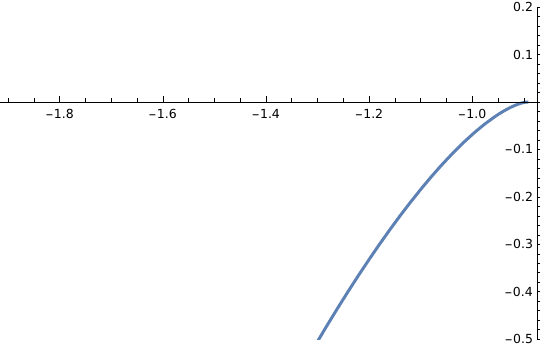}
\caption*{$x=-1.5$}
\end{subfigure}
\begin{subfigure}{.32\textwidth}
\includegraphics[scale=.5]{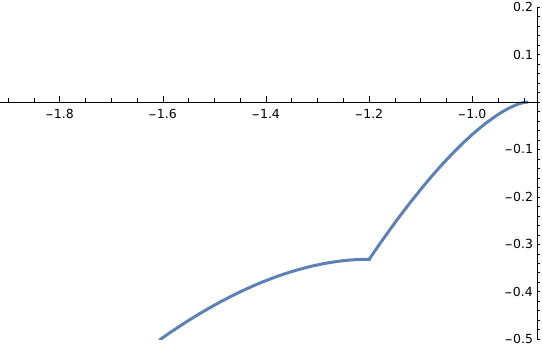}
\caption*{$x=-1-\e^{-\eta}=-1.2$}
\end{subfigure}
\begin{subfigure}{.32\textwidth}
\includegraphics[scale=.5]{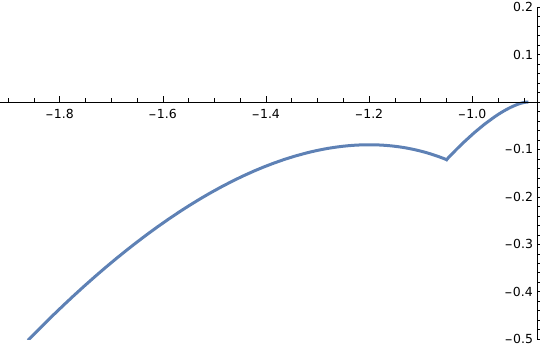}
\caption*{$x=-1.05$}
\end{subfigure}
\begin{subfigure}{.32\textwidth}
\includegraphics[scale=.5]{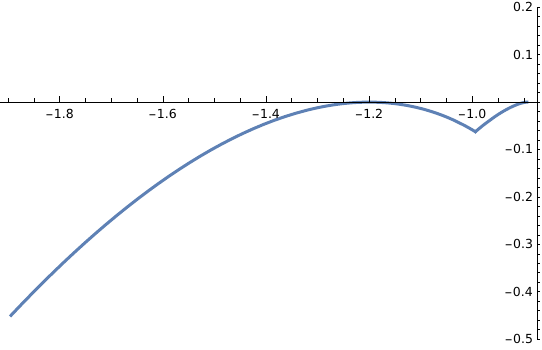}
\caption*{$x=x_*=-0.994\ldots$}
\end{subfigure}
\begin{subfigure}{.32\textwidth}
\includegraphics[scale=.5]{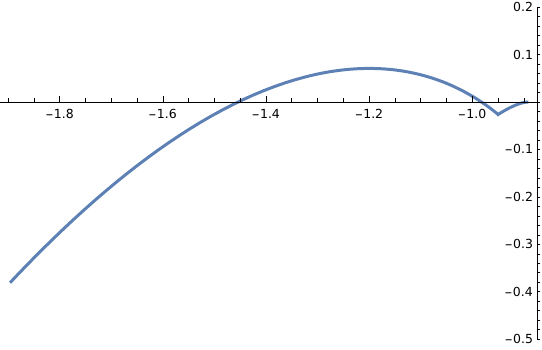}
\caption*{$x=-0.95$}
\end{subfigure}
\begin{subfigure}{.32\textwidth}
\includegraphics[scale=.5]{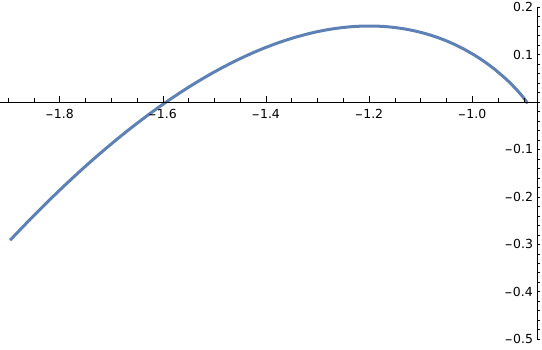}
\caption*{$x=-2\e^{-\eta/2}=-0.894\ldots$}
\end{subfigure}
\caption{Plot of the function~\eqref{eq:functiontoplot} for $\mu<-2\e^{-\eta/2}$ for $\eta=\log 5$ and various values of $x$, see Remark~\ref{remark:BirthOfACut}.}
\label{fig:BirthOfACut}
\end{figure}

It is elementary to check that for $|\mu|<2\e^{-\eta/2}$ we have
\begin{equation}
\label{eq:rescaledVKLS}
\mathfrak h_*(\mu)+\mathbf 1_{(-\infty,0)}(\mu)=\frac 1{2\pi\i}\bigl( g'_+(\mu)- g'_-(\mu)\bigr)=\frac 1\pi\arccos\left(\e^{\eta/2}\frac{\mu}{2}\right)
\end{equation}
and the proof of the first part of Theorem~\ref{thm:minimizer} is complete, because all the conditions of Proposition~\ref{prop:minfromg} are satisfied.

\subsection{Case $x_*<x<2$}\label{sec:minimizationxq<x<2}

In this case we make a ``two-cut'' assumption which, again, will be justified later.
Namely, we assume $I=(a,b)\cup(c,d)$, with
\begin{equation}
a<b<x<c<d.
\end{equation}
(For a heuristic motivation of this assumption, see~Remark~\ref{remark:BirthOfACut}.)
Again, we start from the construction of $g'(z)$, which we define by
\begin{equation}
\label{eq:gprimefrakright}
\begin{aligned}
g'(z)&=
r(z)\left(\int_{(-\infty,a)\cup(b,c)} \frac{\d\nu}{r(\nu)(\nu-z)}-\frac{\eta}{2\pi\i}\int_c^d \frac{\d\nu}{r_+(\nu)(\nu-z)}\right)
\\
&=\pm\i\pi-r(z)\left(\int_d^{+\infty} \frac{\d\nu}{r(\nu)(\nu-z)}+\frac{\eta}{2\pi\i}\int_c^d \frac{\d\nu}{r_+(\nu)(\nu-z)}\right)
\end{aligned}
\end{equation}
where
\begin{equation}
r(z)=\sqrt{(z-a)(z-b)(z-c)(z-d)}
\end{equation}
(analytic for $z\in\mathbb{C}\setminus\bigl([a,b]\cup[c,d]\bigr)$ and $\sim z^2$ as $z\to\infty$) and $r_+(\nu)=\lim_{\varepsilon\downarrow 0}r(\nu+\i\varepsilon)$.
Moreover, in the second line of~\eqref{eq:gprimefrakright} the sign $\pm$ is chosen according to $\pm\Im z>0$ and
the equality of the two lines follows from Cauchy's theorem.

The function $g' (z)$ is analytic for $z\in\mathbb{C}\setminus(-\infty,d]$ and, by the Sokhotski--Plemelj formulas, the boundary values of $g'$ satisfy
\begin{equation}
\label{eq:gprimefrakrightjumps}
\begin{aligned}
&g_+'(\mu)+g_-'(\mu)=-\eta\mathbf 1_{(c,d)}(\mu),&&\mu\in(a,b)\cup(c,d),\\
&g_+'(\mu)-g_-'(\mu)=2\pi\i,&&\mu\in (-\infty,a)\cup(b,c).
\end{aligned}
\end{equation}
As $z\to\infty$, we have
\begin{equation}
\label{eq:gfrakrightexpinfty}
g'(z)=zg_{-2}+\log z+g_{-1}+g_{0}z^{-1}+g_1z^{-2}+O(z^{-3})
\end{equation}
for appropriate coefficients $g_j$ which we will make explicit later.

It is convenient to introduce the elliptic uniformization of the Riemann surface of $r(z)$.
Namely, we consider the Schwarz--Christoffel conformal transformation
\begin{equation}
\label{eq:conformalelliptic}
z\mapsto w(z) = m\int_d^z\frac{\d\nu}{r(\nu)},\qquad m=\frac{\i\pi}{\int_d^c\frac{\d\nu}{r_+(\nu)}}>0.
\end{equation}
(Note that $\int_d^c\frac{\d\nu}{r_+(\nu)}\in\i\mathbb{R}_{>0}$; indeed, the branch of $r(z)$ we consider satisfies $r(\nu)\in\mathbb{R}_{<0}$ for $\nu\in(b,c)$ and $r_+(\nu)\in\i\mathbb{R}_{>0}$ for $\nu\in(c,d)$.)
It maps the half-planes $\lbrace z\in\mathbb{C}:\,\pm\Im z>0\rbrace$ conformally onto the rectangles $\lbrace w\in\mathbb{C}:\,0<\Re w<K,\, 0<\pm\Im w<\pi\rbrace$ (see~Figure~\ref{fig:conformalelliptic}) where
\begin{equation}
K= m\int_c^b\frac{\d\nu}{r(\nu)}>0.
\end{equation}
The inverse map $z=z(w)$ extends to the complex $w$-plane and is the universal cover of the Riemann surface of $r(z)$, realizing the latter as the complex torus $\mathbb{C}/(2K\mathbb{Z}+2\pi\i\mathbb{Z})$.
The involution $w\mapsto -w$ corresponds to the involution that exchanges the two sheets $(z,\pm r(z))$ of the Riemann surface of $r(z)$.
The two points at infinity of the Riemann surface of $r(z)$ correspond to the points $\pm w_\infty$, where
\begin{equation}
\label{eq:defwinfty}
w_\infty= m\int_d^{\infty}\frac{\d\nu}{r(\nu)}\in (0,K).
\end{equation}

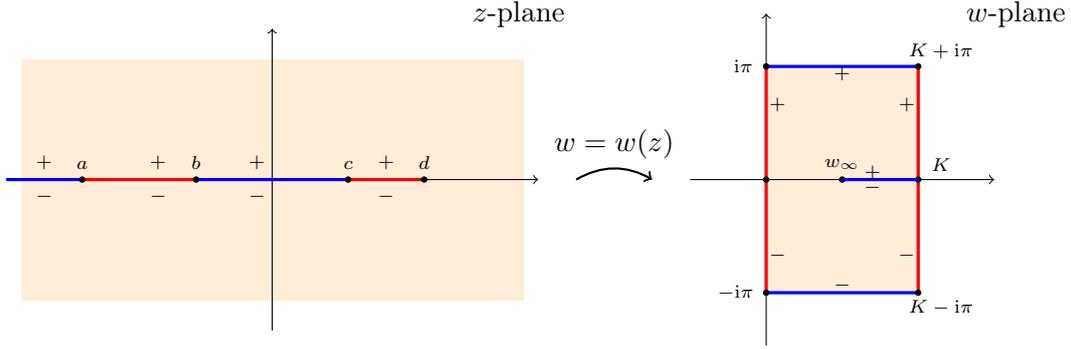
\begin{figure}[t]
\centering
\begin{tikzpicture}
\begin{scope}[shift={(-4,0)}]
\node[below right] at (2.5,2.5) {$z$-plane};

    \fill[orange!15] (-3.3,-1.6) rectangle (3.3,1.6);
    
    \draw[->] (-3.5,0) -- (3.5,0);
    \draw[->] (0,-2) -- (0,2);
    
    \draw[very thick, blue] (-3.5,0) -- (-2.5,0);
    \draw[very thick, red] (-2.5,0) -- (-1,0);
    \draw[very thick, blue] (-1,0) -- (1,0);
    \draw[very thick, red] (1,0) -- (2,0);

    \fill (-2.5,0) circle (1.2pt) node[above] {\tiny{$a$}};
    \fill (-1,0) circle (1.2pt) node[above] {\tiny{$b$}};
    \fill (1,0) circle (1.2pt) node[above] {\tiny{$c$}};
    \fill (2,0) circle (1.2pt) node[above] {\tiny{$d$}};
    
    \node[above] at (-1.5,0) {\tiny{$+$}}; \node[above] at (-3,0) {\tiny{$+$}}; \node[above] at (-.2,0) {\tiny{$+$}}; \node[above] at (1.5,0) {\tiny{$+$}};
    \node[below] at (-1.5,0) {\tiny{$-$}}; \node[below] at (-3,0) {\tiny{$-$}}; \node[below] at (-.2,0) {\tiny{$-$}}; \node[below] at (1.5,0) {\tiny{$-$}};
\end{scope}

\draw[->, thick, bend left] (0,0) to node[above] {$w=w(z)$} (1,0);
    
\begin{scope}[shift={(2.5,0)}]

\node[below right] at (2.5,2.5) {$w$-plane};

    \fill[orange!15] (0,-1.5) rectangle (2,1.5);
    \draw[->] (-1,0) -- (3,0);
    \draw[->] (0,-2.2) -- (0,2.2);
    
    \draw[very thick, blue] (0,-1.5) -- (2,-1.5);
    \draw[very thick, red] (2,-1.5) -- (2,1.5);
    \draw[very thick, blue] (2,1.5) -- (0,1.5);
    \draw[very thick, red] (0,1.5) -- (0,-1.5);
    
    \node at (-.4,-1.5) {\tiny{$-\i\pi$}};
    \node at (-.3,1.5) {\tiny{$\i\pi$}};
    \node at (2.3,-1.7) {\tiny{$K-\i\pi$}};
    \node at (2.3,1.7) {\tiny{$K+\i\pi$}};
    \node at (2.3,.2) {\tiny{$K$}};
    
    \fill (0,-1.5) circle (1.2pt);
    \fill (0,1.5) circle (1.2pt);
    \fill (2,-1.5) circle (1.2pt);
    \fill (2,1.5) circle (1.2pt);
	\fill (0,0) circle (1.2pt);
	\fill (2,0) circle (1.2pt);
    
    \node at (1,-1.4) {\tiny{$-$}};
    \node at (1,1.4) {\tiny{$+$}};
    \node at (0.15,1) {\tiny{$+$}};
    \node at (0.15,-1) {\tiny{$-$}};
    \node at (1.85,1) {\tiny{$+$}};
    \node at (1.85,-1) {\tiny{$-$}};
    \node at (1.4,.1) {\tiny{$+$}};
    \node at (1.4,-.1) {\tiny{$-$}};
    
    \fill (1,0) circle (1.2pt) node[above] {\tiny{$w_\infty$}};
    \draw[very thick, blue] (1,0) -- (2,0);
    \end{scope}
\end{tikzpicture}
\caption{Conformal transformation $w(z)$ defined in~\eqref{eq:conformalelliptic}.}
\label{fig:conformalelliptic}
\end{figure}

To obtain the endpoints $(a, b, c, d)$ we will first determine the parameters $(K, m, w_\infty, d)$.
Three constraints for these parameters are $g_2=g_1=g_0=0$, see~\eqref{eq:gfrakrightexpinfty}.
To find an additional restriction, we introduce
\begin{equation}
\label{eq:gfrakright}
g(z)=\int_d^z g'(y)\,\d y,
\end{equation}
which is analytic for $z\in\mathbb{C}\setminus(-\infty,d]$.
By integrating~\eqref{eq:gprimefrakrightjumps}, we see that the boundary values of $g$ at $(a,b)\cup(c,d)$ satisfy (recall that we are assuming $b<x<c$)
\begin{equation}
g_+(\mu)+g_-(\mu)+V_{\eta,x}(\mu)=\begin{cases}{\ell}_1  & \text{if }\mu\in (a,b)
\\
\ell &\text{if } \mu\in (c,d)
\end{cases}
\end{equation}
with
\begin{equation}
\label{eq:fraklright}
\ell_1=\eta(d-c)-\int_b^c\bigl(g'_+(\nu)+g_-'(\nu)\bigr)\d \nu,\qquad
\ell=\eta(d-x).
\end{equation}
The last condition we need is supplied by the requirement $\ell_1=\ell$, i.e., more explicitly,
\begin{equation}
\label{eq:balance}
\int_b^c\bigl(g'_+(\nu)+g_-'(\nu)\bigr)\d \nu=-\eta(c-x).
\end{equation}

Summarizing, we now have to show that the system of four equations formed by the three conditions $g_j=0$ (for $j=0,1,2$) and by~\eqref{eq:balance} admits a unique solution in the parameters $(K, m, w_\infty, d)$ and that such solution satisfies $b<x<c$.

The first equation, $g_{-2}=0$ is easily rewritten in terms of these parameters because
\begin{equation}
\label{eq:gm2}
g_{-2} = \int_d^{\infty}\frac{\d\nu}{r(\nu)}+\frac{\eta}{2\pi\i}\int_c^d \frac{\d\nu}{r_+(\nu)}=
m^{-1}\left(w_\infty-\frac{\eta}2\right),
\end{equation}
by~\eqref{eq:gprimefrakright} (second line),~\eqref{eq:conformalelliptic}, and~\eqref{eq:defwinfty}.

For the remaining equations, it is convenient to first work out an explicit expression for $g'(z(w))$ (see~Proposition~\ref{prop:h}).
We will work with Weierstrass elliptic functions $\sigma,\zeta,\wp$ with half-periods $K>0$ and $\i\pi$ (see~Appendix~\ref{app:elliptic}). We simply write $\sigma(w)=\sigma(w|K,\i\pi)$, $\zeta(w)=\zeta(w|K,\i\pi)$, and $\wp(w)=\wp(w|K,\i\pi)$ throughout this section.

We will repeatedly use the following well-known fact.

\begin{lemma}
\label{lemma:merodoublyperiodic}
Let $\psi$ be a meromorphic function on $\mathbb{C}$ satisfying $\psi(w+2Kn_0+2\pi\i n_1)=\psi(w)$ for all $n_0,n_1\in\mathbb{Z}$.
Assume that the set of poles of $\psi$ is $\lbrace w_1,\dots,w_N\rbrace+2K\mathbb{Z}+2\pi\i\mathbb{Z}$ and that $\psi(w)=k_i(w-w_i)^{-2}+h_i(w-w_i)^{-1}+O(1)$ as $w\to w_i$. 
Then, $\sum_{i=1}^Nh_i=0$ and, for some $\psi_0\in\mathbb{C}$,
\begin{equation}
\psi(w) = \psi_0+\sum_{i=1}^Nh_i\zeta(w-w_i)+\sum_{i=1}^Nk_i\wp(w-w_i).
\end{equation}
\end{lemma}

As a preliminary computation for what follows, we use this lemma to express $z(w)$ in terms of elliptic functions of $w$.
As $z\to\infty$ we have
\begin{equation}
w(z)=w_\infty-m z^{-1}+O(z^{-2}),
\end{equation}
where $w_\infty$ is defined in~\eqref{eq:defwinfty}, and so
\begin{equation}
z(w) = \frac {m}{w_\infty-w}+O(1),\quad w\to w_\infty.
\end{equation}
Therefore, the function~$z(w)$ is even, doubly periodic, and meromorphic with simple poles at~$\pm w_\infty+2K\mathbb{Z}+2\pi\i\mathbb{Z}$.
By Lemma~\ref{lemma:merodoublyperiodic}, we have
\begin{equation}
\label{eq:zelliptic}
z(w)=d-2m\zeta(w_\infty)+m\zeta(w_\infty-w)+m\zeta(w_\infty+w).
\end{equation}
Here we also use the value $z(0)=d$ to fix the constant (and we exploit the fact that $\zeta$ is odd).
Using~\eqref{eq:Taylorzeta} we see that as $w\to w_\infty$
\begin{equation}
\label{eq:expzelliptic}
m z(w) = \frac 1{w_\infty-w}+c_0+c_1(w_\infty-w)+O\bigl((w_\infty-w)^2\bigr)
\end{equation}
where
\begin{equation}
\label{eq:c0c1}
c_0=m^{-1}d+\zeta(2w_\infty)-2\zeta(w_\infty),\qquad
c_1=\wp(2w_\infty).
\end{equation}
Therefore, as $w\to w_\infty$, we have
\begin{equation}
\label{eq:explogzelliptic}
\log(z(w)) = -\log\biggl(\frac{w_\infty-w}{m}\biggr)+c_0(w_\infty-w)+\bigl(c_1-\frac 12 c_0^2\bigr) (w_\infty-w)^2+O\bigl((w_\infty-w)^3\bigr)
\end{equation}

Next, we introduce
\begin{equation}
\label{eq:f=gzw}
f(w) =g'\bigl(z(w)\bigr).
\end{equation}
Note that $z\mapsto w(z)$ is a conformal transformation of $z\in \mathbb{C}\setminus(-\infty,d]$ onto $w\in \mathcal R$, where
\begin{equation}\mathcal R=\bigl\lbrace w\in\mathbb{C}:\,0<\Re w<K,\,-\pi<\Im w<\pi,\, w\not\in (w_\infty,K)\bigr\rbrace,\end{equation}
see~Figure~\ref{fig:conformalelliptic}.
Hence, the function $f(w)$ is analytic for $w\in \mathcal R$.

\begin{proposition}
\label{prop:h}
Assuming $g_{-2}=0$, for all $w\in\mathcal R$ we have
\begin{equation}
f(w)=Aw-\frac {\eta}2-\log\frac{\sigma\bigl(w_\infty-w\bigr)}{\sigma\bigl(w_\infty+w\bigr)},\qquad A=-\frac{\zeta(\i\pi)}{\i\pi}\eta.
\end{equation}
Here, $\log\bigl(\sigma(w_\infty-w)/\sigma(w_\infty+w)\bigr)$ denotes the branch analytic in the simply connected set $\mathcal R$ which takes real values for $w\in(0,w_\infty)$.
\end{proposition} 
\begin{proof}
The existence of boundary values of $g'(z)$ and the jump conditions~\eqref{eq:gprimefrakrightjumps} imply the relations
\begin{align}
\label{condh1}
f(K+\i u)+f(K-\i u)&=0,&&u\in(0,\pi),\\
\label{condh2}
f(\i u)+f(-\i u)&=-\eta,&&u\in (0,\pi),\\
\label{condh3}
f(u+\i\pi)-f(u-\i\pi)&=2\pi\i,&&u\in(0,K),\\
\label{condh4}
f_+(u)-f_-(u)&=2\pi\i,&&u\in(w_\infty,K),
\end{align}
which also involve the boundary values $f_\pm$ of $f$ at the boundary of $\mathcal R$, see~Figure~\ref{fig:conformalelliptic}.
Moreover, using the assumption $g_{-2}=0$, $f(w)=\log(w_\infty-w)+O(1)$ as $w\to w_\infty$.
Therefore, the derivative $f'(w)=\frac{\d}{\d w}f(w)$ extends to an even doubly periodic meromorphic function of~$w$ with simple poles at~$w=\pm w_\infty+2K\mathbb{Z}+2\pi\i\mathbb{Z}$ such that
\begin{equation}
f'(w)=-\frac 1{w-w_\infty}+O(1),\qquad w\to w_\infty.
\end{equation}
By Lemma~\ref{lemma:merodoublyperiodic} we get
\begin{equation}
\label{eq:hprime}
f'(w)= A-\zeta(w-w_\infty)+\zeta(w+w_\infty)
=A-\frac{\d}{\d w}\log\frac{\sigma(w_\infty-w)}{\sigma(w_\infty+w)}
\end{equation}
for some $A\in\mathbb{C}$.
Therefore, for all $w\in\mathcal R$,
\begin{equation}
f(w)=f(0)+Aw-\log\frac{\sigma(w_\infty-w)}{\sigma(w_\infty+w)},
\end{equation}
where the branch of the logarithm is defined in the statement: indeed, $g'(z)$ takes real values for $z\in(d,+\infty)$, hence $f(w)$ must take real values for $w\in (0,w_\infty)$.
Due to this choice of branch and the properties of $\sigma$, we see that~\eqref{condh4} is satisfied, as well as~\eqref{condh2}, provided that $f(0)=-\frac 12\eta$.
Next, note that the the function $\wt{f}(w) = \frac{\sigma(w_\infty-w)}{\sigma(w_\infty+w)}$ satisfies, thanks to~\eqref{eq:periodicsigma},
\begin{equation}
\label{eq:htilde}
\wt{f}(K-u) = \frac 1{\wt{f}(K+u)}\e^{-4\zeta(K)w_\infty}.
\end{equation}
For $u\in (0,\i\pi)$ we therefore have
\begin{equation}
f(K+u)+f(K-u)=2AK-\eta+4\zeta(K)w_\infty
\end{equation}
and this expression vanishes, in agreement with~\eqref{condh1}, if and only if we take
\begin{equation}
\label{eq:A}
A=-\frac 1{2K}\bigl(4\zeta(K)w_\infty-\eta\bigr)=-\frac{\zeta(\i\pi)}{\i \pi}\eta.
\end{equation}
In the last equality we used~\eqref{eq:gm2} to substitute $w_\infty={\eta}/2$ and the Legendre identity~\eqref{eq:LegendreIdentity} to simplify the result.
This determines the expression for $f(w)$ claimed in the statement, thus completing the proof.
(One could verify by similar means that~\eqref{condh3} is also satisfied, but this is not necessary for the proof.)
\end{proof}

In what follows, we will continue denoting $A=-\frac{\zeta(\i\pi)}{\i\pi}\eta$ for short.

\begin{corollary}
Assuming $g_{-2}=0$, as $w\to w_\infty$ we have
\begin{equation}
\begin{aligned}
f(w)&=-\log(w_\infty-w)+Aw_\infty-\frac{\eta}2+\log\sigma(2w_\infty)-\bigl(A+\zeta(2w_\infty)\bigr)(w_\infty-w)
\\
&\qquad \qquad\qquad\qquad \qquad\qquad\qquad
-\frac 12\wp(2w_\infty)(w_\infty-w)^2+O\bigl((w_\infty-w)^3\bigr).
\end{aligned}
\end{equation}
\end{corollary}

By comparing with~\eqref{eq:gfrakrightexpinfty}, the last corollary implies
\begin{equation}
\begin{aligned}
g_{-1}&=Aw_\infty-\frac{\eta}2+\log\sigma(2w_\infty)-\log m,\qquad\quad
g_0=-m\bigl(A+\zeta(2w_\infty)+c_0\bigr),\\
g_1&=-m^2\biggl(c_1-\frac{c_0^2}2\biggr)+c_0\wt g_0-\frac {m^2}2\wp(2w_\infty). \label{eq:g0gpm1}
\end{aligned}
\end{equation}
Therefore, the equation $g_0=0$ implies, in view of~\eqref{eq:c0c1},
\begin{equation}
\label{eq:qplusg0=0}
A+m^{-1}d+2\zeta(2w_\infty)-2\zeta(w_\infty)=0.
\end{equation}
In order to examine~\eqref{eq:balance}, we will also express
\begin{equation}
\label{eq:lastintegral}
\begin{aligned}
F(w)=g(z(w))&=\int_{d}^{z(w)}g'(y)\d y \\
&= \int_0^wf(y) z'(y)\d y= f(w)z(w)+d\frac{\eta}2-\int_0^wf'(y) z(y)\d y
\end{aligned}
\end{equation}
in terms of elliptic functions of $w$.
We recall that $f(w)=g'(z(w))$ and the explicit values $f(0)=-\frac 12\eta$ and $z(0)=d$.
In order to compute the last integral in~\eqref{eq:lastintegral}, note that, by~\eqref{eq:zelliptic} and~\eqref{eq:hprime}, we have
\begin{equation}
\begin{aligned}
&f'(y) z(y)
\\
&= \biggl(A-\frac{\d}{\d y}\log\frac{\sigma(w_\infty-y)}{\sigma(w_\infty+y)}\biggr)\biggl(d-2m\zeta(w_\infty)-m\frac{\d}{\d y}\log\frac{\sigma(w_\infty-y)}{\sigma(w_\infty+y)}\biggr)
\\
&=A\bigl(d-2m\zeta(w_\infty)\bigr)-\bigl(m A+d-2m\zeta(w_\infty)\bigr)\frac{\d}{\d y}\log\frac{\sigma(w_\infty-y)}{\sigma(w_\infty+y)}+
m\biggl(\frac\d{\d y}\log\frac{\sigma(w_\infty-y)}{\sigma(w_\infty+y)}\biggr)^2.
\end{aligned}
\end{equation}
Let us prove an identity which is useful to integrate the last term in this expression.
\begin{lemma}
We have
\begin{equation}
\label{eq:tobeprovenelliptic}
\begin{aligned}
&\biggl(\frac\d{\d y}\log\frac{\sigma(w_\infty-y)}{\sigma(w_\infty+y)}\biggr)^2=
\\&\qquad\wp(w_\infty-y)+\wp (w_\infty+y)+2\zeta(2w_\infty)\bigl(\zeta(w_\infty-y)+\zeta(w_\infty+y)\bigr)-\frac{\sigma''(2w_\infty)}{\sigma(2w_\infty)}.
\end{aligned}
\end{equation}
\end{lemma}
\begin{proof}
The left-hand side of~\eqref{eq:tobeprovenelliptic} is an even, doubly periodic, and meromorphic function of $y$ with double poles at $\pm w_\infty+2K\mathbb{Z}+2\pi\i\mathbb{Z}$ with Laurent expansion
\begin{equation}
\frac 1{(y-w_\infty)^2}-\frac{2\zeta(2w_\infty)}{y-w_\infty}+\zeta(2w_\infty)^2-2\zeta'(2w_\infty)+O(y-w_\infty)
\end{equation}
as $y\to w_\infty$.
By Lemma~\ref{lemma:merodoublyperiodic}, this function equals the right-hand side of~\eqref{eq:tobeprovenelliptic} (where we use the constant term in the expansion as $y\to w_\infty$ to fix the constant).
\end{proof}

Using this lemma we finally obtain
\begin{equation}
\label{eq:F}
\begin{aligned}
F(w) &= f(w)z(w)+d\frac{\eta}2-\biggl(A\bigl(d-2m\zeta(w_\infty)\bigr)-m\frac{\sigma''(2w_\infty)}{\sigma(2w_\infty)}\biggr)w
\\
&\quad +\bigl(m A+d-2m\zeta(w_\infty)+2m\zeta(2w_\infty)\bigr)\log\frac{\sigma(w_\infty-w)}{\sigma(w_\infty+w)}
\\
&\quad-m\zeta(w_\infty-w)+m\zeta(w_\infty+w).
\end{aligned}
\end{equation}
By~\eqref{eq:qplusg0=0}, the coefficient in front of $\log\frac{\sigma(w_\infty-w)}{\sigma(w_\infty+w)}$ in this expression vanishes, so we can safely ignore the term in the second line from now on.

Equation~\eqref{eq:balance} is equivalent to
\begin{equation}
\begin{aligned}
&\int_d^c\bigl(g'_+(\nu)+g_-'(\nu)\bigr)\d \nu-\int_d^b\bigl(g'_+(\nu)+g_-'(\nu)\bigr)\d \nu 
\\
&\qquad
= F(\i\pi)+F(-\i\pi)-F(K+\i\pi)-F(K-\i\pi)=-\eta(c-x).
\end{aligned}
\end{equation}
Using the explicit expression for $F(w)$ we just derived, see~\eqref{eq:F}, as well as~\eqref{condh1} and~\eqref{condh2}, this condition is equivalent to
\begin{equation}
\biggl(A\bigl(d-2m\zeta(w_\infty)\bigr)-m\frac{\sigma''(2w_\infty)}{\sigma(2w_\infty)}\biggr)2K-4m\zeta(K) = x\eta.
\end{equation}

Summarizing, the four equations that determine $(K,m,w_\infty,d)$ are
\begin{align}
w_\infty&=\frac{\eta}2, \label{eq:endpointselliptic1}\\
\log m&=-\frac{\zeta(\i\pi)}{\i\pi}w_\infty\eta-\frac{\eta}2+\log\sigma(2w_\infty),\label{eq:endpointselliptic2}\\
m^{-1}d&=\frac{\zeta(\i\pi)}{\i\pi}\eta-2\zeta(2w_\infty)+2\zeta(w_\infty), \label{eq:endpointselliptic3}\\
m^{-1}x\eta &=\biggl(-\frac{\zeta(\i\pi)}{\i\pi}\eta\bigl(m^{-1}d-2\zeta(w_\infty)\bigr)-\frac{\sigma''(2w_\infty)}{\sigma(2w_\infty)}\biggr)2K-4\zeta(K). \label{eq:endpointselliptic4}
\end{align}
The first equation determines $w_\infty$ and the third one gives the value of $d$ in terms of $m,K$.
The remaining two equations then read
\begin{equation}
\label{eq:systemfinaltwocut}
m=\mathcal{U}(K),\qquad
-\frac{x\eta}{2m}=\mathcal V(K),
\end{equation}
where, for $K\in(\eta/2,+\infty)$, we denote
\begin{align}
\label{eq:f1Weierstrass}
\mathcal{U}(K)&=\exp\biggl(-\frac{\eta}{2}-\frac{\zeta(\i\pi)}{2\i\pi}\eta^2\biggr)\,\sigma(\eta),
\\
\label{eq:f2}
\mathcal V(K)&=1+K\biggl( \left( \zeta(\eta) - \eta \frac{\zeta(\i \pi)}{\i \pi}  \right)^2 - \wp( \eta ) + 2 \frac{\zeta(\i \pi)}{\i \pi} \biggr).
\end{align}
We observe that one needs to use the identity $\frac{\sigma''}{\sigma}=\zeta^2-\wp$ and the Legendre identity~\eqref{eq:LegendreIdentity} to rewrite~\eqref{eq:endpointselliptic4} as the second equation in~\eqref{eq:systemfinaltwocut}.
We also note that~\eqref{eq:f1Weierstrass} coincides with the definition in~\eqref{eq:f1} by~\eqref{eq:relsigmatheta} and the Legendre identity~\eqref{eq:LegendreIdentity}.

\begin{proposition}
\label{prop:relf1f2}
For all $K\in(\eta/2,+\infty)$ we have
\begin{equation}
\mathcal V(K) = 1-K\frac{\partial}{\partial K}\log\mathcal{U}(K).
\end{equation}
\end{proposition}
\begin{proof}
By~\eqref{eq:f2}, it suffices to show that
\begin{equation}
\label{eq:ZZZZ}
\frac{\partial}{\partial K}\log\mathcal{U}(K) = \wp( \eta )-2 \frac{\zeta(\i \pi)}{\i \pi}-\left( \zeta(\eta) - \eta \frac{\zeta(\i \pi)}{\i \pi}  \right)^2.
\end{equation}
This follows from~\eqref{eq:f1Weierstrass} along with the first equation in~\eqref{eq:partialKsigmazetawp} and~\eqref{eq: partial K zeta i pi}.
\end{proof}

We are interested in solutions satisfying $0<w_\infty<K$, see~\eqref{eq:defwinfty}, hence we restrict to~$K>\eta/2$.

\begin{proposition}
\label{prop:f1f2body}
The function $K\mapsto \mathcal U(K)\mathcal V(K)$ is monotonically increasing for $K\in \bigl(\eta/2,+\infty)$.
It tends to $-\eta$ when $K\downarrow\tfrac 12\eta$ and to $1-\e^{-\eta}$ when $K\to+\infty$.
\end{proposition}
\begin{proof}
The proof of the monotonicity property is rather lengthy and technical. Therefore, we defer it to Appendix~\ref{app:monotonic}.
Actually, both~$\mathcal{U}(K)$ and~$\mathcal{V}(K)$ are monotone in~$K$, as it will be shown in Propositions~\ref{prop:f1monotone} and~\ref{prop:f2monotone}.

In the limit~$K\downarrow\tfrac 12\eta$, the argument~$\eta$ of the Weierstrass functions (appearing in~$\mathcal{U}(K)$ and~$\mathcal{V}(K)$) becomes close to a point in the period lattice.
We use the quasi-periodicity and the homogeneity properties of~$\sigma$, see~\eqref{eq:homogeneity}, to get
\begin{equation}
\begin{aligned}
\sigma(\eta)=\sigma(\eta|K,\i\pi) &=-\sigma(\eta-2K|K,\i\pi)\e^{-2(K-\eta)\zeta(K|K,\i\pi)}
\\
&= \frac{\eta}{2K}\sigma\biggl(-\frac{\eta}{2K}(-\eta+2K)\bigg|\frac{\eta}2,\i\pi\frac{\eta}{2K}\biggr)\e^{-2(K-\eta)\zeta(K|K,\i\pi)}.
\end{aligned}
\end{equation}
Since the Weierstrass functions are continuous functions of the half-periods, this identity combined with~\eqref{eq:Taylorsigma} implies that
\begin{equation}
\sigma(\eta) \sim \e^{-2(K-\eta)\zeta(K)}(2K-\eta),\qquad K\downarrow\tfrac 12\eta.
\end{equation}
By completely similar arguments, we obtain
\begin{equation}
\label{eq:limitssss}
\zeta(\eta) \sim \frac {1}{\eta-2K},\quad
\wp(\eta) \sim \frac {1}{(2K-\eta)^2},\qquad K\downarrow\tfrac 12\eta.
\end{equation}
It is straightforward to deduce that~$\mathcal{U}(K)\sim 2K-\eta$ and~$\mathcal{V}(K)\sim-\frac{\eta}{2K-\eta}$ as~$K\downarrow\tfrac 12\eta$.
Here one needs to use the Legendre identity~\eqref{eq:LegendreIdentity}.

On the other hand, in the limit $K\to+\infty$ the Weierstrass functions degenerate to trigonometric functions, namely
\begin{equation}
\label{eq:triglimitKinfty}
\sigma(w)\to 2\e^{-\frac{w^2}{24}} \sinh \frac{w}{2},\quad
\zeta(w)\to\frac{6 \coth \frac{w}{2}-w}{12},\quad 
\wp(w)\to\frac{5+\cosh w}{24} \biggl(\mathrm{csch}\,\frac{w}{2}\biggr)^2.
\end{equation}
It follows that $\mathcal{U}(K)\to 1-\e^{-\eta}$ and $\mathcal{V}(K)\to 1$ as $K\to+\infty$.
\end{proof}

\begin{figure}[t]
\centering
\includegraphics[scale=.5]{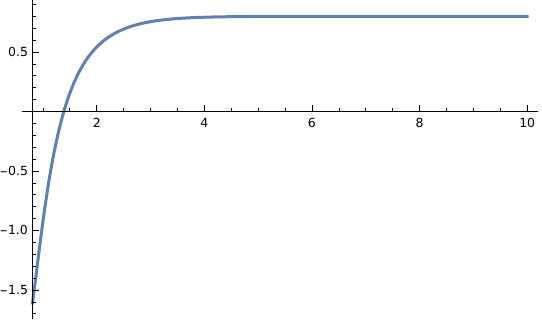}
\caption{The function $K\mapsto \mathcal{U}(K)\mathcal{V}(K)$ when $\eta=\log 5$.}
\label{fig:f1f2}
\end{figure}

The following corollary is immediate, recalling $x_*=-2(1-\e^{-\eta})\eta^{-1}$ from~\eqref{eq:xq} and the function $\mathcal{K}=\mathcal{K}(x)$ given in Definition~\ref{def:K(x)}.

\begin{corollary}
\label{cor:Kstar}
For all $x_*<x<2$ the equation~$\mathcal{U}(K)\mathcal{V}(K)=-\frac{\eta}2x$ has a unique solution $K=\mathcal{K}(x)>\frac 12\eta$.
Hence, for all $x_*<x<2$ the system~\eqref{eq:systemfinaltwocut} has a unique solution $K=\mathcal{K}(x)$, $m=\mathcal{U}\bigl(\mathcal{K}(x)\bigr)$.
\end{corollary}

Having determined the parameters $(K,m,w_\infty,d)$, we can now obtain the following explicit expression for the endpoints $(a,b,c,d)$ by using~\eqref{eq:zelliptic}:
\begin{equation}
\begin{aligned}
\label{eq:abcdapprox}
a&=\mathcal{U}(\mathcal{K})\biggl(\frac{\zeta(\i\pi)}{\i\pi}\eta+\zeta\bigl(\frac {\eta}2+\mathcal{K}\bigr)+\zeta\bigl(\frac {\eta}2-\mathcal{K}\bigr)-2\zeta\bigl(\eta\bigr)\biggr),\\
b&=\mathcal{U}(\mathcal{K})\biggl(\frac{\zeta(\i\pi)}{\i\pi}\eta+\zeta\bigl(\frac {\eta}2+\mathcal{K}+\i\pi\bigr)+\zeta\bigl(\frac {\eta}2-\mathcal{K}-\i\pi\bigr)-2\zeta\bigl(\eta\bigr)\biggr),\\
c&=\mathcal{U}(\mathcal{K})\biggl(\frac{\zeta(\i\pi)}{\i\pi}\eta+\zeta\bigl(\frac {\eta}2+\i\pi\bigr)+\zeta\bigl(\frac {\eta}2-\i\pi\bigr)-2\zeta\bigl(\eta\bigr)\biggr),\\
d&=\mathcal{U}(\mathcal{K})\biggl(\frac{\zeta(\i\pi)}{\i\pi}\eta-2\zeta(\eta)+2\zeta\bigl(\frac{\eta}2\bigr)\biggr),
\end{aligned}
\end{equation}
where $\mathcal{K}=\mathcal{K}(x)$ and it is understood that the half-periods of the Weierstrass functions are $\mathcal{K}$ and $\i\pi$.
These expressions can be simplified to~\eqref{eq:maintheorem:endpoints} by the relation between the Weierstrass $\zeta$ function and the theta functions~\eqref{eq:theta}, see~\eqref{eq:relzetatheta}.
We refer to Figure~\ref{fig:endpoints} for a plot and to Section~\ref{sec:propertiesendpoints} for some further properties of the endpoints.

\begin{remark}
\label{remark:Jacobian}
For convenience in the discussion of Section~\ref{sec:gqplus} below, we define
\begin{equation}
\label{eq:Xi}
\begin{aligned}
\Xi(w_\infty,m,K,d) &= 
\biggl(w_\infty-\frac{\eta}2,
m-\e^{-\eta/2}\,\sigma(2w_\infty)\e^{Aw_\infty},
d+m \bigl(A+2\zeta(2w_\infty)-2\zeta(w_\infty)\bigr),
\\
&\qquad\qquad\frac{x\eta}2+m\biggl(-A(A+2\zeta(2w_\infty))+\frac{\sigma''(2w_\infty)}{\sigma(2w_\infty)}\biggr)K+2\zeta(K)\biggr)
\end{aligned}
\end{equation}
(with $A=-\frac{\zeta(\i\pi)}{\i\pi}\eta$, as above)
such that the system determining the parameters $w_\infty,m,K,d$ can be written as
\begin{equation}
\label{eq:Xieq}
\Xi(w_\infty,m,K,d)=(0,0,0,0).
\end{equation}
The Jacobian determinant of $\Xi$ evaluated at the solution to~\eqref{eq:Xieq} is, up to an irrelevant sign, $\partial_K\bigl(\mathcal{U}(K)\mathcal{V}(K)\bigr)\bigr|_{K=\mathcal{K}(x)}$, which is uniformly away from $0$ as long as $x\in (x_*+\delta,2-\delta)$ for some $\delta>0$, see Appendix~\ref{app:monotonic}.
\end{remark}

For later convenience, we report the values of $g_1$, $g_\infty$, and $\ell$.

\begin{proposition}
\label{prop:g1ginftyqplus}
For all $x\in(x_*,2)$ we have
\begin{equation}
\label{eq:g1ginftyqplus}
\begin{aligned}
g_1(x)&= -\mathcal{U}(\mathcal{K})^2 \left(\wp(\eta) + \frac{\zeta(\i\pi)}{\i\pi} - \frac{\mathcal{V}(\mathcal{K})-1}{2 \mathcal{K}} \right),
\\
g_\infty(x)&=\eta\frac{\mathcal{U}(\mathcal{K})}{2}\left(\frac{\zeta(\i \pi)}{\i \pi}\eta - 2 \zeta(\eta) + 2\zeta\bigl(\frac{\eta}{2}\bigr) + \frac{\mathcal{V}(\mathcal{K})-1}{\mathcal{K}} \right),
\\
\ell(x)&=\eta\mathcal{U}(\mathcal{K})\left(\frac{\zeta(\i \pi)}{\i \pi}\eta- 2 \zeta(\eta) + 2\zeta\bigl(\frac{\eta}{2}\bigr)+2\eta^{-1}\mathcal{V}(\mathcal{K}) \right),
\end{aligned}
\end{equation}
where, as usual, $\mathcal{K}=\mathcal{K}(x)$ and the half-periods of the Weierstrass functions are $\mathcal{K},\i\pi$.
\end{proposition}
\begin{proof}
By~\eqref{eq:g0gpm1} together with the explicit expression of $c_0$ and $c_1$, given in~\eqref{eq:c0c1}, we obtain
\begin{equation}
g_1 = -m^2\left(\frac{3}2\wp(2w_\infty) - \frac{1}2\left( -\frac{\zeta(\i\pi)}{\i\pi}\eta + \zeta(2w_\infty)\right)^2 \right),
\end{equation}
and this expression is further simplified using~\eqref{eq:systemfinaltwocut} and~\eqref{eq:f2}, second line.
The expression for~$g_1$ follows from $w_\infty = \eta/2$ and $m=\mathcal{U}(\mathcal{K})$.

As for $g_\infty$, we use the explicit expression of $F(w)$ in \eqref{eq:F} and expand $F(w(z))$ as $z\to\infty$ using $w(z) = w_\infty - m z^{-1} + O(z^{-2})$. Taking the constant term in $z$ we obtain
\begin{equation}
g_\infty = d \frac{\eta}2 - m\bigl(A + \zeta(2w_\infty)\bigr) - \left(A\bigl(d - 2m\zeta(w_\infty)\bigr) - m\frac{\sigma''(2w_\infty)}{\sigma(2w_\infty)}\right)w_\infty + m \zeta(2 w_\infty).
\end{equation}
This expression is further simplified using \eqref{eq:f2} together with \eqref{eq:systemfinaltwocut} to obtain
\begin{equation}
g_\infty =- \frac{m\eta}2 \left( -\frac{\zeta(\i\pi)}{\i \pi}\eta + 2 \zeta(2w_\infty) - 2\zeta(w_\infty) - \frac{2 \zeta(\i \pi)}{\i \pi} + 2 \frac{\zeta(K)}{K} + \frac{x\eta}{2 m K} \right).
\end{equation}
The desired expression follows from the Legendre identity~\eqref{eq:LegendreIdentity} as well as from $K=\mathcal K$, $w_\infty = \eta/2$, and $m=\mathcal{U}(\mathcal{K})$.

The expression for~$\ell$ is obtained directly by~\eqref{eq:fraklright} and the expression for~$d$ in~\eqref{eq:abcdapprox}.
\end{proof}

Our next task is to show that $b<x<c$.
To this end it is convenient to use a property of this construction which is actually independent of the specific value of $K$ and which we state and prove now.

\begin{proposition}
    \label{prop:bcineq}
    Consider the endpoints $a,b,c,d$ as functions of $K$, with all other parameters $m,w_\infty,d$ determined as functions of $K$ through~\eqref{eq:endpointselliptic1}--\eqref{eq:endpointselliptic3}.
    Then, $-\eta<g_+'(\nu)+g_-'(\nu) <0$ for all $K>\eta/2$ and all $\nu\in(b,c)$.
\end{proposition}
\begin{proof}
In view of~\eqref{eq:f=gzw}, it is sufficient to show that for all $K>\frac 12\eta$, we have
    \begin{equation}
        -\eta<f(u+\i\pi)+f(u-\i\pi)<0\quad \text{for all }u\in(0,K),
    \end{equation}
    where $f(w) = f(w;K,q) = g'\bigl(z(w)\bigr)$ and, as explained in the statement, we regard $a,b,c,d$ as functions of $K$ only, the other parameters $m,w_\infty,d$ being determined in terms of $K$ by~\eqref{eq:endpointselliptic1}--\eqref{eq:endpointselliptic3}.
    Using the explicit expression for $f(w)$ given in Proposition~\ref{prop:h}, we need to show that the inequalities
    \begin{equation}
    -\eta<G(u)=-2\eta\frac{\zeta(\i\pi)}{\i\pi}u-\eta-\log\frac{\sigma\bigl(\frac{\eta}2-u-\i\pi\bigr)}{\sigma\bigl(\frac{\eta}2+u-\i\pi\bigr)}-\log\frac{\sigma\bigl(\frac{\eta}2-u+\i\pi\bigr)}{\sigma\bigl(\frac{\eta}2+u+\i\pi\bigr)}<0
    \end{equation}
hold for all $\eta>0$, $K>\frac 12\eta$, and $u\in (0,K)$.
Since $G(0)=-\eta$ and $G(K)=0$, it is enough to show that $G(u)$ is monotonically increasing for $u\in (0,K)$. 
To this end, it is sufficient to show that, for all $K>\frac 12\eta$, we have $\partial_u G(u)\big|_{u=0}>0$, $\partial_u G(u)\big|_{u=K}>0$, and $\partial_u^2G(u)<0$ for all $u\in(0,K)$.
Using the explicit formulas
\begin{equation}
\begin{aligned}
\frac 12\partial_u G(u) &= \zeta\bigl(\frac{\eta}2+u+\i\pi\bigr)+\zeta\bigl(\frac{\eta}2-u-\i\pi\bigr)-\frac{\zeta(\i\pi)}{\i\pi}\eta,\\
\frac 12\partial_u^2G(u) &= -\wp\bigl(\frac{\eta}2+u+\i\pi\bigr)+
\wp\bigl(-\frac{\eta}2-u-\i\pi\bigr),
\end{aligned}\end{equation}
these three inequalities are equivalent, respectively, to
\begin{equation}
\label{eq:ineqzeta}
\zeta\bigl(\frac \eta2+\i\pi\bigr)+
\zeta\bigl(\frac \eta2-\i\pi\bigr)-\eta\frac{\zeta(\i\pi)}{\i\pi}>0,\quad
\zeta\bigl(\frac \eta2+K+\i\pi\bigr)+
\zeta\bigl(\frac \eta2-K-\i\pi\bigr)-\eta\frac{\zeta(\i\pi)}{\i\pi}>0
\end{equation}
(for $K>0$ and $\eta\in(0,2K)$) and
\begin{equation}
\label{eq:ineqwp}
\wp\left(\frac{\eta}{2} + \i \pi -u \right) -\wp\left(\frac{\eta}2 +\i\pi +u\right)<0\ \ \text{(for $K>0$, $\eta\in (0,2K)$ and $u\in(0,K)$).}
\end{equation}
By~\eqref{eq:ellipticprostaphaeresis}, to prove~\eqref{eq:ineqwp}, it suffices to show that 
    \begin{equation}\label{eq:ineqwpaddition}
        \frac{ \wp'( \frac{\eta}{2} + \i \pi)  \wp'(u) }{ \bigl( \wp( \frac{\eta}{2} + \i \pi) - \wp(u) \bigr)^2 }<0
        \ \ \text{(for $K>0$, $\eta\in (0,2K)$ and $u\in(0,K)$).}
    \end{equation}
    This follows directly from Lemma~\ref{lemma:ineqwprectlattice}.
    Next, to prove the first inequality in~\eqref{eq:ineqzeta}, which is an equality when $\eta=0$, it suffices (taking a derivative in $\eta$ and applying the periodicity of $\wp=-\zeta'$) to show that $-\wp\bigl(\frac \eta 2+\i\pi\bigr)-\frac{\zeta(\i\pi)}{\i\pi}>0$ for all $K>0$ and $\eta\in (0,2K)$.
    Since $\wp'\bigl(\frac \eta 2+\i\pi\bigr)>0$ for all $K>0$ and $\eta\in (0,2K)$, see~Lemma~\ref{lemma:ineqwprectlattice}, it suffices to show $-\wp\bigl(K+\i\pi\bigr)-\frac{\zeta(\i\pi)}{\i\pi}>0$ for all $K\geq 0$.
    This is clear, because from~\eqref{eq:zetaipi} and~\eqref{eq:wpKipi} we obtain that $\frac{\zeta(\i\pi)}{\i\pi}+\wp(K+\i\pi)$ is equal to
    \begin{equation}
     -\frac{1}{2 \bigl(\cosh (K)+1\bigr)}-\frac 12\sum_{n\geq 1}\biggl(\frac{1}{\cosh \bigl( K(2n+1)\bigr)+1}+\frac{1}{\cosh\bigl( K(-2n+1)\bigr)+1}\biggr),
    \end{equation}
which is manifestly negative (sum of negative terms) for all $K>0$.
The proof of the first inequality in~\eqref{eq:ineqzeta} is complete.
The second one is proved by a completely parallel argument; we omit the details.
\end{proof}
\begin{corollary}
We have $b<x<c$.
\end{corollary}
\begin{proof}
    By the above lemma and further taking $K=\mathcal{K}(x)$ according to~\eqref{eq:endpointselliptic4}, which is a rewriting of~\eqref{eq:balance}, we must have
    \begin{equation}
    -(c-b)\eta<\int_{b}^{c}
    \bigl(g_+'(\nu)+g_-'(\nu)\bigr)\d\nu=-\eta(c-x)<0,
    \end{equation}
    which implies $b<x<c$.
\end{proof}    

\begin{proposition}
\label{prop:ineqvariational2cut}
For all $x_*<x<2$, the inequalities
\begin{equation}
g_+(\mu)+g_-(\mu)+ V_{\eta,x}(\mu)
\,\,
\begin{cases}
\geq \ell& \text{if } \mu \in [d,+\infty)
\\
\leq \ell & \text{if } \mu \in (-\infty,a]\cup [b,c]
\end{cases}
\end{equation}
are satisfied, where $\ell=\eta(d-x)$ as in~\eqref{eq:fraklright}.
\end{proposition}

\begin{proof}
The inequalities in the statement can be rewritten more explicitly as
\begin{equation}
    \begin{aligned}
        2 g(\mu)+\eta(\mu-d)&\geq 0,&&\mu\in[d,+\infty),\\
         g_+(\mu)+ g_-(\mu)+\eta(\mu-d)&\leq 0,&&\mu\in[x,c],\\
         g_+(\mu)+ g_-(\mu)+\eta(x-d)&\leq 0,&&\mu\in[b,x],\\
         g_+(\mu)+ g_-(\mu)+\eta(x-d)&\leq 0,&&\mu\in(-\infty,a].
    \end{aligned}
\end{equation}
First, we note that these inequalities are saturated when $\mu=a,b,c,d$.
This follows from some direct consequences (which we list now) of~\eqref{eq:gprimefrakrightjumps} and~\eqref{eq:gfrakright}.
First, when $\mu\to d$ we have $ g(\mu)\to 0$.
Then, when $\mu\to c$ we have
\begin{equation}
     g_+(\mu)+ g_-(\mu)\to\int_{d}^{c}\bigl( g'_+(\nu)+ g'_+(\nu)\bigr)\d\nu=\eta(d-c).
\end{equation}
Moreover, when $\mu\to b$, we have
\begin{equation}
     g_+(\mu)+ g_-(\mu)\to\eta(d-c)+\int_{c}^{b}\bigl( g'_+(\nu)+ g'_+(\nu)\bigr)\d\nu = \eta(d-x),
\end{equation}
see~\eqref{eq:balance}. Finally, when $\mu\to a$, we have
\begin{equation}
     g_+(\mu)+ g_-(\mu)\to\eta(d-x)+\int_{b}^{a}\bigl( g'_+(\nu)+ g'_+(\nu)\bigr)\d\nu = \eta(d-x).
\end{equation}
Hence, to complete the proof it is enough to show that
\begin{equation}
\begin{aligned}
2 g'(\mu)+\eta & > 0,&&\mu\in(d,+\infty),\\
 g_+'(\mu)+ g_-'(\mu)+\eta & > 0,&&\mu\in(x,c),\\
 g_+'(\mu)+ g_-'(\mu) & < 0,&&\mu\in(b,x),\\
 g_+'(\mu)+ g_-'(\mu) & > 0,&&\mu\in(-\infty,a).
\end{aligned}
\end{equation}

The inequalities on~$(b,x)$ and $(x,c)$ follow directly from Proposition~\ref{prop:bcineq}.
The inequalities on $(d,+\infty)$ and $(-\infty,a)$ are equivalent, in terms of the variable $u=w(\mu)$, see~\eqref{eq:conformalelliptic}, and using Proposition~\ref{prop:h}, to (respectively)
\begin{align}
\label{eq:equivalentinequalitydinfty}
-\frac{\zeta(\i\pi)}{\i\pi}\eta u-\log\frac{\sigma\bigl(\frac{\eta}2-u\bigr)}{\sigma\bigl(\frac{\eta}2+u\bigr)}&>0, &&u\in\left(0,\frac{\eta}2\right),\ \ \eta\in(0,2K),
\\
\label{eq:equivalentinequality-inftya}
-\frac{\zeta(\i\pi)}{\i\pi}\eta u-\log\left|\frac{\sigma\bigl(\frac{\eta}2-u\bigr)}{\sigma\bigl(\frac{\eta}2+u\bigr)}\right|&>\frac{\eta}2, &&u\in\left(\frac{\eta}2,K\right),\ \ \eta\in(0,2K).
\end{align}
The first one is an equality when $u\downarrow 0$ and the second one is an equality when $u\uparrow K$ (the former is a trivial assertion, for the second one needs to use the quasi-periodicity properties of $\sigma$, see~\eqref{eq:periodicsigma}, and the Legendre identity, see~\eqref{eq:LegendreIdentity}.)
Therefore, it suffices to show that the left-hand side of~\eqref{eq:equivalentinequalitydinfty} is increasing in $u\in(0,\eta/2)$ and that the left-hand side of~\eqref{eq:equivalentinequality-inftya} is decreasing in $u\in(\eta/2,K)$, namely, that
\begin{align}
    -\frac{\zeta(\i\pi)}{\i\pi}\eta+\zeta\bigl(\frac{\eta}2-u\bigr)+\zeta\bigl(\frac{\eta}2+u\bigr)&>0, &&u\in\left(0,\frac{\eta}2\right),\ \ \eta\in(0,2K),
    \\
    -\frac{\zeta(\i\pi)}{\i\pi}\eta+\zeta\bigl(\frac{\eta}2-u\bigr)+\zeta\bigl(\frac{\eta}2+u\bigr)&<0, &&u\in\left(\frac{\eta}2,K\right),\ \ \eta\in(0,2K).
\end{align}
Again, to study these inequalities, since the left-hand side of the first (second) one diverges to $+\infty$ ($-\infty$) when $u\uparrow\tfrac{\eta}2$ (when $u\downarrow\tfrac{\eta}2$) by~\eqref{eq:Taylorzeta}, it suffices to show the three inequalities
\begin{equation}
\label{eq:needtoestablish1}
-\frac{\zeta(\i\pi)}{\i\pi}\eta+2\zeta\bigl(\frac{\eta}2\bigr)> 0,\quad
\eta\in(0,2K),
\end{equation}
\begin{equation}
\label{eq:needtoestablish2}
-\frac{\zeta(\i\pi)}{\i\pi}\eta+\zeta\bigl(\frac{\eta}2-K\bigr)+\zeta\bigl(\frac{\eta}2+K\bigr)< 0,\quad \eta\in(0,2K),
\end{equation}
(corresponding to the limits $u\downarrow 0$ and $u\uparrow K$) and
\begin{equation}
\wp\bigl(\frac{\eta}2-u\bigr)-\wp\bigl(\frac{\eta}2+u\bigr)=\frac{\wp'(\eta/2)\wp'(u)}{\bigl(\wp(\eta/2)-\wp(u)\bigr)^2}>0,\quad u\in\left(0,\frac{\eta}2\right)\cup\left(\frac{\eta}2,K\right),\ \ \eta\in(0,2K).
\end{equation}
In the last inequality, we used~\eqref{eq:ellipticprostaphaeresis} to rewrite the left-hand side; this inequality follows directly from the fact that $\wp'(u)<0$ for $u\in(0,K)$, see~Lemma~\ref{lemma:ineqwprectlattice}, and so we only need to establish~\eqref{eq:needtoestablish1} and~\eqref{eq:needtoestablish2}.

Let us proceed with~\eqref{eq:needtoestablish1} first: the left-hand side diverges to $+\infty$ when $\eta\downarrow 0$ by~\eqref{eq:Taylorzeta} and converges to $1$ as $\eta\uparrow 2K$ by the Legendre identity~\eqref{eq:LegendreIdentity}. Hence we only need to show it is decreasing in $\eta$ for $\eta\in(0,2K)$, i.e., that
\begin{equation}
\label{eq:itisequivalenttoinequality}
-\frac{\zeta(\i\pi)}{\i\pi}-\wp\bigl(\frac{\eta}2\bigr)< 0,\qquad
\eta\in(0,2K).
\end{equation}
The left-hand side of this inequality diverges to $-\infty$ when $\eta\downarrow 0$. Therefore, since $\wp(\eta/2)$ is decreasing for $\eta\in(0,2K)$ (see Lemma~\ref{lemma:ineqwprectlattice}), it is enough to show this inequality for $\eta=2K$:
\begin{equation}
-\frac{\zeta(\i\pi)}{\i\pi}-\wp(K)< 0,\qquad
\eta\in(0,2K).
\end{equation}
This is clear, because from~\eqref{eq:zetaipi} and~\eqref{eq:wpK} we obtain that $-\frac{\zeta(\i\pi)}{\i\pi}-\wp(K)$ is equal to
\begin{equation}
     -\frac{1}{4\sinh(K/2)^2}-\frac 14\sum_{n\geq 1}\left[\sinh \bigl( (n+\tfrac 12)K\bigr)^{-2}+\sinh \bigl( (-n+\tfrac 12)K\bigr)^{-2}\right],
\end{equation}
which is manifestly negative (sum of negative terms) for all $K>0$. Hence~\eqref{eq:needtoestablish1} is established.

Finally, we prove~\eqref{eq:needtoestablish2}: the left-hand side can be rewritten as $2\zeta(K)-\frac{\zeta(\i\pi)}{\i\pi}\eta+2\zeta\bigl(\frac{\eta}2-K\bigr)$ by~\eqref{eq:periodiczeta}. Hence, it converges to~$0$ when $\eta\downarrow 0$  (recall that $\zeta$ is odd) and it diverges to $-\infty$ when $\eta\uparrow 2K$ by~\eqref{eq:Taylorzeta}.
Hence we only need to show it is decreasing in $\eta$ for $\eta\in(0,2K)$, i.e., that
\begin{equation}
-\frac{\zeta(\i\pi)}{\i\pi}-\wp\bigl(\frac{\eta}2-K\bigr)< 0,\qquad
\eta\in(0,2K).
\end{equation}
This is equivalent to~\eqref{eq:itisequivalenttoinequality}, which we have already shown. The proof is complete.
\end{proof}

\begin{remark}
    \label{remark:strictineqqplus}
    In the proof of the last proposition, we also showed that for any $\epsilon,\delta>0$ there exists $ k>0$ such that for all $x\in(x_*+\delta,2-\delta)$ we have
    \begin{equation}
         g_+(\mu)+ g_-(\mu)+ V_{\eta,x}(\mu) 
        \,\,
        \begin{cases}
            \geq  \ell+ k& \text{if } \mu \in [d+\epsilon,+\infty),
            \\
            \leq  \ell- k & \text{if } \mu \in (-\infty,a-\epsilon]\cup [b+\epsilon,c-\epsilon].
        \end{cases}
    \end{equation}
\end{remark}

\begin{proposition}
    The quantity $\frac{1}{2\pi\i}\bigl(g_+'(\mu)-g_-'(\mu)\bigr)$ is equal to the right-hand side in~\eqref{eq:maintheorem:densitymiddle}.
    Moreover, the corresponding $\mathfrak h_*$ defined by $\mathfrak h_*(\mu)+\mathbf{1}_{(-\infty,0)}(\mu)=\frac{1}{2\pi\i}\bigl(g_+'(\mu)-g_-'(\mu)\bigr)$ satisfies $\mathfrak h_*\in\mathcal H$.
\end{proposition}
\begin{proof}
The first statement is a simple consequence of the definition~\eqref{eq:gprimefrakright} and of the Sokhotski--Plemelj formulas.
For the second statement, we only need to show that $\frac{1}{2\pi\i}\bigl(g_+'(\mu)-g_-'(\mu)\bigr)\in[0,1]$, which is a nontrivial statement only when~$\mu\in(a,b)$ or~$\mu\in(c,d)$.
For the statement on $(a,b)$, we will show that the following inequality holds for all $K>\eta/2>0$:
\begin{equation}
\label{eq: inequality in ab}
    0<\frac 1{2\pi\i}\bigl( f(K+\i u)- f(K-\i u)\bigr)<1,\qquad  u\in(0,\pi).
\end{equation}
By \Cref{prop:h}, with $w_\infty=\eta/2$ as per \eqref{eq:endpointselliptic1}, we have
\begin{equation}
    \frac {f(K+\i u)- f(K-\i u)}{2\pi\i} = 1+\frac{ u}{\pi} + \frac{\zeta(\i \pi)}{\i \pi}(2K-\eta) \frac{ u}{\pi} - \frac{1}{\i \pi} \log \frac{\sigma(K-\frac{\eta}{2} + \i  u)}{\sigma(K-\frac{\eta}{2} - \i  u)}.
\end{equation}
For every $K>0$ and $ u\in (0,\pi)$ fixed, we will show that the function
\begin{equation}\label{eq:ineq ab case proof 1}
    \eta \mapsto 1+\frac{ u}{\pi} + \frac{\zeta(\i \pi)}{\i \pi}(2K-\eta) \frac{ u}{\pi} - \frac{1}{\i \pi} \log \frac{\sigma(K-\frac{\eta}{2} + \i  u)}{\sigma(K-\frac{\eta}{2} - \i  u)}
\end{equation}
is strictly decreasing, in the variable $\eta$, for $\eta \in (0,2K)$ and subsequently we will show that its values at $\eta=0$ and $\eta=2K$ are in the interval~$[0,1]$.
To see that \eqref{eq:ineq ab case proof 1} is a decreasing function of $\eta$, observe that its second derivative
\begin{equation}
    \frac{\wp(K-\frac{\eta}{2}+ \i  u)-\wp(K-\frac{\eta}{2} - \i  u)}{4 \pi \i} = - \frac{\wp'(K-\frac{\eta}{2}) \wp'(\i  u)}{4 \pi \i \bigl(\wp(K-\frac{\eta}{2}) - \wp (\i  u) \bigr)^2}<0
\end{equation}
is manifestly negative for all $\eta \in (0,2K)$ and $ u\in(0, \pi)$, by \Cref{lemma:ineqwprectlattice}. On the other hand, the first derivative of \eqref{eq:ineq ab case proof 1} is
\begin{equation}\label{eq:ineq ab case proof 2}
    -\frac{\zeta(\i \pi)}{\i \pi} \frac{ u}{\pi} + \frac{\zeta(K-\frac{\eta}{2}+ \i  u)-\zeta(K-\frac{\eta}{2} - \i  u)}{2 \pi \i}
\end{equation}
and so we only need to show that its evaluation at $\eta=0$ is negative.
Evaluating \eqref{eq:ineq ab case proof 2} at $\eta=0$, we have
\begin{equation}
    -\frac{\zeta(\i \pi)}{\i \pi} \frac{ u}{\pi} + \frac{\zeta(K+ \i  u)-\zeta(K - \i  u)}{2 \pi \i}.
\end{equation}
To show that the above function is negative we view it as a function of $ u \in (0,\pi)$ for $\eta,K$ fixed.
We observe that it vanishes at $ u =0$ and at $ u=\pi$, while its second derivative with respect to $ u$ is 
\begin{equation}
    \frac{\wp'(K+\i  u) - \wp'(K-\i  u)}{2 \i \pi} = \frac{\wp'(K + \i  u)}{\i \pi} >0,
\end{equation}
where we used \Cref{lemma:ineqwprectlattice}.
This proves that the function
\begin{equation}
    \frac {f(K+\i u)- f(K-\i u)}{2\pi\i}
\end{equation}
is strictly decreasing in $\eta$ and we can prove the inequality \eqref{eq: inequality in ab} by evaluating \eqref{eq:ineq ab case proof 1} at $\eta=0$ and $\eta=2K$. For $\eta=0$ we have
\begin{equation}
    \frac {f(K+\i u)- f(K-\i u)}{2\pi\i}\bigg|_{\eta=0}=1+\frac{ u}{\pi} + 2K \frac{\zeta(\i \pi)}{\i \pi} \frac{ u}{\pi} - \frac{1}{\i \pi} \log \frac{\sigma(K + \i  u)}{\sigma(K - \i  u)} = 1,
\end{equation}
which follows from the quasi-periodicity of the $\sigma$ function. On the other hand, for $\eta=2K$ we have
\begin{equation}
    \frac {f(K+\i u)- f(K-\i u)}{2\pi\i}\bigg|_{\eta=2K}=1+\frac{ u}{\pi} - \frac{1}{\i \pi} \log \frac{\sigma( + \i  u)}{\sigma(- \i  u)} = \frac{ u}{\pi} \in [0,1].
\end{equation}
This concludes the proof of the relevant inequality for $\mu\in(a,b)$. The inequality on $(c,d)$ follows from completely similar arguments and we omit it.
\end{proof}

The proof of the second part in Theorem~\ref{thm:minimizer} is complete.

\subsection{Further remarks on the endpoints}\label{sec:propertiesendpoints}
When $x\downarrow x_*$, $\mathcal{K}(x)\to+\infty$ and, using \eqref{eq:triglimitKinfty} and $\mathcal{U}(K) \to 1-\e^{-\eta}$ as $K\to+\infty$ (as shown in Proposition~\ref{prop:f1f2body}), we have 
\begin{equation}
    a,b\to -1-\e^{-\eta},\quad c\to -2\e^{-\eta/2},\quad d\to 2\e^{-\eta/2}.
\end{equation}
This confirms the emergence, when $x$ increases past $x_*$, of a new ``cut'' $(a,b)$ around the point $-1-\e^{-\eta}$, as anticipated in Remark~\ref{remark:BirthOfACut}.
At the same time, as $x\downarrow x_*$, $(c,d)$ converges to the support $(-2\e^{-\eta/2},2\e^{-\eta/2})$ of the equilibrium measure for $x\leq x_*$.
From~\eqref{eq:maintheorem:densitymiddle} it can also be shown that $\mathfrak h_*$ on $(c,d)$ converges to the minimizer of the case $x<x_*$ as $x\downarrow x_*$.

When $x\uparrow 2$, $\mathcal{K}(x)\to\eta/2$ and, using \eqref{eq:limitssss} and $\mathcal{U}(K)\sim 2K-\eta$ as $K\to\eta/2$ (as shown in Proposition~\ref{prop:f1f2body}), we have 
\begin{equation}
    a\to -2,\qquad b,c,d\to 2,
\end{equation}
namely, $(a,b)$ converges to the support $(-2,2)$ of the Vershik--Kerov--Logan--Shepp shape and $(c,d)$ shrinks to a single point and disappears as $x$ increases past $2$.
From~\eqref{eq:maintheorem:densitymiddle} it can also be shown that $\mathfrak h_*$ on $(a,b)$ converges to the Vershik--Kerov--Logan--Shepp shape as~$x\uparrow 2$.
Moreover, it is not hard to check that, in this limit,
\begin{equation}
\begin{aligned}
d-c &\sim -2\frac{2\wp''(\i\pi+\frac{\eta}2 | \frac{\eta}2,\i\pi)}{\wp(\i\pi+\frac{\eta}2|  \frac{\eta}2,\i\pi)-\wp(\i\pi| \frac{\eta}2,\i\pi)}\left(K-\frac {\eta}2\right)^2,\\
d-b &\sim 2\frac{\wp''(\i\pi| \frac{\eta}2,\i\pi)}{\wp(\i\pi+\frac{\eta}2| \frac{\eta}2,\i\pi)-\wp(\i\pi| \frac{\eta}2,\i\pi)}\left(K-\frac {\eta}2\right)^2.
\end{aligned}
\end{equation}
As a consequence,
\begin{equation}\label{eq:bcclose}
    c-b \sim 2 \frac {\wp''(\i\pi+\frac{\eta}2|\frac{\eta}2,\i\pi)+\wp''(\i\pi|\frac{\eta}2,\i\pi)}{\wp(\i\pi+\frac{\eta}2|\frac{\eta}2,\i\pi)-\wp(\i\pi|\frac{\eta}2,\i\pi)}\left(K-\frac {\eta}2\right)^2.
\end{equation}
This implies that $d-c$ and $d-b$ (and, hence, $c-b$) are all $O((2-x)^2)$ as $x\uparrow 2$.
On the other hand, the coefficient appearing in the asymptotic relation for $c-b$ is very small for small $\eta$ (see~Figure~\ref{fig:bcclose}), explaining the apparently faster convergence of $b$ and $c$ (see~Figure~\ref{fig:endpoints}).

\begin{figure}[t]
\centering
\includegraphics[scale=.5]{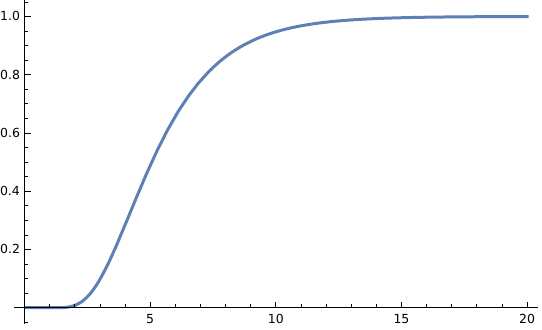}
\caption{The coefficient $2 \frac {\wp''(\i\pi+\frac{\eta}2|\frac{\eta}2,\i\pi)+\wp''(\i\pi|\frac{\eta}2,\i\pi)}{\wp(\i\pi+\frac{\eta}2|\frac{\eta}2,\i\pi)-\wp(\i\pi|\frac{\eta}2,\i\pi)}$ which appears in the asymptotic relation~\eqref{eq:bcclose}, plotted as a function of~$\eta$.}
\label{fig:bcclose}
\end{figure}

\subsection{Case $x\geq 2$}\label{sec:minimizationx>2}
The analysis of this case is completely analogous to that of Section~\ref{sec:minimizationx<xq} and so we omit the details.
Furthermore, this case reduces exactly to the classical Vershik--Kerov--Logan--Shepp analysis~\cite{VershikKerov_LimShape1077,logan_shepp1977variational}.
Indeed, by~\eqref{eq:log-energy},
\begin{equation}
\label{E+E}
\mathcal{E}_{\eta,x}[\mathfrak h]=    \mathcal{E}\left[\mathfrak{h} \right]+\int V_{\eta,x}(\mu) \mathfrak{h}(\mu) \d \mu
\end{equation}
where~$\mathcal E$ is defined in~\eqref{eq:VKLSenergy}.
It was shown in \emph{op.~cit.} that the minimizer of~$\mathcal E$ is~$\mathfrak h_*(\mu)$ given in (3) of Theorem~\ref{thm:minimizer}.
Since such $\mathfrak h_*$ is supported on $[-2,2]$, if $x\geq 2$ and $\mathfrak h\in\mathcal H$ we have
\begin{equation}
\mathcal{E}_{\eta,x}[\mathfrak h]\geq \mathcal E[\mathfrak h]\geq
\mathcal E[\mathfrak h_*]=\mathcal{E}_{\eta,x}[\mathfrak h_*]
\end{equation}
and so $\mathfrak h_*$ is also the unique minimizer of $\mathcal{E}_{\eta,x}$.
The proof of the third part in Theorem~\ref{thm:minimizer} is complete.

\section{Discrete Riemann--Hilbert characterization of the multiplicative average}\label{sec:notationsDeltaNabla}

In this section we recall the discrete Riemann--Hilbert characterization of~$Q(t,s)$ from~\cite{CafassoRuzza23}.

We begin by setting up some notation.
We denote the $2\times 2$ identity matrix by $\boldsymbol {\mathrm{I}}$ and recall the notation $\mathbb{Z}'=\mathbb{Z}+\frac 12$ from the Introduction.
Given $P\subseteq \mathbb{C}$ and $\delta>0$, we denote
\begin{equation}
\label{eq:neighborhood}
\mathcal N_\delta(P)=\bigcup_{p \in P} \bigl\lbrace z\in\mathbb{C}:\, |z-p|<\delta\bigr\rbrace.
\end{equation}
Given $-\pi\leq\nu_1<\nu_2\leq\pi$, we denote
\begin{equation}
\mathcal S_{\nu_1,\nu_2}=\bigl\lbrace z\in\mathbb{C}:\, \nu_1<\arg z<\nu_2\bigr\rbrace.
\end{equation}
Given $\eta>0$, we denote
\begin{equation}
\varsigma(z) = \frac 1{1+\e^{-\eta z}}.
\end{equation}

By the general theory of determinantal point processes, the multiplicative average~$Q(t,s)$, defined in~\eqref{eq:defQ}, can be expressed as a Fredholm determinant
\begin{equation}
Q(t,s)=\mathbb E\biggl[\prod_{i}\left(1-\varsigma(\lambda_i-i+\tfrac 12-s)\right)\biggr]=\det_{\ell^2(\mathbb{Z}')}\bigl(1-\mathscr H_\eta(t,s)\bigr).
\end{equation}
Here, $\mathscr H_\eta(t,s)$ is the operator on $\ell^2(\mathbb{Z}')$ acting via the kernel
\begin{equation}
\label{eq:Mkernel}
\mathsf H_\eta(i,j;t,s) = \sqrt{\varsigma(i-s)}\,\mathsf K(i,j;t)\,\sqrt{\varsigma(j-s)},\qquad i,j\in\mathbb{Z}',
\end{equation}
where $\mathsf K(i,j;t)$ is the discrete Bessel kernel, see~\eqref{eq:dBKernel}.

\begin{remark}
\label{rem:nonzero}
It is well-known that $Q(t,s)\not=0$ for all $\eta>0$, $t>0$, and $s\in\mathbb{Z}$, see~\cite{Borodin2007} or~\cite[Section~1]{CafassoRuzza23}.
\end{remark}

The discrete Bessel kernel belongs to a discrete analog of the class of \emph{integrable operators}~\cite{its1990differential}.
In particular, following A.~Borodin's discrete version~\cite{borodin2000riemann} of the Its--Izergin--Korepin--Slavnov theory of integrable operators, the Fredholm determinant $Q(t,s)$ can be characterized in terms of a discrete Riemann--Hilbert problem, as we now detail.

\begin{dRHp}
\label{drhp:Y}
Find a meromorphic function $\boldsymbol Y:\mathbb{C}\to\mathrm{SL}(2,\mathbb{C})$ with simple poles at $\mathbb{Z}'$ only such that
\begin{equation}
\boldsymbol Y(z)=O(1)\,\biggl(\boldsymbol {\mathrm{I}}+\frac{\boldsymbol W(n)}{z-n}\biggr),\quad z\to n,
\end{equation}
for all $n\in\mathbb{Z}'$ and
\begin{equation}
\boldsymbol Y(z)\to \boldsymbol {\mathrm{I}}
\end{equation}
as $z\to\infty$ uniformly in $\mathbb{C}\setminus\mathcal N_\delta(\mathbb{Z}')$ for any $\delta>0$.
Here, $O(1)$ denotes a $2\times 2$ matrix valued analytic function of $z$ in a neighborhood of $z=n$ and
\begin{equation}
\boldsymbol W(n) = \frac{\varsigma(n-s)}{1-\mathsf H_\eta(n,n;t,s)}
\renewcommand*{\arraystretch}{1.5}
\begin{pmatrix}
t\,\mathrm J_{n-\frac 12}(2t)\,\mathrm J_{n+\frac 12}(2t) & -\mathrm J_{n+\frac 12}(2t)^2  \\  t^2\,\mathrm J_{n-\frac 12}(2t)^2 & -t\,\mathrm J_{n-\frac 12}(2t)\,\mathrm J_{n+\frac 12}(2t) 
\end{pmatrix}
\end{equation}
where $\mathsf H_\eta(m,n;t,s)$ is the kernel appearing in~\eqref{eq:Mkernel}.
\end{dRHp}

The following result has been proven in~\cite{CafassoRuzza23} (with $\varsigma$ replaced by more general functions, see Theorems 3.4 and 3.5 in \emph{op.~cit.} as well as Remark~\ref{rem:nonzero}) and is a consequence of the general theory of (discrete) integrable operators~\cite{its1990differential,borodin2000riemann}.

\begin{theorem}[\cite{CafassoRuzza23}]
\label{thm:CR}
The discrete Riemann--Hilbert problem~\ref{drhp:Y} has a unique solution $\boldsymbol Y$ and we have
\begin{equation}
\boldsymbol Y(z)=\boldsymbol {\mathrm{I}}+z^{-1}\begin{pmatrix}
\alpha & \beta \\ \gamma & -\alpha
\end{pmatrix}+O(z^{-2})
\end{equation}
as $z\to\infty$ uniformly in $\mathbb{C}\setminus\mathcal N_\delta(\mathbb{Z}')$ for any $\delta>0$, with
\begin{equation}\begin{aligned}
\label{eq:alphabetagamma}
\alpha&=\alpha(t,s)=-\frac 12\,t\,\partial_t\log Q(t,s), \\
\beta&=\beta(t,s)=\frac{Q(t,s-1)}{Q(t,s)}-1,\\
\gamma&=\gamma(t,s)= t^2\biggl(\frac{Q(t,s+1)}{Q(t,s)}-1\biggr).
\end{aligned}
\end{equation}
\end{theorem}

Introduce, following~\cite{CafassoRuzza23},
\begin{equation}
\label{eq:Phi}
\boldsymbol \Phi(z) = \begin{pmatrix}
\mathrm J_{z-\frac 12}(2t) & \i\pi\mathrm H^{(1)}_{z-\frac 12}(2t) \\
t\mathrm J_{z+\frac 12}(2t) & \i\pi t \mathrm H^{(1)}_{z+\frac 12}(2t)
\end{pmatrix}
\end{equation}
where $\mathrm H_k^{(1)}(r)$ are the Hankel functions of the first kind.

The following lemma is a straightforward consequence of the results in~\cite[Section~4.1]{CafassoRuzza23}.

\begin{lemma}[\cite{CafassoRuzza23}]
\label{lemma:simplejump}
We have
\begin{equation}
\boldsymbol Y(z)\boldsymbol \Phi(z)=O(1)\left(\boldsymbol {\mathrm{I}}-\frac{\varsigma(n-s)}{z-n}\begin{pmatrix}
    0 & 1 \\ 0 & 0
\end{pmatrix} \right),\quad z\to n,
\end{equation}
for all $n\in\mathbb{Z}'$.
\end{lemma}

In other words, right multiplication by $\boldsymbol \Phi$ simplifies the pole condition in the discrete Riemann--Hilbert problem~\ref{drhp:Y}.
On the other hand, the asymptotic expansion at $z=\infty$ is more involved now.
To deal with it, we will also consider the Hankel functions of the second kind, $H_k^{(2)}(r)$.

\begin{lemma}
\label{lemma:asymptoticsbesselandhankel}
As $z\to\infty$ uniformly in $\mathcal S_{-\pi+\delta,\pi-\delta}$ for any $\delta>0$, we have
\begin{equation}
\begin{aligned}
\mathrm J_{z-\frac 12}(2t)&=\frac 1{\sqrt{2\pi t}}\biggl(1-\bigl(t^2-\tfrac 1{24}\bigr)z^{-1}+O\bigl(z^{-2}\bigr)\biggr) \e^{-z(\log(zt^{-1})-1)},
\\
\mathrm J_{z+\frac 12}(2t) &=\frac 1{\sqrt{2\pi t}}\biggl(tz^{-1}+O\bigl(z^{-2}\bigr)\biggr) \e^{-z(\log(zt^{-1})-1)}.
\end{aligned}
\end{equation}
As $z\to\infty$ uniformly in $\mathcal S_{-\pi+\delta,\pi-\delta}\setminus\mathcal N_\delta(\mathbb{Z}'_{>0})$ for any $\delta>0$, we have
\begin{equation}
\begin{aligned}
\frac{\pi}{\cos\pi z}\,\mathrm J_{-z-\frac 12}(2t) &=\sqrt{\frac{2\pi}{t}}\biggl(1+\bigl(t^2-\tfrac 1{24}\bigr)z^{-1}+O\bigl(z^{-2}\bigr)\biggr) \e^{z(\log(zt^{-1})-1)},
\\
\frac{\pi}{\cos\pi z}\,\mathrm J_{-z+\frac 12}(2t) &=\sqrt{\frac{2\pi}{t}}\biggl(-tz^{-1}+O\bigl(z^{-2}\bigr)\biggr) \e^{z(\log(zt^{-1})-1)}.
\end{aligned}
\end{equation}
For $l=1,2$, as $z\to\infty$ uniformly in $\mathcal S_{-\frac\pi 2+\delta,\frac\pi 2-\delta}$ for any $\delta>0$, we have
\begin{equation}\begin{aligned}
\i\pi\mathrm H^{(l)}_{z-\frac 12}(2t) &=(-1)^{l+1}\sqrt{\frac {2\pi} t}\biggl(tz^{-1}+O\bigl(z^{-2}\bigr)\biggr) \e^{z(\log(zt^{-1})-1)},
\\
\i\pi\mathrm H^{(l)}_{z+\frac 12}(2t) &=(-1)^{l+1}\sqrt{\frac {2\pi} t}\biggl(1+\bigl(t^2-\tfrac 1{24}\bigr)z^{-1}+O\bigl(z^{-2}\bigr)\biggr) \e^{z(\log(zt^{-1})-1)}.
\end{aligned}
\end{equation}
\end{lemma}
\begin{proof}
Using the well-known expansion of Bessel functions
\begin{equation}
\mathrm J_{k}(r)=\biggl(\frac r2\biggr)^k\sum_{n\geq 0}\frac{(-r^2/4)^n}{n!\,\Gamma(k+1+n)},
\end{equation}
we have
\begin{equation}
\mathrm J_{z-\frac 12}(2t)\,\Gamma(\tfrac 12+z)\,t^{\frac 12-z}=\sum_{n\geq 0}\frac{(-t^2)^n}{n!}\,\frac{\Gamma(z+\frac 12)}{\Gamma(z+\frac 12+n)}=1-\frac{t^2}{z+\frac 12}+\sum_{n\geq 2}\frac{(-t^2)^n}{n!\,(z+\frac 12)_n}
\end{equation}
using the Pochhammer symbol $(w)_n=\prod_{j=0}^{n-1} (w+j)$.
Hence,
\begin{equation}
\label{eq:restartfromhere}
\mathrm J_{z-\frac 12}(2t) \,\Gamma(\tfrac 12+z)\,t^{\frac 12-z}\,-\,1\,+\,t^2z^{-1}=\frac {t^2}{2z(z+\frac 12)}\,+\,\sum_{n\geq 2}\frac{(-t^2)^n}{n!\,(z+\frac 12)_n}=O\bigl(z^{-2}\bigr)
\end{equation}
as $z\to\infty$ uniformly in $\mathbb{C}\setminus \mathcal N_\delta(\mathbb{Z}'_{<0})$ for any $\delta>0$.
Stirling's approximation implies that the asymptotic relation
\begin{equation}
\label{eq:Stirling}
\Gamma(\tfrac 12+z)t^{\frac 12-z}=\sqrt{2\pi t}\bigl(1-\tfrac 1{24}z^{-1}+O(z^{-2})\bigr)\e^{z(\log(zt^{-1})-1)}
\end{equation}
holds as $z\to\infty$ uniformly in $\mathcal S_{-\pi+\delta,\pi-\delta}$ for any $\delta>0$.
It is straightforward to combining these facts to get the claimed asymptotic relation for $\mathrm J_{z-\frac 12}(2t)$; the one for $\mathrm J_{z+\frac 12}(2t)$ follows.

By replacing $z\mapsto-z$ and using Euler's reflection formula in~\eqref{eq:restartfromhere} we get
\begin{equation}
\mathrm J_{-z-\frac 12}(2t) \,\frac{\pi}{\cos\pi z}\Gamma(\tfrac 12+z)^{-1}\,t^{\frac 12+z}\,-\,1\,-\,t^2z^{-1}=\frac {t^2}{2z(z-\frac 12)}+\sum_{n\geq 2}\frac{(-t^2)^n}{n!\,(-z+\frac 12)_n}=O\bigl(z^{-2}\bigr)
\end{equation}
as $z\to\infty$ uniformly in $\mathbb{C}\setminus \mathcal N_\delta(\mathbb{Z}'_{>0})$ for any $\delta>0$.
The claimed asymptotic relation for $\mathrm J_{-z-\frac 12}(2t)$ then follows from the Stirling approximation, see~\eqref{eq:Stirling}. The one for $\mathrm J_{-z+\frac 12}(2t)$ is a direct consequence.

The asymptotic relations for the Hankel functions follow from (see~\cite[eq.~10.4.7]{DLMF})
\begin{equation}
\label{eq:connectionformulas1}
\begin{aligned}
\i\pi\mathrm H^{(1)}_{z-\frac 12}(2t)&=\frac{2\pi\i}{1+\e^{2\pi\i z}}\mathrm J_{z-\frac 12}(2t)-\frac\pi{\cos(\pi z)}\mathrm J_{-z+\frac 12}(2t),
\\
\i\pi\mathrm H^{(2)}_{z-\frac 12}(2t)&=\frac{2\pi\i}{1+\e^{-2\pi\i z}}\mathrm J_{z-\frac 12}(2t)+\frac\pi{\cos(\pi z)}\mathrm J_{-z+\frac 12}(2t),
\end{aligned}
\end{equation}
and the proof is complete.
\end{proof}

We set
\begin{equation}
{\renewcommand{\arraystretch}{1.4}
\begin{aligned}
\boldsymbol \Phi_R(z)&=
\begin{pmatrix}
    1 & 0 \\ 0 & t
\end{pmatrix}
\begin{pmatrix}
\mathrm J_{z-\frac 12}(2t) &  -\mathrm J_{-z+\frac 12}(2t) \\
\mathrm J_{z+\frac 12}(2t) & \mathrm J_{-z-\frac 12}(2t)
\end{pmatrix}
\begin{pmatrix}
    1 & 0 \\ 0 & \frac{\pi}{\cos(\pi z)}
\end{pmatrix},
\\
\boldsymbol \Phi_{L}^{+}(z)&=
\begin{pmatrix}
    1 & 0 \\ 0 & t
\end{pmatrix}
\begin{pmatrix}
\mathrm H^{(1)}_{-z+\frac 12}(2t) & -\mathrm J_{-z+\frac 12}(2t)
\\
-\mathrm H^{(1)}_{-z-\frac 12}(2t) & \mathrm J_{-z-\frac 12}(2t)
\end{pmatrix}
\begin{pmatrix}
    \frac {\i}{2}\e^{-\i\pi z} & 0 \\ 0 & 2\pi \e^{\i\pi z}
\end{pmatrix},
\\
\boldsymbol \Phi_{L}^{-}(z)&=
\begin{pmatrix}
    1 & 0 \\ 0 & t
\end{pmatrix}
\begin{pmatrix}
\mathrm H^{(2)}_{-z+\frac 12}(2t) & -\mathrm J_{-z+\frac 12}(2t)
\\
-\mathrm H^{(2)}_{-z-\frac 12}(2t) & \mathrm J_{-z-\frac 12}(2t)
\end{pmatrix}
\begin{pmatrix}
    -\frac {\i}{2}\e^{\i\pi z} & 0 \\ 0 & 2\pi \e^{-\i\pi z}
\end{pmatrix}.
\end{aligned}
}
\end{equation}
By use of the connection formulas~\eqref{eq:connectionformulas1} and
\begin{equation}
\mathrm H^{(1)}_{-z+\frac 12}(2t)=-\i\,\e^{\i\pi z}\mathrm H^{(1)}_{z-\frac 12}(2t),\qquad
\mathrm H^{(2)}_{-z+\frac 12}(2t)=\i\,\e^{-\i\pi z}\mathrm H^{(2)}_{z-\frac 12}(2t),
\end{equation}
see~\cite[eq.~10.4.6]{DLMF}, we get
\begin{equation}
\boldsymbol \Phi_R(z)=\boldsymbol \Phi(z) \begin{pmatrix}
    1 & -\frac{2\pi\i}{1+\e^{2\pi\i z}}\\ 0 & 1
\end{pmatrix}
,\quad
\boldsymbol \Phi_L^\pm(z)=\boldsymbol \Phi_R(z)\boldsymbol C_\pm(z),\quad
\boldsymbol C_\pm(z)=\begin{pmatrix} \frac 1{1+\e^{\pm 2\pi\i z}} & 0 \\ \pm \frac 1{2\pi\i} & 1+\e^{\pm 2\pi\i z} \end{pmatrix}.
\end{equation}

\begin{remark}
    The matrix $\boldsymbol \Phi_R$ coincides with the matrix employed by Borodin in the same context in~\cite[Section~3]{borodin2003discrete}.
    The construction of the other matrices $\boldsymbol \Phi_L^\pm$ is essentially dictated by the Riemann--Hilbert analysis which will be carried out in the following sections.
\end{remark}

The following two propositions will be the main ingredients for the formulation of \emph{continuous} Riemann--Hilbert problems amenable to the nonlinear steepest descent analysis.

\begin{proposition}
\label{prop:analytic}
For all $s\in\mathbb{Z}$, let
\begin{equation}
\begin{aligned}
\boldsymbol \nabla_\pm(z)&=\begin{pmatrix}
1 & \mp2\pi\i(1+\e^{\mp 2\pi\i z})^{-1}\bigl(1-\varsigma(z-s)\bigr) \\ 0 & 1
\end{pmatrix},
\\
\boldsymbol \Delta_\pm(z)&=\begin{pmatrix}
(1+\e^{\pm 2\pi\i z})^{-1} & 0 \\ \pm\frac{1}{2\pi\i}\bigl(1-\varsigma(z-s)\bigr)^{-1} & 1+\e^{\pm 2\pi\i z}
\end{pmatrix}.
\end{aligned}
\end{equation}
The matrices $\boldsymbol Y(z)\boldsymbol \Phi_R(z)\boldsymbol \nabla_\pm(z)$ and $\boldsymbol Y(z)\boldsymbol \Phi_R(z)\boldsymbol \Delta_\pm(z)$ are analytic at $z=n$ for all $n\in\mathbb{Z}'$.
\end{proposition}
\begin{proof}
It follows from Lemma~\ref{lemma:simplejump}.
\end{proof}

\begin{remark}
\label{remark:entire}
    It is important to note that $\varsigma(z)$ has poles on the line $\Re z=0$, whereas $(1-\varsigma(z))^{-1}$ is entire in~$z$.
\end{remark}

Throughout this paper we will use the \emph{Pauli matrix} $\boldsymbol\sigma_3=\begin{pmatrix}
    1 & 0 \\ 0 & -1
\end{pmatrix}$.

\begin{proposition}
\label{prop:asymptotics}
Let $\boldsymbol \Phi=\boldsymbol \Phi(z)$ be any of the matrices $\boldsymbol \Phi_R,\boldsymbol \Phi_{L}^{\pm}$.
The asymptotic relation
\begin{equation}
\boldsymbol \Phi(z)\bigl(\sqrt{2\pi t}\,\e^{z(\log(zt^{-1})-1)}\bigr)^{\boldsymbol\sigma_3}=\boldsymbol {\mathrm{I}}+\begin{pmatrix}
\tfrac 1{24}-t^2 & 1 \\ t^2 & -\tfrac 1{24}+t^2
\end{pmatrix}z^{-1}+O\bigl(z^{-2}\bigr)
\end{equation}
holds in the following regimes:
\begin{itemize}[leftmargin=*]
\item if $\boldsymbol \Phi=\boldsymbol \Phi_R$, as $z\to\infty$ uniformly in $\mathcal S_{-\pi+\delta,\pi-\delta}\setminus\mathcal N_\delta(\mathbb{Z}'_{>0})$ for any $\delta>0$;
\item if $\boldsymbol \Phi=\boldsymbol \Phi_L^+$, as $z\to\infty$ uniformly in $\mathcal S_{\pi-\delta,\pi}$ for any $\delta>0$;
\item if $\boldsymbol \Phi=\boldsymbol \Phi_L^-$, as $z\to\infty$ uniformly in $\mathcal S_{-\pi,-\pi+\delta}$ for any $\delta>0$.
\end{itemize}
\end{proposition}
\begin{proof}
It follows from Lemma~\ref{lemma:asymptoticsbesselandhankel}.
\end{proof}

\section{Nonlinear steepest descent analysis (case \texorpdfstring{$x<x_*$}{x<x\_*})}\label{sec:DeiftZhou1cut}

Throughout this section we are going to assume that $x=s/t<x_*$.

\subsection{Continuous Riemann--Hilbert problem}

Assume $\wt u\in\mathbb{R}$ (which will be fixed later).
Consider the (multi-)contour $\Sigma_M$ in the complex $z$-plane depicted in Figure~\ref{fig:SigmaMqminus}:
\begin{equation}
\Sigma_M=\mathbb{R}\cup\Gamma_{L}^+\cup\Gamma_{L}^-\cup\Gamma_{R}^+\cup\Gamma_{R}^-\cup\gamma_u^+\cup\gamma_u^-,
\end{equation}
where $\gamma_u^\pm$ are smooth contours joining $\wt u$ to the lines $\Im z=\pm\epsilon$ and, denoting $\wt u^\pm$ the intersection points of $\gamma_u^\pm$ with $\Im z=\pm\epsilon$, $\Gamma_R^\pm=[\wt u^\pm,+\infty\pm\i\epsilon)$ and $\Gamma_L^\pm=\bigl([-R_0\pm\i\epsilon,\wt u^\pm]\bigr)\cup\bigl(\e^{\pm\i \delta_0}\mathbb{R}_+-R_0\pm\i\epsilon\bigr)$.
The parameters $\wt u\in\mathbb{R}$, $\epsilon>0$, $R_0>-\wt u$, as well as the specific curves $\gamma_u^\pm$, will be fixed later, while $\delta_0\in(\pi/2,\pi)$ can be fixed arbitrarily.
Let $\Sigma_M^\circ=\Sigma_M\setminus \lbrace {\wt{u}},{\wt{u}}^+,{\wt{u}}^-\rbrace$, oriented as in Figure~\ref{fig:SigmaMqminus}.
The orientation determines $\pm$ sides of $\Sigma_M^\circ$, $+$ to the left-hand side and $-$ to the right-hand side.
Moreover, $\Sigma_M$ divides $\mathbb{C}$ into six connected components which we call $\Omega_1$, $\dots$, $\Omega_6$.
This is illustrated in Figure~\ref{fig:SigmaMqminus}.

\begin{figure}[t]
\centering
\begin{tikzpicture}

\draw[very thin,->] (-6.5,0) -- (4.5,0);
\draw[very thin,->] (0,-2) -- (0,2);

\draw[very thick,->] (-6,0) -- (-3,0) node[below]{\tiny{$-$}} node[above]{\tiny{$+$}};
\draw[very thick,->] (-3,0) -- (1,0) node[below]{\tiny{$-$}} node[above]{\tiny{$+$}};
\draw[very thick] (1,0) -- (4,0);

\draw[very thick] (-4,1) -- (-5,3/2) node[below]{\tiny{$-$}} node[above]{\tiny{$+$}};
\draw[very thick,<-] (-5,3/2) -- (-6,2);

\draw[very thick] (-4,-1) -- (-5,-3/2) node[below]{\tiny{$-$}} node[above]{\tiny{$+$}};
\draw[very thick,<-] (-5,-3/2) -- (-6,-2);

\draw[very thick,->] (-4,1) -- (1,1) node[below]{\tiny{$-$}} node[above]{\tiny{$+$}};
\draw[very thick] (1,1) -- (4,1);

\draw[very thick,->] (-4,-1) -- (1,-1) node[below]{\tiny{$-$}} node[above]{\tiny{$+$}};
\draw[very thick] (1,-1) -- (4,-1);

\draw[very thick,->] 
  (-1,-1) .. controls (-1.3,-0.6) .. (-1.5,-0.5) 
    node[right]{\tiny{$-$}} node[left]{\tiny{$+$}};

\draw[very thick] 
  (-1.5,-0.5) .. controls (-1.8,-0.3) .. (-2,0);

\draw[very thick,->] 
  (-2,0) .. controls (-1.8,0.3) .. (-1.5,0.5) 
    node[right]{\tiny{$-$}} node[left]{\tiny{$+$}};

\draw[very thick] 
  (-1.5,0.5) .. controls (-1.3,0.6) .. (-1,1);

\node at (-2.1,-1/5) {\small{${\wt{u}}$}};
\node at (-1,6/5) {\small{${\wt{u}}^+$}};
\node at (-1,-6/5) {\small{${\wt{u}}^-$}};

\node at (-4,1/2) {\small{$\Omega_6$}};
\node at (-4,-1/2) {\small{$\Omega_1$}};
\node at (1.2,2) {\small{$\Omega_5$}};
\node at (1.2,-2) {\small{$\Omega_2$}};
\node at (2.5,1/2) {\small{$\Omega_4$}};
\node at (2.5,-1/2) {\small{$\Omega_3$}};

\end{tikzpicture}
\caption{$\Sigma_M$, its orientation and corresponding $\pm$ sides and domains $\Omega_i$ for $1\leq i\leq 6$ (case $x<x_*$).}
\label{fig:SigmaMqminus}
\end{figure}
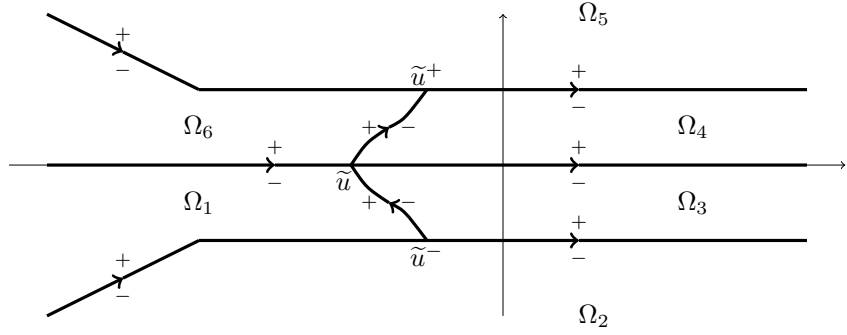

We assume that 
\begin{equation}
x<{\wt{u}}.
\end{equation}
We introduce analytic matrix functions $\boldsymbol M_i:\Omega_i\to\mathrm{SL}(2,\mathbb{C})$ as follows:
\begin{equation}
\begin{aligned}
\boldsymbol M_1(z)&=\boldsymbol Y(tz)\boldsymbol\Phi_R(tz)\boldsymbol\Delta_-(tz)=\boldsymbol Y(tz)\boldsymbol\Phi_{L}^-(tz)\boldsymbol C_-(tz)^{-1}\boldsymbol\Delta_-(tz),\\
\boldsymbol M_2(z)&=\boldsymbol M_5(z)=\boldsymbol Y(tz)\boldsymbol\Phi_R(tz),\\
\boldsymbol M_3(z)&=\boldsymbol Y(tz)\boldsymbol\Phi_R(tz)\boldsymbol\nabla_-(tz),\\
\boldsymbol M_4(z)&=\boldsymbol Y(tz)\boldsymbol\Phi_R(tz)\boldsymbol\nabla_+(tz),\\
\boldsymbol M_6(z)&=\boldsymbol Y(tz)\boldsymbol\Phi_R(tz)\boldsymbol\Delta_+(tz)=\boldsymbol Y(tz)\boldsymbol\Phi_L^+(tz)\boldsymbol C_+(tz)^{-1}\boldsymbol\Delta_+(tz),
\end{aligned}
\end{equation}
with the notations introduced in Section~\ref{sec:notationsDeltaNabla} (and a slight abuse of notation in writing equalities like $\boldsymbol M_2=\boldsymbol M_5$, as these are functions with different domains but defined by the same formula).
Note that the assumption~$x<{\wt{u}}$ is needed to ensure that $\boldsymbol M_3,\boldsymbol M_4$ are analytic in $\Omega_3,\Omega_4$ (respectively), see~Remark~\ref{remark:entire}.

An important property of these matrix functions, which follows from Proposition~\ref{prop:analytic}, is that $\boldsymbol M_i$ is analytic in a proper open neighborhood of $\Omega_i$.
Moreover, we have
\begin{equation}\begin{aligned}
\e^{-z(\log(zt^{-1})-1)\boldsymbol \sigma_3}\boldsymbol \nabla_+(z)\e^{z(\log(zt^{-1})-1)\boldsymbol \sigma_3}&=\boldsymbol {\mathrm{I}}+O(z^{-\infty})&&\text{as }z\to\infty\text{ uniformly in }\Omega_4\setminus\mathcal N_\delta(\mathbb{Z}'_+),\\
\e^{-z(\log(zt^{-1})-1)\boldsymbol \sigma_3}\boldsymbol \nabla_-(z)\e^{z(\log(zt^{-1})-1)\boldsymbol \sigma_3}&=\boldsymbol {\mathrm{I}}+O(z^{-\infty})&&\text{as }z\to\infty\text{ uniformly in }\Omega_3\setminus\mathcal N_\delta(\mathbb{Z}'_+),\\
\e^{-z(\log(zt^{-1})-1)\boldsymbol\sigma_3}\boldsymbol C_+(z)^{-1}\boldsymbol \Delta_+(z)\e^{z(\log(zt^{-1})-1)\boldsymbol\sigma_3}&=\boldsymbol {\mathrm{I}}+O(z^{-\infty})&&\text{as }z\to\infty\text{ uniformly in }\Omega_6\setminus\mathcal N_\delta(\mathbb{Z}'_-),\\
\e^{-z(\log(zt^{-1})-1)\boldsymbol\sigma_3}\boldsymbol C_-(z)^{-1}\boldsymbol \Delta_-(z)\e^{z(\log(zt^{-1})-1)\boldsymbol\sigma_3}&=\boldsymbol {\mathrm{I}}+O(z^{-\infty})&&\text{as }z\to\infty\text{ uniformly in }\Omega_1\setminus\mathcal N_\delta(\mathbb{Z}'_-),
\end{aligned}\end{equation}
for any $\delta>0$.
We infer from Proposition~\ref{prop:asymptotics} that the matrix function $\boldsymbol M:\mathbb{C}\setminus \Sigma_M\to\mathrm{SL}(2,\mathbb{C})$ which equals $(\sqrt{2\pi t})^{\boldsymbol\sigma_3}\boldsymbol M_i$ on $\Omega_i$ is the unique solution to the following Riemann--Hilbert problem.

\begin{cRHp}
\label{cRHp:Mqminus}
Find an analytic function $\boldsymbol M:\mathbb{C}\setminus \Sigma_M\to\mathrm{SL}(2,\mathbb{C})$ such that the following conditions hold true.
\begin{enumerate}[leftmargin=*]
\item Non-tangential boundary values of $\boldsymbol M$ exist and are continuous on $\Sigma_M^\circ$ and satisfy
\begin{equation}
\boldsymbol M_+(z)=\boldsymbol M_-(z)\,\boldsymbol J_M(z),\qquad z\in \Sigma_M^\circ,
\end{equation}
where $\boldsymbol J_M(z)$ is given for $z\in \Sigma_M^\circ$ by 
\begin{equation}
\label{eq:JMqminus}
\boldsymbol J_M(z)=\begin{cases}
\boldsymbol \Delta_\pm(tz)^{\mp 1}&z\in \Gamma_L^\pm,\\
\boldsymbol \nabla_\pm(tz)^{\mp 1}&z\in \Gamma_R^\pm,\\
\boldsymbol \nabla_-(tz)^{-1}\boldsymbol \nabla_+(tz) = \begin{pmatrix}
1 & -2\pi\i\bigl(1-\varsigma(t(z-x))\bigr) \\ 0 & 1
\end{pmatrix}
&z\in ({\wt{u}},+\infty),\\
\boldsymbol \Delta_-(tz)^{-1}\boldsymbol \Delta_+(tz)
&z\in (-\infty,{\wt{u}}),
\\
\boldsymbol \nabla_\pm(tz)^{-1}\boldsymbol \Delta_\pm(tz) &  z\in\gamma_u^\pm.
\end{cases}
\end{equation}
\item We have $\boldsymbol M(z)\e^{tz(\log z-1)\boldsymbol\sigma_3}\to\boldsymbol {\mathrm{I}}$ as $z\to\infty$  uniformly in $\mathbb{C}\setminus \Sigma_M$.
\item We have $\boldsymbol M(z)=O(1)$ as $z\to z_0$ uniformly in $\mathbb{C}\setminus\Sigma_M$ for all $z_0\in\Sigma_M\setminus\Sigma_M^\circ$.
\end{enumerate}
\end{cRHp}

We note from Theorem~\ref{thm:CR} and Lemma~\ref{lemma:asymptoticsbesselandhankel} that
\begin{equation}
\label{eq:refinedexpansionMqminus}
\boldsymbol M(z)=\biggl(\boldsymbol {\mathrm{I}}+\frac 1t\begin{pmatrix}
\alpha+\tfrac 1{24}-t^2 & 2\pi t(\beta+1) \\ \frac{\gamma+t^2}{2\pi t} & -\alpha-\tfrac 1{24}+t^2
\end{pmatrix}z^{-1}+O\bigl(z^{-2}\bigr)\biggr)\,\e^{-tz(\log z-1)\,\boldsymbol\sigma_3}
\end{equation}
as $z\to\infty$ uniformly in $\mathbb{C}\setminus \Sigma_M$.
Here, $\alpha=\alpha(t,xt)$, $\beta=\beta(t,xt)$, and $\gamma=\gamma(t,xt)$ are given in~\eqref{eq:alphabetagamma}.

\subsection{Construction of the \texorpdfstring{$g$}{g}-function}

We will construct the $g$-function ${\wt{g}}(z)$ as a small deformation, when $t$ is large, of the function $g(z)$ employed in Section~\ref{sec:minimizationx<xq} in the context of minimization of the logarithmic energy $\mathcal{E}_{\eta,x}$ for $x\leq x_*$.
To highlight the parallel, we use the same letters for the analogous quantities, with a tilde to distinguish them (and implying the dependence on $t$).

Let us first introduce $\wt V_{\eta,x}(z)$ by
\begin{equation}
\label{eq:vqminus}
\wt V_{\eta,x}(z)=-\frac 1t\log\bigl(1-\varsigma(t(z-x))\bigr)=\frac 1t\log(1+\e^{\eta t (z-x)}).
\end{equation}
We note that
\begin{equation}
\wt V_{\eta,x}'(z)=\frac{\eta}{1+\e^{-\eta t(z-x)}}.
\end{equation}

For any ${\wt{u}}<\wt{v}$, let
\begin{equation}
{\wt{r}}(z) = \sqrt{(z-{\wt{u}})(z-\wt{v})},
\end{equation}
analytic for $z\in\mathbb{C}\setminus[{\wt{u}},\wt{v}]$ and $\sim z$ as $z\to\infty$.

We define
\begin{equation}
\label{eq:gleftprelim}
\begin{aligned}
{\wt{g}}'(z)&={\wt{r}}(z)\left(\int_{-\infty}^{\wt{u}}\frac{\d\nu}{{\wt{r}}(\nu)(\nu-z)}-\frac{1}{2\pi\i}\int_{\wt{u}}^{\wt{v}}\frac{\wt V_{\eta,x}'(\nu)\d\nu}{{\wt{r}}_+(\nu)(\nu-z)}\right)
\\
&=\pm \i\pi-{\wt{r}}(z)\left(\int_{\wt{v}}^{+\infty}\frac{\d\nu}{{\wt{r}}(\nu)(\nu-z)}+\frac{1}{2\pi\i}\int_{\wt{u}}^{\wt{v}}\frac{\wt V_{\eta,x}'(\nu)\d\nu}{{\wt{r}}_+(\nu)(\nu-z)}\right).
\end{aligned}
\end{equation}
(The sign in the second line is determined by~$\pm\Im z>0$.)
We assume ${\wt{u}},\wt{v}\in\mathbb{R}$ with $x<{\wt{u}}<\wt{v}$ and we will shortly determine the values of ${\wt{u}},\wt{v}$.
The function ${\wt{g}}' (z)$ is analytic for $z\in\mathbb{C}\setminus(-\infty,\wt{v}]$ and, by the Sokhotski--Plemelj formulas, the boundary values ${\wt{g}}'_\pm$ from above ($+$) and below ($-$) the real axis exist and are continuous on $(-\infty,{\wt{u}})\cup({\wt{u}},\wt{v})$ and satisfy
\begin{equation}
\begin{aligned}
{\wt{g}}_+'(\mu)-{\wt{g}}_-'(\mu)&=2\pi\i,&&\mu\in (-\infty,{\wt{u}}),\\
{\wt{g}}_+'(\mu)+{\wt{g}}_-'(\mu)&=-\wt V_{\eta,x}'(\mu),&&\mu\in ({\wt{u}} ,\wt{v}).
\end{aligned}
\end{equation}
It is easy to check that we have an alternative expression
\begin{equation}
\label{eq:gprimeqminusexplicit}
{\wt{g}}'(z) =-\frac{{\wt{r}}(z)}{2\pi\i}\int_{\wt{u}}^{\wt{v}} \frac{\wt V_{\eta,x}'(\nu)}{{\wt{r}}_+(\nu)}\frac{\d\nu}{\nu-z} + \log\frac{1+\sqrt{\frac{z-\wt{v}}{z-{\wt{u}}}}}{1-\sqrt{\frac{z-\wt{v}}{z-{\wt{u}}}}}
\end{equation}
where $\sqrt{\frac{z-\wt{v}}{z-{\wt{u}}}}$ (here and below) is analytic for $z\in\mathbb{C}\setminus [{\wt{u}},\wt{v}]$ and $\sim 1$ as $z\to\infty$ and we take the principal branch of the logarithm.
We observe that
\begin{equation}
\label{eq:periodqminus}
\int_{\wt{u}}^{\wt{v}}\bigl({\wt{g}}'_+(\nu)-{\wt{g}}_-'(\nu)\bigr)\d \nu=-2\pi\i {\wt{u}},
\end{equation}
which follows from Cauchy's theorem.

The endpoints ${\wt{u}},\wt{v}$ are fixed, as we are now going to show, by requiring the asymptotic expansion
\begin{equation}
\label{eq:gprimeasympqminus}
{\wt{g}}'(z)=\log z+O(z^{-2}),\quad z\to\infty
\end{equation}
to hold.
Indeed, for arbitrary ${\wt{u}}<\wt{v}$ we have
\begin{equation}
{\wt{g}}'(z)=\log z+{\wt{g}}_{-1}+{\wt{g}}_{0}z^{-1}+{\wt{g}}_1z^{-2}+O(z^{-3}),\quad z\to\infty
\end{equation}
and therefore we want to find ${\wt{u}}<\wt{v}$ such that ${\wt{g}}_{-1}={\wt{g}}_0=0$ (we include here the term of order $z^{-2}$ for later convenience).

\begin{proposition}\label{prop:gapproxendpointsqminus}
    As $t\to+\infty$, uniformly for $x\leq {\wt{u}}-\delta$ (for any $\delta>0$), we have
    \begin{equation}
    {\wt{g}}_{-1}=\log\frac{4\e^{-\eta/2}}{\wt{v}-{\wt{u}}}+O(t^{-\infty}),\
    {\wt{g}}_0=-\frac{{\wt{u}}+\wt{v}}2+O(t^{-\infty}),\
    {\wt{g}}_1=-\frac{3{\wt{u}}^2+2{\wt{u}}\wt{v}+3\wt{v}^2}{16}+O(t^{-\infty}).
    \end{equation}
\end{proposition}
\begin{proof}
We have
\begin{equation}
{\wt{g}}'(z)=-\frac{\eta}2 + \log\frac{1+\sqrt{\frac{z-\wt{v}}{z-{\wt{u}}}}}{1-\sqrt{\frac{z-\wt{v}}{z-{\wt{u}}}}}+\frac{{\wt{r}}(z)}{2\pi\i}\int_{\wt{u}}^{\wt{v}} \frac{\eta-\wt V_{\eta,x}'(\nu)}{{\wt{r}}_+(\nu)}\frac{\d\nu}{\nu-z}.
\end{equation}
The thesis follows from the fact that $\wt V_{\eta,x}'(\nu)=\eta+O(t^{-\infty})$ as $t\to+\infty$, uniformly for $\nu\geq x+\delta$ (for any $\delta>0$).
\end{proof}

It follows from this proposition that the conditions ${\wt{g}}_{-1}={\wt{g}}_0=0$ uniquely determine the endpoints ${\wt{u}},\wt{v}$, provided $t$ is sufficiently large, and, moreover,
\begin{equation}
\label{eq:uvqminus}
{\wt{u}}=-2\e^{-\eta/2}+O(t^{-\infty}),\qquad
\wt{v}=2\e^{-\eta/2}+O(t^{-\infty}),
\end{equation}
as $t\to+\infty$, uniformly for $x\leq x_*$.

From now on we assume that $t$ is sufficiently large and that ${\wt{u}},\wt{v}$ are fixed as just explained.
Under such assumption, in~\eqref{eq:gprimeasympqminus} we have, again thanks to Proposition~\ref{prop:gapproxendpointsqminus},
\begin{equation}
\label{eq:g1qminus}
{\wt{g}}_1=-\e^{-\eta}+O(t^{-\infty})
\end{equation}
as $t\to+\infty$ uniformly for $x\leq x_*$.

Next, we introduce
\begin{equation}
\label{eq:gqminus}
{\wt{g}}(z)=\int_{\wt{v}}^z {\wt{g}}'(y)\,\d y,
\end{equation}
which is analytic for $z\in\mathbb{C}\setminus(-\infty,\wt{v}]$.
The properties of ${\wt{g}}'(z)$ imply that the non-tangential boundary values ${\wt{g}}_\pm(\mu)$ exist for $\mu\in(-\infty,{\wt{u}})\cup({\wt{u}},\wt{v})$ and satisfy
\begin{equation}
\begin{aligned}
{\wt{g}}_+(\mu)-{\wt{g}}_-(\mu)&=2\pi\i \mu,&&\mu\in(-\infty,{\wt{u}}),
\\
{\wt{g}}_+(\mu)+{\wt{g}}_-(\mu)&=-\wt V_{\eta,x}(\mu)+\wt{\ell},&&\mu\in({\wt{u}},\wt{v}),
\end{aligned}
\end{equation}
with
\begin{equation}
\label{eq:ellqminus}
\wt{\ell}=\wt V_{\eta,x}(\wt{v})=\eta(2\e^{-\eta/2}-x)+O(t^{-\infty}).
\end{equation}
Here we also used~\eqref{eq:periodqminus} and the fact that ${\wt{g}}(z)\to 0$ as $z\to\wt{v}$ (by definition of ${\wt{g}}(z)$).

As $z\to\infty$,
\begin{equation}
{\wt{g}}(z)=z(\log z-1)+{\wt{g}}_\infty-{\wt{g}}_1 z^{-1}+O(z^{-2})
\end{equation}
with ${\wt{g}}_\infty$ independent of $z$.
With arguments similar to those in Proposition~\ref{prop:gapproxendpointsqminus} and comparing with~\eqref{eq:frakginftyleft}, we obtain that
\begin{equation}
\label{eq:ginftyqminus}
{\wt{g}}_\infty=\e^{-\eta/2} \eta+O(t^{-\infty}).
\end{equation}

\subsection{Normalization of the continuous Riemann--Hilbert problem}

Let $\Sigma_N=\Sigma_M$ and $\Sigma_N^\circ=\Sigma_M^\circ\setminus\lbrace v\rbrace$, with the same orientation.
Introduce the analytic matrix function $\boldsymbol N:\mathbb{C}\setminus \Sigma_N\to\mathrm{SL}(2,\mathbb{C})$ by
\begin{equation}
\label{eq:defNqminus}
\boldsymbol N(z)=\begin{pmatrix}
-\frac 1{2\pi\i} & 0 \\ 0 & 1
\end{pmatrix}\,\e^{t(\frac {\wt{\ell}} 2-{\wt{g}}_\infty)\boldsymbol\sigma_3}\,\boldsymbol M(z)\,e^{t ({\wt{g}}(z)-\frac {\wt{\ell}} 2)\boldsymbol\sigma_3}\,\begin{pmatrix}
-2\pi\i & 0 \\ 0 & 1
\end{pmatrix}.
\end{equation}
The construction of ${\wt{g}}(z)$ carried out in the previous paragraph ensures that $\boldsymbol N(z)$ is the unique solution to the following Riemann--Hilbert problem.

\begin{cRHp}
\label{cRHp:Nqminus}
Find an analytic function $\boldsymbol N:\mathbb{C}\setminus \Sigma_N\to\mathrm{SL}(2,\mathbb{C})$ such that the following conditions hold true.
\begin{enumerate}[leftmargin=*]
\item Non-tangential boundary values of $\boldsymbol N$ exist and are continuous on $\Sigma_N^\circ$ and satisfy
\begin{equation}
\boldsymbol N_+(z)=\boldsymbol N_-(z)\boldsymbol J_N(z),\qquad z\in \Sigma_N^\circ,
\end{equation}
where $\boldsymbol J_N(z)$ is given for $z\in \Sigma_N^\circ$ by 
\begin{equation}
\label{eq:JNqminus1}
\boldsymbol J_N(z)=
\begin{pmatrix}
-\frac 1{2\pi\i} & 0 \\ 0 & 1
\end{pmatrix}
\e^{-t({\wt{g}}_-(z)-\frac{\wt{\ell}} 2)\boldsymbol\sigma_3}\boldsymbol J_M(z)\e^{t({\wt{g}}_+(z)-\frac{\wt{\ell}} 2)\boldsymbol \sigma_3}
\begin{pmatrix}
-2\pi\i & 0 \\ 0 & 1
\end{pmatrix}
\end{equation}
if $z\in (-\infty,{\wt{u}})\cup({\wt{u}},\wt{v})$ and by
\begin{equation}
\boldsymbol J_N(z)=\begin{pmatrix}
-\frac 1{2\pi\i} & 0 \\ 0 & 1
\end{pmatrix}
\e^{-t({\wt{g}}(z)-\frac{\wt{\ell}} 2)\boldsymbol \sigma_3}\boldsymbol J_M(z)\e^{t({\wt{g}}(z)-\frac{\wt{\ell}} 2)\boldsymbol \sigma_3}
\begin{pmatrix}
-2\pi\i & 0 \\ 0 & 1
\end{pmatrix}
\end{equation}
otherwise.
\item We have $\boldsymbol N(z)\to\boldsymbol {\mathrm{I}}$ as $z\to\infty$  uniformly in $\mathbb{C}\setminus \Sigma_N$.
\item We have $\boldsymbol N(z)=O(1)$ as $z\to z_0$ uniformly in $\mathbb{C}\setminus\Sigma_N$ for all $z_0\in\Sigma_N\setminus\Sigma_N^\circ$.
\end{enumerate}
\end{cRHp}

We note from~\eqref{eq:refinedexpansionMqminus} and~\eqref{eq:defNqminus} that
\begin{equation}
\label{eq:refinedexpansionNqminus}
\boldsymbol N(z)=\boldsymbol {\mathrm{I}}+\begin{pmatrix}
\wh\alpha-t(1+{\wt{g}}_1)+\tfrac 1{24t} & \i\e^{t(\wt{\ell}-2{\wt{g}}_\infty)}\wh\beta \\ -\i\e^{t(2{\wt{g}}_\infty-\wt{\ell})}\wh\gamma & -\wh\alpha+t(1+{\wt{g}}_1)-\tfrac 1{24t}
\end{pmatrix}z^{-1}+O\bigl(z^{-2}\bigr)
\end{equation}
as $z\to\infty$ uniformly in $\mathbb{C}\setminus \Sigma_N$.
Here, $\wh\alpha=\wh\alpha(t,x)$, $\wh\beta=\wh\beta(t,x)$, and $\wh\gamma=\wh\gamma(t,x)$ are
\begin{equation}
\label{eq:wtalphabetagamma}
\begin{aligned}
\wh\alpha(t,x)&=\frac 1t\alpha(t,xt)=-\frac 12\,\partial_t\log Q(t,s)\bigr|_{s=xt}, \\
\wh\beta(t,x)&=\beta(t,xt)+1=\frac{Q(t,xt-1)}{Q(t,xt)},\\
\wh\gamma(t,x)&= \frac 1{t^2}\bigl(\gamma(t,xt)+t^2\bigr)=\frac{Q(t,xt+1)}{Q(t,xt)}.
\end{aligned}
\end{equation}

To write down the jump matrix $\boldsymbol J_N(z)$ in a more explicit way, it is convenient to introduce
\begin{equation}
\label{eq:varphiqminus}
\varphi(z) = 2{\wt{g}}(z)+\wt V_{\eta,x}(z)-\wt{\ell}
\end{equation}
as well as
\begin{equation}
\varphi_1(z) = \varphi(z)\mp 2\pi \i z,\qquad \pm\Im z>0.
\end{equation}

\begin{proposition}\label{prop:varphiqminus}
The following properties hold true, for $t$ sufficiently large.
\begin{enumerate}[leftmargin=*]
    \item The function $\varphi(z)$ is analytic for $z\in\mathbb{C}\setminus\bigl((-\infty,\wt{v}]\cup(\i \mathbb{R}+x)\bigr)$.
    It has non-tangential boundary values $\varphi_\pm(\mu)$ for all $\mu\in(-\infty,{\wt{u}})\cup({\wt{u}},\wt{v})$ such that
    \begin{align}
        \label{eq:jumpphiqminus1}
        \varphi_\pm(\mu)&=\pm({\wt{g}}_+(\mu)-{\wt{g}}_-(\mu)\bigr),&&\mu\in({\wt{u}},\wt{v}),
        \\
        \label{eq:jumpphiqminus2}
        \varphi_+(\mu)-\varphi_-(\mu)&= 4\pi\i\,\mu,&&\mu\in(-\infty,{\wt{u}}).
    \end{align}
    \item There exist a neighborhood of $z=\wt{v}$ and a function $\varphi_{\wt{v}}(z)$ analytic in that neighborhood such that $\varphi(z)=(z-\wt{v})^{3/2}\varphi_{\wt{v}}(z)$ (with principal branch) and that 
    \begin{equation}
    \label{eq:varphivvaluev}
    \begin{aligned}
    \varphi_{\wt{v}}(\wt{v})&=\frac{8}{3\,\sqrt{\wt{v}-{\wt{u}}}}=\frac 43 \e^{\frac 14\eta}+O(t^{-\infty}),
    \\
    \varphi_{\wt{v}}'(\wt{v})&=-\frac{4}{15\,\sqrt{\wt{v}-{\wt{u}}}^3}=-\frac 1{30} \e^{\frac 34\eta}+O(t^{-\infty}).
    \end{aligned}
    \end{equation}
    The neighborhood can be chosen independent of $t$ and $x$, provided $x\leq x_*-\delta$ (for some $\delta>0$).
    \item The function $\varphi_1(z)$ is analytic for $z\in\mathbb{C}\setminus\bigl([{\wt{u}},+\infty)\cup(\i\mathbb{R}+x)\bigr)$.
    There exist a neighborhood of $z={\wt{u}}$ and a function $\varphi_{\wt{u}}(z)$ analytic in that neighborhood such that $\varphi_1(z) = -({\wt{u}}-z)^{3/2}\varphi_{\wt{u}}(z)$ (with principal branch) and that 
    \begin{equation}
    \label{eq:varphiuvalueu}
    \begin{aligned}
    \varphi_{\wt{u}}({\wt{u}})&=\frac{8}{3\,\sqrt{\wt{v}-{\wt{u}}}}=\frac 43 \e^{\frac 14\eta}+O(t^{-\infty}).
    \\
    \varphi_{\wt{u}}'({\wt{u}})&=\frac{4}{15\,\sqrt{\wt{v}-{\wt{u}}}^3}=\frac 1{30} \e^{\frac 34\eta}+O(t^{-\infty}).
    \end{aligned}
    \end{equation}
    The neighborhood can be chosen independent of $t$ and $x$, provided $x\leq x_*-\delta$ (for some $\delta>0$).
\end{enumerate}
\end{proposition}
\begin{proof}
These properties are simple consequences of the definition so we only comment on the proof of the statement about the local structure of $\varphi$ near $\wt{v}$.

First, by~\eqref{eq:gprimeqminusexplicit}
\begin{equation}
\label{eq:varphiprimeqminusexplicit}
\varphi'(z)=2{\wt{g}}'(z)+\wt V_{\eta,x}'(z) = -\frac{{\wt{r}}(z)}{\i\pi}\int_{\wt{u}}^{\wt{v}}\frac{\wt V_{\eta,x}'(\mu)-\wt V_{\eta,x}'(z)}{{\wt{r}}_+(\mu)(\mu-z)}\d\mu+2\log\frac{1+\sqrt{\frac{z-\wt{v}}{z-{\wt{u}}}}}{1-\sqrt{\frac{z-\wt{v}}{z-{\wt{u}}}}}
\end{equation}
and so $\varphi'(z)$ equals $(z-\wt{v})^{1/2}$ times a function of $z$ analytic in a neighborhood of $z=\wt{v}$.
By integrating in $z$, noting that $\varphi(z)\to 0$ as $z\to \wt{v}$, we get $\varphi(z) = (z-\wt{v})^{3/2}\varphi_{\wt{v}}(z)$ for a function $\varphi_{\wt{v}}(z)$ analytic for $z$ in a neighborhood of $\wt{v}$.
Moreover, since $\wt V_{\eta,x}'(z)=\eta+O(t^{-\infty})$ uniformly for $\Re z>x+\delta$ for any $\delta>0$, we can rewrite~\eqref{eq:varphiprimeqminusexplicit} as
\begin{equation}
\label{eq:varphiprimeqminusapprox}
\varphi'(z)=2\log\frac{1+\sqrt{\frac{z-\wt{v}}{z-{\wt{u}}}}}{1-\sqrt{\frac{z-\wt{v}}{z-{\wt{u}}}}}+O(t^{-\infty}\sqrt{z-\wt{v}})
\end{equation}
uniformly for $z$ is in a fixed neighborhood of $\wt{v}$.
The conclusion then follows easily, also using~\eqref{eq:uvqminus}.

For the statement about the local structure of $\varphi_1$ near ${\wt{u}}$ we use a completely similar argument, just using the second line of~\eqref{eq:gleftprelim} in place of the first, and so we omit the details.
\end{proof}

We can use these functions to write the jump matrix $\boldsymbol J_N(z)$ as
\begin{equation}
\label{eq:JNqminusallcases}
\boldsymbol J_N(z)=\begin{cases}
\begin{pmatrix}
(1+\e^{\pm 2\pi\i t z})^{\pm 1} & 0 \\ \e^{t\varphi(z)}& (1+\e^{\pm 2\pi\i t z})^{\mp 1}
\end{pmatrix}
&z\in\Gamma_L^\pm,\\
\begin{pmatrix}
1 & - \frac{\e^{-t\varphi(z)}}{1+\e^{\mp 2\pi\i tz}} \\ 0 & 1
\end{pmatrix}
&z\in\Gamma_R^\pm,\\
\begin{pmatrix}
\e^{t\varphi_+(z)} & 1 \\ 0 & \e^{t\varphi_-(z)}
\end{pmatrix}
&z\in ({\wt{u}},\wt{v}),\\
\begin{pmatrix}
1 & \e^{-t\varphi(z)} \\ 0 & 1
\end{pmatrix}
&z\in (\wt{v},+\infty),\\
\begin{pmatrix}
	1 & 0 \\
	-\e^{t\varphi_1(z)} & 1
\end{pmatrix}
&z\in (-\infty,{\wt{u}}),
\\
\begin{pmatrix}
	1 & \mp \e^{-t\varphi_1(z) } \\
	\mp \e^{t\varphi(z)} & 1 + \e^{\pm 2 \pi \i t z}
\end{pmatrix} & z\in\gamma_u^\pm.
\end{cases}
\end{equation}

\subsection{Lens opening}
We now fix the contours $\gamma_u^\pm$ as the loci where $\Im \varphi_1(z)=0$, $0\leq \pm\Im z\leq \epsilon$, see~Proposition~\ref{prop:varphiqminus}.
In particular, $\varphi_1(z)>0$ on these contours.
We introduce contours $\gamma_v^\pm$ (starting at $\wt{v}$ and ending on the lines $\Im z=\pm\epsilon$) as the loci where $\Im \varphi(z)=0$, $0\leq \pm\Im z\leq \epsilon$.
In particular, $\varphi(z)<0$ on these contours. 
We define
\begin{equation}
\Sigma_T = \Sigma_N\cup\gamma_v^+\cup\gamma_v^-.
\end{equation}
Denoting $\wt{v}^\pm$ the intersection points of $\gamma_v^\pm$ with $\Im z=\pm\epsilon$, we set $\Sigma_T^\circ=\Sigma_T\setminus\lbrace {\wt{u}},{\wt{u}}^+,{\wt{u}}^-,\wt{v},\wt{v}^+,\wt{v}^-\rbrace$.
We orient $\Sigma_T^\circ$ as $\Sigma_N^\circ$, with the additional curves also oriented upwards.
This is illustrated in Figure~\ref{fig:SigmaTqminus}; in particular the curves $\gamma_u^\pm,\gamma_v^\pm,\lbrace\Im z=\pm\epsilon\rbrace,\lbrace\Im z=0\rbrace$ delimit bounded regions $\mathscr{L}_\pm$ (the ``lenses'' in Riemann--Hilbert jargon) and we define
\begin{equation}
\boldsymbol T(z)=
\begin{cases}
\boldsymbol N(z)\begin{pmatrix}
1 & 0 \\ \mp\e^{t\varphi(z)} & 1
\end{pmatrix} & \text{if }z\in\mathscr{L}_\pm,\\
\boldsymbol N(z) & \text{otherwise}.
\end{cases}
\end{equation}

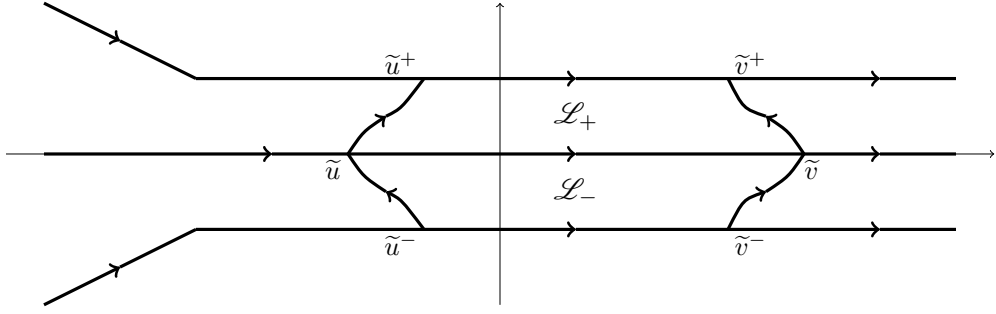
\begin{figure}[t]
\centering
\begin{tikzpicture}

\draw[very thin,->] (-6.5,0) -- (6.5,0);
\draw[very thin,->] (0,-2) -- (0,2);

\draw[very thick,->] (-6,0) -- (-3,0);
\draw[very thick,->] (-3,0) -- (1,0);
\draw[very thick,->] (1,0) -- (5,0) ;
\draw[very thick] (5,0) -- (6,0);

\draw[very thick] (-4,1) -- (-5,3/2);
\draw[very thick,<-] (-5,3/2) -- (-6,2);

\draw[very thick] (-4,-1) -- (-5,-3/2);
\draw[very thick,<-] (-5,-3/2) -- (-6,-2);

\draw[very thick,->] (-4,1) -- (1,1);
\draw[very thick,->] (1,1) -- (5,1);
\draw[very thick] (5,1) -- (6,1);

\draw[very thick,->] (-4,-1) -- (1,-1);
\draw[very thick,->] (1,-1) -- (5,-1);
\draw[very thick] (5,-1) -- (6,-1);

\draw[very thick,->] 
  (-1,-1) .. controls (-1.3,-0.6) .. (-3/2,-1/2);
\draw[very thick] 
  (-3/2,-1/2) .. controls (-1.8,-0.3) .. (-2,0);
\draw[very thick,->] 
  (-2,0) .. controls (-1.8,0.3) .. (-3/2,1/2);
\draw[very thick] 
  (-3/2,1/2) .. controls (-1.3,0.6) .. (-1,1);

\node at (-2.2,-1/5) {\small{${\wt{u}}$}};
\node at (-1.3,6/5) {\small{${\wt{u}}^+$}};
\node at (-1.3,-6/5) {\small{${\wt{u}}^-$}};

\draw[very thick,->] 
  (3,-1) .. controls (3.2,-0.6) .. (7/2,-1/2);
\draw[very thick] 
  (7/2,-1/2) .. controls (3.8,-0.3) .. (4,0);
\draw[very thick,->] 
  (4,0) .. controls (3.8,0.3) .. (7/2,1/2);
\draw[very thick] 
  (7/2,1/2) .. controls (3.2,0.6) .. (3,1);

\node at (4.1,-1/5) {\small{$\wt{v}$}};
\node at (3.3,6/5) {\small{$\wt{v}^+$}};
\node at (3.3,-6/5) {\small{$\wt{v}^-$}};

\node at (1,1/2) {$\mathscr{L}_+$};
\node at (1,-1/2) {$\mathscr{L}_-$};

\end{tikzpicture}
\caption{$\Sigma_{T}$, its orientation, and the ``lenses'' $\mathscr{L}_\pm$ (case $x<x_*$).}
\label{fig:SigmaTqminus}
\end{figure}

It is important to note that we can perform this transformation because $\varphi(z)$ is analytic for $z\in\mathscr{L}_\pm$, see~Proposition~\ref{prop:varphiqminus}.

It is clear that $\boldsymbol T$ solves the following Riemann--Hilbert problem.

\begin{cRHp}
Find an analytic function $\boldsymbol T:\mathbb{C}\setminus \Sigma_T\to\mathrm{SL}(2,\mathbb{C})$ such that the following conditions hold true.
\begin{enumerate}[leftmargin=*]
\item Non-tangential boundary values of $\boldsymbol T$ exist and are continuous on $\Sigma_T^\circ$ and satisfy
\begin{equation}
\boldsymbol T_+(z)=\boldsymbol T_-(z)\boldsymbol J_T(z),\qquad z\in \Sigma_T^\circ,
\end{equation}
where $\boldsymbol J_T(z)$ is given explicitly below.
\item We have $\boldsymbol T(z)\to\boldsymbol {\mathrm{I}}$ as $z\to\infty$  uniformly in $\mathbb{C}\setminus \Sigma_T$.
\item We have $\boldsymbol T(z)=O(1)$ as $z\to z_0$ uniformly in $\mathbb{C}\setminus\Sigma_T$ for all $z_0\in\Sigma_T\setminus\Sigma_T^\circ$.
\end{enumerate}
\end{cRHp}

It follows from the factorization
\begin{equation}
\begin{pmatrix}
\e^{t\varphi_+(z)} & 1 \\ 0 & \e^{t\varphi_-(z)}
\end{pmatrix} =
\begin{pmatrix}
1 & 0 \\ \e^{t\varphi_-(z)} & 1
\end{pmatrix}
\begin{pmatrix}
0 & 1 \\ -1 & 0
\end{pmatrix}
\begin{pmatrix}
1 & 0 \\ \e^{t\varphi_+(z)} & 1
\end{pmatrix}
\end{equation}
that the jump matrix $\boldsymbol J_T(z)$ for $z\in\Sigma_T^\circ$ is given explicitly by
\begin{equation}
\label{eq:JTqminus}
\boldsymbol J_T(z)=\begin{cases}
\begin{pmatrix}
0 & 1\\ -1 & 0
\end{pmatrix}
&\text{if }z\in ({\wt{u}},\wt{v})
\\
\begin{pmatrix}
1 & - \frac{\e^{-t\varphi(z)}}{1+\e^{\mp 2\pi\i tz}} \\ \e^{t\varphi(z)} & \frac 1{1+\e^{\pm2\pi\i t z}}
\end{pmatrix}
&\text{if }z\in ({\wt{u}}^\pm,\wt{v}^\pm)
\\
\begin{pmatrix}
1 & \mp \e^{-t\varphi_1(z)} \\ 0 & 1
\end{pmatrix}
&\text{if } z\in \gamma_u^\pm,
\\
\begin{pmatrix}
1 & 0 \\ \mp \e^{t\varphi(z)} & 1
\end{pmatrix}
&\text{if } z\in \gamma_v^\pm,
\\
\boldsymbol J_N(z)
&\text{otherwise}.
\end{cases}
\end{equation}

The key point of this construction for $x<x_*$ is that the jump matrices $\boldsymbol J_T(z)$ are exponentially close to the identity when $t$ is large, except on $({\wt{u}},\wt{v})$ and in (arbitrarily small) neighborhoods of ${\wt{u}}$ and $\wt{v}$.
Namely, let
\begin{equation}
\label{eq:SigmaTepsilonqminus}
\wt{\Sigma_T}^\epsilon = \Sigma_T^\circ\setminus\bigl(({\wt{u}}-\epsilon,\wt{v}+\epsilon)\cup\gamma_u^+\cup\gamma_u^-\cup\gamma_v^+\cup\gamma_v^-\bigr).
\end{equation}

\begin{proposition}
    \label{prop:JTsmallqminus}
    For any $\epsilon>0$ small enough and for any $\delta>0$, there exists $c>0$ such that
    $\boldsymbol J_T(z)=\boldsymbol {\mathrm{I}}+O\left(\frac{1}{|z|^2+1}\e^{-ct}\right)$ as $t\to+\infty$ uniformly for $z\in \wt{\Sigma_T}^\epsilon$ and for $x\leq x_*-\delta$.
\end{proposition}
\begin{proof}
We start by making some general observations.
First,
\begin{equation}
\label{eq:JTsmallqminus ineq1}
\Re \wt V_{\eta,x}(z)\leq\frac{\log 2}t+\eta[\Re z-x]_+,\quad\text{ for all }z\in\mathbb{C}\text{ and }x\in\mathbb{R},
\end{equation}
and
\begin{equation}
\label{eq:JTsmallqminus ineq2}
\Re \wt V_{\eta,x}(z)=\eta[\Re z-x]_++O(t^{-\infty}),\quad\text{ for all }z\in\mathbb{C}\text{ such that }|\Re z-x|\geq \delta,
\end{equation}
uniformly for any $\delta>0$.

Second, letting $g(z)$ the function constructed in Section~\ref{sec:minimizationx<xq}, see~\eqref{eq:frakgleft}, ${\wt{g}}(z)=g(z)+O(t^{-\infty})$ uniformly for $z\in\mathbb{C}\setminus(-\infty,\wt{v}+\epsilon]$, and, similarly,
${\wt{g}}_\pm(z)=g_\pm(z)+O(t^{-\infty})$ uniformly for $z<{\wt{u}}-\epsilon$, for any $\epsilon>0$, and for $x\leq x_*$.

Third, as it follows from~\eqref{eq:varphiprimeqminusexplicit},
\begin{align}
\label{eq:JTsmallqminus ineq3}
\varphi(z)&\sim 2z\log z,\quad\text{ as }|\Re z|\to\infty,\,z\notin\mathbb{R}_-,
\\
\label{eq:JTsmallqminus ineq4}
\varphi_\pm(z)&\sim 2z\log|z|,\quad\text{ as }z\to-\infty,
\end{align}
which are also uniform for $x\leq x_*$.

Next, we reason separately for each component of $\wt{\Sigma_T}^\epsilon$.

    When $z\in(\wt{v}+\epsilon,+\infty)$ we have $\boldsymbol J_T-\boldsymbol {\mathrm{I}}=O(\e^{-t\varphi(z)})$.
    By the observations above, in particular ~\eqref{eq:JTsmallqminus ineq2} and~\eqref{eq:ellqminus}, we have
    \begin{equation}
        \varphi(z)= 2{\wt{g}}(z)+\wt V_{\eta,x}(z)-\wt{\ell} =2g(z)+\eta[\Re z-x]_++\eta(2\e^{-\eta/2}-x)+O(t^{-\infty})
    \end{equation}
    hence it suffices to recall Remark~\ref{remark:strictineqqminus}, taking for example $c=k/2>0$, to obtain $\varphi(z)>c$ for all $z>\wt{v}+\epsilon$.
    Combining with~\eqref{eq:JTsmallqminus ineq3} it is easy to conclude that $\e^{-t\varphi(z)}=O\bigl(\e^{-c t}/(|z|^2+1)\bigr)$.
    
    When $z\in(-\infty,{\wt{u}}-\epsilon)$ we have $\boldsymbol J_T-\boldsymbol {\mathrm{I}}=O(\e^{t\varphi_1(z)})$.
    By the observations above, in particular~\eqref{eq:JTsmallqminus ineq2} and~\eqref{eq:ellqminus}, we have
    \begin{equation}
    \begin{aligned}
    \varphi_1(z)=\varphi_+(z)-2\pi\i z &= {\wt{g}}_+(z)+{\wt{g}}_-(z)+\wt V_{\eta,x}(z)-\wt{\ell} \\
    &=g_+(z)+g_-(z)+\eta[z-x]_+-\eta(2\e^{-\eta/2}-x) +O(t^{-\infty}),
    \end{aligned}
    \end{equation}
    if $|z-x|\geq \delta$ for a fixed $\delta>0$, and, by~\eqref{eq:JTsmallqminus ineq1},
    \begin{equation}
    \begin{aligned}
    \varphi_1(z) &\leq \frac{\log 2}t+{\wt{g}}_+(z)+{\wt{g}}_-(z)+\eta[z-x]_+-\wt{\ell} \\
    &=\frac{\log 2}t+g_+(z)+g_-(z)+\eta[z-x]_+-\eta(2\e^{-\eta/2}-x)+O(t^{-\infty}),
    \end{aligned}
    \end{equation}
    if $|z-x|\leq \delta$.
    Hence, it suffices to recall Remark~\ref{remark:strictineqqminus}, taking for example $c=k/2>0$, to obtain $\varphi_1(z)<-c$ for all $z<{\wt{u}}-\epsilon$.
    Combining with~\eqref{eq:JTsmallqminus ineq4} it is easy to conclude that $\e^{t\varphi_1(z)}=O\bigl(\e^{-ct}/(|z|^2+1)\bigr)$.
    
    When $z\in \Gamma_R^\pm$ we have $\boldsymbol J_T-\boldsymbol {\mathrm{I}}=O(\e^{-t(\varphi(z)+2\pi\epsilon)})$.
    First, when $z$ is in a neighborhood of $z=\wt{v}$ such that $\varphi(z) = (z-\wt{v})^{3/2}\varphi_{\wt{v}}(z)$ as in Proposition~\ref{prop:varphiqminus}, we see that along the line $\Im z=\pm\epsilon$, to the right of $\wt{v}^\pm$, we have $\Re\varphi(z)>\Re\varphi(\wt{v}^\pm)$.
    By the same proposition and simple estimates, we have $\Re\varphi(\wt{v}^\pm)>-\frac C2 \epsilon^{3/2}$ (for some $C>0$) and so (taking $\epsilon$ small enough) we can make $\Re\varphi(z)+2\pi\epsilon\geq c>0$ for $z$ inside this neighborhood, for some $c>0$.
    When $z$ is outside, the desired bound follows, by continuity, from the one we already proved on $(\wt{v}+\epsilon,+\infty)$, also using~\eqref{eq:JTsmallqminus ineq3}.    
    
    When $z\in\Gamma_L^\pm$ we have $\boldsymbol J_T-\boldsymbol {\mathrm{I}}=O(\e^{t\varphi(z)})=O(\e^{t(\varphi_1(z)-2\pi\epsilon)})$ (we are ignoring the diagonal entries which are easily bounded).
    The argument is completely similar to that on $\Gamma_R^\pm$.
    Namely, when $z$ is in a neighborhood of $z={\wt{u}}$ such that $\varphi_1(z) =-({\wt{u}}-z)^{3/2}\varphi_{\wt{u}}(z)$ as in Proposition~\ref{prop:varphiqminus}, we see that along the line $\Im z=\pm\epsilon$, to the left of ${\wt{u}}^\pm$, we have $\Re\varphi_1(z)<\Re\varphi_1({\wt{u}}^\pm)$.
    By the same proposition and simple estimates, we have $\Re\varphi_1({\wt{u}}^\pm)<2C \epsilon^{3/2}$ and so (taking $\epsilon$ small enough) we can make $\Re\varphi_1(z)-2\pi\epsilon\leq-c<0$ for $z$ inside this neighborhood, for some $c>0$.
    When $z$ is outside, the desired bound follows, by continuity, from the one we already proved on $(-\infty,{\wt{u}}-\epsilon)$, also using~\eqref{eq:JTsmallqminus ineq4}. (Here we can finally fix $R_0$ appearing in the definition of $\Gamma_L^\pm$, namely in the diagonal part of the contour we can simply use the behavior at $\infty$ of $\varphi_1$, see~\eqref{eq:JTsmallqminus ineq4}, and in the horizontal part of the contour we can use the continuity argument.)
    
    When $z\in ({\wt{u}}^\pm,\wt{v}^\pm)$ we have $\boldsymbol J_T-\boldsymbol {\mathrm{I}}=O(\e^{-t(\varphi(z)+2\pi\epsilon)}+\e^{t\varphi(z)}+\e^{-2\pi\epsilon t})$.
    Therefore we are done if we can show that, with $\epsilon>0$ sufficiently small, we can achieve $c<-\varphi(\mu\pm\i\epsilon)<2\pi\epsilon-c$ for some $c>0$, for all $\mu\in(\Re {\wt{u}}^\pm,\Re \wt{v}^\pm)$.
    This can be established using a standard argument from the nonlinear steepest descent asymptotic analysis of Riemann--Hilbert problems.
    Namely, we first observe that $(\Re {\wt{u}}^\pm,\Re \wt{v}^\pm)$ is strictly contained in $({\wt{u}},\wt{v})$ (by the local structure of $\varphi$ and $\varphi_1$ near $z=\wt{v}$ and $z={\wt{u}}$ established in Proposition~\ref{prop:varphiqminus}).
    Hence, for any $\mu\in (\Re {\wt{u}}^\pm,\Re \wt{v}^\pm)$ we have
    \begin{equation}
    \frac{\partial}{\partial y}\Re\varphi(\mu\pm\i y)\biggr|_{y=0}=
    \mp\frac{\partial}{\partial\mu}\Im \varphi_\pm(\mu)=
    -2\arccos\frac{\mu}{\wt{v}-{\wt{u}}}+O(t^{-\infty})
    \end{equation}
    where we first use the Cauchy--Riemann equations and then~\eqref{eq:varphiprimeqminusapprox}.
    Therefore, for $\epsilon>0$ small enough, recalling that $(\Re {\wt{u}}^\pm,\Re \wt{v}^\pm)$ is strictly contained in $({\wt{u}},\wt{v})$, we can achieve
    \begin{equation}
    -2\pi+\wt c<\frac{\d}{\d y}\Re\varphi(\mu\pm\i y)<-\wt c
    \end{equation}
    for all $0<\pm y<\epsilon$ and some $\wt c>0$.
    Integrating this expression for $y$ between $0$ and $\epsilon$, using the fact that $\Re\varphi_\pm =0$ on $({\wt{u}},\wt{v})$, we obtain the desired inequality with $c=\wt c\epsilon$.
\end{proof}

\subsection{Parametrices}
\label{sec:parametricesqminus}
The next step in the asymptotic analysis of Riemann--Hilbert problems is the construction of explicit approximations to $\boldsymbol T$ (called \emph{parametrices}).
Their construction takes advantage of the fact just proved (Proposition~\ref{prop:JTsmallqminus}) that the jump matrix $\boldsymbol J_T$ becomes close to the identity as $t\to+\infty$ everywhere except on~$({\wt{u}},\wt{v})$ and in small (but fixed) neighborhoods of~${\wt{u}}$ and~$\wt{v}$.
Accordingly, we will construct the \emph{outer parametrix} (an approximation to~$\boldsymbol T$ valid away from ${\wt{u}}$ and $\wt{v}$) and the \emph{inner parametrices} (approximations to $\boldsymbol T$ near ${\wt{u}}$ and $\wt{v}$).
The parametrices are constructed following standard procedures of the nonlinear steepest descent method.

\subsubsection{Outer parametrix}\label{sec:outerparaqminus}

The outer parametrix is obtained by neglecting all jumps of $\boldsymbol T$ except on $({\wt{u}},\wt{v})$, which corresponds to the following Riemann--Hilbert problem.

\begin{cRHp}
\label{cRHp:Poutqminus}
Find an analytic function $\boldsymbol P^{\mathrm{out}}:\mathbb{C}\setminus[{\wt{u}},\wt{v}]\to\mathrm{SL}(2,\mathbb{C})$ such that the following conditions hold true.
\begin{enumerate}[leftmargin=*]
\item Non-tangential boundary values of $\boldsymbol P^{\mathrm{out}}$ exist and are continuous on $({\wt{u}},\wt{v})$ and satisfy
\begin{equation}
\boldsymbol P_+^{\mathrm{out}}(z)=\boldsymbol P_-^{\mathrm{out}}(z)\,\begin{pmatrix} 0 & 1 \\ -1 & 0 \end{pmatrix},\qquad z\in ({\wt{u}},\wt{v}).
\end{equation}
\item We have $\boldsymbol P^{\mathrm{out}}(z)\to\boldsymbol {\mathrm{I}}$ as $z\to\infty$  uniformly in~$\mathbb{C}$.
\item We have $\boldsymbol P^{\mathrm{out}}(z)=O\bigl(|z-z_0|^{-1/4}\bigr)$ as $z\to z_0$ for $z_0\in\lbrace{\wt{u}},\wt{v}\rbrace$.
\end{enumerate}
\end{cRHp}

It is well known (e.g., see~\cite{ItsLargeN}) that the unique solution is
\begin{equation}
\label{eq:globalparaqminus}
\boldsymbol P^{{\mathrm{out}}}(z) = \boldsymbol G \left(\frac{z-\wt{v}}{z-{\wt{u}}}\right)^{\frac 14\boldsymbol \sigma_3}\boldsymbol G^{-1},\qquad\boldsymbol G=\frac 1{\sqrt 2}\begin{pmatrix}
1 & \i \\ \i & 1
\end{pmatrix}.
\end{equation}
We note that
\begin{equation}
\label{eq:Poutlargezqminus}
\boldsymbol P^{\mathrm{out}}(z)=
\boldsymbol {\mathrm{I}}+z^{-1}\begin{pmatrix}
0 & \frac \i 4(\wt{v}-{\wt{u}}) \\ \frac \i 4({\wt{u}}-\wt{v}) & 0
\end{pmatrix}+O\bigl(z^{-2}\bigr)
\end{equation}
as $z\to\infty$ uniformly in $\mathbb{C}$.

\subsubsection{Inner Airy parametrices}

By Proposition~\ref{prop:varphiqminus}, the maps
\begin{equation}
z\mapsto\zeta_{\wt{u}}(z)=\left(\frac 34t\varphi_{\wt{u}}(z)\right)^{\frac 23}(z-{\wt{u}}),\qquad
z\mapsto\zeta_{\wt{v}}(z)=\left(\frac 34t\varphi_{\wt{v}}(z)\right)^{\frac 23}(z-\wt{v}),
\end{equation}
are conformal (injective) mappings of neighborhoods $\mathcal{Q}_{\wt{u}}$ of $z={\wt{u}}$ and $\mathcal{Q}_{\wt{v}}$ of $z=\wt{v}$ (respectively).
We may safely assume that 
\begin{equation}
\text{$\mathcal{Q}_{\wt{u}}$ and $\mathcal{Q}_{\wt{v}}$ are (open) squares of side length $2\epsilon$ centered at ${\wt{u}}$ and $\wt{v}$ respectively}
\end{equation}
and that these conformal mappings extend to open neighborhoods of the closures $\overline {\mathcal{Q}_{{\wt{u}}}}$ and $\overline {\mathcal{Q}_{\wt{v}}}$.
We also know from Proposition~\ref{prop:varphiqminus} that we may take $\epsilon$ independent of $x,t$, as long as $x\leq x_*$ and $t$ is sufficiently large.
Moreover, also by Proposition~\ref{prop:varphiqminus},
\begin{equation}
\begin{aligned}
\zeta_{\wt{u}}'(z)\big|_{z={\wt{u}}} &= \left(\frac 34t\varphi_{\wt{u}}({\wt{u}})\right)^{\frac 23}= \e^{-\frac 16\eta}t^{\frac 23}+O(t^{-\infty}),
\\
\zeta_{\wt{v}}'(z)\big|_{z=\wt{v}} &= \left(\frac 34t\varphi_{\wt{v}}(\wt{v})\right)^{\frac 23}=\e^{-\frac 16\eta}t^{\frac 23}+O(t^{-\infty}),
\end{aligned}
\end{equation}
which implies that (possibly taking $\epsilon$ small enough), for some $C>0$ (independent of $x$),
\begin{equation}
\label{eq:conformalboundaryexpands}
\left|\zeta_{\wt{u}}(z)\right|\geq Ct^{\frac 23}\text{  when }z\in\partial \mathcal{Q}_{\wt{u}},
\qquad
\left|\zeta_{\wt{v}}(z)\right| \geq Ct^{\frac 23}\text{  when }z\in\partial \mathcal{Q}_{\wt{v}}.
\end{equation}
Let us also note that $\zeta_{\wt{u}}$ and $\zeta_{\wt{v}}$ map the real line into the real line, that $\zeta_{\wt{u}}$ maps $\gamma_u^\pm$ into the half-line emanating at the origin with argument $\pm\pi/3$, and that $\zeta_{\wt{v}}$ maps $\gamma_v^\pm$ into the half-line emanating at the origin with argument $\pm 2\pi/3$.

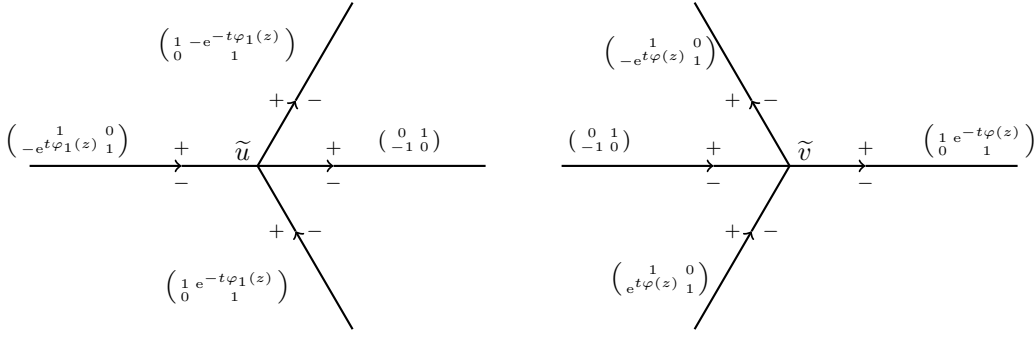
\begin{figure}[t]
\centering
\begin{tikzpicture}
\begin{scope}[shift={(3.5,0)}]

\draw[->,thick] (-3,0) -- (-1,0) node[below]{\tiny{$-$}} node[above]{\tiny{$+$}};
\draw[thick] (-1,0) -- (0,0);
\draw[->,thick] (0,0) -- (1,0) node[below]{\tiny{$-$}} node[above]{\tiny{$+$}};
\draw[thick] (1,0) -- (3,0);

\begin{scope}[rotate=-30]
\draw[->,thick] (0,-2.5) -- (0,-1) node[right]{\tiny{$-$}} node[left]{\tiny{$+$}};
\draw[thick] (0,-1) -- (0,0);
\end{scope}
\begin{scope}[rotate=30]
\draw[->,thick] (0,0) -- (0,1) node[right]{\tiny{$-$}} node[left]{\tiny{$+$}};
\draw[thick] (0,1) -- (0,2.5);
\end{scope}

\node at (1/5,1/5) {\small{$\wt{v}$}};

\node at (-2.5,1/3) {\tiny{$\left(\begin{smallmatrix}0 & 1\\ -1 & 0 \end{smallmatrix}\right)$}};
\node at (2.5,1/3) {\tiny{$\left(\begin{smallmatrix}1 & \e^{-t\varphi(z)}\\ 0 & 1 \end{smallmatrix}\right)$}};
\node at (-1.7,1.5) {\tiny{$\left(\begin{smallmatrix}1 & 0\\ -\e^{t\varphi(z)} & 1 \end{smallmatrix}\right)$}};
\node at (-1.7,-1.5) {\tiny{$\left(\begin{smallmatrix} 1 & 0\\ \e^{t\varphi(z)} & 1 \end{smallmatrix}\right)$}};
\end{scope}

\begin{scope}[shift={(-3.5,0)}]
\draw[->,thick] (-3,0) -- (-1,0) node[below]{\tiny{$-$}} node[above]{\tiny{$+$}};
\draw[->,thick] (-1,0) -- (1,0) node[below]{\tiny{$-$}} node[above]{\tiny{$+$}};
\draw[thick] (1,0) -- (3,0);

\begin{scope}[rotate=30]
\draw[->,thick] (0,-2.5) -- (0,-1) node[right]{\tiny{$-$}} node[left]{\tiny{$+$}};
\draw[thick] (0,-1) -- (0,0);
\end{scope}
\begin{scope}[rotate=-30]
\draw[->,thick] (0,0) -- (0,1) node[right]{\tiny{$-$}} node[left]{\tiny{$+$}};
\draw[thick] (0,1) -- (0,2.5);
\end{scope}

\node at (-1/5,1/5) {\small{${\wt{u}}$}};

\node at (-2.5,1/3) {\tiny{$\left(\begin{smallmatrix}1 & 0 \\ -\e^{t\varphi_1(z)} & 1 \end{smallmatrix}\right)$}};
\node at (2,1/3) {\tiny{$\left(\begin{smallmatrix}0 & 1\\ -1 & 0 \end{smallmatrix}\right)$}};
\node at (-.4,1.6) {\tiny{$\left(\begin{smallmatrix}1 & -\e^{-t\varphi_1(z)} \\ 0 & 1 \end{smallmatrix}\right)$}};
\node at (-.4,-1.6) {\tiny{$\left(\begin{smallmatrix} 1 & \e^{-t\varphi_1(z)} \\ 0 & 1 \end{smallmatrix}\right)$}};
\end{scope}

\end{tikzpicture}
\caption{Jumps of $T(z)$ in a neighborhood of $z={\wt{u}}$ and of $z=\wt{v}$.}
\label{fig:localparaqminus}
\end{figure}

By definition, if $z$ is in $\mathcal{Q}_{\wt{u}}$ and $\mathcal{Q}_{\wt{v}}$ (respectively), as well as in the domain of analyticity of $\varphi_1$ and $\varphi$ (respectively), we have
\begin{equation}
\label{eq:zetatphiqminus}
\frac 43\left(-\zeta_{\wt{u}}(z) \right)^{\frac 32} = -t\varphi_1(z),\qquad
\frac 43\zeta_{\wt{v}}(z) ^{\frac 32} = t\varphi(z).
\end{equation}
Using these identities, we see that the jumps of $\boldsymbol T(z)$ inside $\mathcal{Q}_{\wt{u}}$ and $\mathcal{Q}_{\wt{v}}$ (depicted in Figure~\ref{fig:localparaqminus}) coincide, under these conformal mappings, with the jumps of appropriate Airy model Riemann--Hilbert problem solutions $\boldsymbol\Phi^{\mathrm{Ai},\mathrm{II}}$ and $\boldsymbol\Phi^{\mathrm{Ai},\mathrm{I}}$, defined in~Appendix~\ref{app:Airy} (whose jumps are depicted in~Figure~\ref{fig:AirymodelRHP}).
Therefore, following the routine practice of the nonlinear steepest descent method, we define the inner Airy parametrices by
\begin{equation}
\begin{aligned}
\boldsymbol P^{({\wt{u}})}(z)&=\boldsymbol E^{({\wt{u}})}(z)\boldsymbol \Phi^{\mathrm{Ai},\mathrm{II}}(\zeta_{\wt{u}}(z))&&(z\in \mathcal{Q}_{\wt{u}}),\\
\boldsymbol P^{(\wt{v})}(z)&=\boldsymbol E^{(\wt{v})}(z)\boldsymbol \Phi^{\mathrm{Ai},\mathrm{I}}(\zeta_{\wt{v}}(z))&&(z\in \mathcal{Q}_{\wt{v}}),
\end{aligned}
\end{equation}
where
\begin{equation}
\begin{aligned}
\boldsymbol E^{({\wt{u}})}(z)&=\boldsymbol P^{\mathrm{out}}(z)\boldsymbol G^{-1}\bigl(-\zeta_{\wt{u}}(z)\bigr)^{-\frac 14\boldsymbol\sigma_3}&&(z\in \mathcal{Q}_{\wt{u}}),
\\
\boldsymbol E^{(\wt{v})}(z)&=\boldsymbol P^{\mathrm{out}}(z)\boldsymbol G^{-1}\zeta_{\wt{v}}(z)^{\frac 14\boldsymbol\sigma_3}&&(z\in \mathcal{Q}_{\wt{v}}).
\end{aligned}
\end{equation}
Here, $\boldsymbol P^{\mathrm{out}}$ is the outer parametrix introduced in~\Cref{sec:outerparaqminus}, $\boldsymbol G$ is defined in~\eqref{eq:globalparaqminus}, and $\boldsymbol\Phi^{\mathrm{Ai},\mathrm{I}}$ and $\boldsymbol\Phi^{\mathrm{Ai},\mathrm{II}}$ are defined in~\eqref{eq:defPhiAi1}--\eqref{eq:defPhiAi2} and solve the model Airy Riemann--Hilbert problem~\ref{cRHp:Airy}.
The usual properties of this construction are summarized in the next proposition.

\begin{proposition}
\label{prop:matchingqminus}
Let $z_0\in\lbrace{\wt{u}},\wt{v}\rbrace$.
The matrix $\boldsymbol E^{(z_0)}$ is analytic in $\mathcal{Q}_{z_0}$.
The matrix $\boldsymbol P^{(z_0)}$ is analytic in $\mathcal{Q}_{z_0}\setminus\Sigma_T$ and satisfies the same jump condition as $\boldsymbol T$ on $\Sigma_{T}\cap \mathcal{Q}_{z_0}$.
Moreover, when $z\in\partial \mathcal{Q}_{z_0}$ we have
\begin{equation}
\boldsymbol P^{(z_0)}(z)\boldsymbol P^{\mathrm{out}}(z)^{-1} = \boldsymbol {\mathrm{I}}+t^{-1}\wt{\boldsymbol J}_{R,z_0}(z)+O(t^{-2}),\quad t\to+\infty,
\end{equation}
with error term $O(t^{-1})$ uniform for $z\in\partial \mathcal{Q}_{z_0}$ and for $x\leq x_*-\delta$ (for any~$\delta>0$) and
\begin{equation}
\label{eq:JR1qminus}
\begin{aligned}
\wt{\boldsymbol J}_{R,{\wt{u}}} &= \frac 1{72\,\varphi_{\wt{u}}(z)\,(z-{\wt{u}})^2\,\sqrt{\wt{v}-z}}\begin{pmatrix}
    -7{\wt{u}}+5\wt{v}+2z & \i(7{\wt{u}}+5\wt{v}-12 z) \\
    \i(5{\wt{u}}+7\wt{v}-12 z) & 7{\wt{u}}-5\wt{v}-2z
\end{pmatrix},\\
\wt{\boldsymbol J}_{R,\wt{v}} &= \frac 1{72\,\varphi_{\wt{v}}(z)\,(z-\wt{v})^2\,\sqrt{z-{\wt{u}}}}\begin{pmatrix}
    5{\wt{u}}-7\wt{v}+2z & -\i(5{\wt{u}}+7\wt{v}-12 z) \\ 
    -\i(5{\wt{u}}+7\wt{v}-12 z) & -5{\wt{u}}+7\wt{v}-2z
\end{pmatrix}.
\end{aligned}
\end{equation}
\end{proposition}
\begin{proof}
Since $(\zeta_{\wt{v}}^{1/4})_+=\i(\zeta_{\wt{v}}^{1/4})_-$ on $\mathcal{Q}_{\wt{v}}\cap (-\infty,\wt{v})$, we have
\begin{equation}
\boldsymbol E^{(\wt{v})}_+=\boldsymbol P^{\mathrm{out}}_-\begin{pmatrix}0&1\\-1&0\end{pmatrix}\boldsymbol G^{-1}\begin{pmatrix}
\i&0\\ 0&-\i
\end{pmatrix}(\zeta_{\wt{v}}^{\frac 14\boldsymbol\sigma_3})_-
=\boldsymbol E^{(\wt{v})}_-,\quad \text{on }\mathcal{Q}_{\wt{v}}\cap (-\infty,\wt{v}),
\end{equation}
where one uses the identity
\begin{equation}
\begin{pmatrix}0&1\\-1&0\end{pmatrix}\boldsymbol G^{-1}\begin{pmatrix}
\i&0\\ 0&-\i
\end{pmatrix}=\boldsymbol G^{-1}.
\end{equation}
The isolated singularity of $\boldsymbol E^{(\wt{v})}(z)$ at $z=\wt{v}$ is removable because $\boldsymbol E^{(\wt{v})}(z)=O\bigl((z-\wt{v})^{-1/2}\bigr)$ as $z\to \wt{v}$.
Therefore, $\boldsymbol E^{(\wt{v})}$ is analytic in $\mathcal{Q}_{\wt{v}}$ and the first statement is proved when $z_0=\wt{v}$.
The statement for $z_0={\wt{u}}$ is proven in a completely analogous manner.

The second statement is true by construction.

For the last statement, when $t\to+\infty$ we have $\zeta_{z_0}(z)^{-1}=O(t^{-2/3})$ uniformly for $z\in\partial \mathcal{Q}_{z_0}$ by~\eqref{eq:conformalboundaryexpands}, hence we can use~\eqref{eq:asympPhiAi} to get
\begin{equation}
\begin{aligned}
\boldsymbol P^{({\wt{u}})}(z)\boldsymbol P^{\mathrm{out}}(z)^{-1} &=\boldsymbol P^{\mathrm{out}}(z)\left(\boldsymbol {\mathrm{I}}-\frac1{36\,t\,\varphi_1(z)}
\begin{pmatrix} -1 & 6\i \\ 6\i & 1 \end{pmatrix}
+O(t^{-2})\right)\boldsymbol P^{\mathrm{out}}(z)^{-1}
\\
\boldsymbol P^{(\wt{v})}(z)\boldsymbol P^{\mathrm{out}}(z)^{-1} &=\boldsymbol P^{\mathrm{out}}(z)\left(\boldsymbol {\mathrm{I}}+\frac1{36\,t\,\varphi(z)}
\begin{pmatrix} 1 & 6\i \\ 6\i & -1 \end{pmatrix}
+O(t^{-2})\right)\boldsymbol P^{\mathrm{out}}(z)^{-1}
\end{aligned}
\end{equation}
where we also use~\eqref{eq:zetatphiqminus}.
The proof follows by using the explicit expression for $\boldsymbol P^{\mathrm{out}}(z)$, as well as the local factorizations $\varphi(z)=\varphi_{\wt{v}}(z)(z-\wt{v})^{3/2}$ and $\varphi_1(z) = -\varphi_{{\wt{u}}}(z)({\wt{u}}-z)^{3/2}$, see~Proposition~\ref{prop:varphiqminus}.
\end{proof}

\subsection{Error analysis}

Let us finally fix $\epsilon>0$ sufficiently small such that the results of the previous sections hold true.
Let $\Sigma_R$ be
\begin{equation}
\Sigma_R=(-\infty,{\wt{u}}-\epsilon]\cup[\wt{v}+\epsilon,+\infty)\cup \Gamma_L^+\cup\Gamma_L^-\cup\Gamma_R^+\cup\Gamma_R^-\cup\partial \mathcal{Q}_{{\wt{u}}}\cup\partial \mathcal{Q}_{\wt{v}}
\end{equation}
and let $\Sigma_R^\circ$ be, as usual, the complement in $\Sigma_R$ of the points of intersection of the various contours forming $\Sigma_R$, namely $\Sigma_R^\circ$ consists of the points in $\Sigma_R$ except for ${\wt{u}}-\epsilon,\wt{v}+\epsilon,{\wt{u}}^\pm,\wt{v}^\pm$ and for the vertices of the squares $\mathcal{Q}_{\wt{u}}$ and $\mathcal{Q}_{\wt{v}}$.
The orientation on $\Sigma_R^\circ$ is illustrated in Figure~\ref{fig:SigmaRqminus}.

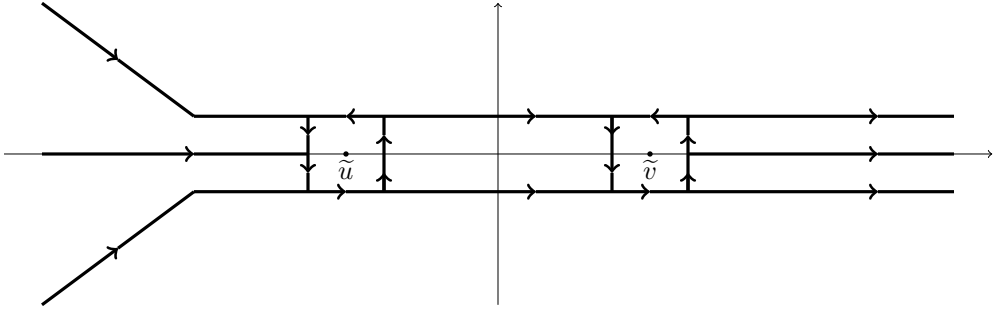
\begin{figure}[t]
\centering
\begin{tikzpicture}

\draw[very thin,->] (-6.5,0) -- (6.5,0);
\draw[very thin,->] (0,-2) -- (0,2);

\draw[very thick,->] (-6,0) -- (-4,0);
\draw[very thick,] (-4,0) -- (-5/2,0);
\draw[very thick,->] (5/2,0) -- (5,0) ;
\draw[very thick] (5,0) -- (6,0);

\draw[very thick, ->] (-5/2,1/2) -- (-5/2,1/4);
\draw[very thick, ->] (-5/2,1/4) -- (-5/2,-1/4);
\draw[very thick] (-5/2,-1/4) -- (-5/2,-1/2);

\draw[very thick, ->] (-3/2,-1/2) -- (-3/2,-1/4);
\draw[very thick, ->] (-3/2,-1/2) -- (-3/2,1/4);
\draw[very thick] (-3/2,1/4) -- (-3/2,1/2);

\draw[very thick, ->] (5/2,-1/2) -- (5/2,-1/4);
\draw[very thick, ->] (5/2,-1/2) -- (5/2,1/4);
\draw[very thick] (5/2,1/4) -- (5/2,1/2);

\draw[very thick, ->] (3/2,1/2) -- (3/2,1/4);
\draw[very thick, ->] (3/2,1/2) -- (3/2,-1/4);
\draw[very thick] (3/2,-1/4) -- (3/2,-1/2);

\draw[very thick] (-4,1/2) -- (-5,5/4);
\draw[very thick,<-] (-5,5/4) -- (-6,2);

\draw[very thick] (-4,-1/2) -- (-5,-5/4);
\draw[very thick,<-] (-5,-5/4) -- (-6,-2);

\draw[very thick] (-4,1/2) -- (-2,1/2);
\draw[very thick,<->] (-2,1/2) -- (1/2,1/2);
\draw[very thick] (1/2,1/2) -- (2,1/2);
\draw[very thick,<->] (2,1/2) -- (5,1/2);
\draw[very thick] (5,1/2) -- (6,1/2);

\draw[very thick,->] (-4,-1/2) -- (-2,-1/2);
\draw[very thick,->] (-2,-1/2) -- (1/2,-1/2);
\draw[very thick,->] (1/2,-1/2) -- (2,-1/2);
\draw[very thick,->] (2,-1/2) -- (5,-1/2);
\draw[very thick] (5,-1/2) -- (6,-1/2);

\fill (-2,0) circle (1pt);
\node at (-2,-1/5) {\small{${\wt{u}}$}};

\fill (2,0) circle (1pt);
\node at (2,-1/5) {\small{$\wt{v}$}};

\end{tikzpicture}
\caption{$\Sigma_R$ (case $x<x_*$).}
\label{fig:SigmaRqminus}
\end{figure}

Introduce the analytic function $\boldsymbol R:\mathbb{C}\setminus\Sigma_R\to\mathrm{SL}(2,\mathbb{C})$ by
\begin{equation}
\boldsymbol R(z) = \begin{cases}
\boldsymbol T(z)\boldsymbol P^{(z_0)}(z)^{-1}&\text{if }z\in \mathcal{Q}_{z_0},\,\,z_0\in\lbrace{\wt{u}},\wt{v}\rbrace,\\
\boldsymbol T(z)\boldsymbol P^{{\mathrm{out}}}(z)^{-1}&\text{otherwise}.
\end{cases}
\end{equation}

By construction, $\boldsymbol R(z)$ is the unique solution to the following Riemann--Hilbert problem.

\begin{cRHp}
\label{cRHp:Rqminus}
Find an analytic function $\boldsymbol R:\mathbb{C}\setminus \Sigma_R\to\mathrm{SL}(2,\mathbb{C})$ such that the following conditions hold true.
\begin{enumerate}[leftmargin=*]
\item Non-tangential boundary values of $\boldsymbol R$ exist and are continuous on $\Sigma_R^\circ$ and satisfy
\begin{equation}
\boldsymbol R_+(z)=\boldsymbol R_-(z)\boldsymbol J_R(z),\qquad z\in \Sigma_R^\circ,
\end{equation}
where $\boldsymbol J_R(z)$ is given for $z\in \Sigma_R^\circ$ by 
\begin{equation}
\label{eq:jumpRqminus}
\boldsymbol J_R(z)=\begin{cases}
\boldsymbol P^{{\mathrm{out}}}(z)\boldsymbol P^{(z_0)}(z)^{-1}&\text{if } z\in[z_0\pm\epsilon-\i\epsilon,z_0\pm\epsilon+\i\epsilon],\,\,z_0\in\lbrace{\wt{u}},\wt{v}\rbrace,\\
\boldsymbol P^{{\mathrm{out}}}(z)\boldsymbol J_T(z)^{\mp 1}\boldsymbol P^{(z_0)}(z)^{-1}&\text{if }z\in[z_0\pm\i\epsilon-\epsilon,z_0\pm\i\epsilon+\epsilon],\,\,z_0\in\lbrace{\wt{u}},\wt{v}\rbrace,\\
\boldsymbol P^{{\mathrm{out}}}(z)\boldsymbol J_T(z)\boldsymbol P^{{\mathrm{out}}}(z)^{-1}&\text{otherwise}.
\end{cases}
\end{equation}
\item We have $\boldsymbol R(z)\to\boldsymbol {\mathrm{I}}$ as $z\to\infty$  uniformly in $\mathbb{C}\setminus \Sigma_R$.
\item We have $\boldsymbol R(z)=O(1)$ as $z\to z_0$ uniformly in $\mathbb{C}\setminus\Sigma_R$ for all $z_0\in\Sigma_R\setminus\Sigma_R^\circ$.
\end{enumerate}
\end{cRHp}

By~\eqref{eq:SigmaTepsilonqminus}, we have $\Sigma^\circ_R=\partial \mathcal{Q}_{\wt{u}}\cup\partial \mathcal{Q}_{\wt{v}}\cup\wt{\Sigma_T}^\epsilon\setminus\lbrace\text{vertices of }\mathcal{Q}_{\wt{u}},\mathcal{Q}_{\wt{v}}\rbrace$.

\begin{proposition}
\label{prop:JRsmallqminus}
For any $\delta>0$ there exists $c,t_*>0$ such that
\begin{equation}
\boldsymbol J_R(z) = \begin{cases}
\boldsymbol {\mathrm{I}}-t^{-1}\wt{\boldsymbol J}_{R,z_0}(z)+O(t^{-2}),&\text{if }z\in\partial \mathcal{Q}_{z_0},\ \ z_0\in\lbrace{\wt{u}},\wt{v}\rbrace,\\
\boldsymbol {\mathrm{I}}+O\left(\frac{1}{|z|^2+1}\e^{-ct}\right),&\text{if }z\in\wt{\Sigma_T}^\epsilon,
\end{cases}
\end{equation}
with $\wt{\boldsymbol J}_{R,z_0}$ given in~\eqref{eq:JR1qminus} and error terms uniform for $t\geq t_*$, $z\in\Sigma^\circ_R$, and $x\leq x_*-\delta$.
\end{proposition}
\begin{proof}
It follows from Propositions~\ref{prop:JTsmallqminus} and~\ref{prop:matchingqminus}.
\end{proof}

\begin{proposition}
\label{prop:finalqminus}
For any $\delta>0$ there exists $t_*>0$ such that
\begin{equation}
\wh\alpha(t,x)= t(1-\e^{-\eta})+O(t^{-2}),\qquad
\log\wh\beta(t,x) = xt\eta-\frac 12\eta+O(t^{-2}).
\end{equation}
with error terms uniform for $t\geq t_*$ and $x\leq x_*-\delta$.
Here, $\wh\alpha(t,x)$ and $\wh\beta(t,x)$ are as in~\eqref{eq:wtalphabetagamma}.
\end{proposition}
\begin{proof}
By Proposition~\ref{prop:JRsmallqminus} we obtain that
\begin{equation}
\label{eq:smallnessqminus}
\|\boldsymbol J_R-\boldsymbol {\mathrm{I}}\|_{p}=O(t^{-1}),\qquad p=1,2,\infty,
\end{equation}
as $t\to+\infty$ uniformly for $x\leq x_*-\delta$ (for any $\delta>0$), where $\|\cdot\|_p$ is the maximum over the four matrix entries of their $L^p(\Sigma_R^\circ)$-norm.
It follows that $\boldsymbol R$ solves a \emph{small-norm} Riemann--Hilbert problem.
In such situation, it is well-known that the solution $\boldsymbol R$ admits the representation
\begin{equation}
\label{eq:smallnormsolutionqminus}
\boldsymbol R(z) = \boldsymbol {\mathrm{I}}+\frac 1{2\pi\i}\int_{\Sigma_R^\circ}\left(\boldsymbol {\mathrm{I}}+\boldsymbol\rho(\mu)\right)\left(\boldsymbol J_R(\mu)-\boldsymbol {\mathrm{I}}\right)\frac{\d \mu}{\mu-z}
\end{equation}
where $\boldsymbol\rho\in L^2(\Sigma_R^\circ)\otimes \mathbb{C}^{2\times 2}$ is the unique solution to
\begin{equation}
\label{eq:rhosmallnormqminus}
\boldsymbol\rho-\mathscr C_R[\boldsymbol\rho] =\mathscr C_-\left[\boldsymbol J_R-\boldsymbol {\mathrm{I}}\right].
\end{equation}
Here $\mathscr C_-$ is the Cauchy projector on $L^2(\Sigma_R^\circ)$ defined by
\begin{equation}
\label{eq:CauchyProjectorOneCut}
\mathscr C_-[\phi] (\mu)= \lim_{\varepsilon\downarrow 0}\int_{\Sigma_R^\circ}\phi(\nu)\frac{\d\nu}{\nu-\mu+\i\epsilon},\qquad \phi\in L^2(\Sigma_R^\circ),
\end{equation}
and $\mathscr C_R$ is the operator on $L^2(\Sigma_R^\circ)\otimes\mathbb{C}^{2\times 2}$ defined by
\begin{equation}
\mathscr C_R[\boldsymbol F] = \mathscr C_-[\boldsymbol F(\boldsymbol J_R-\boldsymbol {\mathrm{I}})],\qquad \boldsymbol F\in L^2(\Sigma_R^\circ)\otimes\mathbb{C}^{2\times 2}.
\end{equation}
The conditions in~\eqref{eq:smallnessqminus} imply (see~\cite{ItsLargeN}, for example) that the operator norm of~$\mathscr C_R$ is $O(t^{-1})$ and so the solution~$\boldsymbol \rho$ to~\eqref{eq:rhosmallnormqminus} can be expressed in terms of a Neumann series 
\begin{equation}
\boldsymbol\rho=\sum_{k\geq 0}\mathscr C_R^k\bigl[\mathscr C_-[\boldsymbol J_R-\boldsymbol {\mathrm{I}}]\bigr]=O(t^{-1}).
\end{equation}
(Here we also use that $\mathscr C_-$ is a bounded operator.)
Plugging this into~\eqref{eq:smallnormsolutionqminus} and using Proposition~\ref{prop:JRsmallqminus}, it follows that
\begin{equation}
\boldsymbol R(z) = \boldsymbol {\mathrm{I}}+\frac 1{2\pi\i}\int_{\Sigma_R^\circ}\frac{\boldsymbol J_R(\mu)-\boldsymbol {\mathrm{I}}}{\mu-z}\d\mu+O(t^{-2})=
\boldsymbol {\mathrm{I}}+t^{-1}\sum_{z_0\in\lbrace {\wt{u}},\wt{v}\rbrace}\frac 1{2\pi\i}\int_{\partial \mathcal{Q}_{z_0}}\frac{\wt{\boldsymbol J}_{R,z_0}(\mu)}{z-\mu}\d\mu+O(t^{-2}).
\end{equation}
Since $\wt{\boldsymbol J}_{R,z_0}$ extends to a meromorphic function inside $\mathcal{Q}_{z_0}$ with a double pole at $z_0$ and no other singularities, see~\eqref{eq:JR1qminus}, by Cauchy's theorem the integrals above reduce to the polar part of $\wt{\boldsymbol J}_{R,z_0}$ at $z_0$, namely
\begin{equation}
\begin{aligned}
\frac 1{2\pi\i}\int_{\partial \mathcal{Q}_{{\wt{u}}}}\frac{\wt{\boldsymbol J}_{R,{\wt{u}}}(\mu)}{z-\mu}\d\mu
&=\frac{5\e^{-\eta/2}}{48(z-{\wt{u}})^2}\begin{pmatrix}
    1 & \i \\ \i & -1
\end{pmatrix}+\frac 1{48(z-{\wt{u}})}\begin{pmatrix}
    1 & -\frac 52\i \\ -\frac 52\i & -1
\end{pmatrix},
\\
\frac 1{2\pi\i}\int_{\partial \mathcal{Q}_{\wt{v}}}\frac{\wt{\boldsymbol J}_{R,\wt{v}}(\mu)}{z-\mu}\d\mu
&=\frac{5\e^{-\eta/2}}{48(z-\wt{v})^2}\begin{pmatrix}
    -1 & \i \\ \i & 1
\end{pmatrix}+\frac 1{48(z-\wt{v})}\begin{pmatrix}
    1 & \frac 52 \i \\ \frac 52\i & -1
\end{pmatrix},
\end{aligned}
\end{equation}
assuming $z\not\in \mathcal{Q}_{{\wt{u}}}\cup \mathcal{Q}_{\wt{v}}$.
In this computation, we also used~\eqref{eq:varphivvaluev} and~\eqref{eq:varphiuvalueu}.
Finally, when $\Im z$ is large we have $\boldsymbol N(z)=\boldsymbol R(z)\boldsymbol P^{{\mathrm{out}}}(z)$ and so, as $t\to+\infty$ and $\Im z\to+\infty$, we have
\begin{equation}
\boldsymbol N(z) = \left(\boldsymbol {\mathrm{I}}+\frac 1{24 zt}\begin{pmatrix}
    1 & 0 \\ 0 & -1
\end{pmatrix}+O\bigl(z^{-1}t^{-2}\bigr)\right)\,\boldsymbol P^{{\mathrm{out}}}(z)
\end{equation}
(uniformly for $x\leq x_*-\delta$ for any $\delta>0$).
By the large-$z$ expansions~\eqref{eq:refinedexpansionNqminus} and~\eqref{eq:Poutlargezqminus}, it follows that
\begin{equation}
\wh\alpha(t,x) = t(1+{\wt{g}}_1)+\left(\frac 1{24}-\frac 1{24}\right)t^{-1}+O(t^{-2}),\quad
\e^{t(\wt{\ell}-2{\wt{g}}_\infty)}\wh\beta(t,x) =\frac 14(\wt{v}-{\wt{u}})+O(t^{-2}),
\end{equation}
as $t\to+\infty$ uniformly for $x\leq x_*-\delta$.
The thesis then follows from~\eqref{eq:uvqminus},~\eqref{eq:g1qminus}, \eqref{eq:ellqminus}, and~\eqref{eq:ginftyqminus}.
\end{proof}

\section{Nonlinear steepest descent analysis (case \texorpdfstring{$x_*<x<2$}{2<x<x\_*})}
\label{sec:DeiftZhou2cut}

Throughout this section we are going to assume that $x_*<x=s/t<2$.
We are once again starting from a continuous Riemann--Hilbert problem, which characterizes the Fredholm determinant $Q(t,s)$, and performing a nonlinear steepest descent analysis on it.
This strategy is similar to the one employed in in the previous section and it involves a series of analogous transformations of the continuous Riemann--Hilbert problem.
Therefore, in order to emphasize the parallel and avoid excessive notation, we will use the same symbols to denote analogous quantities in both sections, even though their definitions may differ (e.g., functions like $\boldsymbol M,{\wt{g}},\boldsymbol N,\boldsymbol T,\ldots$ and contours like $\Sigma_M,\Sigma_N,\Sigma_T,\ldots$).

\subsection{Continuous Riemann--Hilbert problem}

Let ${\wt{a}}<{\wt{b}}<{\wt{c}}$ (to be fixed later).
Consider the (multi-)contour $\Sigma_M$ in the complex $z$-plane depicted in Figure~\ref{fig:SigmaMqplus}:
\begin{equation}
\Sigma_M=\mathbb{R}\cup\Gamma_{L}^-\cup\Gamma_{L}^+\cup\Gamma_{R}^+\cup\Gamma_{R}^-\cup\left(\bigcup_{z_0\in\lbrace{\wt{a}},{\wt{b}},{\wt{c}}\rbrace}\bigl(\gamma_{z_0}^+\cup\gamma_{z_0}^-\bigr)\right),
\end{equation}
where, for $z_0\in\lbrace{\wt{a}},{\wt{b}},{\wt{c}}\rbrace$, $\gamma_{z_0}^\pm$ is a smooth contour joining $z_0$ to the line $\Im z=\pm\epsilon$ and, denoting $z_0^\pm$ the intersection point of $\gamma_{z_0}^\pm$ with $\Im z=\pm\epsilon$, $\Gamma_R^\pm=[{\wt{a}}^\pm,+\infty\pm\i\epsilon)$ and $\Gamma_L^\pm=\bigl([-R_0\pm\i\epsilon,{\wt{a}}^\pm]\bigr)\cup\bigl(\e^{\pm\i \delta_0}\mathbb{R}_+-R_0\pm\i\epsilon\bigr)$.
The parameters ${\wt{a}}<{\wt{b}}<{\wt{c}}$, $\epsilon>0$, $R_0>-{\wt{a}}$, as well as the specific curves $\gamma_{\wt a}^\pm$, $\gamma_{\wt b}^\pm$, and $\gamma_{\wt c}^\pm$, will be fixed later, while $\delta_0\in(\pi/2,\pi)$ can be fixed arbitrarily.
Let $\Sigma_M^\circ=\Sigma_M\setminus \bigcup_{z_0\in\lbrace\wt a,\wt b,\wt c\rbrace}\lbrace z_0,z_0^+,z_0^-\rbrace$, oriented as in Figure~\ref{fig:SigmaMqminus}.
The orientation determines $\pm$ sides of $\Sigma_M^\circ$, $+$ to the left-hand side and $-$ to the right-hand side.
Moreover, $\Sigma_M$ divides $\mathbb{C}$ into ten connected components which we call $\Omega_1$, $\dots$, $\Omega_{10}$.
This is illustrated in Figure~\ref{fig:SigmaMqplus}.

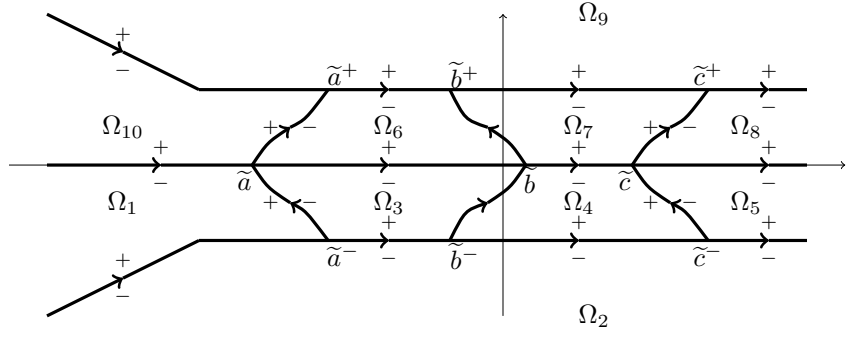
\begin{figure}[t]
\centering
\begin{tikzpicture}

\draw[very thin,->] (-6.5,0) -- (4.5,0);
\draw[very thin,->] (0,-2) -- (0,2);

\draw[very thick,->] (-6,0) -- (-4.5,0) node[below]{\tiny{$-$}} node[above]{\tiny{$+$}};
\draw[very thick,->] (-4.5,0) -- (-1.5,0) node[below]{\tiny{$-$}} node[above]{\tiny{$+$}};
\draw[very thick,->] (-1.5,0) -- (1,0) node[below]{\tiny{$-$}} node[above]{\tiny{$+$}};
\draw[very thick,->] (1,0) -- (3.5,0) node[below]{\tiny{$-$}} node[above]{\tiny{$+$}};
\draw[very thick] (3.5,0) -- (4,0);

\draw[very thick] (-4,1) -- (-5,3/2) node[below]{\tiny{$-$}} node[above]{\tiny{$+$}};
\draw[very thick,<-] (-5,3/2) -- (-6,2);

\draw[very thick] (-4,-1) -- (-5,-3/2) node[below]{\tiny{$-$}} node[above]{\tiny{$+$}};
\draw[very thick,<-] (-5,-3/2) -- (-6,-2);

\draw[very thick,->] (-4,1) -- (-1.5,1) node[below]{\tiny{$-$}} node[above]{\tiny{$+$}};
\draw[very thick,->] (-1.5,1) -- (1,1) node[below]{\tiny{$-$}} node[above]{\tiny{$+$}};
\draw[very thick,->] (1,1) -- (3.5,1) node[below]{\tiny{$-$}} node[above]{\tiny{$+$}};
\draw[very thick] (3.5,1) -- (4,1);

\draw[very thick,->] (-4,-1) -- (-1.5,-1) node[below]{\tiny{$-$}} node[above]{\tiny{$+$}};
\draw[very thick,->] (-1.5,-1) -- (1,-1) node[below]{\tiny{$-$}} node[above]{\tiny{$+$}};
\draw[very thick,->] (1,-1) -- (3.5,-1) node[below]{\tiny{$-$}} node[above]{\tiny{$+$}};
\draw[very thick] (3.5,-1) -- (4,-1);

\begin{scope}[shift={(-1,0)}]
\draw[very thick,->] 
  (-1.3,-1) .. controls (-1.6,-0.6) .. (-1.8,-0.5) 
    node[right]{\tiny{$-$}} node[left]{\tiny{$+$}};
\draw[very thick] 
  (-1.8,-0.5) .. controls (-2.1,-0.3) .. (-2.3,0);
\draw[very thick,->] 
  (-2.3,0) .. controls (-2.1,0.3) .. (-1.8,0.5) 
    node[right]{\tiny{$-$}} node[left]{\tiny{$+$}};
\draw[very thick] 
  (-1.8,0.5) .. controls (-1.6,0.6) .. (-1.3,1);

\draw[very thick,->] 
  (.3,-1) .. controls (.5,-0.6) .. (.8,-1/2);
\draw[very thick] 
  (.8,-1/2) .. controls (1.1,-0.3) .. (1.3,0);
\draw[very thick,->] 
  (1.3,0) .. controls (1.1,0.3) .. (.8,1/2);
\draw[very thick] 
  (.8,1/2) .. controls (.5,0.6) .. (.3,1);
\end{scope}

\begin{scope}[shift={(4,0)}]
\draw[very thick,->] 
  (-1.3,-1) .. controls (-1.6,-0.6) .. (-1.8,-0.5) 
    node[right]{\tiny{$-$}} node[left]{\tiny{$+$}};
\draw[very thick] 
  (-1.8,-0.5) .. controls (-2.1,-0.3) .. (-2.3,0);
\draw[very thick,->] 
  (-2.3,0) .. controls (-2.1,0.3) .. (-1.8,0.5) 
    node[right]{\tiny{$-$}} node[left]{\tiny{$+$}};
\draw[very thick] 
  (-1.8,0.5) .. controls (-1.6,0.6) .. (-1.3,1);
\end{scope}
  
\node at (-3.4,-1/5) {\small{${\wt{a}}$}};
\node at (-2.1,6/5) {\small{${\wt{a}}^+$}};
\node at (-2.1,-6/5) {\small{${\wt{a}}^-$}};

\node at (.35,-1/5) {\small{${\wt{b}}$}};
\node at (-.5,6/5) {\small{${\wt{b}}^+$}};
\node at (-.5,-6/5) {\small{${\wt{b}}^-$}};

\node at (1.6,-1/5) {\small{${\wt{c}}$}};
\node at (2.7,6/5) {\small{${\wt{c}}^+$}};
\node at (2.7,-6/5) {\small{${\wt{c}}^-$}};

\node at (-5,1/2) {\small{$\Omega_{10}$}};
\node at (-5,-1/2) {\small{$\Omega_1$}};
\node at (1.2,2) {\small{$\Omega_9$}};
\node at (1.2,-2) {\small{$\Omega_2$}};
\node at (-1.5,1/2) {\small{$\Omega_6$}};
\node at (-1.5,-1/2) {\small{$\Omega_3$}};
\node at (1,1/2) {\small{$\Omega_7$}};
\node at (1,-1/2) {\small{$\Omega_4$}};
\node at (3.2,1/2) {\small{$\Omega_8$}};
\node at (3.2,-1/2) {\small{$\Omega_5$}};

\end{tikzpicture}
\caption{$\Sigma_M$, its orientation and corresponding $\pm$ sides and domains $\Omega_i$ for $1\leq i\leq 10$ (case $x_*<x<2$).}
\label{fig:SigmaMqplus}
\end{figure}

We assume that 
\begin{equation}
{\wt{a}}<{\wt{b}}<x<{\wt{c}}.
\end{equation}
We introduce analytic matrix functions $\boldsymbol M_i:\Omega_i\to\mathrm{SL}(2,\mathbb{C})$ as follows:
\begin{equation}
\begin{aligned}
\boldsymbol M_1(z)&=\boldsymbol M_4(z)=\boldsymbol Y(tz)\boldsymbol \Phi_R(tz)\boldsymbol \Delta_-(tz)=\boldsymbol Y(tz)\boldsymbol \Phi_{L}^-(tz)\boldsymbol C_-(tz)^{-1}\boldsymbol \Delta_-(tz),\\
\boldsymbol M_2(z)&=\boldsymbol M_9(z)=\boldsymbol Y(tz)\boldsymbol \Phi_R(tz),\\
\boldsymbol M_3(z)&=\boldsymbol M_5(z)=\boldsymbol Y(tz)\boldsymbol \Phi_R(tz)\boldsymbol \nabla_-(tz),\\
\boldsymbol M_6(z)&=\boldsymbol M_8(z)=\boldsymbol Y(tz)\boldsymbol \Phi_R(tz)\boldsymbol \nabla_+(tz),\\
\boldsymbol M_7(z)&=\boldsymbol M_{10}(z)=\boldsymbol Y(tz)\boldsymbol \Phi_R(tz)\boldsymbol \Delta_+(tz)=\boldsymbol Y(tz)\boldsymbol \Phi_L^+(tz)\boldsymbol C_+(tz)^{-1}\boldsymbol \Delta_+(tz),\\
\end{aligned}
\end{equation}
with the notations introduced in Section~\ref{sec:notationsDeltaNabla} (and a minor abuse of notation in writing equalities like $\boldsymbol M_1=\boldsymbol M_4$, as these are functions with different domains but defined by the same formula).
Note that the assumption~${\wt{b}}<x<{\wt{c}}$ is necessary if we want $\boldsymbol M_3,\boldsymbol M_5,\boldsymbol M_6,\boldsymbol M_8$ to be analytic in $\Omega_3,\Omega_5,\Omega_6,\Omega_8$ (respectively), see~Remark~\ref{remark:entire}

An important property of these matrix functions, which follows from Proposition~\ref{prop:analytic}, is that $\boldsymbol M_i$ is analytic in a proper open neighborhood of $\Omega_i$.
Moreover, we have
\begin{equation}
\begin{aligned}
\e^{-z(\log(zt^{-1})-1)\boldsymbol\sigma_3}\boldsymbol \nabla_+(z)\e^{z(\log(zt^{-1})-1)\boldsymbol\sigma_3}&=\boldsymbol {\mathrm{I}}+O(z^{-\infty})&&\text{as }z\to\infty\text{ uniformly in }\Omega_8\setminus\mathcal N_\delta(\mathbb{Z}'_+),\\
\e^{-z(\log(zt^{-1})-1)\boldsymbol\sigma_3}\boldsymbol \nabla_-(z)\e^{z(\log(zt^{-1})-1)\boldsymbol\sigma_3}&=\boldsymbol {\mathrm{I}}+O(z^{-\infty})&&\text{as }z\to\infty\text{ uniformly in }\Omega_5\setminus\mathcal N_\delta(\mathbb{Z}'_+),\\
\e^{-z(\log(zt^{-1})-1)\boldsymbol\sigma_3}\boldsymbol C_+(z)^{-1}\boldsymbol\Delta_+(z)\e^{z(\log(zt^{-1})-1)\boldsymbol\sigma_3}&=\boldsymbol {\mathrm{I}}+O(z^{-\infty})&&\text{as }z\to\infty\text{ uniformly in }\Omega_{10}\setminus\mathcal N_\delta(\mathbb{Z}'_-),\\
\e^{-z(\log(zt^{-1})-1)\boldsymbol\sigma_3}\boldsymbol C_-(z)^{-1}\boldsymbol\Delta_-(z)\e^{z(\log(zt^{-1})-1)\boldsymbol\sigma_3}&=\boldsymbol {\mathrm{I}}+O(z^{-\infty})&&\text{as }z\to\infty\text{ uniformly in }\Omega_1\setminus\mathcal N_\delta(\mathbb{Z}'_-),
\end{aligned}\end{equation}
for any $\delta>0$.
We infer from Proposition~\ref{prop:asymptotics} that the matrix function $\boldsymbol M:\mathbb{C}\setminus \Sigma_M\to\mathrm{SL}(2,\mathbb{C})$ which equals $(\sqrt{2\pi t})^{\boldsymbol\sigma_3}\boldsymbol M_i$ on $\Omega_i$ is the unique solution to the following Riemann--Hilbert problem.

\begin{cRHp}
\label{cRHp:Mqplus}
Find an analytic function $\boldsymbol M:\mathbb{C}\setminus \Sigma_M\to\mathrm{SL}(2,\mathbb{C})$ such that the following conditions hold true.
\begin{enumerate}[leftmargin=*]
\item Non-tangential boundary values of $\boldsymbol M$ exist and are continuous on $\Sigma_M^\circ$ and satisfy
\begin{equation}
\boldsymbol M_+(z)=\boldsymbol M_-(z)\,\boldsymbol J_M(z),\qquad z\in \Sigma_M^\circ,
\end{equation}
where $\boldsymbol J_M(z)$ is given for $z\in \Sigma_M^\circ$ by 
\begin{equation}
\label{eq:JMqplus}
\boldsymbol J_M(z)=\begin{cases}
\boldsymbol \Delta_\pm(tz)^{\mp 1}&
z\in \Gamma_L^{\pm} \cup ({\wt{b}}^\pm,{\wt{c}}^\pm),
\\
\boldsymbol\nabla_\pm(tz)^{\mp 1}&z\in({\wt{a}}^\pm,{\wt{b}}^\pm)\cup({\wt{c}}^\pm,+\infty\pm\i\epsilon),
\\
\boldsymbol\nabla_-(tz)^{-1}\boldsymbol\nabla_+(tz) = \begin{pmatrix}
1 & -2\pi\i\bigl(1-\varsigma(t(z-x))\bigr) \\ 0 & 1
\end{pmatrix}
&z\in ({\wt{a}},{\wt{b}})\cup({\wt{c}},+\infty),
\\ 
\boldsymbol\Delta_-(tz)^{-1}\boldsymbol\Delta_+(tz)
&z\in (-\infty,{\wt{a}})\cup({\wt{b}},{\wt{c}}),
\\
\boldsymbol\nabla_\pm(tz)^{-1}\boldsymbol\Delta_\pm(tz) & z\in \gamma_{\wt a}^\pm\cup\gamma_{\wt c}^\pm,
\\
\boldsymbol\Delta_\pm(tz)^{-1}\boldsymbol\nabla_\pm(tz) & z\in \gamma_{\wt b}^\pm.
\end{cases}
\end{equation}
\item We have $\boldsymbol M(z)\e^{tz(\log z-1)\boldsymbol\sigma_3}\to\boldsymbol {\mathrm{I}}$ as $z\to\infty$  uniformly in $\mathbb{C}\setminus \Sigma_M$.
\item We have $\boldsymbol M(z)=O(1)$ as $z\to z_0$ uniformly in $\mathbb{C}\setminus\Sigma_M$ for all $z_0\in\Sigma_M\setminus\Sigma_M^\circ$.
\end{enumerate}
\end{cRHp}

We note from Theorem~\ref{thm:CR} that
\begin{equation}
\label{eq:refinedexpansionMqplus}
M(z)=\biggl(\boldsymbol {\mathrm{I}}+\frac 1t\begin{pmatrix}
\alpha+\tfrac 1{24}-t^2 & 2\pi t(\beta+1) \\ \frac{\gamma+t^2}{2\pi t} & -\alpha-\tfrac 1{24}+t^2
\end{pmatrix}z^{-1}+O\bigl(z^{-2}\bigr)\biggr)\,\e^{-tz(\log z-1)\,\boldsymbol\sigma_3}
\end{equation}
as $z\to\infty$ uniformly in $\mathbb{C}\setminus \Sigma_M$.
Here, $\alpha=\alpha(t,xt)$, $\beta=\beta(t,xt)$, and $\gamma=\gamma(t,xt)$ are given in~\eqref{eq:alphabetagamma}.

\subsection{Construction of the \texorpdfstring{$g$}{g}-function}\label{sec:gqplus}

In this case, we will construct the $g$-function ${\wt{g}}(z)$ as a small deformation, when $t$ is large, of the function $g(z)$ employed in Section~\ref{sec:minimizationxq<x<2} in the context of minimization of the logarithmic energy $\mathcal{E}_{\eta,x}$ for $x_*<x<2$.
We will use the function $\wt V_{\eta,x}(z)$ defined in~\eqref{eq:vqminus} in this section too.

For any ${\wt{a}}<{\wt{b}}<{\wt{c}}<{\wt{d}}$, let
\begin{equation}
    {\wt{r}}(z)=\sqrt{(z-{\wt{a}})(z-{\wt{b}})(z-{\wt{c}})(z-{\wt{d}})}
\end{equation} 
analytic for $z\in\mathbb{C}\setminus\bigl([{\wt{a}},{\wt{b}}]\cup[{\wt{c}},{\wt{d}}]\bigr)$ and $\sim z^2$ as $z\to\infty$.

We define
\begin{equation}
\label{eq:gprimeqplus}
\begin{aligned}
{\wt{g}}'(z)&={\wt{r}}(z)\left(\int_{(-\infty,{\wt{a}})\cup({\wt{b}},{\wt{c}})} \frac{\d\nu}{{\wt{r}}(\nu)(\nu-z)}-\frac{1}{2\pi\i}\int_{({\wt{a}},{\wt{b}})\cup({\wt{c}},{\wt{d}})} \frac{\wt V_{\eta,x}'(\nu)\d\nu}{{\wt{r}}_+(\nu)(\nu-z)}\right)
\\
&=\pm\i\pi-{\wt{r}}(z)\left(\int_{\wt{d}}^{+\infty} \frac{\d\nu}{{\wt{r}}(\nu)(\nu-z)}+\frac{1}{2\pi\i}\int_{({\wt{a}},{\wt{b}})\cup({\wt{c}},{\wt{d}})} \frac{\wt V_{\eta,x}'(\nu)\d\nu}{{\wt{r}}_+(\nu)(\nu-z)}\right)
\end{aligned}
\end{equation}
where, in the last line, the sign is determined by $\pm\Im z>0$ and the equality follows from Cauchy's theorem.
Here we assume ${\wt{a}},{\wt{b}},{\wt{c}},{\wt{d}}\in\mathbb{R}$ with ${\wt{b}}<x<{\wt{c}}$ and we will shortly determine the values of ${\wt{a}},{\wt{b}},{\wt{c}},{\wt{d}}$.
The function ${\wt{g}}'(z)$ is analytic for $z\in\mathbb{C}\setminus(-\infty,{\wt{d}}]$ and, by the Sokhotski--Plemelj formulas, the boundary values ${\wt{g}}'_\pm$ from above ($+$) and below ($-$) the real axis exist and are continuous for all $\mu\in (-\infty,{\wt{d}})\setminus\lbrace {\wt{a}},{\wt{b}},{\wt{c}}\rbrace$ and satisfy
\begin{equation}
\begin{aligned}
\label{eq:jumpgprimeqplus}
{\wt{g}}_+'(\mu)+{\wt{g}}_-'(\mu)&=-\wt V_{\eta,x}'(\mu),&&\mu\in ({\wt{a}},{\wt{b}})\cup({\wt{c}},{\wt{d}}),\\
{\wt{g}}_+'(\mu)-{\wt{g}}_-'(\mu)&=2\pi\i,&&\mu\in (-\infty,{\wt{a}})\cup({\wt{b}},{\wt{c}}).
\end{aligned}
\end{equation}

The endpoints ${\wt{a}},{\wt{b}},{\wt{c}},{\wt{d}}$ are fixed by similar arguments as in~Section~\ref{sec:minimizationxq<x<2}. Namely, we first require that in the asymptotic expansion
\begin{equation}
\label{eq:gprimeasympqplus}
{\wt{g}}'(z)=z{\wt{g}}_{-2}+\log z+{\wt{g}}_{-1}+{\wt{g}}_{0}z^{-1}+{\wt{g}}_1z^{-2}+O(z^{-3}),\quad z\to\infty,
\end{equation}
we have
\begin{equation}\label{eq:condqplus1}
{\wt{g}}_{-2}={\wt{g}}_{-1}={\wt{g}}_0=0.
\end{equation}
(We include in~\eqref{eq:gprimeasympqplus} the term of order $z^{-2}$ for later convenience.)
A fourth condition is determined by introducing
\begin{equation}
\label{eq:gqplus}
{\wt{g}}(z)=\int_{\wt{d}}^z {\wt{g}}'(y)\,\d y,
\end{equation}
which is analytic for $z\in\mathbb{C}\setminus(-\infty,{\wt{d}}]$.
By~\eqref{eq:jumpgprimeqplus} we have
\begin{equation}
{\wt{g}}_+(\mu)+{\wt{g}}_-(\mu)=-\wt V_{\eta,x}(\mu)+\begin{cases}\wt{\ell}_1 & \text{if }\mu\in ({\wt{a}},{\wt{b}})
\\
\wt \ell &\text{if } \mu\in ({\wt{c}},{\wt{d}})
\end{cases}
\end{equation}
with
\begin{equation}
\wt{\ell}_1=\wt V_{\eta,x}({\wt{d}})-\wt V_{\eta,x}({\wt{c}})+\wt V_{\eta,x}({\wt{b}})-\int_{\wt{b}}^{\wt{c}}\bigl({\wt{g}}'_+(\mu)+{\wt{g}}_-'(\mu)\bigr)\d \mu,\qquad
\wt\ell=\wt V_{\eta,x}({\wt{d}}).
\end{equation}
The last condition needed for the determination of the endpoints ${\wt{a}},{\wt{b}},{\wt{c}},{\wt{d}}$ is supplied by the requirement $\wt{\ell}_1=\wt\ell$.

Similarly to Section~\ref{sec:minimizationxq<x<2}, it is convenient to work with the elliptic uniformization of the Riemann surface of ${\wt{r}}(z)$. Namely, we introduce
\begin{equation}
\label{eq:newelliptic}
{\wt{m}}= \frac{\i\pi}{\int_{\wt{d}}^{\wt{c}}\frac{\d\nu}{{\wt{r}}_+(\nu)}},\quad
{\wt{K}} = {\wt{m}}\int_{\wt{c}}^{\wt{b}}\frac{\d\nu}{{\wt{r}}(\nu)},\quad
{\wt{w}}_\infty = {\wt{m}}\int_{\wt{d}}^{+\infty}\frac{\d\nu}{{\wt{r}}(\nu)}.
\end{equation}
as well as 
\begin{equation}
\label{eq:newconformalelliptic}
{\wt{w}}(z)={\wt{m}}\int_{\wt{d}}^z\frac{\d\nu}{{\wt{r}}(\nu)}.
\end{equation}
Writing
\begin{equation}
\label{eq:twolinesg}
\begin{aligned}
{\wt{g}}'(z)&=
{\wt{r}}(z)\left(\int_{(-\infty,{\wt{a}})\cup({\wt{b}},{\wt{c}})} \frac{\d\nu}{{\wt{r}}(\nu)(\nu-z)}-\frac{\eta}{2\pi\i}\int_{\wt{c}}^{\wt{d}} \frac{\d\nu}{{\wt{r}}_+(\nu)(\nu-z)}\right)
\\
&\qquad-\frac{{\wt{r}}(z)}{2\pi\i}\int_{\wt{a}}^{\wt{b}} \frac{\wt V_{\eta,x}'(\nu)}{{\wt{r}}_+(\nu)}\frac{\d\nu}{\nu-z}-\frac{{\wt{r}}(z)}{2\pi\i}\int_{\wt{c}}^{\wt{d}} \frac{\wt V_{\eta,x}'(\nu)-\eta}{{\wt{r}}_+(\nu)}\frac{\d\nu}{\nu-z},
\end{aligned}
\end{equation}
the second line is $O\left((1+|z|)t^{-\infty}\right)$ uniformly in $z\in\mathbb{C}\setminus(-\infty,{\wt{d}}]$ (assuming ${\wt{b}}<x<{\wt{c}}$) and the first one can be rewritten using the same arguments as in Section~\ref{sec:minimizationxq<x<2}, thus yielding
\begin{equation}
\label{eq:gprimeellipticprelimit}
{\wt{g}}'(z)=-\frac{\zeta(\i\pi|{\wt{K}},\i\pi)}{\i\pi}\eta{\wt{w}}(z)-\frac{\eta}2-\log\frac{\sigma\bigl({\wt{w}}_\infty-{\wt{w}}(z)|{\wt{K}},\i\pi\bigr)}{\sigma\bigl({\wt{w}}_\infty+{\wt{w}}(z)|{\wt{K}},\i\pi\bigr)}+O\left((1+|z|)t^{-\infty}\right).
\end{equation}

This discussion shows that the system ${\wt{g}}_{-2}={\wt{g}}_{-1}={\wt{g}}_0=\wt{\ell}_1-\wt\ell=0$ has the form
\begin{equation}
\Xi({\wt{w}}_\infty,{\wt{m}},{\wt{K}},{\wt{d}})+O(t^{-\infty}) = (0,0,0,0)
\end{equation}
where $\Xi$ has been introduced in~\eqref{eq:Xi} and where the remainder $O(t^{-\infty})$ as $t\to+\infty$ is uniform for $x_*+\delta \leq x\leq 2-\delta$ and for ${\wt{a}}<{\wt{b}}<{\wt{c}}<{\wt{d}}$ and ${\wt{b}}+\delta\leq x\leq {\wt{c}}-\delta$ (for any $\delta>0$).
In other words, the system ${\wt{g}}_{-2}={\wt{g}}_{-1}={\wt{g}}_0=\wt{\ell}_1-\wt\ell=0$ is a small deformation, of order $O(t^{-\infty})$, of the system determining $w_\infty,m,K,d$ considered in Section~\ref{sec:minimizationxq<x<2}.
Since the Jacobian determinant of that system is uniformly away from zero if $K$ is bounded away from infinity (see~Remark~\ref{remark:Jacobian}), we obtain from the implicit function theorem that, for $t$ sufficiently large, the system ${\wt{g}}_{-2}={\wt{g}}_{-1}={\wt{g}}_0=\wt{\ell}_1-\wt\ell=0$ determines uniquely the parameters ${\wt{w}}_\infty,{\wt{m}},{\wt{K}},{\wt{d}}$ introduced in~\eqref{eq:newelliptic}, and, moreover,
\begin{equation}
\label{eq:parameterscloseqplus}
{\wt{w}}_\infty =\frac{\eta}2+O(t^{-\infty}),\quad
{\wt{m}} =\mathcal{U}\bigl(\mathcal{K}(x)\bigr)+O(t^{-\infty}),\quad
{\wt{K}} =\mathcal{K}(x)+O(t^{-\infty}).
\end{equation}
In turn, these parameters determine the endpoints ${\wt{a}},{\wt{b}},{\wt{c}},{\wt{d}}$ by formulas analogous to those shown in Section~\ref{sec:minimizationxq<x<2}, see~\eqref{eq:abcdapprox}.
In particular,
\begin{equation}
\label{eq:endpointscloseqplus}
{\wt{a}}=a+O(t^{-\infty}),\quad
{\wt{b}}=b+O(t^{-\infty}),\quad
{\wt{c}}=c+O(t^{-\infty}),\quad
{\wt{d}}=d+O(t^{-\infty}),
\end{equation}
as $t\to+\infty$ uniformly for $x\in [x_*+\delta,2-\delta]$ (for any $\delta>0$), where $a,b,c,d$ are the endpoints in~\eqref{eq:maintheorem:endpoints}.

From now on we assume that $t$ is sufficiently large and that ${\wt{a}},{\wt{b}},{\wt{c}},{\wt{d}}$ are determined as we just explained.

As $z\to\infty$,
\begin{equation}
{\wt{g}}(z)=z(\log z-1)+{\wt{g}}_\infty-{\wt{g}}_1 z^{-1}+O(z^{-2})
\end{equation}
with ${\wt{g}}_1$ and ${\wt{g}}_\infty$ independent of $z$.
With similar arguments we obtain that
\begin{equation}
\label{eq:ginftyqplus}
{\wt{g}}_1=g_1+O(t^{-\infty}),\qquad {\wt{g}}_\infty=g_\infty+O(t^{-\infty}),
\end{equation}
where $g_1$ and $g_\infty$ have been explicitly computed in~Proposition~\ref{prop:g1ginftyqplus} and the error terms are uniform for $x\in[x_*+\delta,2-\delta]$ for any $\delta>0$.
Moreover, we have
\begin{equation}
\label{eq:jumpqplus}
\begin{aligned}
{\wt{g}}_+(\mu)-{\wt{g}}_-(\mu) &= 2\pi\i \mu, &&\mu\in (-\infty,{\wt{a}}),\\
{\wt{g}}_+(\mu)-{\wt{g}}_-(\mu) &= 2\pi\i(\mu + {\wt{\mathcal{L}}}), && \mu\in ({\wt{b}},{\wt{c}})\\
{\wt{g}}_+(\mu)+{\wt{g}}_-(\mu)&=-\wt V_{\eta,x}(\mu)+\wt{\ell} &&\mu\in ({\wt{a}},{\wt{b}})\cup({\wt{c}},{\wt{d}}),
\end{aligned}
\end{equation}
where $\wt{\mathcal{L}}=\wt{\mathcal{L}}(x)$ is given by
\begin{equation}
\label{eq:Omega2cut}
{\wt{\mathcal{L}}}=-{\wt{c}}-\frac 1{2\pi\i}\int_{\wt{c}}^{\wt{d}}\bigl({\wt{g}}'_+(\mu)-{\wt{g}}_-'(\mu)\bigr)\d \mu.
\end{equation}
In deriving~\eqref{eq:jumpqplus}, we have also used
\begin{equation}
\int_{({\wt{a}},{\wt{b}})\cup({\wt{c}},{\wt{d}})}\bigl({\wt{g}}'_+(\mu)-{\wt{g}}_-'(\mu)\bigr)\d \mu=-2\pi\i ({\wt{a}}-{\wt{b}}+{\wt{c}}),
\end{equation}
which follows from Cauchy's theorem.
Moreover,
\begin{equation}
\label{eq:asympellqplus}
\wt{\ell}= \wt V_{\eta,x}({\wt{d}})  = \eta(d-x)+O(t^{-\infty}),
\end{equation}
as $t\to+\infty$, uniformly for $x\in[x_*+\delta,2-\delta]$ for any $\delta>0$.

\begin{proposition}
\label{prop:L}
    We have ${\wt{\mathcal{L}}}=\mathcal{L}+O(t^{-\infty})$ with
    \begin{equation}
    \label{eq:L}
    \mathcal{L}(x) = \mathcal{U}\bigl(\mathcal{K}(x)\bigr)\frac{\mathcal{V}\bigl(\mathcal{K}(x)\bigr) - 1}{\mathcal{K}(x)}=-\left.\frac{\partial \mathcal{U}(K)}{\partial K}\right|_{K=\mathcal{K}(x)}.
    \end{equation}
\end{proposition}
\begin{proof}
 By~\eqref{eq:Omega2cut}, we have
 \begin{equation}
 {\wt{\mathcal{L}}} = -{\wt{c}}+\frac 1{2\pi\i}\int_{\wt{d}}^{\wt{c}}\bigl({\wt{g}}'_+(\mu)-{\wt{g}}_-'(\mu)\bigr)\d \mu = -{\wt{c}} +\frac 1{2\pi\i}\left({\wt{g}}_+({\wt{c}})-{\wt{g}}_-({\wt{c}})\right).
 \end{equation}
Therefore ${\wt{\mathcal{L}}}=\mathcal{L}+O(t^{-\infty})$ with $\mathcal{L}=-c +\frac 1{2\pi\i}\left(F(\i\pi)-F(-\i\pi)\right)$, where $F$ is defined in~\eqref{eq:lastintegral}.
Using the expression for $F$ given in~\eqref{eq:F}, we evaluate $\frac{1}{2 \pi \i}\left(F(\i \pi) - F(-\i\pi) \right)$ using the following identities: first,
\begin{equation}
\left(f(w)z(w)+d\frac{\eta}2\right)\Big|_{w = -\i \pi}^{w = \i \pi} =
f(w)z(w)\Big|_{w = -\i \pi}^{w = \i \pi} =(g'_+(c)-g'_-(c)) c=2\pi\i \,c,
\end{equation}
stemming from $z(\pm\i\pi)=c$ and~\eqref{eq:gprimefrakrightjumps}; then,
\begin{equation}
-w\left(A\bigl(d - 2 m\zeta(w_\infty)\bigr)+m\frac{\sigma''(2w_\infty)}{\sigma(2w_\infty)}\right)\Big|_{w = -\i \pi}^{w = \i \pi} = 2\pi\i\,m \left( \left(\frac{\zeta(\i \pi)}{\i\pi}\eta - \zeta(\eta) \right)^2 - \wp(\eta)\right)
\end{equation}
stemming from the explicit expression of $A = -\frac{\zeta(\i \pi)}{\i \pi}\eta$, the expression \eqref{eq:endpointselliptic3} for $d$, and the identity $\frac{\sigma''(w)}{\sigma(w)} = \zeta^2(w) - \wp(w)$; finally,
\begin{equation}
m(\zeta(w_\infty + w) - \zeta(w_\infty - w))\Big|_{w = -\i \pi}^{w = \i \pi} = 2m\left( \zeta(w_\infty + \i \pi) - \zeta(w_\infty - \i \pi)\right) = 4m\zeta(\i \pi).
\end{equation}
Combinining these identities and using~\eqref{eq:f1} and~\eqref{eq:f2} completes the proof.
\end{proof}

\subsection{Normalization of the continuous Riemann--Hilbert problem}

Let $\Sigma_N=\Sigma_M$ and $\Sigma_N^\circ=\Sigma_M^\circ\setminus\lbrace {\wt{d}}\rbrace$, with the same orientation.
Introduce the analytic matrix function $\boldsymbol N:\mathbb{C}\setminus \Sigma_N\to\mathrm{SL}(2,\mathbb{C})$ by
\begin{equation}
\label{eq:defNqplus}
\boldsymbol N(z)=\begin{pmatrix}
-\frac 1{2\pi\i} & 0 \\ 0 & 1
\end{pmatrix}\,\e^{t(\frac{\wt{\ell}} 2-{\wt{g}}_\infty)\boldsymbol\sigma_3}\,\boldsymbol M(z)\,\e^{t ({\wt{g}}(z)-\frac {\wt{\ell}} 2)\boldsymbol\sigma_3}\,\begin{pmatrix}
-2\pi\i & 0 \\ 0 & 1
\end{pmatrix}.
\end{equation}
The construction of ${\wt{g}}(z)$ carried out in the previous paragraph ensures that $\boldsymbol N(z)$ is the unique solution to the following Riemann--Hilbert problem.

\begin{cRHp}
\label{cRHp:Nqplus}
Find an analytic function $\boldsymbol N:\mathbb{C}\setminus \Sigma_N\to\mathrm{SL}(2,\mathbb{C})$ such that the following conditions hold true.
\begin{enumerate}[leftmargin=*]
\item Non-tangential boundary values of $\boldsymbol N$ exist and are continuous on $\Sigma_N^\circ$ and satisfy
\begin{equation}
\boldsymbol N_+(z)=\boldsymbol N_-(z)\boldsymbol J_N(z),\qquad z\in \Sigma_N^\circ,
\end{equation}
where $\boldsymbol J_N(z)$ is given for $z\in \Sigma_N^\circ$ by 
\begin{equation}
\label{eq:JNqplus1}
\boldsymbol J_N(z)=
\begin{pmatrix}
-\frac 1{2\pi\i} & 0 \\ 0 & 1
\end{pmatrix}
\e^{-t({\wt{g}}_-(z)-\frac{\wt{\ell}} 2)\boldsymbol\sigma_3}\boldsymbol J_M(z)\e^{t({\wt{g}}_+(z)-\frac{\wt{\ell}} 2)\boldsymbol\sigma_3}
\begin{pmatrix}
-2\pi\i & 0 \\ 0 & 1
\end{pmatrix}
\end{equation}
if $z\in (-\infty,{\wt{a}})\cup({\wt{a}},{\wt{b}})\cup({\wt{b}},{\wt{c}})\cup({\wt{c}},{\wt{d}})$ and by
\begin{equation}
\boldsymbol J_N(z)=\begin{pmatrix}
-\frac 1{2\pi\i} & 0 \\ 0 & 1
\end{pmatrix}
\e^{-t({\wt{g}}(z)-\frac{\wt{\ell}} 2)\boldsymbol\sigma_3}\boldsymbol J_M(z)\e^{t({\wt{g}}(z)-\frac{\wt{\ell}} 2)\boldsymbol\sigma_3}
\begin{pmatrix}
-2\pi\i & 0 \\ 0 & 1
\end{pmatrix}
\end{equation}
otherwise.
\item We have $\boldsymbol N(z)\to\boldsymbol {\mathrm{I}}$ as $z\to\infty$  uniformly in $\mathbb{C}\setminus \Sigma_N$.
\item We have $\boldsymbol N(z)=O(1)$ as $z\to z_0$ uniformly in $\mathbb{C}\setminus\Sigma_N$ for all $z_0\in\Sigma_N\setminus\Sigma_N^\circ$.
\end{enumerate}
\end{cRHp}

We note from~\eqref{eq:refinedexpansionMqplus} and~\eqref{eq:defNqplus} that
\begin{equation}
\label{eq:refinedexpansionNqplus}
\boldsymbol N(z)=\boldsymbol {\mathrm{I}}+\begin{pmatrix}
\wh\alpha-t(1+{\wt{g}}_1)+\tfrac 1{24t} & \i\e^{t(\wt{\ell}-2{\wt{g}}_\infty)}\wh\beta \\ -\i\e^{t(2{\wt{g}}_\infty-\wt{\ell})}\wh\gamma & -\wh\alpha+t(1+{\wt{g}}_1)-\tfrac 1{24t}
\end{pmatrix}z^{-1}+O\bigl(z^{-2}\bigr)
\end{equation}
as $z\to\infty$ uniformly in $\mathbb{C}\setminus \Sigma_N$.
Here, $\wh\alpha=\wh\alpha(t,x)$, $\wh\beta=\wh\beta(t,x)$, and $\wh\gamma=\wh\gamma(t,x)$ are as in~\eqref{eq:wtalphabetagamma}

To write down the jump matrix $\boldsymbol J_N(z)$ in a more explicit way, it is convenient to introduce
\begin{equation}
\label{eq:varphiqplus}
\varphi(z) = 2{\wt{g}}(z)+\wt V_{\eta,x}(z)-\wt{\ell}
\end{equation}
as well as
\begin{equation}
\varphi_1(z) = \varphi(z)\mp2\pi \i z,\quad
\varphi_2(z) = \varphi(z)\mp 2\pi\i (z+{\wt{\mathcal{L}}}),
\qquad \pm\Im z>0.
\end{equation}

We introduce the following quantities, with $z_0\in\lbrace a,b,c,d\rbrace$:
\begin{equation}\label{eq:Tz0_Sz0}
T_{z_0} = \prod_{z\in\lbrace a,b,c,d\rbrace\setminus\lbrace z_0\rbrace}|z_0-z|^{1/2},\quad
S_{z_0} =\pm\frac 16 \sum_{z\in\lbrace a, b, c, d\rbrace\setminus\lbrace z_0\rbrace}(z_0-z)^{-1} ,
\end{equation}
where the sign is $+$ for $z_0\in\lbrace a,c\rbrace$ and $-$ for $z_0\in\lbrace b,d\rbrace$, as well as
\begin{equation}
\label{eq:Az0Bz0}
\begin{aligned}
A_{z_0} &= \frac {8\mathcal{U}(\mathcal{K})}{3T_{z_0}}C_{z_0},
\\
B_{z_0} &= \frac {8\mathcal{U}(\mathcal{K})}{5T_{z_0}}\left(S_{z_0}C_{z_0}-\frac {2\mathcal{U}(\mathcal{K})^2}{3T_{z_0}^2}\left(\wp'(\frac{\eta}2+w(z_0))+\wp'(\frac{\eta}2-w(z_0))\right)\right),
\\
C_{z_0} &= \frac 1{2K}\left(\frac{\vartheta_{11}'}{\vartheta_{11}}\bigl(\frac {\frac{\eta}2+w(z_0)}{2K}\bigr)+\frac{\vartheta_{11}'}{\vartheta_{11}}\bigl(\frac {\frac{\eta}2-w(z_0)}{2K}\bigr)+\eta\right)
\\
&=-\frac{\zeta(\i\pi)}{\i\pi}\eta+\zeta\bigl(\frac{\eta}{2}+w(z_0)\bigr)+\zeta\bigl(\frac{\eta}{2}-w(z_0)\bigr),\\
\end{aligned}
\end{equation}
where $\vartheta_{11}(w)=\vartheta_{11}(w|\i\pi/\mathcal{K})$ and $\frac{\vartheta_{11}'}{\vartheta_{11}}$ is the log-derivative in the argument of the theta function, and the half-periods of $\wp'$ and $\zeta$ are $\mathcal{K},\i\pi$, and $\eta=-\log q$.
The equivalence of the two expressions for $C_{z_0}$ follows from~\eqref{eq:LegendreIdentity} and~\eqref{eq:relzetatheta}.

We recall that $w(z)$ is the conformal transformation defined in~\eqref{eq:conformalelliptic}.
In particular, $w(a)=K$, $w(b)=\mathcal{K}+\i\pi$, $w(c)=\i\pi$, and $w(d)=0$, and we have the expansion
\begin{equation}
\label{eq:expansionwendpoints}
w(z)=w(z_0)\pm\frac{2m}{T_{z_0}}\biggl((\pm(z-z_0))^{1/2}+S_{z_0}(\pm(z-z_0))^{3/2}+O\bigl((z-z_0)^{5/2}\bigr)\biggr)
\end{equation}
as $z\to z_0\in\lbrace a,b,c,d\rbrace$, where the sign is~$+$ if~$z_0\in\lbrace b,d\rbrace$ and~$-$ if~$z_0\in\lbrace a,c\rbrace$.

\begin{proposition}
\label{prop:varphiqplus}
The following properties hold true, for $t$ sufficiently large.
\begin{enumerate}[leftmargin=*]
    \item The function $\varphi(z)$ is analytic for $z\in\mathbb{C}\setminus\bigl((-\infty,{\wt{d}}]\cup(\i \mathbb{R}+x)\bigr)$.
    It has non-tangential boundary values $\varphi_\pm(\mu)$ for all $\mu\in(-\infty,{\wt{d}})\setminus\lbrace {\wt{a}},{\wt{b}},{\wt{c}}\rbrace$ such that
    \begin{align}
        \label{eq:jumpphiqplus1}
        \varphi_\pm(\mu)&=\pm({\wt{g}}_+(\mu)-{\wt{g}}_-(\mu)\bigr),&&\mu\in({\wt{a}},{\wt{b}})\cup({\wt{c}},{\wt{d}}),
        \\
        \label{eq:jumpphiqplus2}
        \varphi_+(\mu)-\varphi_-(\mu)&= 4\pi\i\,\mu,&&\mu\in(-\infty,{\wt{a}}),
        \\
        \label{eq:jumpphiqplus3}
        \varphi_+(\mu)-\varphi_-(\mu)&= 4\pi\i (\mu+{\wt{\mathcal{L}}}),&&\mu\in({\wt{b}},{\wt{c}}).
    \end{align}
    where ${\wt{\mathcal{L}}}$ is defined in~\eqref{eq:Omega2cut}. 
    \item There exist a neighborhood of $z={\wt{d}}$ and a function $\varphi_{\wt{d}}(z)$ analytic in that neighborhood such that $\varphi(z)=(z-{\wt{d}})^{3/2}\varphi_{\wt{d}}(z)$ (with principal branch) and that
    \begin{equation}
        \varphi_{\wt{d}}({\wt{d}}) = A_d+O(t^{-\infty}),
        \qquad
        \varphi_{\wt{d}}'({\wt{d}}) = B_d+O(t^{-\infty}),
    \end{equation}

    \item The function $\varphi_1(z)$ is analytic for $z\in\mathbb{C}\setminus\bigl([{\wt{a}},+\infty)\cup(\i\mathbb{R}+x)\bigr)$.
    Moreover, there exist a neighborhood of $z={\wt{a}}$ and a function $\varphi_{\wt{a}}(z)$ analytic in that neighborhood and such that $\varphi_1(z) = -({\wt{a}}-z)^{3/2}\varphi_{\wt{a}}(z)$ (with principal branch) and that
    \begin{equation}
        \varphi_{\wt{a}}({\wt{a}}) = -A_a+O(t^{-\infty}),
        \qquad
        \varphi_{\wt{a}}'({\wt{a}}) = -B_a+O(t^{-\infty}),
    \end{equation}

    \item The function $\varphi_2(z)$ is analytic for $z\in\mathbb{C}\setminus\bigl((-\infty,{\wt{b}}]\cup[{\wt{c}},+\infty)\cup(\i\mathbb{R}+x)\bigr)$.
    Moreover, there exist neighborhoods of $z={\wt{b}}$ and $z={\wt{c}}$ as well as functions $\varphi_{\wt{b}}(z)$ and $\varphi_{\wt{c}}(z)$ analytic in these neighborhoods and such that
    \begin{equation}
        \varphi_2(z) = -(z-{\wt{b}})^{3/2}\varphi_{\wt{b}}(z),\quad
        \varphi_2(z) = -({\wt{c}}-z)^{3/2}\varphi_{\wt{c}}(z),
    \end{equation}
    (valid in the respective neighborhoods, with principal branches of the square roots) and that
    \begin{equation}
    \begin{aligned}
        \varphi_{\wt{b}}({\wt{b}}) &= A_b+O(t^{-\infty}),
        &&\varphi_{\wt{b}}'({\wt{b}}) =B_b+O(t^{-\infty}),
        \\
        \varphi_{\wt{c}}({\wt{c}}) &= A_c+O(t^{-\infty}),
        &&\varphi_{\wt{c}}'({\wt{c}}) = B_c+O(t^{-\infty}),
    \end{aligned}
    \end{equation}
\end{enumerate}
The neighborhoods in the above statements can be chosen independent of $t$ and $x$, provided $x\in [x_*+\delta,2-\delta]$ for some $\delta>0$.
\end{proposition}

\begin{proof}
These properties are simple consequences of the definition so we only comment on the proof of the statement about the local structures at $z={\wt{a}},{\wt{b}},{\wt{c}},{\wt{d}}$.

First, by~\eqref{eq:gprimeqplus} (first line) we get
\begin{equation}
\label{eq:varphiprimeqplusexplicit}
\varphi'(z)=2{\wt{g}}'(z)+\wt V_{\eta,x}'(z) = 
2{\wt{r}}(z)\left(\int_{(-\infty,{\wt{a}})\cup({\wt{b}},{\wt{c}})} \frac{\d\mu}{{\wt{r}}(\mu)(\mu-z)}-\frac{1}{2\pi\i}\int_{({\wt{a}},{\wt{b}})\cup({\wt{c}},{\wt{d}})} \frac{\wt V_{\eta,x}'(\mu)-\wt V_{\eta,x}'(z)}{{\wt{r}}_+(\mu)(\mu-z)}\d\mu\right)
\end{equation}
and so $\varphi'(z)$ equals $(z-{\wt{d}})^{1/2}$ times a function of $z$ analytic in a neighborhood of $z={\wt{d}}$.
By integrating in $z$, noting that $\varphi(z)\to 0$ as $z\to {\wt{d}}$, we get $\varphi(z)=(z-{\wt{d}})^{3/2}\varphi_{\wt{d}}(z)$ for a function $\varphi_{\wt{d}}(z)$ analytic for $z$ in a neighborhood of $z={\wt{d}}$.
Moreover, since $\wt V_{\eta,x}'(z)=\eta+O(t^{-\infty})$ for $z$ in a neighborhood of $z={\wt{d}}$, by~\eqref{eq:gprimeellipticprelimit} we get
\begin{equation}
\label{eq:nonsodavverocomechiamarti}
\frac{\varphi'(z)}{(z-{\wt{d}}) ^{1/2}} = -2\frac{\zeta(\i\pi|{\wt{K}},\i\pi)}{\i\pi}\eta\frac{{\wt{w}}(z)}{(z-{\wt{d}})^{1/2}}-\frac 2{(z-{\wt{d}})^{1/2}}\log\frac{\sigma\bigl({\wt{w}}_\infty-{\wt{w}}(z)|{\wt{K}},\i\pi\bigr)}{\sigma\bigl({\wt{w}}_\infty+{\wt{w}}(z)|{\wt{K}},\i\pi\bigr)}+O(t^{-\infty}).
\end{equation}
It is now straightforward to complete the proof, also using the expansion~\eqref{eq:expansionwendpoints}.

The second line of~\eqref{eq:gprimeqplus} gives (assuming $\pm\Im z>0$)
\begin{equation}
\begin{aligned}
-\varphi_1'(z)=-\varphi_2'(z)&=-2{\wt{g}}'(z)-\wt V_{\eta,x}'(z)\mp 2\pi\i 
\\
&= 2{\wt{r}}(z)\left(\int_{{\wt{d}}}^{+\infty} \frac{\d\mu}{{\wt{r}}(\mu)(\mu-z)}+\frac{1}{2\pi\i}\int_{({\wt{a}},{\wt{b}})\cup({\wt{c}},{\wt{d}})} \frac{\wt V_{\eta,x}'(\mu)-\wt V_{\eta,x}'(z)}{{\wt{r}}_+(\mu)(\mu-z)}\d\mu\right)
\end{aligned}
\end{equation}
and the other claims are proved in a similar way.
\end{proof}

\begin{remark}
    \label{remark:positivitycoefficientsvarphiabcdqplus}
    From the formulas of this proposition we observe that that $\varphi_{z_0}(z_0)=c_{z_0}+O(t^{-\infty})$ for some positive constant $c_{z_0}>0$, for $z_0\in\lbrace\wt a,\wt b,\wt c,\wt d\rbrace$.
    This follows from the variational analysis of Section~\ref{sec:minimizationxq<x<2}, which ensures that that $C_{b},C_{c},C_{d}>0$ and~$C_{a}<0$, see~\eqref{eq:ineqzeta}, \eqref{eq:needtoestablish1}, and~\eqref{eq:needtoestablish2}, as well as from the fact that $\mathcal{U}(K)>0$ (see the proof of Proposition~\ref{prop:f1f2body}) and that $T_{z_0}>0$.
\end{remark}

We can use these functions to write the jump matrix $\boldsymbol J_N(z)$ as
\begin{equation}
\label{eq:JNqmplusallcases}
\boldsymbol J_N(z)=\begin{cases}
\begin{pmatrix}
(1+\e^{\pm 2\pi\i t z})^{\pm 1} & 0 \\ \e^{t\varphi(z)}& (1+\e^{\pm 2\pi\i t z})^{\mp 1}
\end{pmatrix}
&z\in\Gamma_L^\pm,\\
\begin{pmatrix}
1 & - \frac{\e^{-t\varphi(z)}}{1+\e^{\mp 2\pi\i tz}} \\ 0 & 1
\end{pmatrix}
&z\in ({\wt{a}}^\pm,{\wt{b}}^\pm)\cup({\wt{c}}^\pm,+\infty\pm\i\epsilon),\\
\begin{pmatrix}
\e^{t\varphi_+(z)} & 1 \\ 0 & \e^{t\varphi_-(z)}
\end{pmatrix}
&z\in ({\wt{a}},{\wt{b}})\cup({\wt{c}},{\wt{d}}),\\
\begin{pmatrix}
1 & \e^{-t\varphi(z)} \\ 0 & 1
\end{pmatrix}
&z\in ({\wt{d}},+\infty),\\
\begin{pmatrix}
	1 & 0 \\
	-\e^{t\varphi_1(z)} & 1
\end{pmatrix}
&z\in (-\infty,{\wt{a}}),
\\
\begin{pmatrix}
	\e^{2\pi\i t{\wt{\mathcal{L}}}} & 0 \\
	-\e^{t\varphi_2(z)} & 	\e^{-2\pi\i t{\wt{\mathcal{L}}}}
\end{pmatrix}
&z\in ({\wt{b}},{\wt{c}}),
\\
\begin{pmatrix}
	1 & \mp \e^{-t\varphi_1(z) } \\
	\mp \e^{t\varphi(z)} & 1 + \e^{\pm 2 \pi \i t z}
\end{pmatrix} & z\in\gamma_{\wt a}^\pm\cup\gamma_{\wt c}^\pm.
\\
\begin{pmatrix}
	1+\e^{\pm2\pi\i tz} & \pm \e^{-t\varphi_1(z) } \\
	\pm \e^{t\varphi(z)} & 1 
\end{pmatrix} & z\in \gamma_{\wt b}^\pm.
\end{cases}
\end{equation}

\subsection{Lens opening}

Thanks to Proposition~\ref{prop:varphiqplus} we now fix the contours $\gamma_{\wt a}^\pm,\gamma_{\wt b}^\pm,\gamma_{\wt c}^\pm$: we take $\gamma_{\wt a}^\pm$ as the loci where $\Im \varphi_1(z)=0$ and $0\leq \pm\Im z\leq \epsilon$; similarly, we take $\gamma_{\wt b}^\pm$ and $\gamma_{\wt c}^\pm$ as the loci $\Im \varphi_2(z)=0$ and $0\leq \pm\Im z\leq \epsilon$.
We also introduce contours $\gamma_{\wt d}^\pm$ (starting at ${\wt{d}}$ and ending on the lines $\Im z=\pm\epsilon$) as the loci where $\Im \varphi(z)=0$, $0\leq \pm\Im z\leq \epsilon$.
We define
\begin{equation}
\Sigma_T = \Sigma_N\cup\gamma_{\wt d}^+\cup\gamma_{\wt d}^-.
\end{equation}
Denoting ${\wt{d}}^\pm$ the intersection points of $\gamma_{\wt d}^\pm$ with $\Im z=\pm\epsilon$, we set $\Sigma_T^\circ=\Sigma_T\setminus\bigcup_{z_0\in\lbrace\wt a,\wt b,\wt c,\wt d\rbrace}\lbrace z_0,z_0^+,z_0^-\rbrace$.
We orient $\Sigma_T^\circ$ as $\Sigma_N^\circ$, with the additional curves also oriented upwards.
This is illustrated in Figure~\ref{fig:SigmaTqplus}; in particular the curves $\gamma_{z_0}^\pm$ (where $z_0\in\lbrace\wt a,\wt b,\wt c,\wt d\rbrace$), $\lbrace\Im z=\pm\epsilon\rbrace$, and $\lbrace\Im z=0\rbrace$ delimit bounded regions $\mathscr{L}_\pm$ (the ``lenses'') and we define
\begin{equation}
\boldsymbol T(z)=
\begin{cases}
\boldsymbol N(z)\begin{pmatrix}
1 & 0 \\ \mp\e^{t\varphi(z)} & 1
\end{pmatrix} & \text{if }z\in\mathscr{L}_\pm,\\
\boldsymbol N(z) & \text{otherwise}.
\end{cases}
\end{equation}

\begin{figure}[t]
\centering
\begin{tikzpicture}

\draw[very thin,->] (-6.5,0) -- (6.5,0);
\draw[very thin,->] (0,-2) -- (0,2);

\draw[very thick,->] (-6,0) -- (-4.5,0);
\draw[very thick,->] (-4.5,0) -- (-1.5,0);
\draw[very thick,->] (-1.5,0) -- (1,0);
\draw[very thick,->] (1,0) -- (3.5,0);
\draw[very thick,->] (3.5,0) -- (5.5,0);
\draw[very thick] (5.5,0) -- (6,0);

\draw[very thick] (-4,1) -- (-5,3/2);
\draw[very thick,<-] (-5,3/2) -- (-6,2);

\draw[very thick] (-4,-1) -- (-5,-3/2);
\draw[very thick,<-] (-5,-3/2) -- (-6,-2);

\draw[very thick,->] (-4,1) -- (-1.5,1);
\draw[very thick,->] (-1.5,1) -- (1,1);
\draw[very thick,->] (1,1) -- (3.5,1);
\draw[very thick,->] (3.5,1) -- (5.5,1);
\draw[very thick] (5.5,1) -- (6,1);

\draw[very thick,->] (-4,-1) -- (-1.5,-1);
\draw[very thick,->] (-1.5,-1) -- (1,-1);
\draw[very thick,->] (1,-1) -- (3.5,-1);
\draw[very thick,->] (3.5,-1) -- (5.5,-1);
\draw[very thick] (5.5,-1) -- (6,-1);

\begin{scope}[shift={(-1,0)}]
\draw[very thick,->] 
  (-1.3,-1) .. controls (-1.6,-0.6) .. (-1.8,-0.5) ;
\draw[very thick] 
  (-1.8,-0.5) .. controls (-2.1,-0.3) .. (-2.3,0);
\draw[very thick,->] 
  (-2.3,0) .. controls (-2.1,0.3) .. (-1.8,0.5) ;
\draw[very thick] 
  (-1.8,0.5) .. controls (-1.6,0.6) .. (-1.3,1);

\draw[very thick,->] 
  (.3,-1) .. controls (.5,-0.6) .. (.8,-1/2);
\draw[very thick] 
  (.8,-1/2) .. controls (1.1,-0.3) .. (1.3,0);
\draw[very thick,->] 
  (1.3,0) .. controls (1.1,0.3) .. (.8,1/2);
\draw[very thick] 
  (.8,1/2) .. controls (.5,0.6) .. (.3,1);
\end{scope}

\begin{scope}[shift={(4,0)}]
\draw[very thick,->] 
  (-1.3,-1) .. controls (-1.6,-0.6) .. (-1.8,-0.5) ;
\draw[very thick] 
  (-1.8,-0.5) .. controls (-2.1,-0.3) .. (-2.3,0);
\draw[very thick,->] 
  (-2.3,0) .. controls (-2.1,0.3) .. (-1.8,0.5) ;
\draw[very thick] 
  (-1.8,0.5) .. controls (-1.6,0.6) .. (-1.3,1);
\end{scope}

\begin{scope}[shift={(3.7,0)}]
\draw[very thick,->] 
  (.3,-1) .. controls (.5,-0.6) .. (.8,-1/2);
\draw[very thick] 
  (.8,-1/2) .. controls (1.1,-0.3) .. (1.3,0);
\draw[very thick,->] 
  (1.3,0) .. controls (1.1,0.3) .. (.8,1/2);
\draw[very thick] 
  (.8,1/2) .. controls (.5,0.6) .. (.3,1);
\end{scope}
  
\node at (-3.4,-1/5) {\small{${\wt{a}}$}};
\node at (-2.1,6/5) {\small{${\wt{a}}^+$}};
\node at (-2.1,-6/5) {\small{${\wt{a}}^-$}};

\node at (.35,-1/5) {\small{${\wt{b}}$}};
\node at (-.5,6/5) {\small{${\wt{b}}^+$}};
\node at (-.5,-6/5) {\small{${\wt{b}}^-$}};

\node at (1.6,-1/5) {\small{${\wt{c}}$}};
\node at (2.7,6/5) {\small{${\wt{c}}^+$}};
\node at (2.7,-6/5) {\small{${\wt{c}}^-$}};

\begin{scope}[shift={(4.7,0)}]
\node at (.35,-1/5) {\small{${\wt{d}}$}};
\node at (-.5,6/5) {\small{${\wt{d}}^+$}};
\node at (-.5,-6/5) {\small{${\wt{d}}^-$}};
\end{scope}

\node at (-1.5,1/2) {\small{$\mathscr L_+$}};
\node at (-1.5,-1/2) {\small{$\mathscr L_-$}};
\node at (3.4,1/2) {\small{$\mathscr L_+$}};
\node at (3.4,-1/2) {\small{$\mathscr L_-$}};

\end{tikzpicture}
\caption{$\Sigma_T$, its orientation, and the ``lenses'' $\mathscr{L}_\pm$ (case $x_*<x<2$).}
\label{fig:SigmaTqplus}
\end{figure}
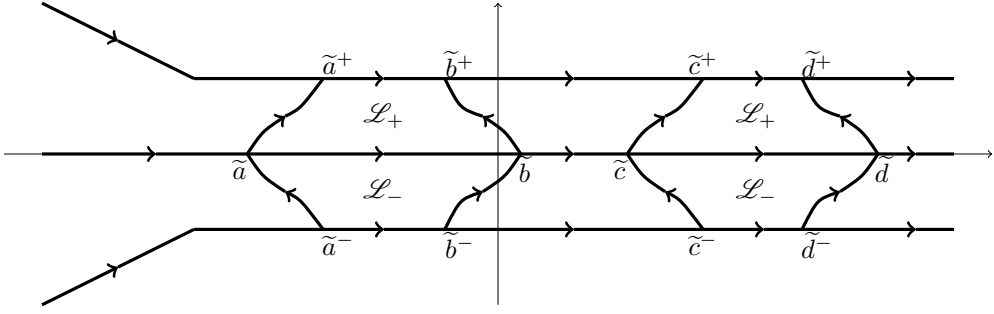

It is important to note that we can perform this transformation because $\varphi(z)$ is analytic for $z\in\mathscr{L}_\pm$, see~Proposition~\ref{prop:varphiqplus}.

It is clear that $\boldsymbol T$ solves the following Riemann--Hilbert problem.

\begin{cRHp}
Find an analytic function $\boldsymbol T:\mathbb{C}\setminus \Sigma_T\to\mathrm{SL}(2,\mathbb{C})$ such that the following conditions hold true.
\begin{enumerate}[leftmargin=*]
\item Non-tangential boundary values of $\boldsymbol T$ exist and are continuous on $\Sigma_T^\circ$ and satisfy
\begin{equation}
\boldsymbol T_+(z)=\boldsymbol T_-(z)\boldsymbol J_T(z),\qquad z\in \Sigma_T^\circ,
\end{equation}
where $\boldsymbol J_T(z)$ is given explicitly below.
\item We have $\boldsymbol T(z)\to\boldsymbol {\mathrm{I}}$ as $z\to\infty$  uniformly in $\mathbb{C}\setminus \Sigma_T$.
\item We have $\boldsymbol T(z)=O(1)$ as $z\to z_0$ uniformly in $\mathbb{C}\setminus\Sigma_T$ for all $z_0\in\Sigma_T\setminus\Sigma_T^\circ$.
\end{enumerate}
\end{cRHp}

It follows from the factorization
\begin{equation}
\begin{pmatrix}
\e^{t\varphi_+(z)} & 1 \\ 0 & \e^{t\varphi_-(z)}
\end{pmatrix} =
\begin{pmatrix}
1 & 0 \\ \e^{t\varphi_-(z)} & 1
\end{pmatrix}
\begin{pmatrix}
0 & 1 \\ -1 & 0
\end{pmatrix}
\begin{pmatrix}
1 & 0 \\ \e^{t\varphi_+(z)} & 1
\end{pmatrix}
\end{equation}
that the jump matrix $\boldsymbol J_T(z)$ for $z\in\Sigma_T^\circ$ is given explicitly by
\begin{equation}
\label{eq:JTqplus}
\boldsymbol J_T(z)=\begin{cases}
\begin{pmatrix}
0 & 1\\ -1 & 0
\end{pmatrix}
&\text{if }z\in ({\wt{a}},{\wt{b}})\cup({\wt{c}},{\wt{d}})
\\
\begin{pmatrix}
1 & -\frac{\e^{-t\varphi(z)}}{1+\e^{\mp2\pi\i t z}} \\ \e^{t\varphi(z)} & \frac{1}{1+\e^{\pm 2\pi\i t z}}
\end{pmatrix}
&\text{if }z\in ({\wt{a}}^\pm,{\wt{b}}^\pm)\cup({\wt{c}}^\pm,{\wt{d}}^\pm)
\\
\begin{pmatrix}
1 & \mp \e^{-t\varphi_1(z)} \\ 0 & 1
\end{pmatrix}
&\text{if }z\in\gamma_{\wt a}^\pm\cup\gamma_{\wt c}^\pm\text{ or }z\in \gamma_{\wt b}^\mp,
\\
\begin{pmatrix}
1 & 0 \\ \mp \e^{t\varphi(z)} & 1
\end{pmatrix}
&\text{if }z\in\gamma_{\wt d}^\pm,
\\
\boldsymbol J_N(z)
&\text{otherwise}.
\end{cases}
\end{equation}

Similarly to the previous section, the key point of this construction is that the jump matrices $\boldsymbol J_T(z)$ are exponentially close to the identity except on $({\wt{a}},{\wt{b}})\cup({\wt{b}},{\wt{c}})\cup({\wt{c}},{\wt{d}})$ and in (arbitrarily small) neighborhoods of ${\wt{a}},{\wt{b}},{\wt{c}},{\wt{d}}$; moreover, $\boldsymbol J_T(z)$ is exponentially close to $\e^{2\pi\i t{\wt{\mathcal{L}}}\boldsymbol\sigma_3}$ uniformly for $z\in[{\wt{b}}+\epsilon,{\wt{c}}-\epsilon]$.
Namely, let
\begin{equation}
\label{eq:SigmaTepsilonqplus}
\wt{\Sigma_T}^\epsilon = \Sigma_T^\circ\setminus\bigl(({\wt{a}}-\epsilon,{\wt{d}}+\epsilon)\cup\left(\bigcup_{z_0\in\lbrace\wt a,\wt b,\wt c,\wt d\rbrace}\bigl(\gamma_{z_0}^+\cup\gamma_{z_0}^-\bigr)\right).
\end{equation}

\begin{proposition}\label{prop:JTsmallqplus}
    For any $\epsilon>0$ small enough and for any $\delta>0$, there exists $c>0$ such that the following estimates hold uniformly for $x\in[x_*+\delta,2-\delta]$.
    
    \noindent \textit{(1)} We have $\boldsymbol J_T(z)=\boldsymbol {\mathrm{I}}+O\left(\frac{1}{|z|^2+1}\e^{-ct}\right)$ as $t\to+\infty$ uniformly for $z\in \wt{\Sigma_T}^\epsilon$.
     
    \noindent \textit{(2)} We have $\boldsymbol J_T(z)\e^{-2\pi\i t{\wt{\mathcal{L}}}\boldsymbol\sigma_3}=\boldsymbol {\mathrm{I}}+O\left(\e^{-ct}\right)$ as $t\to+\infty$ uniformly for $z\in({\wt{b}}+\epsilon,{\wt{c}}-\epsilon)$.
\end{proposition}
\begin{proof}
The proof follows by completely similar arguments (which are standard in the Deift--Zhou nonlinear steepest descent method, see, in particular,~\cite{bleher2011uniform}) to those we presented in full detail in the proof of Proposition~\ref{prop:JTsmallqminus} and so we will be brief here.
More specifically, if $g(z)$ is the function constructed in Section~\ref{sec:minimizationxq<x<2}, we have $\wt g(z)=g(z)+O(t^{-\infty})$ (possibly in the sense of boundary values along the real axis) uniformly for $z$ in the complex plane except for fixed neighborhoods of the endpoints and uniformly for $x\in[x_*+\delta,2-\delta]$.
Then, the estimates on the real parts of the contour~$\wt{\Sigma_T}^\epsilon$ and on $(\wt b+\epsilon,\wt c-\epsilon)$ follow from the variational inequalities established in Proposition~\ref{prop:ineqvariational2cut}, which hold in the strict sense of Remark~\ref{remark:strictineqqplus}, as well as from the asymptotic relation $g(z)\sim z(\log z-1)$ for $z\to\infty$.
These estimates then extend by continuity to the desired estimates on the remaining parts of $\wt{\Sigma_T}^\epsilon$ except close to the endpoints $a,b,c,d$ (in which case the statement follows instead from the local structure of $\varphi,\varphi_1,\varphi_2$ near these points, see Proposition~\ref{prop:varphiqplus}) and except on $\bigl((a,b)\cup(c,d)\bigr)\pm\i\epsilon$ (where we instead resort to a standard argument based on the Cauchy--Riemann equations as in the end of the proof of Proposition~\ref{prop:JTsmallqminus}).
\end{proof}

\subsection{Parametrices}

\subsubsection{Outer parametrix}\label{sec:outerparaqplus}

The outer parametrix is obtained by neglecting all jumps of $ \boldsymbol T$ except on $({\wt{a}},{\wt{b}})\cup({\wt{b}},{\wt{c}})\cup({\wt{c}},{\wt{d}})$ and approximating the jump on~$({\wt{b}},{\wt{c}})$ with $\e^{2\pi\i t{\wt{\mathcal{L}}}\boldsymbol\sigma_3}$, which corresponds to the following Riemann--Hilbert problem.

\begin{cRHp}
\label{cRHp:Poutqplus}
Find an analytic function $\boldsymbol P^{\mathrm{out}}:\mathbb{C}\setminus[{\wt{a}},{\wt{d}}]\to\mathrm{SL}(2,\mathbb{C})$ such that the following conditions hold true.
\begin{enumerate}[leftmargin=*]
\item Non-tangential boundary values of $\boldsymbol P^{\mathrm{out}}$ exist and are continuous on $({\wt{a}},{\wt{b}})\cup({\wt{b}},{\wt{c}})\cup({\wt{c}},{\wt{d}})$ and satisfy
\begin{equation}
\label{eq:jumpPouttwocut}
\begin{aligned}
\boldsymbol P_+^{\mathrm{out}}(z)&=\boldsymbol P_-^{\mathrm{out}}(z)\begin{pmatrix} 0 & 1 \\ -1 & 0 \end{pmatrix},&& z\in ({\wt{a}},{\wt{b}})\cup({\wt{c}},{\wt{d}}),
\\
\boldsymbol P_+^{\mathrm{out}}(z)&=\boldsymbol P_-^{\mathrm{out}}(z)\e^{2\pi\i t{\wt{\mathcal{L}}}\boldsymbol\sigma_3},&& z\in ({\wt{b}},{\wt{c}}).
\end{aligned}
\end{equation}
\item We have $\boldsymbol P^{\mathrm{out}}(z)=\boldsymbol {\mathrm{I}}+O(z^{-1})$ as $z\to\infty$ uniformly in $\mathbb{C}$.
\item We have $\boldsymbol P^{\mathrm{out}}(z)=O\bigl(|z-z_0|^{-1/4}\bigr)$ as $z\to z_0$ with $z_0\in\lbrace\wt a,\wt b,\wt c,\wt d\rbrace$.
\end{enumerate}
\end{cRHp}

The solution to this type of model Riemann--Hilbert problem (in the general multi-cut case) is well known to be explicitly expressed in terms of the Riemann Theta function associated with an elliptic or hyperelliptic curve; see~\cite{DeiftEtAl_multicut}.
We report the solution (in our case), following the literature, which is constructed by using the functions
\begin{equation}
\label{eq:xitilde}
\wt{\xi}(z) =\left(\frac{(z-{\wt{b}})(z-{\wt{d}})}{(z-{\wt{a}})(z-{\wt{c}})}\right)^{\frac 14}
\end{equation}
and (the notation for elliptic theta functions is introduced in Section~\ref{sec:apptheta})
\begin{equation}
\label{eq:chitilde}
\wt{\chi}(w) = \frac{\vartheta_{11}(\frac{w}{2{\wt{K}}}-t{\wt{\mathcal{L}}}|\i\pi {\wt{K}}^{-1})}{\vartheta_{11}(\frac w{2{\wt{K}}}|\i\pi {\wt{K}}^{-1})}.
\end{equation}
The function $\wt{\xi}(z)$ is analytic for $z\in\mathbb{C}\setminus\bigl([{\wt{a}},{\wt{b}}]\cup[{\wt{c}},{\wt{d}}]\bigr)$, tends to $1$ as $z\to\infty$, and admits boundary values on $({\wt{a}},{\wt{b}})$ and $({\wt{c}},{\wt{d}})$ satisfying
\begin{equation}
\wt{\xi}_+(\mu)=\i\wt{\xi}_-(\mu),\qquad \mu\in ({\wt{a}},{\wt{b}})\cup({\wt{c}},{\wt{d}}).
\end{equation}
Moreover, by the automorphy properties~\eqref{eq:periodictheta}, we have
\begin{equation}
\wt{\chi}(w+2{\wt{K}})=\wt{\chi}(w),\qquad \wt{\chi}(w+2\i\pi) = \wt{\chi}(w)\e^{2\pi\i t{\wt{\mathcal{L}}}}.
\end{equation}
By using these properties it is easy to check that
\begin{equation}
\wh{\boldsymbol P}^{\mathrm{out}}(z) = \begin{pmatrix}
    \frac{\wt{\xi}(z)+\wt{\xi}(z)^{-1}}{2} \wt{\chi}({\wt{w}}(z)-{\wt{w}}_1)&
     \frac{\wt{\xi}(z)-\wt{\xi}(z)^{-1}}{2\i}\wt{\chi}(-{\wt{w}}(z)-{\wt{w}}_1)
     \\
     -\frac{\wt{\xi}(z)-\wt{\xi}(z)^{-1}}{2\i}\wt{\chi}({\wt{w}}(z)-{\wt{w}}_2)&
     \frac{\wt{\xi}(z)+\wt{\xi}(z)^{-1}}{2}\wt{\chi}(-{\wt{w}}(z)-{\wt{w}}_2)
\end{pmatrix}
\end{equation}
satisfies the desired jump condition~\eqref{eq:jumpPouttwocut}, for any ${\wt{w}}_1,{\wt{w}}_2$.
To enforce the normalization at $\infty$ we set
\begin{equation}
\label{eq:Poutexplicit2cut}
\boldsymbol P^{\mathrm{out}}(z) = \wh{\boldsymbol P}^{\mathrm{out}}(\infty)^{-1}\wh{\boldsymbol P}^{\mathrm{out}}(z) 
=\begin{pmatrix}
    \wt{\chi}({\wt{w}}_\infty-{\wt{w}}_1)^{-1} & 0 \\ 0 & \wt{\chi}(-{\wt{w}}_\infty-{\wt{w}}_2)^{-1}
\end{pmatrix}\wh{\boldsymbol P}^{\mathrm{out}}(z).
\end{equation}
We claim that choosing 
\begin{equation}
\label{eq:w1w2Pout}
{\wt{w}}_1={\wt{w}}_\infty-{\wt{K}}+\i\pi\quad\text{and}\quad {\wt{w}}_2=-{\wt{w}}_\infty+{\wt{K}}+\i\pi
\end{equation}
then $\boldsymbol P^{\mathrm{out}}(z)$ is the unique solution to the Riemann--Hilbert problem~\ref{cRHp:Poutqplus}.
To see this, note that the only singularities of~$\wt{\chi}(w)$ are simple poles at~$w\in\mathbb{Z}+\frac{\i\pi}{{\wt{K}}}\mathbb{Z}$ and recall that the conformal map $z\mapsto {\wt{w}}(z)$ satisfies $\Re{\wt{w}}(z)>0$ for all $z\in\mathbb{C}\setminus[{\wt{a}},{\wt{d}}]$.
Hence, $\wt{\chi}({\wt{w}}(z)-{\wt{w}}_1)$ and $\wt{\chi}(-{\wt{w}}(z)-{\wt{w}}_2)$ have no poles.
However, $\wt{\chi}(-{\wt{w}}(z)-{\wt{w}}_1)$ and $\wt{\chi}({\wt{w}}(z)-{\wt{w}}_2)$ have poles when ${\wt{w}}(z)=-{\wt{w}}_1$ and when ${\wt{w}}(z)={\wt{w}}_2$, respectively.
Making use of the identity
\begin{equation}
\label{eq:identityxitheta}
\wt{\xi}(z)^2=
\frac
{\vartheta_{11}(\frac{{\wt{w}}(z)}{2{\wt{K}}})\vartheta_{11}(\frac{{\wt{w}}(z)-{\wt{K}}-\i\pi}{2{\wt{K}}})
\vartheta_{11}(\frac{{\wt{w}}_\infty-{\wt{K}}}{2{\wt{K}}})\vartheta_{11}(\frac{{\wt{w}}_\infty-\i\pi}{2{\wt{K}}})}
{\vartheta_{11}(\frac{{\wt{w}}(z)-{\wt{K}}}{2{\wt{K}}})\vartheta_{11}(\frac{{\wt{w}}(z)-\i\pi}{2{\wt{K}}})
\vartheta_{11}(\frac{{\wt{w}}_\infty}{2{\wt{K}}})\vartheta_{11}(\frac{{\wt{w}}_\infty-{\wt{K}}-\i\pi}{2{\wt{K}}})}
\end{equation}
(we use the short-hand notation $\vartheta_{11}(w)=\mathfrak \vartheta_{11}(w|\i\pi{\wt{K}}^{-1})$ in this equation)
one easily checks, using~\eqref{eq:periodictheta} and the fact that $\vartheta_{11}$ is odd, that when ${\wt{w}}(z)=-{\wt{w}}_1$ or ${\wt{w}}(z)={\wt{w}}_2$ we have $\wt{\xi}(z)^2=1$ and so the poles of $\wt{\chi}(-{\wt{w}}(z)-{\wt{w}}_1)$ and $\wt{\chi}({\wt{w}}(z)-{\wt{w}}_2)$ are canceled by zeros of $\wt{\xi}(z)-\wt{\xi}(z)^{-1}$.
Finally, the identity~\eqref{eq:identityxitheta} follows from the fact that both sides are meromorphic functions on the Riemann surface of ${\wt{r}}(z)$ with the same poles and zeros, and that both sides equal~$1$ when $z\to\infty$.

From the explicit representation~\eqref{eq:Poutexplicit2cut} and the expansions
\begin{equation}
\frac{\wt{\xi}(z)+\wt{\xi}(z)^{-1}}{2} = 1+O(z^{-2}),\qquad
\frac{\wt{\xi}(z)-\wt{\xi}(z)^{-1}}{2} = \frac{{\wt{a}}-{\wt{b}}+{\wt{c}}-{\wt{d}}}{4z}+O(z^{-2}),
\end{equation}
as well as ${\wt{w}}(z)-{\wt{w}}_\infty=-{\wt{m}} z^{-1}+O(z^{-2})$, as $z\to\infty$, see~\eqref{eq:newconformalelliptic}, we obtain
\begin{equation}
\label{eq:Poutlargezqplus}
\boldsymbol P^{\mathrm{out}}(z)=
\boldsymbol {\mathrm{I}}+z^{-1}
\begin{pmatrix}
{\wt{p}}_0 & \i{\wt{p}}_+\\
\i{\wt{p}}_- & -{\wt{p}}_0
\end{pmatrix}
+O\bigl(z^{-2}\bigr)
\end{equation}
as $z\to\infty$ uniformly in $\mathbb{C}$, where we also used~\eqref{eq:Taylorsigma}, where ${\wt{p}}_0={\wt{p}}_0(x,t)$ and ${\wt{p}}_\pm={\wt{p}}_\pm(x,t)$ are given by
\begin{equation}
{\wt{p}}_0(x,t) = -{\wt{m}}\frac{\frac{\d}{\d w}\wt{\chi}(w-{\wt{w}}_1)}{\wt{\chi}(w-{\wt{w}}_1)}\bigg|_{w={\wt{w}}_\infty}
=-\frac{{\wt{m}}}{2{\wt{K}}}
\left(
\frac{\vartheta_{11}'(\frac{-{\wt{K}}+\i\pi}{2{\wt{K}}}-t{\wt{\mathcal{L}}})}{\vartheta_{11}(\frac{-{\wt{K}}+\i\pi}{2{\wt{K}}}-t{\wt{\mathcal{L}}})}
-
\frac{\vartheta_{11}'(\frac{-{\wt{K}}+\i\pi}{2{\wt{K}}})}{\vartheta_{11}(\frac{-{\wt{K}}+\i\pi}{2{\wt{K}}})}
\right)
\end{equation}
where $\vartheta_{11}(w)=\mathfrak \vartheta_{11}(w|\i\pi{\wt{K}}^{-1})$ and $\vartheta_{11}'(w)=\frac{\d}{\d w}\mathfrak \vartheta_{11}(w|\i\pi{\wt{K}}^{-1})$,
\begin{equation}
{\wt{p}}_+(x,t) =\frac 14(-{\wt{a}}+{\wt{b}}-{\wt{c}}+{\wt{d}})\frac{\wt{\chi}(-2{\wt{w}}_\infty+{\wt{K}}-\i\pi)}{\wt{\chi}({\wt{K}}-\i\pi)},
\end{equation}
and ${\wt{p}}_-$ has a similar expression which will not be needed in what follows.

\begin{proposition}\label{prop:p0pplus}
We have
\begin{equation}
\label{eq:p0pplusclosemathfrakqplus}
{\wt{p}}_0(x,t)= p_0(x,t)+O(t^{-\infty}),\qquad
{\wt{p}}_+(x,t)= p_+(x,t)+O(t^{-\infty}),
\end{equation}
uniformly for $x\in[x_*+\delta,2-\delta]$, for any $\delta>0$, where
\begin{equation}
\label{eq:p0pplustheta}
p_0=\e^{\frac {\eta}2\bigl(\frac{\eta}{2\mathcal{K}}-1\bigr)}\frac{\vartheta_{11}(\frac{\eta}{2\mathcal{K}}|\frac{\i\pi}{\mathcal{K}})}{\vartheta_{11}'(0|\frac{\i\pi}{\mathcal{K}})}\frac{\vartheta'(t\mathcal{L}|\frac{\i\pi}{\mathcal{K}})}{\vartheta(t\mathcal{L}|\frac{\i\pi}{\mathcal{K}})},
\quad
p_+=\e^{\frac {\eta}2\bigl(\frac{\eta}{2\mathcal{K}}-1\bigr)}\frac{\vartheta(t\mathcal{L}+\frac{\eta}{2\mathcal{K}}|\frac{\i\pi}{\mathcal{K}})}{\vartheta(t\mathcal{L}|\frac{\i\pi}{\mathcal{K}})}.
\end{equation}
Here, $\mathcal{L}$ is as in Proposition~\ref{prop:L} and, as usual, a prime $'$ denotes derivative with respect to the argument of the theta function.
\end{proposition}
\begin{proof}
By~\eqref{eq:f1},~\eqref{eq:parameterscloseqplus}, and~\eqref{eq:L}, we immediately see that the first equation in~\eqref{eq:p0pplusclosemathfrakqplus} holds true with
\begin{equation}
p_0=
-\frac{\mathcal{U}(\mathcal{K})}{2\mathcal{K}}
\left(
\frac{\vartheta_{11}'(\frac{-\mathcal{K}+\i\pi}{2\mathcal{K}}-t\mathcal{L}|\frac{\i\pi}{\mathcal{K}})}{\vartheta_{11}(\frac{-\mathcal{K}+\i\pi}{2\mathcal{K}}-t\mathcal{L}|\frac{\i\pi}{\mathcal{K}}))}
-
\frac{\vartheta_{11}'(\frac{-\mathcal{K}+\i\pi}{2\mathcal{K}}|\frac{\i\pi}{\mathcal{K}}))}{\vartheta_{11}(\frac{-\mathcal{K}+\i\pi}{2\mathcal{K}}|\frac{\i\pi}{\mathcal{K}}))}
\right)
\end{equation}
and if suffices to check that this expression agrees with the one in~\eqref{eq:p0pplustheta}.
This is easily verified thanks to the identity
\begin{equation}
\frac{\vartheta_{11}'(z|\tau)}{\vartheta_{11}(z|\tau)}=\i\pi+\frac{\vartheta'(z+\frac{1+\tau}{2}|\tau)}{\vartheta(z+\frac{1+\tau}{2}|\tau)},
\end{equation}
which follows from the definition of $\vartheta_{11}$, see~\eqref{eq:theta}, and from~\eqref{eq:zerothetaprime}.

Similarly, it follows from~\eqref{eq:parameterscloseqplus} and~\eqref{eq:L} that the second equation in~\eqref{eq:p0pplusclosemathfrakqplus} holds true with
\begin{equation}
\begin{aligned}
p_+&=\frac 14(-a+b-c+d)\frac{\vartheta_{11}(\frac{-\eta+\mathcal{K}-\i\pi}{2\mathcal{K}}-t\mathcal{L}|\frac{\i\pi}{\mathcal{K}})}{\vartheta_{11}(\frac{-\eta+\mathcal{K}-\i\pi}{2\mathcal{K}}|\frac{\i\pi}{\mathcal{K}})}\frac{\vartheta_{11}(\frac{\mathcal{K}-\i\pi}{2\mathcal{K}}|\frac{\i\pi}{\mathcal{K}})}{\vartheta_{11}(\frac{\mathcal{K}-\i\pi}{2\mathcal{K}}-t\mathcal{L}|\frac{\i\pi}{\mathcal{K}})},
\\
&=\frac 14(-a+b-c+d)\frac{\vartheta(\frac{\eta}{2\mathcal{K}}+t\mathcal{L}|\frac{\i\pi}{\mathcal{K}})}{\vartheta(\frac{\eta}{2\mathcal{K}}|\frac{\i\pi}{\mathcal{K}})}\frac{\vartheta(0|\frac{\i\pi}{\mathcal{K}})}{\vartheta(t\mathcal{L}|\frac{\i\pi}{\mathcal{K}})},
\end{aligned}
\end{equation}
where the last equality follows again from the relation between~$\vartheta_{11}$ and~$\vartheta$, see~\eqref{eq:theta}.
Next, we simplify the expression $a-b+c-d$. By~\eqref{eq:maintheorem:endpoints}, we have
\begin{equation}
a-b+c-d=\frac{\mathcal{U}(\mathcal{K})}{\mathcal{K}}\biggl(
\frac{\vartheta'_{01}}{\vartheta_{01}}(\frac{\eta}{4\mathcal{K}}| \frac{ \i\pi}{ \mathcal{K}})
-\frac{\vartheta'}{\vartheta}(\frac{\eta}{4\mathcal{K}}| \frac{ \i\pi}{ \mathcal{K}})
+\frac{\vartheta'_{10}}{\vartheta_{10}}(\frac{\eta}{4\mathcal{K}}| \frac{ \i\pi}{ \mathcal{K}})
-\frac{\vartheta'_{11}}{\vartheta_{11}}(\frac{\eta}{4\mathcal{K}}| \frac{ \i\pi}{ \mathcal{K}})
\biggr)
\end{equation}
We note the identity (for all $z\in\mathbb{C}$ and $\tau\in\mathbb{C}$ with $\Im\tau>0$)
\begin{equation}
\frac{\vartheta_{01}'}{\vartheta_{01}}\bigl(\frac z2\big|\tau\bigr)
-\frac{\vartheta'}{\vartheta}\bigl(\frac z2\big|\tau\bigr)
+\frac{\vartheta_{10}'}{\vartheta_{10}}\bigl(\frac z2\big|\tau\bigr)
-\frac{\vartheta_{11}'}{\vartheta_{11}}\bigl(\frac z2\big|\tau\bigr)
=-2\frac{\vartheta'_{11}(0|\tau)}{\vartheta(0|\tau)}\frac{\vartheta(z|\tau)}{\vartheta_{11}(z|\tau)}
\end{equation}
which is easily proved as both sides are anti-periodic with respect to the period lattice $\mathbb{Z}+\tau\mathbb{Z}$ (namely, both sides pick up a minus sign when $z\mapsto z+1$ or when $z\mapsto z+\tau$) and both sides have simple poles only when $z\in\mathbb{Z}+\tau\mathbb{Z}$, with residue $-2$ when $z=0$. Such claims follow from the definitions~\eqref{eq:theta}, the quasi-periodicity properties~\eqref{eq:periodictheta}, and the fact that the only zeros of $\vartheta_{11}(z|\tau)$ are $z\in\mathbb{Z}+\tau\mathbb{Z}$. 
The proof is completed by straightforward computations.
\end{proof}

For later convenience, we compute the average of $p_0(t)$ and $\log p_+(t)$ over the period.

\begin{proposition}\label{prop:mean}
    We have
    \begin{equation}
    \int_0^{1/\mathcal{L}} p_0(t)\d t = 0,\qquad \mathcal{L} \int_0^{1/\mathcal{L}} \log p_+(t)\d t 
    =\frac{\eta}2\left(\frac{\eta}{2\mathcal{K}}-1\right).
    \end{equation}
\end{proposition}
\begin{proof}
Since $\vartheta(z|\tau)=\vartheta(z+1|\tau)$ and $\vartheta(z|\tau)>0$ for all $z\in\mathbb{R}$ and $\tau\in\i\mathbb{R}_{>0}$, the integral $\int_{0}^{1/L}\frac{\vartheta'(tL|\tau)}{\vartheta(tL|\tau)}\d t$ vanishes and the integral $\int_{0}^{1/L}\log\vartheta(tL+z_0|\tau)\d t$ is independent of $z_0\in\mathbb{R}$ for all $t,L>0$ and $\tau\in\i\mathbb{R}_{>0}$.
The proof follows.
\end{proof}

\subsubsection{Inner Airy parametrices}

Let $z_0$ be any of the points ${\wt{a}},{\wt{b}},{\wt{c}},{\wt{d}}$.
By Proposition~\ref{prop:varphiqplus}, the maps
\begin{equation}
    z\mapsto\zeta_{z_0}(z)=\left(\frac 34t\varphi_{z_0}(z)\right)^{\frac 23}(z-z_0),\qquad
    z_0\in\lbrace\wt a,\wt b,\wt c,\wt d\rbrace,
\end{equation}
are conformal (injective) mappings of neighborhoods $\mathcal{Q}_{z_0}$ of $z=z_0$.
We may safely assume that 
\begin{equation}
\text{$\mathcal{Q}_{z_0}$ is an (open) square of side length $2\epsilon$ centered at $z_0$}
\end{equation}
and that these conformal mappings extend to open neighborhoods of the closures $\overline {\mathcal{Q}_{z_0}}$.
We can also assume that $\mathcal{Q}_{\wt{a}},\mathcal{Q}_{\wt{b}},\mathcal{Q}_{\wt{c}},\mathcal{Q}_{\wt{d}}$ are pairwise disjoint.
Moreover, by Remark~\ref{remark:positivitycoefficientsvarphiabcdqplus}), there exists $C>0$ such that
\begin{equation}
\zeta_{z_0}'(z)\big|_{z=z_0} = \left(\frac 34t\varphi_{z_0}(z_0)\right)^{\frac 23}\geq \left(\frac 34Ct\right)^{\frac 23}
\end{equation}
which implies that, for some $C'>0$,
\begin{equation}
\label{eq:conformalboundaryexpandsqplus}
\left|\zeta_{z_0}(z)\right|\geq C't^{\frac 23}\text{  when }z\in\partial \mathcal{Q}_{z_0}.
\end{equation}
Let us also note that $\zeta_{z_0}$ maps the real line into the real line and it maps $\gamma_{z_0}^\pm$ into the half-line emanating at the origin with argument $\pm\pi/3$ (when $z_0\in\lbrace\wt a,\wt c\rbrace$) or argument $\pm 2\pi/3$ (when $z_0\in\lbrace\wt b,\wt d\rbrace$).

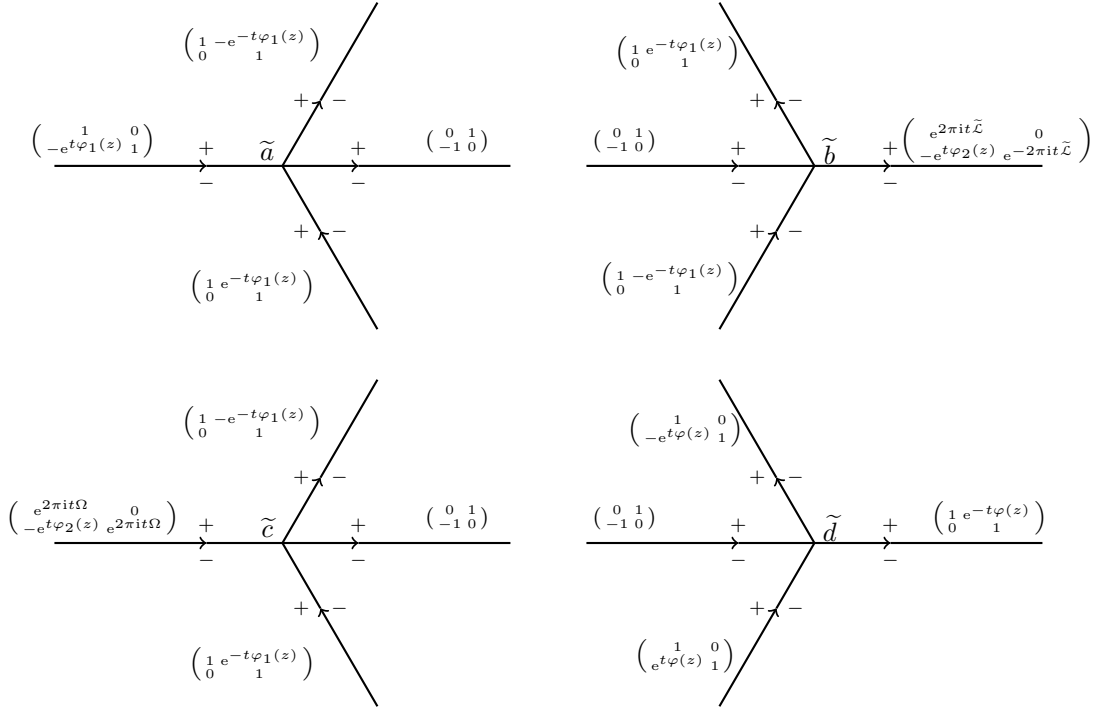
\begin{figure}[t]
\centering
\begin{tikzpicture}
\begin{scope}[shift={(3.5,-2.5)}]

\draw[->,thick] (-3,0) -- (-1,0) node[below]{\tiny{$-$}} node[above]{\tiny{$+$}};
\draw[thick] (-1,0) -- (0,0);
\draw[->,thick] (0,0) -- (1,0) node[below]{\tiny{$-$}} node[above]{\tiny{$+$}};
\draw[thick] (1,0) -- (3,0);

\begin{scope}[rotate=-30]
\draw[->,thick] (0,-2.5) -- (0,-1) node[right]{\tiny{$-$}} node[left]{\tiny{$+$}};
\draw[thick] (0,-1) -- (0,0);
\end{scope}
\begin{scope}[rotate=30]
\draw[->,thick] (0,0) -- (0,1) node[right]{\tiny{$-$}} node[left]{\tiny{$+$}};
\draw[thick] (0,1) -- (0,2.5);
\end{scope}

\node at (1/5,1/5) {\small{${\wt{d}}$}};

\node at (-2.5,1/3) {\tiny{$\left(\begin{smallmatrix}0 & 1\\ -1 & 0 \end{smallmatrix}\right)$}};
\node at (2.3,1/3) {\tiny{$\left(\begin{smallmatrix}1 & \e^{-t\varphi(z)}\\ 0 & 1 \end{smallmatrix}\right)$}};
\node at (-1.7,1.5) {\tiny{$\left(\begin{smallmatrix}1 & 0\\ -\e^{t\varphi(z)} & 1 \end{smallmatrix}\right)$}};
\node at (-1.7,-1.5) {\tiny{$\left(\begin{smallmatrix} 1 & 0\\ \e^{t\varphi(z)} & 1 \end{smallmatrix}\right)$}};
\end{scope}

\begin{scope}[shift={(-3.5,-2.5)}]
\draw[->,thick] (-3,0) -- (-1,0) node[below]{\tiny{$-$}} node[above]{\tiny{$+$}};
\draw[->,thick] (-1,0) -- (1,0) node[below]{\tiny{$-$}} node[above]{\tiny{$+$}};
\draw[thick] (1,0) -- (3,0);

\begin{scope}[rotate=30]
\draw[->,thick] (0,-2.5) -- (0,-1) node[right]{\tiny{$-$}} node[left]{\tiny{$+$}};
\draw[thick] (0,-1) -- (0,0);
\end{scope}
\begin{scope}[rotate=-30]
\draw[->,thick] (0,0) -- (0,1) node[right]{\tiny{$-$}} node[left]{\tiny{$+$}};
\draw[thick] (0,1) -- (0,2.5);
\end{scope}

\node at (-1/5,1/5) {\small{${\wt{c}}$}};

\node at (-2.5,1/3) {\tiny{$\left(\begin{smallmatrix}\e^{2\pi\i t\Omega} & 0 \\ -\e^{t\varphi_2(z)} & \e^{2\pi\i t\Omega}\end{smallmatrix}\right)$}};
\node at (2.3,1/3) {\tiny{$\left(\begin{smallmatrix}0 & 1\\ -1 & 0 \end{smallmatrix}\right)$}};
\node at (-.4,1.6) {\tiny{$\left(\begin{smallmatrix}1 & -\e^{-t\varphi_1(z)} \\ 0 & 1 \end{smallmatrix}\right)$}};
\node at (-.4,-1.6) {\tiny{$\left(\begin{smallmatrix} 1 & \e^{-t\varphi_1(z)} \\ 0 & 1 \end{smallmatrix}\right)$}};
\end{scope}

\begin{scope}[shift={(3.5,2.5)}]

\draw[->,thick] (-3,0) -- (-1,0) node[below]{\tiny{$-$}} node[above]{\tiny{$+$}};
\draw[thick] (-1,0) -- (0,0);
\draw[->,thick] (0,0) -- (1,0) node[below]{\tiny{$-$}} node[above]{\tiny{$+$}};
\draw[thick] (1,0) -- (3,0);

\begin{scope}[rotate=-30]
\draw[->,thick] (0,-2.5) -- (0,-1) node[right]{\tiny{$-$}} node[left]{\tiny{$+$}};
\draw[thick] (0,-1) -- (0,0);
\end{scope}
\begin{scope}[rotate=30]
\draw[->,thick] (0,0) -- (0,1) node[right]{\tiny{$-$}} node[left]{\tiny{$+$}};
\draw[thick] (0,1) -- (0,2.5);
\end{scope}

\node at (1/5,1/5) {\small{${\wt{b}}$}};

\node at (-2.5,1/3) {\tiny{$\left(\begin{smallmatrix}0 & 1\\ -1 & 0 \end{smallmatrix}\right)$}};
\node at (2.4,1/3) {\tiny{$\left(\begin{smallmatrix}\e^{2\pi\i t{\wt{\mathcal{L}}}} & 0 \\ -\e^{t\varphi_2(z)} & \e^{-2\pi\i t{\wt{\mathcal{L}}}} \end{smallmatrix}\right)$}};
\node at (-1.8,1.5) {\tiny{$\left(\begin{smallmatrix}1 & \e^{-t\varphi_1(z)}\\ 0 & 1 \end{smallmatrix}\right)$}};
\node at (-1.9,-1.5) {\tiny{$\left(\begin{smallmatrix} 1 & -\e^{-t\varphi_1(z)}\\ 0 & 1 \end{smallmatrix}\right)$}};
\end{scope}

\begin{scope}[shift={(-3.5,2.5)}]
\draw[->,thick] (-3,0) -- (-1,0) node[below]{\tiny{$-$}} node[above]{\tiny{$+$}};
\draw[->,thick] (-1,0) -- (1,0) node[below]{\tiny{$-$}} node[above]{\tiny{$+$}};
\draw[thick] (1,0) -- (3,0);

\begin{scope}[rotate=30]
\draw[->,thick] (0,-2.5) -- (0,-1) node[right]{\tiny{$-$}} node[left]{\tiny{$+$}};
\draw[thick] (0,-1) -- (0,0);
\end{scope}
\begin{scope}[rotate=-30]
\draw[->,thick] (0,0) -- (0,1) node[right]{\tiny{$-$}} node[left]{\tiny{$+$}};
\draw[thick] (0,1) -- (0,2.5);
\end{scope}

\node at (-1/5,1/5) {\small{${\wt{a}}$}};

\node at (-2.5,1/3) {\tiny{$\left(\begin{smallmatrix}1 & 0 \\ -\e^{t\varphi_1(z)} & 1 \end{smallmatrix}\right)$}};
\node at (2.3,1/3) {\tiny{$\left(\begin{smallmatrix}0 & 1\\ -1 & 0 \end{smallmatrix}\right)$}};
\node at (-.4,1.6) {\tiny{$\left(\begin{smallmatrix}1 & -\e^{-t\varphi_1(z)} \\ 0 & 1 \end{smallmatrix}\right)$}};
\node at (-.4,-1.6) {\tiny{$\left(\begin{smallmatrix} 1 & \e^{-t\varphi_1(z)} \\ 0 & 1 \end{smallmatrix}\right)$}};
\end{scope}

\end{tikzpicture}
\caption{Jumps of $T(z)$ in neighborhoods of $z={\wt{a}},{\wt{b}},{\wt{c}},{\wt{d}}$.}
\label{fig:localparaqplus}
\end{figure}

By definition,
\begin{equation}
\begin{aligned}
&\mbox{if $z\in \mathcal{Q}_{\wt{a}}\setminus({\wt{a}},+\infty)$,}&&t\varphi_1(z)=-\frac 43\left(-\zeta_{\wt{a}}(z) \right)^{\frac 32},
\\
&\mbox{if $z\in \mathcal{Q}_{\wt{b}}\setminus(-\infty,{\wt{b}})$,}&&t\varphi_2(z)=-\frac 43\zeta_{\wt{b}}(z) ^{\frac 32},
\\
&\mbox{if $z\in \mathcal{Q}_{\wt{c}}\setminus({\wt{c}},+\infty)$,}&&t\varphi_2(z)= -\frac 43\left(-\zeta_{\wt{c}}(z) \right)^{\frac 32} ,
\\
&\mbox{if $z\in \mathcal{Q}_{\wt{d}}\setminus(-\infty,{\wt{d}})$,}&&t\varphi(z)= \frac 43\zeta_{\wt{d}}(z) ^{\frac 32} .
\end{aligned}
\end{equation}
Using these identities, we see that the jumps of $\boldsymbol T(z)$ inside $\mathcal{Q}_{\wt{a}},\mathcal{Q}_{\wt{b}},\mathcal{Q}_{\wt{c}},\mathcal{Q}_{\wt{d}}$ (depicted in Figure~\ref{fig:localparaqplus}) can be exactly solved, using these conformal mappings, in terms of appropriate Airy model Riemann--Hilbert problem solutions $\boldsymbol \Phi^{\mathrm{Ai},\mathrm{I}}$, $\boldsymbol \Phi^{\mathrm{Ai},\mathrm{II}}$, and $\boldsymbol \Phi^{\mathrm{Ai},\mathrm{III}}$, defined in~Appendix~\ref{app:Airy} (whose jumps are depicted in~Figure~\ref{fig:AirymodelRHP}).
Therefore, following the routine practice of the nonlinear steepest descent method, we define the inner Airy parametrices by
\begin{equation}
\begin{aligned}
\boldsymbol P^{({\wt{a}})}(z)&=\boldsymbol E^{({\wt{a}})}(z)\,\boldsymbol\Phi^{\mathrm{Ai},\mathrm{II}}(\zeta_{\wt{a}}(z)) &&(z\in \mathcal{Q}_{\wt{a}}),
\\
\boldsymbol P^{({\wt{b}})}(z)&=\boldsymbol E^{({\wt{b}})}(z)\boldsymbol\Phi^{\mathrm{Ai},\mathrm{III}}(\zeta_{\wt{b}}(z))\,\e^{\pm\i\pi t{\wt{\mathcal{L}}}\boldsymbol\sigma_3} &&(z\in \mathcal{Q}_{\wt{b}}),
\\
\boldsymbol P^{({\wt{c}})}(z)&=\boldsymbol E^{({\wt{c}})}(z)\,\boldsymbol\Phi^{\mathrm{Ai},\mathrm{II}}(\zeta_{\wt{c}}(z))\,\e^{\pm\i\pi t{\wt{\mathcal{L}}}\boldsymbol\sigma_3} &&(z\in \mathcal{Q}_{\wt{c}}),
\\
\boldsymbol P^{({\wt{d}})}(z)&=\boldsymbol E^{({\wt{d}})}(z)\,\boldsymbol\Phi^{\mathrm{Ai},\mathrm{I}}(\zeta_{\wt{d}}(z)) &&(z\in \mathcal{Q}_{\wt{d}}),
\end{aligned}
\end{equation}
where the sign is chosen according to $\pm\Im z>0$ and
\begin{equation}
\begin{aligned}
\boldsymbol E^{({\wt{a}})}(z)&=\boldsymbol P^{\mathrm{out}}(z)\,\boldsymbol G^{-1}\,\bigl(-\zeta_{\wt{a}}(z)\bigr)^{-\frac 14\boldsymbol\sigma_3} &&(z\in \mathcal{Q}_{\wt{a}}),
\\
\boldsymbol E^{({\wt{b}})}(z)&=\boldsymbol P^{\mathrm{out}}(z)\,\e^{\mp\i\pi t{\wt{\mathcal{L}}} \boldsymbol \sigma_3}\,\boldsymbol G\, \zeta_{\wt{b}}(z)^{-\frac 14\boldsymbol\sigma_3}&&(z\in \mathcal{Q}_{\wt{b}}),
\\
\boldsymbol E^{({\wt{c}})}(z)&=\boldsymbol P^{\mathrm{out}}(z)\,\e^{\mp\i\pi t{\wt{\mathcal{L}}}\boldsymbol\sigma_3}\, \boldsymbol G^{-1}\,\bigl(-\zeta_{\wt{c}}(z)\bigr)^{-\frac 14\boldsymbol\sigma_3} &&(z\in \mathcal{Q}_{\wt{c}}),
\\
\boldsymbol E^{({\wt{d}})}(z)&=\boldsymbol P^{\mathrm{out}}(z)\,\boldsymbol G^{-1}\,\zeta_{\wt{d}}(z)^{\frac 14\boldsymbol\sigma_3} &&(z\in \mathcal{Q}_{\wt{d}}),
\end{aligned}
\end{equation}
where, again, the sign is determined by $\pm\Im z>0$, and $\boldsymbol P^{\mathrm{out}}$ is the outer parametrix introduced in see~\Cref{sec:outerparaqplus}, $\boldsymbol G$ is defined in~\eqref{eq:Gmatrix}, and $\boldsymbol \Phi^{\mathrm{Ai},\mathrm{I}}$, $\boldsymbol\Phi^{\mathrm{Ai},\mathrm{II}}$, and $\Phi^{\mathrm{Ai},\mathrm{III}}$ are defined in~\eqref{eq:defPhiAi1}--\eqref{eq:defPhiAi2} and solve the model Airy Riemann--Hilbert problem~\ref{cRHp:Airy}.
The usual properties of this construction are summarized in the next proposition.

\begin{proposition}
\label{prop:matchingqplus}
Let $z_0\in\lbrace\wt a,\wt b,\wt c,\wt d\rbrace$.
The matrix $\boldsymbol E^{(z_0)}$ is analytic in $\mathcal{Q}_{z_0}$.
The matrix $\boldsymbol P^{(z_0)}$ is analytic in $\mathcal{Q}_{z_0}\setminus\Sigma_T$ and satisfies the same jump condition as $\boldsymbol T$ on $\Sigma_{T}^\circ\cap \mathcal{Q}_{z_0}$.
Moreover, when $z\in\partial \mathcal{Q}_{z_0}$ we have 
\begin{equation}
\boldsymbol P^{(z_0)}(z)\boldsymbol P^{\mathrm{out}}(z)^{-1} = \boldsymbol {\mathrm{I}}+t^{-1}\wt{\boldsymbol J}_{R,z_0}(z)+O(t^{-2}),\quad t\to+\infty,
\end{equation}
with error term $O(t^{-1})$ uniform for $z\in\partial \mathcal{Q}_{z_0}$ and for $x\in[x_*+\delta,2-\delta]$ (for any~$\delta>0$) and
\begin{equation}
\label{eq:JR1qplus}
\begin{aligned}
\wt{\boldsymbol J}_{R,{\wt{a}}}(z) &= -\frac{1}{36\,\varphi_1(z)}
\boldsymbol P^{{\mathrm{out}}}(z)
\begin{pmatrix}
    -1 & 6\i \\ 6\i & 1
\end{pmatrix}
\boldsymbol P^{{\mathrm{out}}}(z)^{-1},
\\
\wt{\boldsymbol J}_{R,{\wt{b}}}(z) &=-\frac{1}{36\,\varphi_2(z)}
\boldsymbol P^{{\mathrm{out}}}(z)
\begin{pmatrix}
    -1 & -6\i\e^{\mp 2\pi\i t {\wt{\mathcal{L}}}} \\ -6\i\e^{\pm 2\pi\i t {\wt{\mathcal{L}}}} & 1
\end{pmatrix}
\boldsymbol P^{{\mathrm{out}}}(z)^{-1},
\\
\wt{\boldsymbol J}_{R,{\wt{c}}}(z) &= -\frac{1}{36\,\varphi_2(z)}
\boldsymbol P^{{\mathrm{out}}}(z)
\begin{pmatrix}
    -1 & 6\i\e^{\mp 2\pi\i t {\wt{\mathcal{L}}}} \\ 6\i\e^{\pm 2\pi\i t {\wt{\mathcal{L}}}} & 1
\end{pmatrix}
\boldsymbol P^{{\mathrm{out}}}(z)^{-1},
\\
\wt{\boldsymbol J}_{R,{\wt{d}}}(z) &= \frac{1}{36\,\varphi(z)}
\boldsymbol P^{{\mathrm{out}}}(z)
\begin{pmatrix}
    1 & 6\i \\ 6\i & -1
\end{pmatrix}
\boldsymbol P^{{\mathrm{out}}}(z)^{-1}.
\end{aligned}
\end{equation}
Here, $\boldsymbol P^{{\mathrm{out}}}$ is the outer parametrix defined in~\eqref{eq:Poutexplicit2cut} and the sign $\pm$ is determined by $\pm\Im z>0$.
\end{proposition}

We omit the proof because it consists of verifications and computations which are typical of the nonlinear steepest descent method and are completely analogous to those in the proof of Proposition~\ref{prop:matchingqminus}.

As usual in this type of analysis (see for example~\cite{bleher2010exact,bothner2015asymptotic}), the matrices $\wt{\boldsymbol J}_{R,z_0}$ (with $z_0\in\lbrace\wt a,\wt b,\wt c,\wt d\rbrace$) extend to meromorphic functions of $z$ in $\mathcal{Q}_{z_0}$ whose only singularity is a double pole at $z=z_0$, as it is readily verified by using the analytic properties of $\boldsymbol P^{\mathrm{out}}$ and of $\varphi,\varphi_1,\varphi_2$.
We omit this standard check.

In the final step of the Riemann--Hilbert asymptotic analysis we will need asymptotics for the $(1,1)$- and $(1,2)$-entries of $\res{z=z_0}\wt{\boldsymbol J}_{R,z_0}\d z$ when $t$ is large and $z_0\in\lbrace\wt a,\wt b,\wt c,\wt d\rbrace$, which we now provide.

First, it is immediate to see that 
\begin{equation}
\label{eq:JtildeJqplus}
\res{z=\wt a}\wt{\boldsymbol J}_{R,\wt a}\d z=\res{z=a}\boldsymbol J_{R,a}\d z+O(t^{-\infty}),
\end{equation}
and similarly for $b,c,d$, where
\begin{equation}
\label{eq:JR1qplus_tlimit}
\begin{aligned}
\boldsymbol J_{R,a}(z) &= \frac{1}{36(a-z)^{3/2}\bigl(A_a-B_a(z-a)\bigr)}
\boldsymbol U(z)
\begin{pmatrix}
    -1 & 6\i \\ 6\i & 1
\end{pmatrix}
\boldsymbol U(z)^{-1},
\\
\boldsymbol J_{R,b}(z) &=\frac{1}{36(z-b)^{3/2}\bigl(A_b+B_b(z-b)\bigr)}
\boldsymbol U(z)
\begin{pmatrix}
    -1 & -6\i\e^{\mp 2\pi\i t \mathcal{L}} \\ -6\i\e^{\pm 2\pi\i t \mathcal{L}} & 1
\end{pmatrix}
\boldsymbol U(z)^{-1},
\\
\boldsymbol J_{R,c}(z) &=\frac{1}{36(c-z)^{3/2}\bigl(A_c-B_c(z-c)\bigr)}
\boldsymbol U(z)
\begin{pmatrix}
    -1 & 6\i\e^{\mp 2\pi\i t \mathcal{L}} \\ 6\i\e^{\pm 2\pi\i t \mathcal{L}} & 1
\end{pmatrix}
\boldsymbol U(z)^{-1},
\\
\boldsymbol J_{R,d}(z) &= \frac{1}{36(z-d)^{3/2}\bigl(A_d+B_d(z-d)\bigr)}
\boldsymbol U(z)
\begin{pmatrix}
    1 & 6\i \\ 6\i & -1
\end{pmatrix}
\boldsymbol U(z)^{-1}.
\end{aligned}
\end{equation}
Here, the signs $\pm$ are chosen according to $\pm\Im z>0$, the expressions $A_{z_0}$ and $B_{z_0}$ have been defined in~\eqref{eq:Az0Bz0} (see~Proposition~\ref{prop:varphiqplus}), and
\begin{equation}
\boldsymbol U(z)=\begin{pmatrix}
    \frac{\xi(z)+\xi(z)^{-1}}{2}\frac{\chi(w(z)-w_1)}{\chi(w_\infty-w_1)}&
     \frac{\xi(z)-\xi(z)^{-1}}{2\i}\frac{\chi(-w(z)-w_1)}{\chi(w_\infty-w_1)}
     \\
     -\frac{\xi(z)-\xi(z)^{-1}}{2\i}\frac{\chi(w(z)-w_2)}{\chi(-w_\infty-w_2)}&
     \frac{\xi(z)+\xi(z)^{-1}}{2}\frac{\chi(-w(z)-w_2)}{\chi(-w_\infty-w_2)}
\end{pmatrix}
\end{equation}
where, see~\eqref{eq:xitilde} and~\eqref{eq:chitilde},
\begin{equation}
\label{eq:xichi}
\xi(z) =\left(\frac{(z-b)(z-d)}{(z-a)(z-c)}\right)^{\frac 14},\qquad
\chi(w) = \frac{\vartheta_{11}(\frac{w}{2\mathcal{K}}-t\mathcal{L}|\frac{\i\pi}{\mathcal{K}})}{\vartheta_{11}(\frac w{2\mathcal{K}}|\frac{\i\pi}{\mathcal{K}})},
\end{equation}
as well as, see~\eqref{eq:w1w2Pout},
\begin{equation}
w_1=w_\infty-\mathcal{K}+\i\pi\quad\text{and}\quad w_2=-w_\infty+\mathcal{K}+\i\pi.
\end{equation}
After some lengthy (albeit elementary) computations which we omit, we obtain
\begin{equation}
\label{eq:explicitJR1qplus1}
\begin{aligned}
&\res{z=d}(\boldsymbol J_{R,d}(z))_{1,1}\d z=
\frac 1{A_d \,\chi(w_\infty-w_1)\,\chi(-w_\infty-w_2)}
\times
\\
&\quad
\times\biggl[\frac{m}{18T_d}(\chi'(-w_1)\chi(-w_2)-\chi(-w_1)\chi'(-w_2))
\\
&\qquad -\frac{m^2}{36T_d(d-b)}(5\chi''(-w_1)\chi(-w_2)+14\chi'(-w_1)\chi'(-w_2)+5\chi(-w_1)\chi''(-w_2))
\\
&\qquad +\frac{\sqrt{\frac{(d-a)(d-c)}{d-b}}}{144} \left(10 \frac{B_d}{A_d}-\frac 5{d-a}+\frac 5{d-b}-\frac 5{d-c}+14\frac{d-b}{(d-a)(d-c)}\right)\chi(-w_1)\chi(-w_2)
\biggr],
\end{aligned}
\end{equation}
\begin{equation}
\begin{aligned}
&\res{z=a}(\boldsymbol J_{R,a}(z))_{1,1}\d z=
\frac 1{A_a \,\chi(w_\infty-w_1)\,\chi(-w_\infty-w_2)}
\times
\\
&\quad
\times\biggl[-\frac{m}{18T_a}(\chi'(\mathcal{K}-w_1)\chi(\mathcal{K}-w_2)-\chi(\mathcal{K}-w_1)\chi'(\mathcal{K}-w_2))
\\
&\qquad -\frac{m^2}{36T_a(c-a)}(5\chi''(\mathcal{K}-w_1)\chi(\mathcal{K}-w_2)+14\chi'(\mathcal{K}-w_1)\chi'(\mathcal{K}-w_2)+5\chi(\mathcal{K}-w_1)\chi''(\mathcal{K}-w_2))
\\
&\qquad +\frac{\sqrt{\frac{(b-a)(d-a)}{c-a}}}{144} \left(10 \frac{B_a}{A_a}+\frac 5{a-b}-\frac 5{a-c}+\frac 5{a-d}+14\frac{c-a}{(b-a)(d-a)}\right)\chi(\mathcal{K}-w_1)\chi(\mathcal{K}-w_2)
\biggr],
\end{aligned}
\end{equation}
\begin{equation}
\begin{aligned}
&\res{z=b}(\boldsymbol J_{R,b}(z))_{1,1}\d z=
\frac {\e^{-2\pi\i t \mathcal{L}}}{A_b \,\chi(w_\infty-w_1)\,\chi(-w_\infty-w_2)}
\times
\\
&\quad
\times\biggl[-\frac{m}{18T_b}(\chi'(\i\pi+\mathcal{K}-w_1)\chi(\i\pi+\mathcal{K}-w_2)-\chi(\i\pi+\mathcal{K}-w_1)\chi'(\i\pi+\mathcal{K}-w_2))
\\
&\qquad +\frac{m^2}{36T_b(d-b)}\bigl(5\chi''(\i\pi+\mathcal{K}-w_1)\chi(\i\pi+\mathcal{K}-w_2)+14\chi'(\i\pi+\mathcal{K}-w_1)\chi'(\i\pi+\mathcal{K}-w_2)
\\
& \qquad \qquad \qquad \qquad \qquad \qquad \qquad \qquad \qquad \qquad \qquad \qquad +5\chi(\i\pi+\mathcal{K}-w_1)\chi''(\i\pi+\mathcal{K}-w_2)\bigr)
\\
&\qquad +\frac{\sqrt{\frac{(b-a)(c-b)}{d-b}}}{144} \left(-10 \frac{B_b}{A_b}+\frac 5{b-a}+\frac 5{b-c}-\frac 5{b-d} 
-14\frac{d-b}{(b-a)(c-b)}\right)
\\
& \qquad \qquad \qquad \qquad \qquad \qquad \qquad
\times\chi(\i\pi+\mathcal{K}-w_1)\chi(\i\pi+\mathcal{K}-w_2)
\biggr],
\end{aligned}
\end{equation}
\begin{equation}
\label{eq:explicitJR1qplus7}
\begin{aligned}
&\res{z=c}(\boldsymbol J_{R,c}(z))_{1,1}\d z=
\frac {\e^{-2\pi\i t \mathcal{L}}}{A_c \,\chi(w_\infty-w_1)\,\chi(-w_\infty-w_2)}
\times
\\
&\quad
\times\biggl[-\frac{m}{18T_c}(\chi'(\i\pi-w_1)\chi(\i\pi-w_2)-\chi(\i\pi-w_1)\chi'(\i\pi-w_2))
\\
&\qquad -\frac{m^2}{36T_c(c-a)}(5\chi''(\i\pi-w_1)\chi(\i\pi-w_2)+14\chi'(\i\pi-w_1)\chi'(\i\pi-w_2)+5\chi(\i\pi-w_1)\chi''(\i\pi-w_2))
\\
&\qquad +\frac{\sqrt{\frac{(c-b)(d-c)}{c-a}}}{144} \left(10 \frac{B_c}{A_c}-\frac 5{c-a}+\frac 5{c-b}+\frac 5{c-d}+14\frac{c-a}{(c-b)(d-c)}\right)\chi(\i\pi-w_1)\chi(\i\pi-w_2)
\biggr].
\end{aligned}
\end{equation}
The coefficients $A_{z_0},B_{z_0},T_{z_0}$ (for $z_0=a,b,c,d$) were defined in \eqref{eq:Az0Bz0}, \eqref{eq:Tz0_Sz0} and $m=\mathcal{U}(\mathcal{K}(x))$ as in \eqref{eq:systemfinaltwocut}.
The residues $\res{z=z_0}(\boldsymbol J_{R,z_0}(z))_{1,2}\d z$ (for $z_0=a,b,c,d$) can be expressed by completely analogous formulas.
We omit them as their explicit form will not be needed in what follows, as we will only use the fact that they are smooth functions of $x,t$, periodic in $t$ of period~$1/\mathcal{L}(x)$.

\subsection{Error analysis}

Let us finally fix $\epsilon>0$ sufficiently small such that the results of the previous sections hold true.
Let $\Sigma_R$ be
\begin{equation}
\Sigma_R=(-\infty,{\wt{a}}-\epsilon]\cup[{\wt{b}}+\epsilon,{\wt{c}}-\epsilon]\cup[{\wt{d}}+\epsilon,+\infty)\cup \Gamma_L^+\cup\Gamma_L^-\cup\Gamma_R^+\cup\Gamma_R^-\cup\biggl(\bigcup_{z_0\in\lbrace{\wt{a}},{\wt{b}},{\wt{c}},{\wt{d}}\rbrace}\partial \mathcal{Q}_{z_0}\biggr)
\end{equation}
and let $\Sigma_R^\circ$ be, as usual, the complement in $\Sigma_R$ of the points of intersection of the various contours forming $\Sigma_R$, namely $\Sigma_R^\circ$ consists of the points in $\Sigma_R$ except for ${\wt{a}}-\epsilon,{\wt{b}}+\epsilon,{\wt{c}}-\epsilon,{\wt{d}}+\epsilon,{\wt{a}}^\pm,{\wt{b}}^\pm,{\wt{c}}^\pm,{\wt{d}}^\pm$ and for the vertices of the squares $\mathcal{Q}_{\wt{a}}$, $\mathcal{Q}_{\wt{b}}$, $\mathcal{Q}_{\wt{c}}$, and $\mathcal{Q}_{\wt{d}}$.
The orientation on $\Sigma_R^\circ$ is illustrated in Figure~\ref{fig:SigmaRqminus}.
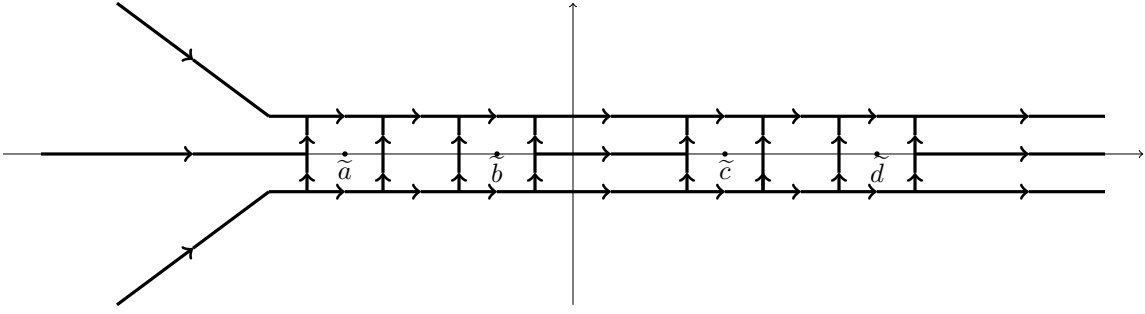
\begin{figure}[t]
\centering
\begin{tikzpicture}

\draw[very thin,->] (-7.5,0) -- (7.5,0);
\draw[very thin,->] (0,-2) -- (0,2);

\draw[very thick,->] (-7,0) -- (-5,0);
\draw[very thick] (-5,0) -- (-7/2,0);
\draw[very thick,->] (-1/2,0) -- (1/2,0);
\draw[very thick] (1/2,0) -- (3/2,0);
\draw[very thick,->] (9/2,0) -- (6,0) ;
\draw[very thick] (6,0) -- (7,0);

\draw[very thick, ->] (-7/2,-1/2) -- (-7/2,-1/4);
\draw[very thick, ->] (-7/2,-1/4) -- (-7/2,1/4);
\draw[very thick] (-7/2,1/4) -- (-7/2,1/2);

\draw[very thick, ->] (-5/2,-1/2) -- (-5/2,-1/4);
\draw[very thick, ->] (-5/2,-1/4) -- (-5/2,1/4);
\draw[very thick] (-5/2,1/4) -- (-5/2,1/2);

\draw[very thick, ->] (-3/2,-1/2) -- (-3/2,-1/4);
\draw[very thick, ->] (-3/2,-1/4) -- (-3/2,1/4);
\draw[very thick] (-3/2,1/4) -- (-3/2,1/2);

\draw[very thick, ->] (-1/2,-1/2) -- (-1/2,-1/4);
\draw[very thick, ->] (-1/2,-1/4) -- (-1/2,1/4);
\draw[very thick] (-1/2,1/4) -- (-1/2,1/2);

\draw[very thick, ->] (3/2,-1/2) -- (3/2,-1/4);
\draw[very thick, ->] (3/2,-1/4) -- (3/2,1/4);
\draw[very thick] (3/2,1/4) -- (3/2,1/2);

\draw[very thick, ->] (9/2,-1/2) -- (9/2,-1/4);
\draw[very thick, ->] (9/2,-1/4) -- (9/2,1/4);
\draw[very thick] (9/2,1/4) -- (9/2,1/2);

\draw[very thick, ->] (7/2,-1/2) -- (7/2,-1/4);
\draw[very thick, ->] (7/2,-1/4) -- (7/2,1/4);
\draw[very thick] (7/2,1/4) -- (7/2,1/2);

\draw[very thick, ->] (5/2,-1/2) -- (5/2,-1/4);
\draw[very thick, ->] (5/2,-1/2) -- (5/2,1/4);
\draw[very thick] (5/2,1/4) -- (5/2,1/2);

\draw[very thick] (-4,1/2) -- (-5,5/4);
\draw[very thick,<-] (-5,5/4) -- (-6,2);

\draw[very thick] (-4,-1/2) -- (-5,-5/4);
\draw[very thick,<-] (-5,-5/4) -- (-6,-2);

\draw[very thick,->] (-4,1/2) -- (-3,1/2);
\draw[very thick,->] (-3,1/2) -- (-2,1/2);
\draw[very thick,->] (-2,1/2) -- (-1,1/2);
\draw[very thick,->] (-1,1/2) -- (1/2,1/2);
\draw[very thick,->] (1/2,1/2) -- (2,1/2);
\draw[very thick,->] (2,1/2) -- (3,1/2);
\draw[very thick,->] (3,1/2) -- (4,1/2);
\draw[very thick,->] (4,1/2) -- (6,1/2);
\draw[very thick] (6,1/2) -- (7,1/2);

\draw[very thick,->] (-4,-1/2) -- (-3,-1/2);
\draw[very thick,->] (-3,-1/2) -- (-2,-1/2);
\draw[very thick,->] (-2,-1/2) -- (-1,-1/2);
\draw[very thick,->] (-1,-1/2) -- (1/2,-1/2);
\draw[very thick,->] (1/2,-1/2) -- (2,-1/2);
\draw[very thick,->] (2,-1/2) -- (3,-1/2);
\draw[very thick,->] (3,-1/2) -- (4,-1/2);
\draw[very thick,->] (4,-1/2) -- (6,-1/2);
\draw[very thick] (6,-1/2) -- (7,-1/2);

\fill (-3,0) circle (1pt);
\node at (-3,-1/5) {\small{${\wt{a}}$}};

\fill (-1,0) circle (1pt);
\node at (-1,-1/5) {\small{${\wt{b}}$}};

\fill (2,0) circle (1pt);
\node at (2,-1/5) {\small{${\wt{c}}$}};

\fill (4,0) circle (1pt);
\node at (4,-1/5) {\small{${\wt{d}}$}};

\end{tikzpicture}
\caption{$\Sigma_R$ (case $x_*<x<2$).}
\label{fig:SigmaRqplus}
\end{figure}
Introduce the analytic function $\boldsymbol R:\mathbb{C}\setminus\Sigma_R\to\mathrm{SL}(2,\mathbb{C})$ by
\begin{equation}
\boldsymbol R(z) = \begin{cases}
\boldsymbol T(z)\boldsymbol P^{(z_0)}(z)^{-1}&\text{if }z\in \mathcal{Q}_{z_0},\,\,z_0\in\lbrace\wt a,\wt b,\wt c,\wt d\rbrace,\\
\boldsymbol T(z)\boldsymbol P^{{\mathrm{out}}}(z)^{-1}&\text{otherwise}.
\end{cases}
\end{equation}
By construction, $\boldsymbol R(z)$ is the unique solution to the following Riemann--Hilbert problem.
\begin{cRHp}
\label{cRHp:Rqplus}
Find an analytic function $\boldsymbol R:\mathbb{C}\setminus \Sigma_R\to\mathrm{SL}(2,\mathbb{C})$ such that the following conditions hold true.
\begin{enumerate}[leftmargin=*]
\item Non-tangential boundary values of $\boldsymbol R$ exist and are continuous on $\Sigma_R^\circ$ and satisfy
\begin{equation}
\boldsymbol R_+(z)=\boldsymbol R_-(z)\boldsymbol J_R(z),\qquad z\in \Sigma_R^\circ,
\end{equation}
where $\boldsymbol J_R(z)$ is given for $z\in \Sigma_R^\circ$ by 
\begin{equation}
\label{eq:jumpRqplus}
\boldsymbol J_R(z)=\begin{cases}
\boldsymbol P^{{\mathrm{out}}}(z)\boldsymbol P^{(z_0)}(z)^{-1}&\text{if } z\in\partial \mathcal{Q}_{z_0},\,\,z_0\in\lbrace\wt a,\wt b,\wt c,\wt d\rbrace\\
\boldsymbol P^{{\mathrm{out}}}_-(z)\boldsymbol J_T(z)\e^{-2\pi\i t{\wt{\mathcal{L}}}\boldsymbol\sigma_3}\boldsymbol P^{{\mathrm{out}}}_-(z)^{-1}&\text{if }z\in({\wt{b}}+\epsilon,{\wt{c}}-\epsilon)\\
\boldsymbol P^{{\mathrm{out}}}(z)\boldsymbol J_T(z)\boldsymbol P^{{\mathrm{out}}}(z)^{-1}&\text{otherwise}.
\end{cases}
\end{equation}
\item We have $\boldsymbol R(z)\to\boldsymbol {\mathrm{I}}$ as $z\to\infty$  uniformly in $\mathbb{C}\setminus \Sigma_R$.
\item We have $\boldsymbol R(z)=O(1)$ as $z\to z_0$ uniformly in $\mathbb{C}\setminus\Sigma_R$ for all $z_0\in\Sigma_R\setminus\Sigma_R^\circ$.
\end{enumerate}
\end{cRHp}
Recalling the definition of $\wt{\Sigma_T}^\epsilon$, see~\eqref{eq:SigmaTepsilonqplus}, we have $\Sigma^\circ_R=\biggl(\bigcup_{z_0\in\lbrace\wt a,\wt b,\wt c,\wt d\rbrace}\partial \mathcal{Q}_{z_0}\biggr)\cup({\wt{b}}+\epsilon,{\wt{c}}-\epsilon)\cup\wt{\Sigma_T}^\epsilon\setminus\lbrace\text{vertices of }\mathcal{Q}_{\wt{a}},\mathcal{Q}_{\wt{b}},\mathcal{Q}_{\wt{c}},\mathcal{Q}_{\wt{d}}\rbrace$.
\begin{proposition}
\label{prop:JRsmallqplus}
For any $\delta>0$ there exists $c,t_*>0$ such that
\begin{equation}
\boldsymbol J_R(z) = \begin{cases}
\boldsymbol {\mathrm{I}}-t^{-1}\wt{\boldsymbol J}_{R,z_0}(z)+O(t^{-2}),&\text{if }z\in\partial \mathcal{Q}_{\wt{a}}\cup\partial \mathcal{Q}_{\wt{b}}\cup\partial \mathcal{Q}_{\wt{c}}\cup\partial \mathcal{Q}_{\wt{d}},\\
\boldsymbol {\mathrm{I}}+O\left(\frac{1}{|z|^2+1}\e^{-ct}\right),&\text{if }z\in\wt{\Sigma_T}^\epsilon\cup({\wt{b}}+\epsilon,{\wt{c}}-\epsilon),
\end{cases}
\end{equation}
with $\wt{\boldsymbol J}_{R,z_0}$ given in~\eqref{eq:JR1qplus} and error terms uniform for $t\geq t_*$, $z\in\Sigma^\circ_R$, and $x\in[x_*+\delta,2-\delta]$.
\end{proposition}
\begin{proof}
It follows from Propositions~\ref{prop:JTsmallqplus} and~\ref{prop:matchingqplus}, noting that  $\boldsymbol P^{{\mathrm{out}}}(z)=O(1)$ uniformly for $z\in\mathbb{C}\setminus\left(\mathcal{Q}_{\wt{a}}\cup \mathcal{Q}_{\wt{b}}\cup \mathcal{Q}_{\wt{c}}\cup \mathcal{Q}_{\wt{d}}\right)$.
\end{proof}
\begin{proposition}
\label{prop:finalqplus}
For any $\delta>0$ there exists $t_*>0$ such that
\begin{equation}
\begin{aligned}
\wh\alpha(t,x) &=  t(1+g_1(x))+p_0(x,t)+t^{-1}\mathcal X(x,t)+O(t^{-2}),\\
\log\wh\beta(t,x) &= t\bigl(2g_\infty(x)-\ell(x)\bigr)+\log p_+(x,t)+t^{-1}\mathcal Y(x,t)+O(t^{-2}).
\end{aligned}
\end{equation}
with error terms uniform for $t\geq t_*$ and $x\in[x_*+\delta,2-\delta]$. Here, $\wh\alpha(t,x)$ and $\wh\beta(t,x)$ are defined in~\eqref{eq:wtalphabetagamma}, $g_1(x)$ and $g_\infty(x)$ are given in~\eqref{eq:g1ginftyqplus}, while $p_0(x,t)$ and $p_+(x,t)$ are defined in~\eqref{eq:p0pplustheta} (in particular, $p_0(x,t)$ and $\log p_+(x,t)$ are smooth in $x,t$ and periodic in $t$ with period $1/\mathcal{L}(x)$), and $\mathcal X(x,t)$ and $\mathcal Y(x,t)$ are given by
\begin{equation}
\label{eq:mathcalXY}
\begin{aligned}
\mathcal X(x,t)&=-\frac 1{24}+\sum_{z_0\in\lbrace a,b,c,d\rbrace}\left( \res{z=z_0}\boldsymbol J_{R,z_0}(z)\d z\right)_{1,1},
\\
\mathcal Y(x,t)&=\frac{1}{\i p_+(x,t)}\sum_{z_0\in\lbrace a,b,c,d\rbrace}\left( \res{z=z_0}\boldsymbol J_{R,z_0}(z)\d z\right)_{1,2},
\end{aligned}
\end{equation}
which are smooth functions of $x,t$ periodic in $t$ with period $1/\mathcal{L}(x)$, see~\eqref{eq:explicitJR1qplus1}--\eqref{eq:explicitJR1qplus7}.
\end{proposition}
\begin{proof}
By Proposition~\ref{prop:JRsmallqplus} we obtain that
\begin{equation}
\label{eq:smallnessqplus}
\|\boldsymbol J_R-\boldsymbol {\mathrm{I}}\|_{p}=O(t^{-1}),\qquad p=1,2,\infty,
\end{equation}
as $t\to+\infty$ uniformly for $x\in[x_*+\delta,2-\delta]$ (for any $\delta>0$), where $\|\cdot\|_p$ is the maximum over the four matrix entries of their $L^p(\Sigma_R^\circ)$-norm.
It follows that $\boldsymbol R$ solves a small-norm Riemann--Hilbert problem and so, by the same (standard) arguments in the proof of Proposition~\ref{prop:finalqminus}, we have
\begin{equation}
\boldsymbol R(z) = \boldsymbol {\mathrm{I}}+\frac 1{2\pi\i}\int_{\Sigma_R}\frac{\boldsymbol J_R(\mu)-\boldsymbol {\mathrm{I}}}{\mu-z}\d\mu+O(t^{-2})=
\boldsymbol {\mathrm{I}}+t^{-1}\sum_{z_0\in\lbrace {\wt{a}},{\wt{b}},{\wt{c}},{\wt{d}}\rbrace}\frac 1{2\pi\i}\int_{\partial \mathcal{Q}_{z_0}}\frac{\wt{\boldsymbol J}_{R,z_0}(\mu)}{z-\mu}\d\mu+O(t^{-2}).
\end{equation}
Since $\wt{\boldsymbol J}_{R,z_0}$ extends to a meromorphic function inside $\mathcal{Q}_{z_0}$ with a double pole at $z_0$ and no other singularities, see~\eqref{eq:JR1qplus}, by Cauchy's theorem the integrals above reduce to the polar part of $\wt{\boldsymbol J}_{R,z_0}$ at $z_0$, namely
\begin{equation}
\frac 1{2\pi\i}\int_{\partial \mathcal{Q}_{z_0}}\frac{\wt{\boldsymbol J}_{R,z_0}(\mu)}{z-\mu}\d\mu=
\frac{\wt{\boldsymbol J}_{R,z_0}^{(2)}}{(z-z_0)^2}+\frac{\wt{\boldsymbol J}_{R,z_0}^{(1)}}{z-z_0},
\end{equation}
assuming $z\not\in \mathcal{Q}_{z_0}$, where $\wt{\boldsymbol J}_{R,z_0}^{(0)}$ and $\wt{\boldsymbol J}_{R,z_0}^{(1)}$ are
\begin{equation}
\wt{\boldsymbol J}_{R,z_0}^{(2)}=\res{z=z_0}(z-z_0)\wt{\boldsymbol J}_{R,z_0}(z)\d z,\qquad
\wt{\boldsymbol J}_{R,z_0}^{(1)}=\res{z=z_0}\wt{\boldsymbol J}_{R,z_0}(z)\d z.
\end{equation}
When $\Im z$ is large, we have $\boldsymbol N(z)=\boldsymbol R(z)\boldsymbol P^{{\mathrm{out}}}(z)$ and so, as $t\to+\infty$ and $\Im z\to+\infty$, we have
\begin{equation}
\boldsymbol N(z) = \left(\boldsymbol {\mathrm{I}}+\frac 1{zt}\sum_{z_0\in\lbrace {\wt{a}},{\wt{b}},{\wt{c}},{\wt{d}}\rbrace}
\wt{\boldsymbol J}^{(1)}_{R,z_0}+O\bigl(z^{-1}t^{-2}\bigr)\right)\,\boldsymbol P^{{\mathrm{out}}}(z)
\end{equation}
(uniformly for $x\in [x_*+\delta,2-\delta]$ for any $\delta>0$).
By the large-$z$ expansions~\eqref{eq:refinedexpansionNqplus} and~\eqref{eq:Poutlargezqplus}, it follows that
\begin{equation}
\begin{aligned}
\wh\alpha(t,x) &= t(1+{\wt{g}}_1)+{\wt{p}}_0(t)+\wt{\mathcal X}(x,t)t^{-1}+O(t^{-2}),
\\
\e^{t(\wt{\ell}-2{\wt{g}}_\infty)}\wh\beta(t,x) &={\wt{p}}_+(t)\bigl(1+\wt{\mathcal Y}(x,t)t^{-1}+O(t^{-2})\bigr)
\end{aligned}
\end{equation}
with
\begin{equation}
\begin{aligned}
\wt{\mathcal X}(x,t)&=-\frac 1{24}+\sum_{z_0\in\lbrace {\wt{a}},{\wt{b}},{\wt{c}},{\wt{d}}\rbrace} \left(\wt{\boldsymbol J}_{R,z_0}^{(1)}\right)_{1,1},
\\
\wt{\mathcal Y}(x,t)&=\frac{1}{\i \wt{p}_+(x,t)}\sum_{z_0\in\lbrace {\wt{a}},{\wt{b}},{\wt{c}},{\wt{d}}\rbrace} \left(\wt{\boldsymbol J}^{(1)}_{R,z_0}\right)_{1,2},
\end{aligned}
\end{equation}
as $t\to+\infty$ uniformly for $x\in[x_*+\delta,2-\delta]$.
By~\eqref{eq:JtildeJqplus} and~\eqref{eq:p0pplusclosemathfrakqplus}, in the same regime we have $\wt{\mathcal X}(x,t)=\mathcal X(x,t)+O(t^{-\infty})$ and $\wt{\mathcal Y}(x,t)=\mathcal Y(x,t)+O(t^{-\infty})$ with $\mathcal X(x,t)$ and $\mathcal Y(x,t)$ as in the statement, see~\eqref{eq:mathcalXY}.
The thesis then follows from~\eqref{eq:g1ginftyqplus}, \eqref{eq:ginftyqplus}, \eqref{eq:asympellqplus}, \eqref{eq:p0pplusclosemathfrakqplus}, and~\eqref{eq:p0pplustheta}.
The fact that $\mathcal X(x,t)$ and $\mathcal Y(x,t)$ are smooth in $x,t$ and periodic in $t$ with period $1/\mathcal{L}(x)$ is manifest from the explicit formulas~\eqref{eq:explicitJR1qplus1}--\eqref{eq:explicitJR1qplus7} and the fact that $\chi(w)$, defined in~\eqref{eq:xichi}, is anti-periodic in $t$ with period $1/\mathcal{L}(x)$, and so are $\chi'(w)$ and $\chi''(w)$.
\end{proof}

\section{Asymptotics for the multiplicative average} \label{sec:proof of main thm}

\subsection{Proof for $x<x_*$}

For any $t=t_L>t_{L-1}>\dots>t_1>t_0$ we have
\begin{equation}
    Q(t,xt) = Q(t_0,xt_0)+\sum_{i=0}^{L-1}\left[\int_{t_{i}}^{t_{i+1}}\partial_\tau\log Q(\tau,xt_{i+1})\d\tau+\log\frac{Q(t_{i},xt_{i+1})}{Q(t_{i},xt_{i})}\right]
\end{equation}
If we fix $t_i=t_0-\frac i x$ (we note that $x<0$ because $x_*<0$), we can use~\eqref{eq:wtalphabetagamma} to rewrite the previous identity as
\begin{equation}
\begin{aligned}
\log Q(t,xt)
&=\log Q(t_0,xt_0)+\sum_{i=0}^{L-1}\left[-2\int_{t_{i}}^{t_{i+1}}\wh\alpha(xt_{i+1}/\tau,\tau)\d\tau+\log\wh\beta(x,t_{i})\right]
\\
&=\log Q(t_0,xt_0)+\sum_{i=0}^{L-1}\left[-2\int_{t_{i}}^{t_{i+1}}\tau(1-\e^{-\eta})\d\tau+\left(t_{i}x-\frac 12\right)\eta\right]+\sum_{i=0}^{L-1}U_i
\end{aligned}
\label{eq:proofx<xq}
\end{equation}
where 
\begin{equation}
U_i=-2\int_{t_{i}}^{t_{i+1}}\bigl(\wh\alpha(xt_{i+1}/\tau,\tau)-\tau(1-\e^{-\eta})\bigr)\d\tau+\log\wh\beta(x,t_{i})-\left(t_{i}x-\frac 12\right)\eta.
\end{equation}
Hence, we obtain
\begin{equation}
\begin{aligned}
\log Q(t,xt)
&=\log Q(t_0,xt_0) 
-(1-\e^{-\eta})(t^2-t_0^2)
+\frac{L(2t_0x-L)}2\eta+\sum_{i=0}^{L-1}U_i
\\
&=\log Q(t_0,xt_0) -(t^2-t_0^2)\mathcal{F}(x)+\sum_{i=0}^{L-1}U_i
\end{aligned}
\end{equation}
where we recall that $t=t_L$ and so $L=(t_0-t)x$ and that $\mathcal{F}(x) = \eta\frac{x^2}2+1-\e^{-\eta}$ when $x<x_*$.
It follows from Proposition~\ref{prop:finalqminus} that $U_i=O(t_i^{-2})$ (uniformly in $i$ provided we choose $t_0$ large enough) such that we can write
\begin{equation}
\log Q(t,xt) = -t^2\mathcal{F}(x) + C(t_0,x) - \sum_{i=L}^{+\infty}U_i=-t^2\mathcal{F}(x) + C(t_0,x) + O(t^{-1})
\end{equation}
where $C(t_0,x)=\log Q(t_0,xt_0)+t_0^2\mathcal{F}(x)+\sum_{i=0}^{+\infty}U_i$ is a constant independent of $L,t_L$.
Clearly, this constant cannot depend on $t_0$ and Theorem~\ref{thm:main} is proved for $x<x_*$.

\subsection{Proof for $x_*<x<2$}

We need a few preliminary lemmas.
In the first one we prove some identities involving certain quantities introduced throughout this paper, specifically, $\mathcal{U}$, $\mathcal{K}$ (given in~Definition~\ref{def:K(x)}), $\mathcal{L}$, $\mathcal{F}$ (given in Definition~\ref{def:Fq}), $g_1$, $g_\infty$ and~$\ell$ (given in Proposition~\ref{prop:g1ginftyqplus}), and $p_0$ (given in Proposition~\ref{prop:p0pplus}).

\begin{lemma}
\label{lemma:finalidentities}
    The following identities hold true for all $x\in(x_*,2)$:
    \begin{align}
    \label{eq:finalidentityLprime}
    \frac{\d \mathcal{L}(x)}{\d x}&=-\frac{\eta}{2\mathcal{K}(x)},
    \\
    \label{eq:finalidentityL-Lprime}
    \mathcal{L}(x)-x\frac{\d\mathcal{L}(x)}{\d x}&=
    -\frac{\mathcal{U}\bigl(\mathcal{K}(x)\bigr)}{\mathcal{K}(x)},
    \\
    \label{eq:finalidentityp00}
    -2p_0(0,t) &= \frac{\partial}{\partial t}\log\vartheta\bigl(t\mathcal{L}(0)\big|\frac{\i\pi}{\mathcal{K}(0)}\bigr),
    \\
    \label{eq:finalidentityFsecond}
    \frac{\d^2\mathcal{F}(x)}{\d x^2}&=\eta\biggl(1-\frac{\eta}{2\mathcal{K}(x)}\biggr),
    \\
    \label{eq:finalidentityFprime}
    \frac{\d\mathcal{F}(x)}{\d x} &= \eta\bigl(x+\mathcal{L}(x)\bigr)=2g_\infty(x)-\ell(x),
    \\
    \label{eq:compatibility}
    \frac{\d g_1(x)}{\d x}&=\frac 12\biggl(1-x\frac{\d}{\d x}\biggr)\bigl(2g_\infty(x)-\ell(x)\bigr),
    \\
    \label{eq:finalidentityF}
    \mathcal{F}(x) &= \frac{x}{2}\bigl(2g_\infty(x)-\ell(x)\bigr)+1+g_1(x).
    \end{align}
\end{lemma}
\begin{proof}
By the defining relation 
\begin{equation}
\label{eq:definingrel}
\bigl(\mathcal{U}(K)-K\partial_K\mathcal{U}(K)\bigr)\big|_{K=\mathcal{K}(x)}=-\frac{\eta x}{2}
\end{equation}
(see Definition~\ref{def:K(x)}) we obtain
\begin{equation}
\label{eq:diffdefiningrel}
\mathcal{K}(x)\frac{\mathrm{d}\mathcal{K}(x)}{\mathrm{d}x}\left.\frac{\partial^2 \mathcal{U}(K)}{\partial K^2}\right|_{K=\mathcal{K}(x)} = -\frac{\eta}{2}.
\end{equation}
Then,~\eqref{eq:finalidentityLprime} is immediate from the definition of $\mathcal{L}$ in~\eqref{eq:defL}.

Equation~\eqref{eq:finalidentityL-Lprime} follows from~\eqref{eq:finalidentityLprime} and the definition of $\mathcal{L}$ in~\eqref{eq:defL}. 

To prove~\eqref{eq:finalidentityp00}, by~\eqref{eq:p0pplustheta} it suffices to show the identity
\begin{equation}
\mathcal{L}(0)=-2\e^{\frac{\eta}2(\frac{\eta}{2\mathcal{K}(0)}-1)}\frac{\vartheta_{11}\bigl(\frac{\eta}{2\mathcal{K}(0)}\big|\frac{\i\pi}{\mathcal{K}(0)}\bigr)}{\vartheta_{11}'\bigl(0\big|\frac{\i\pi}{\mathcal{K}(0)}\bigr)}.
\end{equation}
This follows by setting~$x=0$ into~\eqref{eq:finalidentityL-Lprime} and recalling the expression of~$\mathcal{U}(K)$ in~\eqref{eq:f1}.

Equation~\eqref{eq:finalidentityFsecond} is immediate from~\eqref{eq:finalidentityLprime} and the definition~\eqref{eq:Fq} of the rate function~$\mathcal{F}$.

Next, by~\eqref{eq:Fq} we immediately have $\frac{\d\mathcal{F}(x)}{\d x}=\eta\bigl(x+\mathcal{L}(x)\bigr)$.
It is easy to verify, using~\eqref{eq:g1ginftyqplus} and the defining relation~\eqref{eq:definingrel}, that we also have $2g_\infty(x)-\ell(x)=\eta\bigl(x+\mathcal{L}(x)\bigr)$.

For the proof of~\eqref{eq:compatibility}, we first note that the right-hand side equals $-\frac{\eta\mathcal{U}(\mathcal{K})}{2\mathcal{K}}$ by~\eqref{eq:finalidentityFprime} and~\eqref{eq:finalidentityL-Lprime}.
Therefore, it is enough to show that $\frac{\mathrm{d}g_1(x)}{\mathrm{d}x}=-\frac{\eta\mathcal{U}(\mathcal{K})}{2\mathcal{K}}$.
To perform this computation, it is convenient to introduce $\mathcal{Z} (K)= \zeta(\eta|K,\i\pi)-\frac{\zeta(\i\pi|K,\i\pi)}{\i\pi}\eta$, such that we can rewrite~\eqref{eq:ZZZZ} as
\begin{equation}
\label{eq:dKlogUZ}
\partial_K\log \mathcal{U}(K) = \wp(\eta|K,\i\pi)-2\frac{\zeta(\i\pi|K,\i\pi)}{\i\pi}-\mathcal Z(K)^2
\end{equation}
and so, using the expression for $g_1(x)$ in~\eqref{eq:g1ginftyqplus} as well as Proposition~\ref{prop:relf1f2},
\begin{equation}
g_1(x)= -\frac{\mathcal{U}\bigl(\mathcal{K}(x)\bigr)^2}{2} \biggl(3\wp\bigl(\eta\big|\mathcal{K}(x),\i\pi\bigr) -\mathcal{Z}(\mathcal{K}(x)\bigr)^2 \biggr).
\end{equation}
We note that (also using Lemma~\ref{lemma:WeierstrassDerivativesinK})
\begin{equation}
\label{eq:derivatives Upsilon}
-\partial_\eta \mathcal{Z}(K) = \wp(\eta|K,\i\pi)+\frac{\zeta(\i\pi|K,\i\pi)}{\i\pi},\qquad
\partial_K\mathcal{Z}(K) = \partial_\eta\bigl( \wp(\eta|K,\i\pi)-\mathcal{Z}(K)^2\bigr).
\end{equation}
In view of the chain rule and~\eqref{eq:diffdefiningrel}, the relation \eqref{eq:compatibility} is a consequence of the identity
\begin{equation}\label{eq:first compatibility equiv}
    \frac{\partial_K^2\mathcal{U}(K)}{\mathcal{U}(K)} = \frac{\partial_K\mathcal{U}(K)}{\mathcal{U}(K)} \bigl(3 \wp(\eta) - \mathcal{Z}(K)^2 \bigr) + \frac{1}{2} \partial_K \bigl( 3 \wp(\eta) - \mathcal{Z}(K)^2 \bigr).
\end{equation}
To prove the latter, we express the left-hand side as
\begin{equation}
\label{eq:secondlogdereasy}
    \frac{\partial_K^2\mathcal{U}(K)}{\mathcal{U}(K)} =  \left( \frac{\partial_K\mathcal{U}(K)}{\mathcal{U}(K)} \right)^2 + \partial_K \frac{\partial_K\mathcal{U}(K)}{\mathcal{U}(K)}
\end{equation} 
such that we can evaluate all terms in \eqref{eq:first compatibility equiv} by using~\eqref{eq:dKlogUZ}.
Subtracting the left-hand side from the right-hand side of \eqref{eq:first compatibility equiv} several simplifications occur, as shown below
\begin{equation}
    \begin{aligned}
        &\frac{\partial_K^2\mathcal{U}}{\mathcal{U}} - \frac{\partial_K\mathcal{U}}{\mathcal{U}} \bigl(3 \wp(\eta) - \mathcal{Z}^2 \bigr) - \frac{1}{2} \partial_K \bigl( 3 \wp(\eta) - \mathcal{Z}^2 \bigr)
        \\
        &= -\mathcal{Z} \partial_K \mathcal{Z} - \frac{1}{2} \partial_K \bigl( \wp(\eta) + \frac{\zeta(\i \pi)}{\i \pi} \bigr) + 2 \mathcal{Z}^2 \bigl( \wp(\eta) + \frac{\zeta(\i \pi)}{\i \pi} \bigr) 
        \\
        & \qquad 
        - \frac{3}{2} \partial_K  \frac{\zeta(\i \pi)}{\i \pi} - 2 \wp(\eta)^2 + 2 \wp(\eta) \frac{\zeta(\i \pi)}{\i \pi} + 4 \left( \frac{\zeta(\i \pi)}{\i \pi} \right)^2 
        \\
        &= -\mathcal{Z}\partial_K \mathcal{Z} + \frac{1}{2} \partial_K \partial_\eta \mathcal{Z} - 2 \mathcal{Z}^2 \partial_\eta \mathcal{Z} + \left( \partial_\eta \mathcal{Z} \right)^2 - \frac{1}{2} \partial_\eta^2\wp(\eta) = 0,
    \end{aligned}
\end{equation}
where we use~\eqref{eq:dKlogUZ}, \eqref{eq:derivatives Upsilon}, \eqref{eq:secondlogdereasy}, \eqref{eq:partialKsigmazetawp}, and~\eqref{eq: partial K zeta i pi}, along with straightforward, although lengthy, algebraic manipulations.
The proof of~\eqref{eq:compatibility} is complete.

Finally, by integrating~\eqref{eq:finalidentityFprime} it follows that left-hand and right-hand sides in~\eqref{eq:finalidentityF} are equal up to an additive constant so it suffices to verify the identity in the limit $x\downarrow x_*$. In this limit, $\mathcal{K}(x)\to+\infty$, $\mathcal{U}(\mathcal{K})\to 1-\e^{-\eta}$, $\mathcal{V}(\mathcal{K})\to 1$ (see Proposition~\ref{prop:f1f2body}).
As proved above, $2g_\infty(x)-\ell(x)=\eta\bigl(x+\mathcal{L}(x)\bigr)$ and so, since $\mathcal{L}(x)\to 0$ in this limit, $2g_\infty(x)-\ell(x)\to \eta x_*$ as $x\downarrow x_*$.
Using the same estimates, as well as those in~\eqref{eq:triglimitKinfty}, it is easy to check that $g_1(x)\to -\e^{-\eta}$ such that the claimed identity holds as $x\downarrow x_*$ and the proof is complete.
\end{proof}

\begin{remark}
One expects~\eqref{eq:compatibility} to hold true because it is a \emph{compatibility condition} of the large-$t$ expansions of~$\widehat{\alpha}(t,x)$ and~$\widehat{\beta}(t,x)$ of Proposition~\ref{prop:finalqplus} in view of the identity
\begin{equation}
\widehat{\alpha}(t,x+\tfrac 1t)-\widehat{\alpha}(t,x) = -\frac 12\biggl(\frac{\partial}{\partial t}-\frac xt\frac{\partial}{\partial x}\biggr)\log\widehat{\beta}(t,x),
\end{equation}
which follows directly from the definitions of~$\widehat{\alpha}(t,x)$ and~$\widehat{\beta}(t,x)$ given in~\eqref{eq:wtalphabetagamma}.
\end{remark}

\begin{remark}
\label{remark:explicitFWeierstrass}
By combining~\eqref{eq:finalidentityF} with the explicit expressions for~$g_1$ and~$2g_\infty-\ell$ from Proposition~\ref{prop:g1ginftyqplus} and using the defining relation~\eqref{eq:definingrel}, we obtain
\begin{equation}
\label{eq:explicitFWeierstrass}
\mathcal{F}(x) = 1+\frac{\eta}{2}\biggl(1-\frac{\eta}{2\mathcal{K}}\biggr)x^2-\frac{3\eta\mathcal{U}(\mathcal{K})}{4\mathcal{K}}x-\mathcal{U}(\mathcal{K})^2\biggl(\wp(\eta)+\frac{\zeta(\mathcal{K})}{\mathcal{K}}\biggr),
\end{equation}
where, as usual, $\mathcal{K}=\mathcal{K}(x)$ is given in Definition~\ref{def:K(x)}, and the half-periods of the Weierstrass elliptic functions are~$\mathcal{K}$ and~$\i\pi$.
We obtain~\eqref{eq:explicitF} from the relation between $\wp$ and $\vartheta_{11}$ given in~\eqref{eq:relwptheta}.
\end{remark}

The next lemmas clarify some simple properties of integrals involving periodic functions.

\begin{lemma}
\label{lemma:periodiceasy}
    Let $\psi:\mathbb{R}\to\mathbb{R}$ be a smooth periodic function with period $\mathcal{P}$ such that $\int_0^\mathcal{P}\psi(t)\d t=0$. Then,
    \begin{equation}
    \left|\int_{t}^{+\infty}\frac{\psi(\tau)}{\tau}\d\tau \right|\leq C_\psi\,t^{-1}
    \end{equation}
    for all $t>0$ and for a suitable constant $C_\psi>0$ depending on $\psi$.
\end{lemma}
\begin{proof}
The assumption that $\int_0^\mathcal{P}\psi(t)\d t=0$ implies that there exists a smooth periodic function $\Psi:\mathbb{R}\to\mathbb{R}$ with the same period $\mathcal{P}$ such that $\frac{\d \Psi(t)}{\d t}=\psi(t)$. Hence, integrating by parts,
\begin{equation}
\int_t^T\frac{\psi(\tau)}{\tau}\d\tau = \int_t^T\frac{\Psi(\tau)}{\tau^2}\d\tau+\frac{\Psi(T)}{T}-\frac{\Psi(t)}{t}
\end{equation}
and the proof is straightforward after letting $T\to+\infty$.
\end{proof}
\begin{lemma}
\label{lemma:periodiclesseasy}
    Let $I\subseteq \mathbb{R}$ be an open interval.
    Let $\psi:I\times\mathbb{R}\to\mathbb{R}$ be a smooth function such that $\psi(x,t+1)=\psi(x,t)$ for all~$(x,t)\in I\times\mathbb{R}$ and such that $\int_0^{1}\psi(x,t)\d t=0$ for all~$x\in I$.
    Let $L:I\to \mathbb{R}$ be a smooth function such that $L(x)\not=0$ and $\frac{\d L(x)}{\d x}\not=0$ for all $x\in I$.
    Then, if $J$ is a closed interval contained in $I$, we have
    \begin{equation}
    \left|\int_J\psi\bigl(y,tL(y)\bigr)\d y\right|\leq C_{\psi,L,J}\,t^{-1}
    \end{equation}
    for all $t>0$ and for a suitable constant $C_{\psi,J}>0$ depending on $\psi$, $L$, and $J$ only.
\end{lemma}
\begin{proof}
The assumption that $\int_0^1\psi(x,t)\d t=0$ for all $x\in I$ implies that there exists a smooth function $\Psi:I\times \mathbb{R}\to\mathbb{R}$ such that $\partial_t\Psi(x,t)=\psi(x,t)$ and $\Psi(x,t+1)=\Psi(x,t)$ for all $(x,t)\in I\times\mathbb{R}$. 
In particular,
\begin{equation}
\label{eq:behchissachenomedarti}
|\Psi(x,t)|\leq M\qquad\text{ and }\qquad |\partial_x \Psi(x,t)|\leq M
\end{equation}
for all $(x,t)\in J\times \mathbb{R}$ and for some constant $M$ depending only on $\psi$ and $J$.
By the identity
\begin{equation}
\psi\bigl(y,t L(y)\bigr) = \frac 1 {t\,L'(y)}\biggl[\frac{\d}{\d y}\Psi\bigl(y,t L(y)\bigr)\,-\,\bigl(\partial_y\Psi(y,T)\bigr)\big|_{T=t L(y)}\biggr]
\end{equation}
(where a prime denotes a derivative in $y$) to complete the proof of the lemma it suffices to note that~\eqref{eq:behchissachenomedarti} implies that
\begin{equation}
\int_J\frac 1{L'(y)}\frac{\d}{\d y}\Psi\bigl(y,t L(y)\bigr)\d y =
\int_J\frac{L''(y)}{L'(y)^2}\Psi\bigl(y,t L(y)\bigr)\d y
+\frac{\Psi\bigl(x_2,t L(x_2)\bigr)}{L'(x_2)}
-\frac{\Psi\bigl(x_1,t L(x_1)\bigr)}{L'(x_1)}
\end{equation}
(where we integrated by parts and denoted $J=[x_1,x_2]$) and
\begin{equation}
\int_J\frac 1{L'(y)}\bigl(\partial_y\Psi(y,T)\bigr)\big|_{T=t L(y)}\d y
\end{equation}
are both bounded by a constant depending only on $\psi$, $L$, and $J$ .
\end{proof}

Finally, the following lemmas provide estimates for sums of rapidly oscillating terms, which commonly arise in the van der Corput method in analytic number theory.

\begin{lemma}[Chapter I, Theorem 6.3 in \cite{tenenbaum2015introduction}]
\label{lemma:tenen1}
    Let $I\subset\mathbb{R}$ be an interval and let $\psi\in C^2(I,\mathbb{R})$ be a function such that $r = \inf_{I} | \psi''| > 0$.
    Then,
    \begin{equation}
        \left| \int_I \e^{2 \pi \i \psi(y) } \d y \right| \le \frac{4}{\sqrt{\pi r}}.
    \end{equation}
\end{lemma}

\begin{lemma}[Chapter I, Theorem 6.4 in \cite{tenenbaum2015introduction}]
\label{lemma:tenen2}
    Let $I\subset \mathbb{R}$ be an interval and let $\psi \in C^1(I,\mathbb{R})$ be a function such that $\psi'$ is monotone on $I$.
    Then,
    \begin{equation}
        \sum_{n\in I} \e^{2 \pi \i \psi(n)} = \sum_{\alpha - 1 < \nu < \beta +1} \int_{I}\e^{2 \pi \i [\psi(y) - \nu y]} \d y +O\bigl( \log(\beta - \alpha + 2) \bigr)
    \end{equation}
    where $\alpha = \inf_{I} \psi'$ and $\beta = \sup_{I} \psi'$.
\end{lemma}

\begin{lemma} \label{lem:van der corput}
    Let $I\subset\mathbb{R}$ be an interval and let $\psi \in C^2(I,\mathbb{R})$ be a function such that $r = \inf_{I} | \psi'' | > 0$.
    Then,
    \begin{equation}
        \left| \sum_{n\in I} \e^{2 \pi \i \psi(n)} \right| \le \frac{4}{\sqrt{\pi r}} (\beta - \alpha+2) +O\bigl( \log(\beta - \alpha + 2) \bigr)
    \end{equation}
    where $\alpha = \inf_{I} \psi'$ and $\beta = \sup_{I} \psi'$.
\end{lemma}
\begin{proof}
    It follows from Lemmas~\ref{lemma:tenen1} and~\ref{lemma:tenen2}.
\end{proof}

\begin{lemma} \label{lem:kusmin-landau_weighted}
    Let $I,\psi,r,\alpha,\beta$ be as in Lemma~\ref{lem:van der corput} and let $\phi\in C^1(I,\mathbb{R})$.
    Then,
    \begin{equation}
        \left| \sum_{n \in I} \phi(n) \e^{2 \pi \i \psi(n)} \right| \le  \left(  \frac{4(\beta - \alpha+2)}{\sqrt{\pi r}} + O\bigl(\log(\beta-\alpha+2)\bigr) \right)\left(\sup_{I} |\phi|  + |I| \sup_{I} | \phi' | \right).
    \end{equation}
\end{lemma}
\begin{proof}
    Let $I= [n_0,n_1]$. By Abel's summation by part,
\begin{equation}
    \sum_{n=n_0}^{n_1} \phi(n) \e^{ 2 \pi \i \psi(n)} = \phi(n_1) \sum_{n=n_0}^{n_1} \e^{2 \pi \i  \psi(n)} - \sum_{n=n_0}^{n_1-1} \bigl(\phi(n+1) - \phi(n) \bigr) \sum_{m=n_0}^n \e^{ 2 \pi \i  \psi(m)}.
\end{equation}
Combining~\Cref{lem:van der corput} with the estimate $|\phi(n+1)-\phi(n)| \le \sup_{ I} | \phi' |$ yields the desired bound.
\end{proof}

We are now ready to give the proof.
We separate the analysis according to~$s=0$, $s\in\mathbb{Z}_{>0}$, or~$s\in\mathbb{Z}_{<0}$.

First, let $s=0$. We have, for any $0<t_0<t$, see~\eqref{eq:wtalphabetagamma},
\begin{equation}
\begin{aligned}
    \log Q(t,0)&=\log Q(t_0,0)-2\int_{t_0}^t\wh\alpha(0,\tau)\d\tau
    \\
    &=\log Q(t_0,0)-2\int_{t_0}^t\biggl(\wh\alpha(0,\tau)-\tau\bigl(1+g_1(0)\bigr)-p_0(0,\tau)-\frac{\mathcal X(0,\tau)}{\tau}\biggr)\d\tau
\\
&\quad
-(t^2-t_0^2)\bigl(1+g_1(0)\bigr)
+\log\frac{\vartheta\bigl(t\mathcal{L}(0)|\frac{\i\pi}{\mathcal{K}(0)}\bigr)}{\vartheta\bigl(t_0 \mathcal{L}(0)|\frac{\i\pi}{\mathcal{K}(0)}\bigr)}
-2\overline{\mathcal X}(0)\log\frac{t}{t_0}
-2\int_{t_0}^t\frac{\mathcal X_0(0,\tau)}{\tau}\d\tau
\end{aligned}
\end{equation}
where we use the identity~\eqref{eq:finalidentityp00} and we decompose $\mathcal X(0,t)$, which is periodic in~$t$ with period $1/\mathcal{L}(0)$, into a constant independent of $t$ and a zero-mean periodic function of $t$ as
\begin{equation}
\mathcal X(0,t) = \overline{\mathcal X}(0)+\mathcal X_0(0,t),\qquad \overline{\mathcal X}(0)=\mathcal{L}(0)\int_0^{1/\mathcal{L}(0)}\mathcal{X}(0,\tau)\d\tau.
\end{equation}
Since $1+g_1(0)=\mathcal{F}(0)$ by~\eqref{eq:finalidentityF}, we obtain
\begin{equation}
\log Q(t,0)=-t^2\mathcal{F}(0)+\log\vartheta\bigl(t\mathcal{L}(0)\big|\frac{\i\pi}{\mathcal{K}(0)}\bigr)-2\overline{\mathcal X}(0)\log t+C(t_0)+O(t^{-1})
\end{equation}
where
\begin{equation}
C(t_0)=\log Q(t_0,0)+t_0^2\mathcal{F}(0)-\log\vartheta\bigl(t_0\mathcal{L}(0)\big|\frac{\i\pi}{\mathcal{K}(0)}\bigr)+2\overline{\mathcal X}(0)\log t_0+\int_{t_0}^{+\infty}\frac{\mathcal X_0(0,\tau)}\tau\d\tau
\end{equation}
and we used Proposition~\ref{prop:finalqplus} as well as the estimate
\begin{equation}
\int_{t}^{+\infty}\frac{\mathcal X_0(0,\tau)}\tau\d\tau=O(t^{-1}),\qquad t\to+\infty,
\end{equation}
which follows from Lemma~\ref{lemma:periodiceasy}.
It is evident that $C(t_0)$ cannot actually depend on $t_0$ and so it is an absolute constant (depending only on~$q$) and the theorem is proved in this case, namely we have shown that
\begin{equation}
\label{eq:finals=0}
\log Q(t,0)=-t^2\mathcal{F}(0)+\log\vartheta\bigl(t\mathcal{L}(0)\big|\frac{\i\pi}{\mathcal{K}(0)}\bigr)+\mathcal{A}\log t+\log\mathcal{C}(0)+O(t^{-1})
\end{equation}
for a suitable constant $\mathcal{C}(0)\neq0$ depending on $\eta$ only and
\begin{equation}
\label{eq:Afinal}
\mathcal{A}=-2\overline{\mathcal X}(0) =-2 \mathcal{L}(0) \int_0^{1/\mathcal{L}(0)} \mathcal{X}(0,\tau) \d \tau,
\end{equation}
which is also a constant depending on $\eta$ only.
We recall that $\mathcal X(x,t)$ is given explicitly by~\eqref{eq:mathcalXY} and~\eqref{eq:explicitJR1qplus1}--\eqref{eq:explicitJR1qplus7}. See \Cref{fig:A} for a plot of $\mathcal{A}$ as a function of $\eta$.

\begin{figure}
    \centering
    \includegraphics[width=0.5\linewidth]{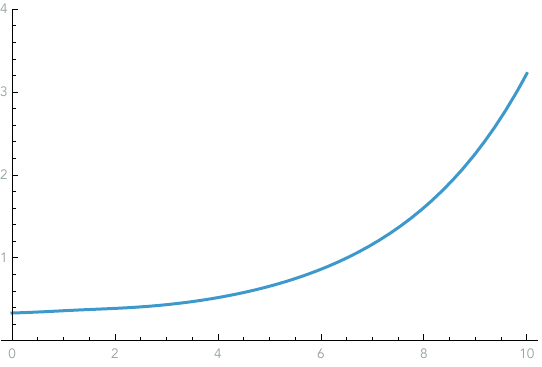}
    \caption{Plot of the function $\mathcal{A} = \mathcal{A}(\eta)$.}
    \label{fig:A}
\end{figure}

Next, let $s=xt\in\mathbb{Z}_{>0}$ with $x\in (0,2)$.
We have
\begin{equation}
\label{eq:s>0_1}
\log Q(t,s)=\log Q(t,0)+\sum_{i=1}^{s}\log\frac{Q(t,i)}{Q(t,i-1)}=\log Q(t,0)-\sum_{i=1}^{s}\log\wh\beta\bigl(\frac it,t\bigr).
\end{equation}
Setting
\begin{equation}
V_i=\log\wh\beta\bigl(\frac it,t\bigr) -t \mathcal{F}'\bigl(\frac it\bigr)-\log p_+\bigl(\frac it,t\bigr)-t^{-1}\mathcal Y\bigl(\frac it,t\bigr),
\end{equation}
by Proposition~\ref{prop:finalqplus} and~\eqref{eq:finalidentityFprime} we have $V_i=O(t^{-2})$, uniformly for $i\in\lbrace 1,\dots, s\rbrace$, hence 
\begin{equation}
\begin{aligned}
\sum_{i=1}^s\log\wh\beta(\frac it,t\bigr) &= t \sum_{i=1}^s\mathcal{F}'\bigl(\frac it\bigr)+\sum_{i=1}^s\log p_+\bigl(\frac it,t\bigr)+t^{-1}\sum_{i=1}^s\mathcal Y\bigl(\frac it,t\bigr)+\sum_{i=1}^sV_i 
\\
\label{eq:s>0_2}
&=t \sum_{i=1}^s\mathcal{F}'\bigl(\frac it\bigr)+\sum_{i=1}^s\log p_+\bigl(\frac it,t\bigr)+t^{-1}\sum_{i=1}^s\mathcal Y\bigl(\frac it,t\bigr)+ O(t^{-1}).
\end{aligned}
\end{equation}
Let us recall the Euler--Maclaurin formulas
\begin{align}
\label{eq:EM1}
\sum_{i=1}^s\psi(i) &= \int_0^s \psi(y)\d y+\frac{\psi(s)-\psi(0)}{2}+\int_0^s\psi'(y)B_1\bigl(y-\lfloor y\rfloor\bigr)\d y,
\\
\label{eq:EM2}
\sum_{i=1}^s\psi(i) &= \int_0^s \psi(y)\d y+\frac{\psi(s)-\psi(0)}{2}+\frac{\psi'(s)-\psi'(0)}{12}-\frac 12\int_0^s\psi''(y)B_2\bigl(y-\lfloor y\rfloor\bigr)\d y,
\end{align}
which hold for any smooth function $\psi:\mathbb{R}\to\mathbb{R}$, where $B_1(y)=y-\frac 12$ and $B_2(y)=y^2-y+\frac 16$ are the first and second Bernoulli polynomials.
Applying~\eqref{eq:EM2} to $\psi(s)=t\mathcal{F}'(s/t)$ we obtain
\begin{equation}
\label{eq:s>0_3}
t \sum_{i=1}^s\mathcal{F}'\bigl(\frac it\bigr) = 
t^2\bigl(\mathcal{F}(x)-\mathcal{F}(0)\bigr)
+\frac t2\bigl(\mathcal{F}'(x)-\mathcal{F}'(0)\bigr)
+\frac 1{12}\bigl(\mathcal{F}''(x)-\mathcal{F}''(0)\bigr)
+O(t^{-1}).
\end{equation}
Next we decompose, thanks to~\eqref{eq:p0pplustheta} and~\eqref{eq:finalidentityFsecond},
\begin{equation}
\log p_+(x,t) = -\frac 12\mathcal F''_q(x)+\log p_+^0(x,t),\quad
p_+^ 0(x,t)=\frac{\vartheta(t\mathcal{L}(x)+\frac{\eta}{2K(x)}|\frac{\i\pi}{\mathcal{K}(x)})}{\vartheta(t\mathcal{L}(x)|\frac{\i\pi}{\mathcal{K}(x)})}.
\end{equation}
Applying the Euler--Maclaurin formula~\eqref{eq:EM1} to $\psi(s)=\mathcal{F}''(s/t)$ we obtain
\begin{equation}
\label{eq:s>0_4}
\sum_{i=1}^s\log p_+\bigl(\frac it,t\bigr)=
\sum_{i=1}^s\log p_+^0\bigl(\frac it,t\bigr)
-\frac t2\bigl(\mathcal{F}'(x)-\mathcal{F}'(0)\bigr)
-\frac 14\bigl(\mathcal{F}''(x)-\mathcal{F}''(0)\bigr)+O(t^{-1}).
\end{equation}
Next, we turn to the sum
\begin{equation}\label{eq:sum_p0+}
\begin{aligned}
\sum_{i=1}^s\log p_+^0\bigl(\frac it,t\bigr) &= \sum_{i=1}^s\log\frac{\vartheta\bigl(t\mathcal{L}\bigl(\frac it\bigr)+\frac{\eta}{2\mathcal{K}(\frac it)}\big|\i\pi \mathcal{K}\bigl(\frac it\bigr)^{-1}\bigr)}{\vartheta\bigl(t\mathcal{L}\bigl(\frac it\bigr)\big|\i\pi \mathcal{K}\bigl(\frac it\bigr)^{-1}\bigr)}
\\
&=\sum_{i=1}^s\log\frac{\vartheta\bigl(t\mathcal{L}\bigl(\frac it\bigr)+\frac{\eta}{2\mathcal{K}(\frac it)}\big|\i\pi \mathcal{K}\bigl(\frac it\bigr)^{-1}\bigr)}{\vartheta\bigl(t\mathcal{L}\bigl(\frac{i-1}t\bigr)\big|\i\pi \mathcal{K}\bigl(\frac {i-1}t\bigr)^{-1}\bigr)}
+\sum_{i=1}^s\log\frac{\vartheta\bigl(t\mathcal{L}\bigl(\frac {i-1}t\bigr)\big|\i\pi \mathcal{K}\bigl(\frac {i-1}t\bigr)^{-1}\bigr)}{\vartheta\bigl(t\mathcal{L}\bigl(\frac it\bigr)\big|\i\pi \mathcal{K}\bigl(\frac it\bigr)^{-1}\bigr)}
\\
&=
\sum_{i=1}^s\log\frac{\vartheta\bigl(t\mathcal{L}\bigl(\frac it\bigr)+\frac{\eta}{2\mathcal{K}(\frac it)}\big|\i\pi \mathcal{K}\bigl(\frac it\bigr)^{-1}\bigr)}{\vartheta\bigl(t\mathcal{L}\bigl(\frac{i-1}t\bigr)\big|\i\pi \mathcal{K}\bigl(\frac {i}t\bigr)^{-1}\bigr)}
+
\sum_{i=1}^s\log\frac{\vartheta\bigl(t\mathcal{L}\bigl(\frac{i-1}t\bigr)\big|\i\pi \mathcal{K}\bigl(\frac {i}t\bigr)^{-1}\bigr)}{\vartheta\bigl(t\mathcal{L}\bigl(\frac{i-1}t\bigr)\big|\i\pi \mathcal{K}\bigl(\frac {i-1}t\bigr)^{-1}\bigr)}
\\
&\,\,\,\,\,\,\,
+\log\frac{\vartheta\bigl(t\mathcal{L}(0)\big|\i\pi \mathcal{K}(0)^{-1}\bigr)}{\vartheta\bigl(t\mathcal{L}(x)\big|\i\pi \mathcal{K}(x)^{-1}\bigr)}
\end{aligned}
\end{equation}
Let us focus on the first sum.
By the Jacobi triple product \eqref{eq:triple_product}, we have
\begin{equation}
    \begin{split}
        &\log\frac{\vartheta\bigl(t\mathcal{L}\bigl(\frac it\bigr)+\frac{\eta}{2\mathcal{K}(\frac it)}\big|\i\pi \mathcal{K}\bigl(\frac it\bigr)^{-1}\bigr)}{\vartheta\bigl(t\mathcal{L}\bigl(\frac{i-1}t\bigr)\big|\i\pi \mathcal{K}\bigl(\frac {i}t\bigr)^{-1}\bigr)} 
        \\
        &= \sum_{\varepsilon \in \{\pm 1\}} \sum_{m \ge 1} \log \left[ \frac{1+\exp\{- \frac{ (2m-1) \pi^2}{\mathcal{K}(\frac{i}{t})}+2\pi \i \varepsilon [t \mathcal{L}(\frac{i}{t}) -\mathcal{L}'(\frac{i}{t})] \} }{ 1+\exp\{- \frac{ (2m-1) \pi^2}{\mathcal{K}(\frac{i}{t})}+2\pi \i \varepsilon [t \mathcal{L}(\frac{i}{t}) -\mathcal{L}'(\frac{i}{t}) +\frac{\mathcal{L}''(\frac{i}{t})}{2t} + O(\frac{1}{t^2})] \} } \right]
        \\
        &= \sum_{\varepsilon \in \{\pm 1\}} \sum_{m \ge 1} \frac{\i \pi \varepsilon \mathcal{L}''(\frac{i}{t})}{t} \frac{\exp\{- \frac{ (2m-1) \pi^2}{\mathcal{K}(\frac{i}{t})}+2\pi \i \varepsilon [t \mathcal{L}(\frac{i}{t}) -\mathcal{L}'(\frac{i}{t})] \} }{1 + \exp\{- \frac{ (2m-1) \pi^2}{\mathcal{K}(\frac{i}{t})}+2\pi \i \varepsilon [t \mathcal{L}(\frac{i}{t}) -\mathcal{L}'(\frac{i}{t})] \} } + O \left( \frac{1}{t^2} \right)
        \\
        &= \sum_{\varepsilon \in \{\pm 1\}} \sum_{m, \ell \ge 1}  \frac{ (-1)^{\ell-1} \i \pi \varepsilon \mathcal{L}''(\frac{i}{t})}{t} \exp\biggl\{- \frac{ \ell(2m-1) \pi^2}{\mathcal{K}(\frac{i}{t})}+2\pi \i \ell \varepsilon \biggl[t \mathcal{L} \biggl( \frac{i}{t} \biggr) -\mathcal{L}'\biggl(\frac{i}{t} \biggr) \biggr] \biggr\} + O\biggl( \frac{1}{t^2} \biggr).
    \end{split}
\end{equation}
To perform the summation over $i$ we will use Lemma~\ref{lem:kusmin-landau_weighted} with $\psi=\psi_{\varepsilon,\ell}$, $\phi=\phi_{\varepsilon,\ell,m}$:
\begin{equation}
    \psi_{\varepsilon,\ell}(i) = \varepsilon \ell t \mathcal{L} \left( \frac{i}{t} \right),\quad 
    \phi_{\varepsilon,\ell,m}(i) =  \frac{ (-1)^{\ell-1} \i \pi \varepsilon \mathcal{L}''(\frac{i}{t})}{t} \exp\left\{- \frac{ \ell(2m-1) \pi^2}{\mathcal{K}(\frac{i}{t})} - 2\pi \i \ell \varepsilon \mathcal{L}' \left( \frac{i}{t} \right)  \right\}.
\end{equation}
We observe that, by \eqref{eq:finalidentityLprime}, we have
\begin{equation}
    \psi_{\varepsilon,\ell}'(i) = -\frac{\ell \varepsilon \eta}{2 \mathcal{K}(\frac{i}{t})}, 
    \qquad
    \psi_{\varepsilon,\ell}''(i) = \frac{\ell \varepsilon \eta \mathcal{K}'(\frac{i}{t})}{2t \mathcal{K}(\frac{i}{t})^2}, \qquad
    \text{for}~ 1\le i \le xt,
\end{equation}
and since $\mathcal{K}$ is a monotonically increasing function, there exists a constant $D>0$ depending on $x$, such that
\begin{equation}
    \frac{\ell}{D}  \le  |\psi_{\varepsilon,\ell}'(i)| \le  \ell D, \qquad 
    |\psi_{\varepsilon,\ell}''(i)| \ge \frac{\ell}{t D}
    \qquad
    \text{for}~ 1\le i \le xt.
\end{equation}
Moreover, by direct inspection of the function $\phi_{\varepsilon,\ell,m}$, there exists $\delta>0$ independent of $i,m,\ell, \varepsilon$ such that, taking the constant~$D$ large enough,
\begin{equation}
    \left|\phi_{\varepsilon,\ell,m}(i) \right| \le \frac{\e^{-\ell(2m-1) \delta }}{t} D,
    \qquad
    \left| \phi_{\varepsilon,\ell,m}'(i) \right| \le \frac{\e^{-\ell(2m-1) \delta }}{t^2} D, \qquad
    \text{for}~ 1\le i \le xt,
\end{equation}
and the hypotheses of Lemma~\ref{lem:kusmin-landau_weighted} are satisfied.
We obtain
\begin{equation}
    \left| \sum_{i=1}^{xt} \phi_{\varepsilon,\ell,m}(i) \e^{2 \pi \i \psi_{\varepsilon,\ell}(i)} \right| \leq  \left[  \frac{4(\ell D - \frac{\ell}{D} +2) \sqrt{D}}{\sqrt{\pi \ell}} \sqrt{t} + O \bigl( \log(\ell) \bigr) \right]  (1 + x) \frac{D}{t} \e^{-\ell(2 m -1)\delta}
\end{equation}
and summing over $\varepsilon,\ell, m$ we obtain
\begin{equation}
    \begin{split}
        &\left| \sum_{i=1}^{xt} \log\frac{\vartheta\bigl(t\mathcal{L}\bigl(\frac it\bigr)+\frac{\eta}{2\mathcal{K}(\frac it)}\big|\i\pi \mathcal{K}\bigl(\frac it\bigr)^{-1}\bigr)}{\vartheta\bigl(t\mathcal{L}\bigl(\frac{i-1}t\bigr)\big|\i\pi \mathcal{K}\bigl(\frac {i}t\bigr)^{-1}\bigr)} \right| 
        \\
        &\le \sum_{\varepsilon \in \{\pm 1\} } \sum_{m, \ell \ge 1 } \left[  \frac{4(\ell D - \frac{\ell}{D} +2) \sqrt{D}}{\sqrt{\pi \ell}} \sqrt{t} + O \bigl( \log(\ell) \bigr) \right]  (1 + x) \frac{D}{t} \e^{-\ell(2 m -1)\delta} +  O\left(\frac{1}{t}\right)
        \\
        & = O ( 1/\sqrt{t} ),
    \end{split}
\end{equation}
We can now consider the second sum in \eqref{eq:sum_p0+}.
By the Jacobi triple product~\eqref{eq:triple_product} we have
\begin{equation}
    \begin{split}
        &\sum_{i=1}^s\log\frac{\vartheta\bigl(t\mathcal{L}\bigl(\frac{i-1}t\bigr)\big|\i\pi \mathcal{K}\bigl(\frac {i}t\bigr)^{-1}\bigr)}{\vartheta\bigl(t\mathcal{L}\bigl(\frac{i-1}t\bigr)\big|\i\pi \mathcal{K}\bigl(\frac {i-1}t\bigr)^{-1}\bigr)} 
        \\
        &= \log \prod_{m \ge 1} \frac{\left(1-\exp\left\{ -  \frac{2 m \pi^2}{\mathcal{K}(0)}  \right\} \right)}{ \left(1-\exp\left\{ -  \frac{2 m \pi^2}{\mathcal{K}(\frac{s}{t})} \right\}\right) } 
        + \sum_{i=1}^s \sum_{ \substack{ m \ge 1 \\ \varepsilon \in \{\pm 1 \} }} \log \frac{\left(1+\exp\left\{ -  \frac{(2 m-1) \pi^2}{\mathcal{K}(\frac{i-1}{t})} +\varepsilon 2 \pi \i t \mathcal{L}(\frac{i}{t})  \right\} \right)}{ \left(1+\exp\left\{ -  \frac{(2 m -1)\pi^2}{\mathcal{K}(\frac{i}{t})} +\varepsilon 2 \pi \i t \mathcal{L}(\frac{i}{t})   \right\} \right) }
        \\
        &= \log \prod_{m \ge 1} \frac{\left(1-\exp\left\{ -  \frac{2 m \pi^2}{\mathcal{K}(0)}  \right\} \right)}{ \left(1-\exp\left\{ -  \frac{2 m \pi^2}{\mathcal{K}(\frac{s}{t})} \right\}\right)} + O \left( \frac{1}{\sqrt{t}} \right).
    \end{split}
\end{equation}
The first term is a function independent of $t$, while the asymptotic behavior of second term follows from van der Corput estimates such as those used above, which we will not repeat.

Finally, we apply the Euler--Maclaurin formula~\eqref{eq:EM1} with $\psi(s)=t^{-1}\mathcal Y(s/t,t)$ to obtain
\begin{equation}
\label{eq:s>0_6}
t^{-1}\sum_{i=1}^s\mathcal Y\bigl(\frac it,t\bigr)=\int_0^x\mathcal Y(y,t)\d y+O(t^{-1}).
\end{equation}
Letting $\overline{\mathcal Y}(x)=\mathcal{L}(x)\int_0^{1/\mathcal{L}(x)}\mathcal Y(x,t)\d t$ be the average over the period in $t$, we have, by Lemma~\ref{lemma:periodiclesseasy},
\begin{equation}
\label{eq:s>0_7}
\int_0^x\mathcal Y(y,t)\d y = \int_0^x\overline{\mathcal Y}(y)\d y+O(t^{-1}).
\end{equation}

Combining~\eqref{eq:finals=0},~\eqref{eq:s>0_1},~\eqref{eq:s>0_2},~\eqref{eq:s>0_3},~\eqref{eq:s>0_4},~\eqref{eq:s>0_6}, and~\eqref{eq:s>0_7}, we obtain
\begin{equation}
\log Q(t,s) = -t^2\mathcal{F}(x)+\log\vartheta\bigl(t\mathcal{L}(x)\big|\i\pi \mathcal{K}(x)^{-1}\bigr)+\log\mathcal C(x)+O\bigl(\frac 1{\sqrt t}\bigr)
\end{equation}
where
\begin{equation}
\log \mathcal C(x)=\log\mathcal C(0)+\frac 16\bigl(\mathcal{F}''(x)-\mathcal{F}''(0)\bigr)+\log \prod_{m \ge 1} \frac{\left(1-\exp\left\{ -  \frac{2 m \pi^2}{\mathcal{K}(0)}  \right\} \right)}{ \left(1-\exp\left\{ -  \frac{2 m \pi^2}{\mathcal{K}(x)} \right\}\right)}+\int_0^x\overline{\mathcal Y}(y)\d y
\end{equation}
is a constant depending on $x$ only.
The proof for $s\in\mathbb{Z}_{<0}$ follows completely similar lines and so we omit it.

\section{Properties of the rate function: phase transitions and large-\texorpdfstring{$\eta$}{η} limit}
\label{sec:propertiesrateproof}
\subsection{Phase transitions: proof of \Cref{thm:phase transition}}

\begin{proof}[Proof of \Cref{thm:phase transition}, (\ref{item:TW phase transition})]
By the periodicity properties~\eqref{eq:periodicsigma}, \eqref{eq:periodiczeta}, and~\eqref{eq:periodicwp} and the Laurent expansions at the origin, see~\eqref{eq:Taylorsigma}, \eqref{eq:Taylorzeta}, \eqref{eq:Taylorwp}, of the Weierstrass elliptic functions, we produce a Taylor expansion of the rate function $\mathcal{F}(x)$ as $x\uparrow 2$.
Namely, denoting $\varepsilon=2K-\eta$, we have, by~\eqref{eq:f1Weierstrass} and~\eqref{eq:f2},
\begin{equation}
    \begin{aligned}
        \mathcal{U}(K) &=\varepsilon \exp\left\{ -\frac{\eta}{4K}\varepsilon - \frac{\zeta(K)}{K} \varepsilon^2  \right\} + O(\varepsilon^5)
        \\
        &= \varepsilon \left[ 1- \frac{\eta}{4K}\varepsilon + \left(\frac{\eta^2}{32 K^2}-\frac{\zeta(K)}{2K} \right) \varepsilon^2 + \left( \frac{\eta \zeta(K)}{8 K^2} - \frac{\eta^3}{384 K^3} \right)\varepsilon^3 \right] + O(\varepsilon^5)
    \end{aligned}
\end{equation}
and
\begin{equation}
    \mathcal{V}(K) = -\frac{\eta}{\varepsilon} + \frac{\eta^2}{4K} + \frac{\eta}{K} \zeta(K) \varepsilon + O(\varepsilon^2).
\end{equation}
Therefore, from \eqref{eq:x} we have
\begin{equation}\label{eq:expansion x epsilon}
    x(K) = 2- 2\varepsilon +\frac{8+3 \eta -24 \zeta(\frac{\eta}{2})}{4\eta}\varepsilon^2+O(\varepsilon^3).
\end{equation}
Combining the expansions of $\mathcal{U}(K)$ and $x(K)$ and expanding the $\wp$ function appearing in $\mathcal{F}$, as in formula ~\eqref{eq:explicitFWeierstrass}, we obtain, after some basic algebra
\begin{equation}
    \mathcal{F}(x) = \frac{2}{3}\varepsilon^3 + O(\varepsilon^4),
\end{equation}
which holds for $\varepsilon>0$.
Inverting the expansion \eqref{eq:expansion x epsilon}, we express $\varepsilon$ in terms of $x$ as
\begin{equation}
    \varepsilon = \frac{2-x}{2} + O(2-x)^2,
\end{equation}
and the proof is complete.
\end{proof}

\begin{proof}[Proof that $\mathcal{F}(x)$ is $C^2(\mathbb{R})$ and of \Cref{thm:phase transition}, (\ref{item:BOAC phase transition})]
    The function $\mathcal{F}(x)$ is smooth for $x\in\mathbb{R}\setminus\lbrace x_*,2\rbrace$, hence it is $C^2(\mathbb{R}\setminus\lbrace x_*\rbrace)$ (the previous point implies that $\mathcal{F}(x)$ is $C^2$ at $x=2$). Therefore, we only need to show that $\mathcal{F}(x)$ is $C^2$ at $x=x_*$ and that $\mathcal{F}''(x)$ is not H\"older continuous at $x_*$ for any H\"older exponent.
    Combining the homogeneity properties, see~\eqref{eq:homogeneity}, and modular symmetries of Weierstrass elliptic functions, see~\eqref{eq:modular} with $a_{11}=a_{22}=0$ and $a_{12}=-a_{21}=1$, we get
    \begin{equation}
    \label{eq:modular symmetries}
            \sigma(z|K,\i\pi) = \i \pi \sigma \biggl( \frac{z}{\i \pi} \bigg| 1,\i\frac{K}{\pi} \biggr),
            \quad
            \zeta(z|K,\i\pi) = \frac{1}{\i \pi} \zeta \biggl(\frac{z}{\i \pi} \bigg| 1,\i\frac{K}{\pi} \biggr),
            \quad
            \wp(z|K,\i\pi) = \frac{1}{\i \pi} \zeta \biggl(  \frac{z}{\i \pi} \bigg| 1,\i\frac{K}{\pi} \biggr).
    \end{equation}
    Applying \eqref{eq:sigma trig expansion}--\eqref{eq:wp trig expansion} to the right-hand sides of these equalities, noting that the (square of the) elliptic nome of the lattice $2 \mathbb{Z} + 2 \i \frac{K}{\pi}\mathbb{Z}$ is~$\sfq=\e^{-2K}$, see~\eqref{eq:nome}, we get
    \begin{align}
        &\begin{aligned}
            \sigma(z|K,\i\pi) &= 2 \e^{-\frac{z^2}{24}} \sinh\left( \frac{z}{2} \right) \left[ 1 - \e^{-2K} \left( 4 \sinh \left( \frac{z}{2} \right)^2 -z^2 \right) \right.
            \\
            & \qquad + \e^{-4K} \left.  \left( 6+5z^2 +\frac{z^4}{2} -(6+2z^2) \cosh(z) \right) \right] +O(\e^{-6K}),
        \end{aligned}
    \\
        &\begin{aligned}
        \zeta(z|K,\i\pi) &= -\frac{z}{12} + \frac{ 1 }{2} \coth \left(\frac{z}{2}\right) - 2 \e^{-2K} \left[ \sinh (z)  -z \right] 
            \\
            & \qquad - 2 \e^{-4 K} \left[ \sinh (z) + \sinh (2z)  - 3z \right] + O(\e^{-6K}),
        \end{aligned}
    \\
        &\begin{aligned}
            \wp(z|K,\i\pi)&= \frac{1}{12} + \frac{ 1 }{4} \frac{1}{\sinh \left(\frac{z}{2}\right)^2} + 2 \e^{-2K} \left[ \cosh (z)  -1 \right] 
            \\
            & \qquad + 2 \e^{-4K} \left[ \cosh (z) + 2\cosh (2z)  - 3 \right] + O(\e^{-6 K}).
        \end{aligned}
    \end{align}
    Plugging these estimates in~\eqref{eq:f1Weierstrass} and~\eqref{eq:f2} we obtain
    \begin{align}
            \mathcal{U}(K) &= (1-\e^{-\eta}) 
            - \frac{(1-\e^{-\eta})^3}{\e^{-\eta}}\e^{-2K} -3 \frac{(1-\e^{-\eta})^3}{\e^{-\eta}} \e^{-4K}+ O(\e^{-6K}),
    \\
            \mathcal{V}(K) &= 1-K \biggl(\frac{ 2 (1-\e^{-\eta})^2}{\e^{-\eta}} \e^{-2 K} + \frac{2 (1-\e^{-\eta})^2 \left(\e^{-2\eta}+4 \e^{-\eta}+1\right)}{\e^{-2\eta}} \e^{-4 K}  \biggr) +O(K \e^{-6K}),
    \end{align}
    and, as a result,
    \begin{equation}
        x(K) = x_* + (1+2K) \e^{-2K} \frac{2 (1-\e^{-\eta})^3}{\eta\e^{-\eta}} + (1+4 K) \e^{-4 K} \frac{6 (1-\e^{-\eta})^3}{\eta\e^{-\eta}} + O(K \e^{-6K}).
    \end{equation}
    Finally, we can expand the rate function $\mathcal{F}(x)$, as in ~\eqref{eq:explicitFWeierstrass}, around $x_*$, where we have
    \begin{equation}
        \begin{split}
            \mathcal{F}(x(K)) = \mathcal{F}(x_*) - 2(1-\e^{-\eta}) [x(K)-x_*] + \frac{\eta}{2} [x(K)-x_*]^2 + \mathsf R(K),
        \end{split}
    \end{equation}
    and the function $\mathsf R$ is
    \begin{equation}
       \mathsf R(K)=-(1+4K) \e^{-4K} \frac{(1-\e^{-\eta})^6}{\e^{-2\eta}} + O(K^2 \e^{-6K}).
    \end{equation}
    It is immediate to observe that $\mathsf R(K)$ satisfies
    \begin{equation}
         \lim_{K \to +\infty}\frac{\mathsf R(K)}{(x(K)-x_*)^{2}} = 0,
         \qquad
         \lim_{K \to +\infty}\frac{\mathsf R(K)}{(x(K)-x_*)^{2+\varepsilon}} = +\infty,
    \end{equation}
    for all $\varepsilon>0$, which completes the proof.
\end{proof}

\subsection{Large-$\eta$ limit} \label{subs: limit q0}
We start by computing the expansion of the relevant quantities in the limit $\eta\to+\infty$.
\begin{lemma}
As $\eta,K\to+\infty$ with $K>\eta/2$, we have
\begin{equation}
\label{eq:largeetaUV}
\mathcal{U}(K)\sim 1-\e^{\eta-2K},\qquad
\mathcal{V}(K)\sim -2K\e^{\eta-2K}.
\end{equation}
For all $x\in(0,2)$, as $\eta\to+\infty$ we have
\begin{align}
\mathcal{K}(x)&=\frac \eta 2 +\frac 12\log(2/x)+O(\eta^{-1}),\\
\label{eq:largeetaL}
\mathcal{L}(x)&=-x+\eta^{-1}\bigl(x\log(2/x)+x-2\bigr)+O(\eta^{-2}).
\end{align}
\end{lemma}
\begin{proof}
From the expansions \eqref{eq:sigma trig expansion}--\eqref{eq:wp trig expansion}, after using the modular symmetries \eqref{eq:modular symmetries}, we get
\begin{equation}
\frac{\zeta(\i \pi)}{\i \pi} \sim -\frac{1}{12},
\qquad
\sigma(\eta) \sim \e^{-\eta^2/24} \left( 1 - \e^{\eta-2K} \right).
\end{equation}
In view of~\eqref{eq:f1Weierstrass} and~\eqref{eq:f2}, \eqref{eq:largeetaUV} follow.
The expansion of~$\mathcal K(x)$ follows from the defining relation $-\frac{\eta}2x=\mathcal{U}\bigl(\mathcal{K}(x)\bigr)\mathcal{V}\bigl(\mathcal{K}(x)\bigr)$ (see Definition~\ref{def:Fq}) and~\eqref{eq:largeetaUV}. By~\eqref{eq:L}, the expansion of~$\mathcal{L}(x)=-\bigl(\mathcal{U}(\mathcal{K}(x))+\frac{\eta x}{2}\bigr)/\mathcal{K}(x)$ follows.
\end{proof}
By~\eqref{eq:Fq} and~\eqref{eq:largeetaL}, we obtain
\begin{equation}
    \lim_{\eta\to+\infty}\mathcal{F} (x) = 
    \begin{cases}
        + \infty \qquad & \text{if } x\leq 0,
        \\
        \frac{1}{2}x^2\log \frac{2}{x} + \frac{3}{4}x^2 - 2 x + 1 & \text{if } 0<x<2,
        \\
        0 \qquad & \text{if } x\geq 2.
    \end{cases}
\end{equation}
In the right-hand side we recognize the lower-tail large deviation rate function of the length of the longest increasing subsequence in a Poisson random environment computed by Sepp{\"a}l{\"a}inen in~\cite{seppalainen_98_increasing}.

\begin{remark}
\label{rem:A}
     Although the rate function $\mathcal{F}$ converges, in the large-$\eta$ limit, to the appropriate rate function of the $\eta = \infty$ case, as shown above, we do not expect the same to be true for the subleading terms in the large-$t$ expansion of~$\log Q(t,xt)$. Indeed, \Cref{fig:A} suggests that the coefficient $\mathcal{A}$ diverges in the limit $\eta \to + \infty$. It is known that, in the continuous setting, the coefficient of the logarithmic term captures topological properties of the equilibrium measure; see \cite{charlier2025asymptotics} or \cite{byun2025free} for the analysis of a 2D log-gas.
    
    A situation similar to ours occurs in the large-$N$ expansion of the partition function $Z_N(\Delta)$ of the six-vertex model with domain-wall boundary conditions in a $N \times N$ square lattice. In the disordered phase ($|\Delta|<1$), P.~Bleher and V.~Fokin \cite{bleher2006exact} proved that $Z_N(\Delta) \propto N^\kappa e^{N^2 f}$, whereas at the disordered-ferroelectric border line ($|\Delta|=1$), P.~Bleher and K.~Liechty \cite{bleher2009exact_critical_line} proved that $Z_N(\Delta=1) \propto N^{\tilde{\kappa}} e^{\sqrt{N}g + N^2 f}$. While the free energy~$f$ is continuous at the transition point $\Delta = 1$, the same is (obviously) not true for the rest of the large-$N$ expansion of $\log Z(\Delta)$.
    
    This point motivates a separate Riemann--Hilbert analysis in the $\eta = \infty$ case, which we hope to address in a future work.
\end{remark}

\appendix

\section{Elliptic functions}\label{app:elliptic}

In this appendix, we recall the definitions of the Weierstrass $\wp$, $\zeta$, and $\sigma$ functions --- collectively (and with a slight abuse of terminology) referred to as \emph{Weierstrass elliptic functions} --- as well as present some of their well-known properties which are needed in this work.
We also recall the main properties of the elliptic theta functions defined in~\eqref{eq:theta} and their relation to Weierstrass elliptic functions.

\subsection{Weierstrass elliptic functions: general period lattice}
Given complex numbers $\omega_1,\omega_2$ such that $\Im(\omega_2/\omega_1)>0$, the \emph{Weierstrass elliptic functions with half-periods} $\omega_1,\omega_2$ are defined by
\begin{equation}
\label{eq:defWeierstrassgeneral}
\begin{aligned}
\sigma(z|\omega_1,\omega_2)&=z\prod_{\substack{l\in 2\omega_1\mathbb{Z}+2\omega_2\mathbb{Z} \\ l\not=0}}\biggl(\biggl(1-\frac zl\biggr)\exp\biggl(\frac zl+\frac{z^2}{2l^2}\biggr)\biggr),\\
\zeta(z|\omega_1,\omega_2)&=\frac{\sigma'(z|\omega_1,\omega_2)}{\sigma(z|\omega_1,\omega_2)},\quad
\wp(z|\omega_1,\omega_2)=-\zeta'(z|\omega_1,\omega_2),
\end{aligned}
\end{equation}
If the half-periods are fixed or clear from the context, we shall omit them from the notation and simply write $\sigma(z)=\sigma(z|\omega_1,\omega_2)$, $\zeta(z)=\zeta(z|\omega_1,\omega_2)$, and $\wp(z)=\wp(z|\omega_1,\omega_2)$.

For all $\kappa\in\mathbb{C}\setminus\lbrace 0\rbrace$ we have the following homogeneity properties
\begin{equation}
\label{eq:homogeneity}
\begin{split}
    \sigma(z|\omega_1,\omega_2) &= \frac 1\kappa\,\sigma(\kappa z|\kappa \omega_1,\kappa \omega_2),
    \\ 
    \zeta(z|\omega_1,\omega_2) &= \kappa\,  \zeta(\kappa z|\kappa \omega_1,\kappa \omega_2),
    \\ 
    \wp(z|\omega_1,\omega_2) &= \kappa^{2}\,\wp(\kappa z|\kappa \omega_1,\kappa \omega_2).
\end{split}
\end{equation}

If $\wt\omega_1=a_{11}\omega_1+a_{12}\omega_2$ and $\wt\omega_2=a_{21}\omega_1+a_{22}\omega_2$ for some $a_{11},a_{12},a_{21},a_{22}\in\mathbb{Z}$ such that $a_{11}a_{22}-a_{12}a_{21}=1$, then we have the following \emph{modular symmetries}
\begin{equation}
\label{eq:modular}
\begin{split}
    \sigma(z |\omega_1,\omega_2) &= \wt\omega_2 \,\sigma\Big(\frac{z}{\wt\omega_2} \Big|\wt\omega_1,\wt\omega_2\Big), 
    \\
    \zeta(z |\omega_1,\omega_2) &= \wt\omega_2^{-1} \, \zeta\Big(\frac{z}{\wt\omega_2}  \Big|\wt\omega_1,\wt\omega_2\Big), 
    \\
    \wp(z|\omega_1,\omega_2) &= \,\wt\omega_2^{-2}  \wp\Big(\frac{z}{\wt\omega_2}  \Big| \wt\omega_1,\wt\omega_2\Big).
\end{split}
\end{equation}

Next, we assume that $\omega_1$ and $\omega_2$ are fixed.
The Weierstrass $\sigma$ function is entire, odd, and satisfies
\begin{equation}
\label{eq:periodicsigma}
\sigma(z+\omega_j)=-\e^{2\zeta(\omega_j)z}\sigma(z-\omega_j),\quad j=1,2.
\end{equation}
The set of zeros of $\sigma$ coincides with $2\omega_1\mathbb{Z}+2\omega_2\mathbb{Z}$ and 
\begin{equation}
\label{eq:Taylorsigma}
\sigma(z)=z+O(z^ 5),\quad z\to 0.
\end{equation}
The Weierstrass $\zeta$ function is meromorphic, odd, and satisfies
\begin{equation}
\label{eq:periodiczeta}
\zeta(z+2\omega_j)=\zeta(z)+2\zeta(\omega_j),\quad j=1,2.
\end{equation}
The set of poles of $\zeta$ coincides with $2\omega_1\mathbb{Z}+2\omega_2\mathbb{Z}$ and
\begin{equation}
\label{eq:Taylorzeta}
\zeta(z)=z^{-1}+O(z^3),\quad z\to 0.
\end{equation}
Finally, the Weierstrass $\wp$ function is meromorphic, even, and doubly-periodic, namely
\begin{equation}
\label{eq:periodicwp}
\wp(z+2\omega_j)=\wp(z),\quad j=1,2.
\end{equation}
The set of poles of $\wp$ coincides with $2\omega_1\mathbb{Z}+2\omega_2\mathbb{Z}$ and
\begin{equation}
\label{eq:Taylorwp}
\wp(z)=z^{-2}+O(z^2),\quad w\to 0.
\end{equation}
Moreover,
\begin{equation}
\label{zeroswpprimegeneral}
\wp'(z) = 0 \iff z\in\lbrace \omega_1,\omega_2,\omega_1+\omega_2\rbrace+ 2\omega_1\mathbb{Z}+2\omega_2\mathbb{Z}
\end{equation}
and the zeros of $\wp'$ are simple.

The Weierstrass $\wp$ function satisfies the ordinary differential equation
\begin{equation}
\label{eq:ODEwp1storder}
\wp'(z)^2=4\wp(z)^3-g_2\wp (z)-g_3 = (\wp (z)-e_1)(\wp (z)-e_2)(\wp (z)-e_3)
\end{equation}
where
\begin{equation}
\label{eq:g2g3Weierstrass}
g_2 = 60\sum_{\substack{l\in 2\omega_1\mathbb{Z}+2\omega_2\mathbb{Z} \\ l\not=0}}l^{-4},\qquad
g_3 = 140\sum_{\substack{l\in 2\omega_1\mathbb{Z}+2\omega_2\mathbb{Z} \\ l\not=0}}l^{-6},
\end{equation}
and
\begin{equation}
\label{eq:e1e2e3}
e_1=\wp(\omega_1),\quad
e_2=\wp(\omega_1+\omega_2),\quad
e_3=\wp(\omega_2).
\end{equation}
It follows from~\eqref{eq:ODEwp1storder} that
\begin{equation}
\label{eq:ODEwp2ndorder}
\wp''(z)=6\wp(z)^2-\frac 12 g_2.
\end{equation}

The following addition formulas hold true:
\begin{align}
\label{eq:additionzeta}
\zeta(z_1+z_2) &= \zeta(z_1)+\zeta(z_2)+\frac{1}{2}\frac{\wp'(z_1)-\wp'(z_2)}{\wp(z_1)-\wp(z_2)}\,,
\\
\label{eq:additionwp}
\wp(z_1+z_2) &= - \wp(z_1) - \wp(z_2) + \frac{1}{4} \left( \frac{\wp'(z_1) - \wp'(z_2)}{\wp(z_1) - \wp(z_2)} \right)^2\,.
\end{align}
Since $\zeta$ and $\wp'$ are odd and $\wp$ is even, we obtain
\begin{align}
    \label{eq:zetaprostaphaeresis}
    \zeta(z_1+z_2) + \zeta(z_1-z_2) - 2\zeta(z_1) &= \frac{\wp'(z_1)}{\wp(z_1) - \wp(z_2)}\,, 
    \\
    \label{eq:ellipticprostaphaeresis}
    \wp(z_1+z_2) -\wp(z_1-z_2) &= \frac{\wp'(z_1) \wp'(z_2)}{\bigl(\wp(z_1) - \wp(z_2)\bigr)^2}.
\end{align}

We report the following trigonometric expansions, see~\cite[Eq.s 23.8.1, 23.8.2, 23.8.6]{DLMF}:
\begin{align} 
    \label{eq:sigma trig expansion}
 \sigma(z) &= \frac{2 \omega_1}{\pi} \exp \left( \frac{\zeta(\omega_1) z^2}{2\omega_1} \right) \sin \left( \frac{\pi z}{2 \omega_1} \right) \prod_{n \ge 1} \frac{1 - 2 \sfq^{n} \cos (\pi z/\omega_1) + \sfq^{2n}}{(1-\sfq^{n})^2},
\\ \label{eq:zeta trig expansion}
	 \zeta(z) &= z \frac{\zeta(\omega_1)}{\omega_1}  + \frac{\pi}{2 \omega_1} \cot\left(  \frac{\pi z}{2 \omega_1} \right) + \frac{2 \pi }{\omega_1} \sum_{n \ge 1} \frac{\sfq^n }{1-\sfq^n} \sin \left( \frac{n \pi z}{\omega_1} \right),
\\ \label{eq:wp trig expansion}
	 \wp(z) &= -\frac{\zeta(\omega_1)}{\omega_1}  + \frac{\pi^2}{4 \omega_1^2} \sin \left(  \frac{\pi z}{2 \omega_1} \right)^{-2} - \frac{2 \pi^2 }{\omega_1^2} \sum_{n \ge 1} n \frac{\sfq^n }{1-\sfq^n} \cos \left( \frac{n \pi z}{\omega_1} \right),
\end{align}
where $\sfq$ is the (square of the) \emph{elliptic nome} of the lattice $2\omega_1\mathbb{Z}+2\omega_2\mathbb{Z}$, namely,
\begin{equation}
\label{eq:nome}
\sfq=\exp\bigl( 2 \pi \i \omega_2/\omega_1\bigr).
\end{equation}
Moreover, we have the expansion
\begin{equation}
	\zeta(\omega_1) = \frac{\pi^2}{12 \omega_1} - \frac{\pi^2}{2 \omega_1} \sum_{n \ge 1} \sinh \left( \frac{n \pi^2}{\omega_1} \right)^{-2}.
\end{equation}

\subsection{Weierstrass elliptic functions: rectangular period lattice}

We now restrict attention to the case $\omega_{1} = K > 0$ and $\omega_{2} = \i\pi$ needed in this paper.
First, we note the \emph{Legendre identity} (which is valid for general half-periods):
\begin{equation}
\label{eq:LegendreIdentity}
\omega_2\zeta(\omega_1)-\omega_1\zeta(\omega_2)=\frac{\i\pi}2
\ \Rightarrow\
\i\pi\zeta(K)-K\zeta(\i\pi)=\frac{\i\pi}2.
\end{equation}
Next, specializing~\eqref{eq:defWeierstrassgeneral}, we have
\begin{equation}
\label{eq:defWeierstrassfunctionsmn}
\begin{aligned}
\sigma(z)&=
z\prod_{(m,n)\in\mathbb{Z}^2\setminus\lbrace(0,0)\rbrace}\biggl[\biggl(1-\frac z{2(nK+m\i\pi)}\biggr)\exp\biggl(\frac z{2(nK+m\i\pi)}+\frac{z^2}{8(nK+m\i\pi)^2}\biggr)\biggr],
\\
\zeta(z) &= \frac 1{z}+\sum_{(m,n)\in\mathbb{Z}^2\setminus\lbrace(0,0)\rbrace}\biggl(\frac{1}{z-2(nK+m\i\pi)}+\frac{1}{2(nK+m\i\pi)}+\frac{z}{4(nK+m\i\pi)^2}\biggr),
\\
\wp(z) &= \frac 1{z^2}+\sum_{(m,n)\in\mathbb{Z}^2\setminus\lbrace(0,0)\rbrace }\biggl(\frac{1}{\bigl(z-2(nK+m\i\pi)\bigr)^2}-\frac{1}{4(nK+m\i\pi)^2}\biggr).
\end{aligned}
\end{equation}
Summing over $m$ and performing elementary algebraic manipulations we obtain:
\begin{align}
    &\begin{aligned} \label{eq:wZ}
        \zeta(z) = &\frac{ -z + 6 \coth(\frac{z}{2}) }{12} 
        \\
        &+ \frac{1}{2} \sum_{n=1}^{\infty} \left\{ \frac{z}{\sinh(K n)^2} - \coth\left(K n- \frac{z}{2} \right) + \coth\left(K n + \frac{z}{2} \right) \right\},
    \end{aligned}
    \\
    &\begin{aligned}
        \label{eq:wpK}
	\wp(z) = \frac{5+ \cosh(z)}{24 \sinh( \frac{z}{2} )^2} - \frac{1}{4} \sum_{n=1}^{\infty} \left\{ \frac{2}{\sinh(K n)^2} - \frac{1}{\sinh(K n-\frac{z}{2})^2} - \frac{1}{\sinh(K n+\frac{z}{2})^2} \right\},
    \end{aligned}
    \\
    & \label{eq:zetaipi}
    \frac{\zeta(\i\pi)}{\i\pi}=- \frac{1}{12} + \frac{1}{2} \sum_{n \ge 1}  \sinh(n K)^{-2},
    \\
    &  
    \begin{aligned} \label{eq:wpKipi}
        \wp(K+\i\pi)&= \frac{\cosh(K)-5}{12\bigl(\cosh(K)+1\bigr)}\\
        &\quad-\frac 12\sum_{n\geq 1}\biggl[\frac{1}{\cosh \bigl( K(2n+1)\bigr)+1}+\frac{1}{\cosh \bigl( K(-2n+1)\bigr)+1}+\frac{1}{\sinh(K n)^2} \biggr] 
    \end{aligned}
\end{align}

We report some properties of $\wp$ and $\wp'$ in the following lemma.

\begin{lemma}\label{lemma:ineqwprectlattice}
Let $\wp(z)=\wp(z|K,\i\pi)$.
\begin{enumerate}[leftmargin=*]
\item The functions $\wp(\i u),\wp(u+\i\pi), \wp(K+ \i u)$ and $\wp'(u +\i\pi)$ are real-valued for $u\in\mathbb{R}$, while $\wp'(\i u),\wp'(K+ \i u)$ are purely imaginary for $u\in\mathbb{R}$.
\item The zeros of $\wp'$ are simple and are located exactly at $\lbrace K,\i\pi,K+\i\pi\rbrace+2K\mathbb{Z}+2\pi\i\mathbb{Z}$.
\item We have $\wp'(u)<0$ and $\wp'(u+\i\pi)>0$ for $u\in (0,K)$, while $\frac{\wp'(K+\i u)}{\i}>0$ and $\frac{\wp'(\i u)}{\i}<0$ for $u\in (0,\pi)$.
\end{enumerate}
\end{lemma}
\begin{proof}
The first claim follows, for example, from~\eqref{eq:wpK}.
The second one follows from~\eqref{zeroswpprimegeneral}.
Finally, by~\eqref{eq:Taylorwp} we have $\wp(u)= u^{-2}+O(1) $ as $u \to 0$, and so $\wp'(u)<0$ for $u\in(0,K)$; on the other hand, it is known that $\wp (K)>\wp (K+\i\pi)>\wp (\i\pi)$ (see~\cite[Eq.~23.5.1]{DLMF}) and so $\wp'(u+\i\pi)>0$ for $u\in (0,K)$ and $\frac{\wp'(K+\i u)}{\i}>0$ for $u\in (0,\pi)$, and the third claim is also proved.
\end{proof}

Finally, one can express derivatives in $K$ of Weierstrass elliptic functions as follows.

\begin{lemma}
\label{lemma:WeierstrassDerivativesinK}
    Let $\sigma(z)=\sigma(z|K,\i\pi)$, $\sigma(z)=\zeta(z|K,\i\pi)$, $\wp(z)=\wp(z|K,\i\pi)$.
    We have
    \begin{equation}
    \label{eq:partialKsigmazetawp}
    \begin{aligned}
        \partial_K\log\sigma(z) &= \wp(z)-\zeta(z)^2-\frac 1{12}g_2z^2+2\frac{\zeta(\i\pi)}{\i\pi}\left(z\zeta(z)-1\right),
        \\
        \partial_K\zeta(z) &= \wp'(z)+2\zeta(z)\wp(z)-\frac 1{6}g_2z+2\frac{\zeta(\i\pi)}{\i\pi}\left(\zeta(z)-z\wp(z)\right),
        \\
        \partial_K\wp(z) &= -\wp''(z)+2\wp(z)^2-2\zeta(z)\wp'(z)+\frac 1{6}g_2+2\frac{\zeta(\i\pi)}{\i\pi}\left(2\wp(z)+z\wp'(z)\right)),
    \end{aligned}
    \end{equation}
    where $g_2$ is defined in~\eqref{eq:g2g3Weierstrass}.
\end{lemma}

In particular, since $\wp'(\i\pi)=0$ (see Lemma~\ref{lemma:ineqwprectlattice}),
\begin{equation}\label{eq: partial K zeta i pi}
\partial_K\zeta(\i\pi) = 2\frac{\zeta(\i\pi)^2}{\i\pi}-\frac{\i\pi}{6}g_2.
\end{equation}

\subsection{Elliptic theta functions}\label{sec:apptheta}

Recall the theta functions defined in~\eqref{eq:theta}, which are all entire functions of~$z$.
We have the following quasi-periodicity relations, for $a,b\in\lbrace 0,1\rbrace$:
\begin{equation}
\label{eq:periodictheta}
\vartheta_{ab}(z+1|\tau) = (-1)^b\vartheta_{ab}(z|\tau), \quad
\vartheta_{ab}(z+\tau|\tau) = (-1)^a\e^{-\pi \i \tau - 2\pi \i z}\vartheta_{ab}(z|\tau),
\end{equation}
where $\vartheta_{00}=\vartheta$.
The functions $\vartheta$, $\vartheta_{10}$, and $\vartheta_{01}$ are even, while $\vartheta_{11}$ is odd.
In particular,
\begin{equation}
\label{eq:zerothetaprime}
\vartheta'(0|\tau) = 0. 
\end{equation}
The set of zeros (in $z$) of $\vartheta_{11}(z|\tau)$ coincides with $\mathbb{Z}+\tau\mathbb{Z}$.

We note the relation to Weierstrass elliptic functions:
if $\tau=\frac{\omega_2}{\omega_1}$ and $\Im \tau>0$, 
\begin{align}
\label{eq:relsigmatheta}
\sigma(z|\omega_1,\omega_2)&=2\omega_1\exp\biggl(\frac{\zeta(\omega_1|\omega_1,\omega_2)}{2\omega_1}z^2\biggr)\frac{\vartheta_{11}(\frac z{2\omega_1}|\tau)}{\vartheta_{11}'(0|\tau)},
\\
\label{eq:relzetatheta}
\zeta(z|\omega_1,\omega_2)&=\frac{\zeta(\omega_1|\omega_1,\omega_2)}{\omega_1}z+\frac{1}{2\omega_1}\frac{\vartheta'_{11} \bigl(\frac z{2\omega_1}\big|\tau\bigr) }{\vartheta_{11} \bigl(\frac z{2\omega_1}\big|\tau\bigr) },
\\
\label{eq:relwptheta}
\wp(z|\omega_1,\omega_2)&=-\frac{\zeta(\omega_1|\omega_1,\omega_2)}{\omega_1}-\frac{\d^2}{\d z^2}\log\vartheta_{11}\bigl(\frac z{2\omega_1}\big|\tau\bigr).
\end{align}

Finally, we recall the \textit{Jacobi triple product identity}
\begin{equation}
\label{eq:triple_product}
    \vartheta(z|\tau) = \prod_{m \ge 1 } \left( 1 - \e^{2 m \pi \i \tau} \right) \left( 1 + \e^{(2m-1) \pi \i \tau + 2 \pi \i z } \right) \left( 1 + \e^{(2m-1) \pi \i \tau - 2 \pi \i z } \right).
\end{equation}

\section{Airy model Riemann--Hilbert problem}\label{app:Airy}

In this section, we recall the well-known Airy model Riemann--Hilbert problem and its solution, which is the key ingredient in the construction of the inner parametrices. For our purposes, it is convenient to present it in three (equivalent) formulations, labeled $\mathrm{I}$, $\mathrm{II}$, and $\mathrm{III}$.

Let $\Sigma_{\mathrm{Ai},\mathrm{I}}=\Sigma_{\mathrm{Ai},\mathrm{III}}=\mathbb{R}\cup\e^{\frac 23\pi\i}\mathbb{R}_{+}\cup\e^{-\frac 23\pi\i}\mathbb{R}_{+}$, $\Sigma_{\mathrm{Ai},\mathrm{II}}=\mathbb{R}\cup\e^{\frac 13\pi\i}\mathbb{R}_{+}\cup\e^{-\frac 13\pi\i}\mathbb{R}_{+}$. We also let $\Sigma_{\mathrm{Ai},X}^\circ=\Sigma_{\mathrm{Ai},X}\setminus\lbrace 0\rbrace$ for $X=\mathrm{I},\mathrm{II},\mathrm{III}$.
We orient these contours (see~Figure~\ref{fig:AirymodelRHP}) so that $\mathbb{R}$ points to the right and the diagonal lines point upwards.
This defines $\pm$-sides of $\Sigma_{\mathrm{Ai},X}^\circ$ (as usual, $+$ to the left, $-$ to the right) for $X=\mathrm{I},\mathrm{II},\mathrm{III}$.

Letting $\mathrm{Ai}$ be the Airy function and $\mathcal{A}_j(\zeta) = \sqrt{2\pi}\e^{\frac {2j}3\pi\i}\mathrm{Ai}\bigl(\e^{\frac {2j}3\pi\i}\zeta\bigr)$ (for $j=0,1,2$), we define
\begin{equation}
\label{eq:defPhiAi1}
\boldsymbol \Phi^{\mathrm{Ai},\mathrm{I}}(\zeta) = 
\begin{cases}
\begin{pmatrix}
-\mathcal{A}_2(\zeta) & \mathcal{A}_1(\zeta) \\ \i \mathcal{A}_2'(\zeta) & -\i\mathcal{A}_1'(\zeta)
\end{pmatrix}
\e^{\frac 23\zeta^{\frac 32}\boldsymbol\sigma_3},
& \text{if }-\pi<\arg\zeta<-\frac 23\pi,\\
\begin{pmatrix}
\mathcal{A}_0(\zeta) & \mathcal{A}_1(\zeta) \\ -\i \mathcal{A}_0'(\zeta) & -\i\mathcal{A}_1'(\zeta)
\end{pmatrix}
\e^{\frac 23\zeta^{\frac 32}\boldsymbol\sigma_3},
& \text{if }-\frac 23\pi<\arg\zeta<0,\\
\begin{pmatrix}
-\mathcal{A}_1(\zeta) & -\mathcal{A}_2(\zeta) \\ \i \mathcal{A}_1'(\zeta) & \i\mathcal{A}_2'(\zeta)
\end{pmatrix}
\e^{\frac 23\zeta^{\frac 32}\boldsymbol\sigma_3},
& \text{if }0<\arg\zeta<\frac 23\pi,\\
\begin{pmatrix}
\mathcal{A}_0(\zeta) & -\mathcal{A}_2(\zeta) \\ -\i \mathcal{A}_0'(\zeta) & \i\mathcal{A}_2'(\zeta)
\end{pmatrix}
\e^{\frac 23\zeta^{\frac 32}\boldsymbol\sigma_3},
& \text{if }\frac 23\pi<\arg\zeta<\pi,
\end{cases}
\end{equation}
and
\begin{equation}
\label{eq:defPhiAi2}
\boldsymbol \Phi^{\mathrm{Ai},\mathrm{II}}(\zeta)=
\begin{pmatrix}
    0&1\\1&0
\end{pmatrix}
\boldsymbol\Phi^{\mathrm{Ai},\mathrm{I}}(-\zeta)
\begin{pmatrix}
    0&1\\1&0
\end{pmatrix},\quad
\boldsymbol\Phi^{\mathrm{Ai},\mathrm{III}}(\zeta)=
\boldsymbol\Phi^{\mathrm{Ai},\mathrm{I}}(\zeta)^{-\top}.
\end{equation}

Let $\boldsymbol J_{\mathrm{Ai},\mathrm{I}}:\Sigma_{\mathrm{Ai},\mathrm{I}}^\circ\to \mathrm{SL}(2,\mathbb{C})$ be defined as follows:
\begin{equation}
\boldsymbol J_{\mathrm{Ai},\mathrm{I}}(\zeta) = \begin{cases}
\begin{pmatrix}
0 & 1 \\ -1 & 0 
\end{pmatrix},&\text{if }\zeta<0,\\
\begin{pmatrix}
1 & \e^{-\frac 43 \zeta^{3/2}} \\ 0 & 1 
\end{pmatrix},&\text{if }\zeta>0,\\
\begin{pmatrix}
1 &  0\\ \mp\e^{\frac 43 \zeta^{3/2}} & 1
\end{pmatrix},&\text{if }\zeta\in\e^{\pm\frac 23\pi\i}\mathbb{R}_+.
\end{cases}
\end{equation}
Similarly, let $\boldsymbol J_{\mathrm{Ai},\mathrm{II}}:\Sigma_{\mathrm{Ai},\mathrm{I}}^\circ\to \mathrm{SL}(2,\mathbb{C})$ be 
\begin{equation}
\boldsymbol J_{\mathrm{Ai},\mathrm{II}}(\zeta) = \begin{pmatrix}
    0&1\\1&0
\end{pmatrix}
\boldsymbol J_{\mathrm{Ai},\mathrm{I}}(-\zeta)^{-1}\begin{pmatrix}
    0&1\\1&0
\end{pmatrix} = \begin{cases}
\begin{pmatrix}
0 & 1 \\ -1 & 0 
\end{pmatrix},&\text{if } \zeta>0,\\
\begin{pmatrix}
1 & 0 \\ -\e^{-\frac 43 (-\zeta)^{3/2}}  & 1 
\end{pmatrix},&\text{if }\zeta<0,\\
\begin{pmatrix}
1 & \mp\e^{\frac 43 (-\zeta^{3/2})} \\ 0 & 1
\end{pmatrix},&\text{if }\zeta\in\e^{\pm\frac 13\pi\i}\mathbb{R}_+,
\end{cases}
\end{equation}
and let $\boldsymbol J_{\mathrm{Ai},\mathrm{III}}:\Sigma_{\mathrm{Ai},\mathrm{III}}^\circ\to\mathrm{SL}(2,\mathbb{C})$ be
\begin{equation}
\boldsymbol J_{\mathrm{Ai},\mathrm{III}}(\zeta) = \boldsymbol J_{\mathrm{Ai},\mathrm{I}}(\zeta)^{-\top} = \begin{cases}
\begin{pmatrix}
0 & 1 \\ -1 & 0 
\end{pmatrix},&\text{if }\zeta<0,\\
\begin{pmatrix}
1 & 0 \\ -\e^{-\frac 43 \zeta^{3/2}} & 1 
\end{pmatrix},&\text{if }\zeta>0,\\
\begin{pmatrix}
1 & \pm\e^{\frac 43 \zeta^{3/2}} \\ 0 & 1
\end{pmatrix},&\text{if }\zeta\in\e^{\pm\frac 23\pi\i}\mathbb{R}_+.
\end{cases}
\end{equation}
Here and below, we take the branch of $\zeta^{3/2}$ analytic for $\zeta\in\mathbb{C}\setminus\mathbb{R}_{\leq 0}$ and positive for $\zeta>0$, and the branch of $(-\zeta)^{3/2}$ analytic for $\zeta\in\mathbb{C}\setminus\mathbb{R}_{\geq 0}$ and positive for $\zeta<0$.
This is illustrated in Figure~\ref{fig:AirymodelRHP}.

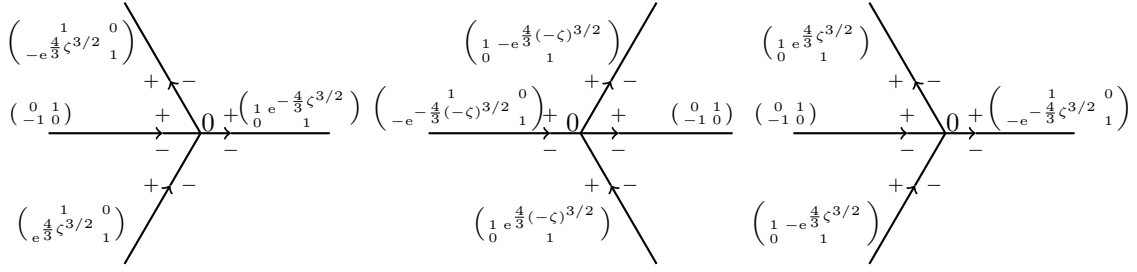
\begin{figure}[t]
\centering
\begin{tikzpicture}

\begin{scope}[shift={(-5,0)}]
    \draw[->,thick] (-2,0) -- (-.5,0) node[below]{\tiny{$-$}} node[above]{\tiny{$+$}};
    \draw[thick] (-.5,0) -- (0,0);
    \draw[->,thick] (0,0) -- (.4,0) node[below]{\tiny{$-$}} node[above]{\tiny{$+$}};
    \draw[thick] (.4,0) -- (1.7,0);

    \begin{scope}[rotate=-30]
    \draw[->,thick] (0,-2) -- (0,-.8) node[right]{\tiny{$-$}} node[left]{\tiny{$+$}};
    \draw[thick] (0,-.8) -- (0,0);
    \end{scope}
    \begin{scope}[rotate=30]
    \draw[->,thick] (0,0) -- (0,.8) node[right]{\tiny{$-$}} node[left]{\tiny{$+$}};
    \draw[thick] (0,.8) -- (0,2);
    \end{scope}

    \node at (1/10,1/7) {\small{$0$}};

    \node at (-2.1,1/4) {\tiny{$\left(\begin{smallmatrix}0 & 1\\ -1 & 0 \end{smallmatrix}\right)$}};
    \node at (1.3,1/3) {\tiny{$\left(\begin{smallmatrix}1 & \e^{-\frac 43\zeta^{3/2}}\\ 0 & 1 \end{smallmatrix}\right)$}};
    \node at (-1.7,1.2) {\tiny{$\left(\begin{smallmatrix}1 & 0\\ -\e^{\frac 43\zeta^{3/2}} & 1 \end{smallmatrix}\right)$}};
    \node at (-1.7,-1.2) {\tiny{$\left(\begin{smallmatrix} 1 & 0\\ \e^{\frac 43\zeta^{3/2}} & 1 \end{smallmatrix}\right)$}};
\end{scope}

\draw[->,thick] (-2,0) -- (-.4,0) node[below]{\tiny{$-$}} node[above]{\tiny{$+$}};
\draw[thick] (-.4,0) -- (0,0);
\draw[->,thick] (0,0) -- (.5,0) node[below]{\tiny{$-$}} node[above]{\tiny{$+$}};
\draw[thick] (.5,0) -- (2,0);

\begin{scope}[rotate=30]
\draw[->,thick] (0,-2) -- (0,-.8) node[right]{\tiny{$-$}} node[left]{\tiny{$+$}};
\draw[thick] (0,-.8) -- (0,0);
\end{scope}
\begin{scope}[rotate=-30]
\draw[->,thick] (0,0) -- (0,.8) node[right]{\tiny{$-$}} node[left]{\tiny{$+$}};
\draw[thick] (0,.8) -- (0,2);
\end{scope}

\node at (-1/10,1/7) {\small{$0$}};

\node at (-1.6,1/3) {\tiny{$\left(\begin{smallmatrix}1 & 0 \\ -\e^{-\frac 43(-\zeta)^{3/2}} & 1 \end{smallmatrix}\right)$}};
\node at (1.6,1/4) {\tiny{$\left(\begin{smallmatrix}0 & 1\\ -1 & 0 \end{smallmatrix}\right)$}};
\node at (-.5,1.2) {\tiny{$\left(\begin{smallmatrix}1 & -\e^{\frac 43(-\zeta)^{3/2}} \\ 0 & 1 \end{smallmatrix}\right)$}};
\node at (-.5,-1.2) {\tiny{$\left(\begin{smallmatrix} 1 & \e^{\frac 43(-\zeta)^{3/2}} \\ 0 & 1 \end{smallmatrix}\right)$}};

\begin{scope}[shift={(4.8,0)}]
    \draw[->,thick] (-2,0) -- (-.5,0) node[below]{\tiny{$-$}} node[above]{\tiny{$+$}};
    \draw[thick] (-.5,0) -- (0,0);
    \draw[->,thick] (0,0) -- (.4,0) node[below]{\tiny{$-$}} node[above]{\tiny{$+$}};
    \draw[thick] (.4,0) -- (1.7,0);

    \begin{scope}[rotate=-30]
    \draw[->,thick] (0,-2) -- (0,-.8) node[right]{\tiny{$-$}} node[left]{\tiny{$+$}};
    \draw[thick] (0,-.8) -- (0,0);
    \end{scope}
    \begin{scope}[rotate=30]
    \draw[->,thick] (0,0) -- (0,.8) node[right]{\tiny{$-$}} node[left]{\tiny{$+$}};
    \draw[thick] (0,.8) -- (0,2);
    \end{scope}

    \node at (1/10,1/7) {\small{$0$}};

    \node at (-2.1,1/4) {\tiny{$\left(\begin{smallmatrix}0 & 1\\ -1 & 0 \end{smallmatrix}\right)$}};
    \node at (1.5,1/3) {\tiny{$\left(\begin{smallmatrix}1 & 0 \\ -\e^{-\frac 43\zeta^{3/2}} & 1 \end{smallmatrix}\right)$}};
    \node at (-1.7,1.2) {\tiny{$\left(\begin{smallmatrix}1 & \e^{\frac 43\zeta^{3/2}}\\ 0 & 1 \end{smallmatrix}\right)$}};
    \node at (-1.7,-1.2) {\tiny{$\left(\begin{smallmatrix} 1 & -\e^{\frac 43\zeta^{3/2}}\\ 0 & 1 \end{smallmatrix}\right)$}};
\end{scope}

\end{tikzpicture}
\caption{Jumps of $\boldsymbol\Phi^{\mathrm{Ai},\mathrm{I}}$ (left), $\boldsymbol\Phi^{\mathrm{Ai},\mathrm{II}}$ (center), and $\boldsymbol\Phi^{\mathrm{Ai},\mathrm{III}}$ (right).}
\label{fig:AirymodelRHP}
\end{figure}

It is well-known that $\boldsymbol\Phi^{\mathrm{Ai},X}$ solves the following Riemann--Hilbert problem, for each $X=\mathrm{I},\mathrm{II},\mathrm{III}$.

\begin{cRHp}[Airy model Riemann--Hilbert problem]
\label{cRHp:Airy}
Let $X=\mathrm{I},\mathrm{II},\mathrm{III}$.
Find an analytic function $\Phi^{\mathrm{Ai},X}:\mathbb{C}\setminus \Sigma_{\mathrm{Ai},X}\to\mathrm{SL}(2,\mathbb{C})$ such that the following conditions hold true.
\begin{enumerate}[leftmargin=*]
\item Non-tangential boundary values of $\boldsymbol \Phi^{\mathrm{Ai},X}$ exist and are continuous on $\Sigma_{\mathrm{Ai},X}^\circ$ and satisfy
\begin{equation}
\boldsymbol \Phi^{\mathrm{Ai},X}_+(\zeta)=\boldsymbol\Phi^{\mathrm{Ai},X}_-(\zeta)\boldsymbol J_{\mathrm{Ai},X}(\zeta),\qquad \zeta\in \Sigma_{\mathrm{Ai},X}^\circ.
\end{equation}
\item As $\zeta\to\infty$ uniformly in $\mathbb{C}\setminus \Sigma_{\mathrm{Ai},X}$ we have
\begin{equation}
\label{eq:asympPhiAi}
\boldsymbol\Phi^{\mathrm{Ai},X}(\zeta)=
\begin{cases}
\zeta^{-\frac 14\boldsymbol\sigma_3}\boldsymbol G
\left(\boldsymbol {\mathrm{I}}+\frac{1}{48}\left(\begin{smallmatrix}
    1 & 6\i \\ 6\i & -1
\end{smallmatrix}\right)\zeta^{-3/2}+O(\zeta^{-3})\right)
,&\text{if }X=\mathrm{I},
\\
(-\zeta)^{\frac 14\boldsymbol\sigma_3}\boldsymbol G
\left(\boldsymbol {\mathrm{I}}+\frac{1}{48}\left(\begin{smallmatrix}
    -1 & 6\i \\ 6\i & 1
\end{smallmatrix}\right)(-\zeta)^{-3/2}+O(\zeta^{-3})\right)
,&\text{if }X=\mathrm{II},
\\
\zeta^{\frac 14\boldsymbol\sigma_3}\boldsymbol G^{-1}
\left(\boldsymbol {\mathrm{I}}-\frac{1}{48}\left(\begin{smallmatrix}
    1 & 6\i \\ 6\i & -1
\end{smallmatrix}\right)\zeta^{-3/2}+O(\zeta^{-3})\right)
,&\text{if }X=\mathrm{III},
\end{cases}
\end{equation}
where
\begin{equation}
\label{eq:Gmatrix}
\boldsymbol G=\frac {1}{\sqrt 2}\begin{pmatrix}1 & \i \\ \i & 1\end{pmatrix}.
\end{equation}
\item We have $\boldsymbol\Phi^{\mathrm{Ai},X}(\zeta)=O(1)$ as $\zeta\to 0$ uniformly in $\mathbb{C}\setminus\Sigma_{\mathrm{Ai},X}$.
\end{enumerate}
\end{cRHp}

\section{Monotonicity of Weierstrass elliptic functions associated with rectangular lattices as functions of the modulus}\label{app:monotonic}

In this appendix we establish monotonicity properties in~$K$ of certain non-linear combinations of the Weierstrass elliptic functions with half-periods $K>0$ and $\i\pi$.

\begin{proposition}\label{prop:f1monotone}
    The function $K\mapsto \mathcal{U}(K)$, see~\eqref{eq:f1} and \eqref{eq:f1Weierstrass}, is strictly increasing for $K\geq \frac 12\eta$.
\end{proposition}
\begin{proof}
    From the expression \eqref{eq:f1Weierstrass}, it suffices to show that both $K \mapsto -\frac{ \zeta(\i \pi) }{\i \pi}$ and  $K \mapsto \log \sigma(\eta)$ are strictly increasing for $K\geq \frac 12\eta$. The first statement follows directly from~\eqref{eq:zetaipi}, because each of the terms $-\sinh(n K)^{-2}$ is strictly increasing for $K>0$ and $n\geq 1$, and for the second one we reason as follows.
    From the definition of the Weierstrass $\sigma$ function, see~\eqref{eq:defWeierstrassfunctionsmn}, we have, after some basic algebra,
    \begin{equation}\label{eq: series derivative log sigma}
        \begin{aligned}
            \partial_K \log \sigma (\eta) &= \sum_{ (m,n) \in \mathbb{Z}^2\setminus(0,0) } \frac{m \eta^3}{4(K m+ \i n \pi)^3(2Km + 2 \i n \pi -\eta)}
            \\
            & = \sum_{m \in \mathbb{Z} \setminus \{ 0 \} } m \left[\coth \left(K m-\frac{\eta}{2}\right)-\coth (K m)\right]-\frac{1}{4} m \eta \frac{(\eta \coth (K m)+2)}{\sinh(K m)^2},
        \end{aligned}
    \end{equation}
    where in the second equality we evaluated the summation over $n $ using the Mittag-Leffler expansion
    \begin{equation}
            \sum_{n \in \mathbb{Z}} \frac{1}{(a+n \i \pi )^3(2a + 2 n \i \pi -\eta)} 
            = \frac{1}{\eta^3} \biggl[4 \coth \left(a-\frac{\eta}{2}\right) -4 \coth (a) - \frac{\eta (\eta \coth (a)+2)}{\sinh(a)^2} \biggr].
    \end{equation}
    This proves that
    \begin{equation}
        \begin{split}
            \partial_K \log \sigma(\eta)
            &=\sum_{m \ge 1} m \left[\coth \left(K m+\frac{\eta}{2}\right)  +\coth \left(K m-\frac{\eta}{2}\right) - \frac{\eta^2 \cosh (K m)}{2 \sinh(K m)^3 } -2 \coth (K m) \right]
            \\
            &=\sum_{m \ge 1} 2 m \frac{\cosh (K m)}{\sinh(K m)^3 } \left[ \frac{\sinh(\frac{\eta}{2})^2 \sinh(K m)^2}{\sinh(Km + \frac{\eta}{2}) \sinh(Km - \frac{\eta}{2})} - \frac{\eta^2}{4} \right],
        \end{split}
    \end{equation}
    where in the first equality we grouped together the $m$th and $(-m)$th terms in the earlier summation and in the second equality we manipulated the summand using standard identities of hyperbolic functions. Each of the terms in the square brackets in the right-hand side is strictly positive for $K>\eta/2$.
    To see this, denoting $X=Km$ and $Y=z/2$, observe that
    \begin{equation}
        \partial_X\left[ \frac{\sinh(Y)^2 \sinh(X)^2}{\sinh(X +Y) \sinh(X - Y)} - Y^2 \right] = -\frac{\sinh(2X) \sinh(Y)^4}{\sinh(X-Y)^2\sinh(X+Y)^2} < 0,
    \end{equation}
    for $X>Y>0$ and taking the $X\to \infty$ limit of the same expression we have
    \begin{equation}
        \lim_{X \to + \infty}\frac{   \sinh(Y)^2 \sinh (X)^2 }{\sinh (X-Y) \sinh (X+Y) }-  Y^2 = \sinh(Y)^2 - Y^2 >0.
    \end{equation}
    This completes the proof.
\end{proof}

The proof of the monotonicity of $\mathcal{V}$ is more involved.
We start by the definition~\eqref{eq:f2},
\begin{equation}\label{eq: f_2 in terms of a b}
	\mathcal{V}(K) = [\mathscr{A}(\eta,K)^2 + \mathscr{B}(\eta,K)]K +1,
\end{equation}
where we introduce
\begin{equation}\label{eq: a and b def}
	\mathscr{A}(\eta,K) = \zeta(\eta) - \frac{\zeta(\i \pi)}{\pi \i }\eta ,
\qquad
	\mathscr{B}(\eta,K) = - \wp(\eta) + 2\frac{\zeta(\i \pi)}{\i \pi}.
\end{equation}

\begin{proposition}\label{prop:f2monotone}
    The function $K\mapsto \mathcal{V}(K)$ is strictly increasing for $K>\frac 12\eta$.
\end{proposition}

\begin{proof}
	By \eqref{eq: f_2 in terms of a b} it suffices to show that $K\to \mathscr{A}(\eta,K)^2 + \mathscr{B}(\eta,K)$ is increasing for $K\geq \frac 12\eta>0$.
    To prove the latter statement, we will make use of multiple known expansions of the Weierstrass elliptic functions to prove that the first derivative in~$K$ of $\mathscr{A}(\eta,K)^2+\mathscr{B}(\eta,K)$ is strictly positive for all $0<\eta/2<K$.
    We split the proof into two parts, corresponding to $0<K\leq 3$ and $K > 3$: the first case is addressed in~\Cref{prop:monotonicity small K} and the second in~\Cref{prop:increasing property K large}. 
\end{proof}

\subsection{Small $K$: expansion in trigonometric functions}

In this section we assume $\eta\in (0,6)$ and $K \in (\eta/2,3]$.
By the trigonometric expansions~\eqref{eq:sigma trig expansion} and \eqref{eq:zeta trig expansion} we have
\begin{equation}
	\mathscr{A}(\eta,K) = \mathscr{A}_0(\eta,K) + R(\eta,K),
\end{equation}
where
\begin{align}
	\mathscr{A}_0(\eta,K) &= \frac{\eta}{2K} +\frac{\pi}{2K} \cot\left( \frac{\pi \eta}{2 K} \right),
\\
\label{eq:R}
		R(\eta,K) &=  \frac{2\pi}{K} \sum_{n \ge 1 } \frac{\sfq^n }{1-\sfq^n} \sin\left( \frac{n \pi \eta}{K} \right)
		 = \sum_{n\ge 1} \sum_{\ell \ge 1} \sfq^{n \ell} \frac{2 \pi}{K} \sin \left( \frac{n \pi \eta}{K} \right),
\end{align}
and (a notation which we will use throughout this section)
\begin{equation}\sfq=\e^{-2\pi^2/K}
\end{equation}
is the elliptic nome of the lattice $2K\mathbb{Z}+2\pi\i\mathbb{Z}$, see~\eqref{eq:nome}.

The function $\mathscr{B}(\eta,K)$ introduced in \eqref{eq: a and b def} has a similar trigonometric expansion
\begin{equation}
	\mathscr{B}(\eta,K) = \mathscr{B}_0(\eta,K) + S(\eta,K),
\end{equation}
where 
\begin{align}
	\mathscr{B}_0(\eta,K)  &= \frac{\pi^2}{4 K^2} - \frac{1}{K} - \frac{\pi^2}{4K^2} \csc\left( \frac{\pi \eta}{2 K} \right) ^2
    \\
    \nonumber
		S(\eta,K) &= \frac{2\pi^2}{K^2} \sum_{n \ge 1 }n \frac{\sfq^n }{1-\sfq^n} \cos\left( \frac{n \pi \eta}{K} \right) - \frac{3 \pi^2}{2 K^2} \sum_{n\ge 1} \mathrm{csch} \left( \frac{n \pi^2}{K} \right)^2
		\\
    \label{eq:S}
		& = \sum_{n\ge 1} \sum_{\ell \ge 1} \frac{2 \pi^2}{K^2} \left[  n \cos \left( \frac{n \pi \eta}{K} \right) - 3 \ell \right] \sfq^{ n \ell }.
\end{align}
In the expressions \eqref{eq:R} and \eqref{eq:S} we used the elementary expansions
\begin{equation}
	\frac{\sfq^n}{1-\sfq^n} = \sum_{\ell \ge 1} \sfq^{n \ell},
	\qquad
	\mathrm{csch}\left( \frac{n \pi^2}{K} \right)^2 = \sum_{ \ell \ge 0} 4 \ell \sfq^{\ell n}.
\end{equation}
We also define
\begin{equation}
	\widetilde{R}(\eta,K) = \frac{2 \pi }{K} \sin\left( \frac{\pi \eta}{K} \right) \sfq,
\qquad
	\widetilde{S}(\eta,K) = \frac{2 \pi ^2}{K^2} \left( \cos \left(\frac{\pi \eta}{K}\right) -3 \right) \sfq,
\end{equation}
which are the first coefficients in the $\sfq$-expansion of the functions $R,S$.

\begin{proposition} \label{prop:monotonicity small K}
	Let $\eta\in (0,6)$ and let $K \in (z/2,3]$. Then, we have
	\begin{equation}
		 \partial_K \left[ \mathscr{A}(\eta,K)^2 + \mathscr{B}(\eta,K) \right] > 0. 
	\end{equation}
\end{proposition}

The proof of \Cref{prop:monotonicity small K} relies on a series of bounds which we state and prove below.

\begin{proof}
	We have
	\begin{align*}
		&
		\partial_K \left[ \mathscr{A}(\eta,K)^2 + \mathscr{B}(\eta,K) \right] 
		\\
		&= \partial_K \left[ \mathscr{A}_0(\eta,K)^2+2\mathscr{A}_0(\eta,K) R(\eta,K) + R(\eta,K)^2 + \mathscr{B}_0(\eta,K) + S(\eta,K) \right]
		\\
		& =  \partial_K \left[ \mathscr{A}_0(\eta,K)^2 + 2\mathscr{A}_0(\eta,K) \widetilde{R}(\eta,K)  + \mathscr{B}_0(\eta,K)  + \widetilde{S}(\eta,K) \right]  
		\\
		& \quad +  \partial_K \left[ 2\mathscr{A}_0(\eta,K) (R(\eta,K)-\widetilde{R}(\eta,K))  + (S(\eta,K) -\widetilde{S}(\eta,K))\right] + \partial_K R(\eta,K)^2 .
	\end{align*}
	Estimates for each of the three terms in the last expression are provided in Lemmas~\ref{lem:bound derivative a^2+b + first order}, \ref{lem:bound derivative higher order}, and~\ref{lem:bound derivative R^2}, respectively.
    Combining these results we get
	\begin{equation}
			\partial_K \left[ \mathscr{A}(\eta,K)^2 + \mathscr{B}(\eta,K) \right] > \frac{\eta^2}{K^4} \sfq  \left( 156 - \frac{163000}{K^2} \sfq  \right)
			> 130 \frac{\eta^2}{K^4} \sfq  > 0,
	\end{equation}
	where in the second inequality we used the fact that $156 - \frac{163000}{K^2} \sfq  =   156 - \frac{163000}{K^2}\e^{-2 \pi^2/K}$
	is strictly decreasing for $K \in (0,3)$ and that it evaluates to $130.858\ldots$ at~$K=3$.
    This completes the proof.
\end{proof}

\begin{lemma} \label{lem:derivative afrak**2 + bfrak}
	Fix $\eta\in (0,2\pi)$. For $K\in\left( \eta/2, \pi^2/2 \right)$, the function
	$
		\partial_K \left[ \mathscr{A}_0(\eta,K)^2 + \mathscr{B}_0(\eta,K) \right]
	$
	is strictly decreasing and strictly positive, bounded below as
	\begin{equation}
		\partial_K \left[ \mathscr{A}_0 (\eta,K)^2 + \mathscr{B}_0(\eta,K) \right] > \frac{\eta^2}{2K^3}\left( \frac{\pi^2}{2K} -1 \right).
	\end{equation}
\end{lemma}
\begin{proof}
	An explicit computation shows that
	\begin{equation}
		\mathscr{A}_0(\eta,K)^2 + \mathscr{B}_0(\eta,K) = -\frac{1}{K} + \frac{\eta^2}{4 K^2}+\frac{\pi  \eta}{2 K^2} \cot \left(\frac{\pi  \eta}{2 K}\right).
	\end{equation}
	Using the Taylor expansion
    $
		\cot(x) = 2 \sum_{\ell \ge 0} \frac{(-1)^\ell }{(2 \ell)!}B_{2 \ell} (2x)^{2\ell-1}
	$
	where $B_{n}$ is the $n$th Bernoulli number, the Taylor expansion of the right-hand side is
	\begin{equation}
			\frac{\eta^2 }{4 K^2} - \frac{\pi^2 \eta^2}{12 K^3} + \sum_{\ell \ge 2} \frac{(-1)^\ell B_{2 \ell}}{(2 \ell)!}  \frac{ ( \pi \eta )^{2\ell} }{ K^{2\ell +1} },
	\end{equation}
	where we used $B_0 = 1$ and $B_2=\frac 16$.
	We can then compute the Taylor expansion of the first and second derivatives of $\mathscr{A}_0(\eta,K)^2 + \mathscr{B}_0(\eta,K)$ as
	\begin{align}
		\partial_K \left[ \mathscr{A}_0(\eta,K)^2 + \mathscr{B}_0(\eta,K) \right] &= -\frac{\eta^2}{2 K^3} + \frac{\pi ^2 \eta^2}{4 K^4} + \sum_{\ell \ge 2} (-1)^{\ell+1} B_{2 \ell} \frac{(2 \ell + 1) }{(2 \ell)!} \frac{ ( \pi \eta )^{2\ell} }{ K^{2\ell +2} },
	\\
		\partial_K^2 \left[ \mathscr{A}_0(\eta,K)^2 + \mathscr{B}_0(\eta,K) \right] &= \frac{3 \eta^2}{2 K^4}-\frac{\pi ^2 \eta^2}{K^5} + \sum_{\ell \ge 2} (-1)^\ell B_{2 \ell} \frac{(2 \ell + 1)(2 \ell +2) }{(2 \ell)!}  \frac{ ( \pi \eta )^{2\ell} }{ K^{2\ell + 3} }.
	\end{align}
	Since every coefficient $(-1)^\ell B_{2 \ell}$ is strictly negative for $\ell \ge 2$ we have
	\begin{align}
		\partial_K \left[ \mathscr{A}_0(\eta,K)^2 + \mathscr{B}_0(\eta,K) \right] &> \frac{\eta^2}{2K^3}\left( \frac{\pi^2}{2K} -1 \right),
	\\
		\partial_K^2 \left[ \mathscr{A}_0(\eta,K)^2 + \mathscr{B}_0(\eta,K) \right] &< - \frac{\eta^2}{K^4} \left( \frac{\pi^2}{K} - \frac{3}{2} \right) ,
	\end{align}
    which completes the proof.
\end{proof}

\begin{lemma} \label{lem:bound derivative a^2+b + first order}
	Let $\eta\in(0,6)$ and $K\in (\eta/2,3]$. Then, we have
	\begin{equation}
		  \partial_K \left[ \mathscr{A}_0(\eta,K)^2 + 2\mathscr{A}_0(\eta,K) \widetilde{R}(\eta,K)  + \mathscr{B}_0(\eta,K) + \widetilde{S}(\eta,K) \right] > \frac{156 \eta^2}{K^4} \sfq.
	\end{equation}
\end{lemma}
\begin{proof}
An explicit calculation shows that
\begin{align*}
	&\partial_K \left[ 2 \mathscr{A}_0(\eta,K) \widetilde{R}(\eta,K) + \widetilde{S}(\eta,K) \right] 
	\\
	& =  \frac{2 \pi }{K^4} \sfq 
	\left[ 2 \left(2 \pi ^2-K\right) \eta \sin \left(\frac{\pi  \eta}{K}\right)+4 \pi  \left(\pi ^2-K\right) \left(\cos \left(\frac{\pi  \eta}{K}\right)-1\right) -\pi  \eta^2 \cos \left(\frac{\pi  \eta}{K}\right) \right]
	\\
	& \ge 
    \frac{2 \pi^2 \eta^2 }{K^4} \e^{-2\pi^2/K} \left( 1 - \frac{2 \pi ^2}{K} - \frac{2 \pi^4}{K^2} \right),
\end{align*}
where in the last inequality we used the basic bounds $\sin(\pi \eta /K) \ge -\pi \eta /K$, $\cos(\pi \eta /K)-1 \ge - (\pi \eta)^2 / 2K^2$ and $\cos(\pi \eta /K) \le 1$.
This shows that
\begin{align*}
	& \partial_K \left[ \mathscr{A}_0(\eta,K)^2 + \mathscr{B}_0(\eta,K) \right] +  \partial_K \left[ 2\mathscr{A}_0(\eta,K) \widetilde{R}(\eta,K)  + \widetilde{S}(\eta,K) \right]
	\\
	&
	> \frac{\eta^2}{2K^3}\left( \frac{\pi^2}{2K} -1 \right) + \frac{2 \pi^2 \eta^2 }{K^4} \sfq \left( 1 - \frac{2 \pi ^2}{K} - \frac{2 \pi^4}{K^2} \right)
	\\
	& = 
    \frac{2 \eta^2}{K^4} \sfq \left[ \frac{1}{8} \sfq^{-1}   \left(\pi ^2-2 K\right) +\pi ^2 -\frac{2 \pi ^4}{K} -\frac{2 \pi ^6}{K^2} \right] 
	\\
	& > \frac{2 \eta^2}{K^4} \sfq \left[ \frac{810}{K^2} \left(\pi ^2-2 K\right) +\pi ^2 -\frac{2 \pi ^4}{K} -\frac{2 \pi ^6}{K^2} \right] 
	\\
	& > \frac{2 \eta^2}{K^4} \sfq \left[ \frac{ 6071 }{K^2}-\frac{1815}{K}+ 9 \right]
	\\
	& > \frac{156 \eta^2}{K^4} \sfq .
\end{align*}
In the last three inequalities we used, respectively: the elementary bound $\frac{1}{8} \sfq^{-1}=\frac{1}{8}\e^{2\pi^2/K}>\frac{810}{K^2}$ for $K\in(0,3]$; numerical evaluations of expressions involving powers of $\pi$; the fact that the function
$\frac{ 6071 }{K^2}-\frac{1815}{K}+ 9$
is strictly decreasing for $K\in (0,3]$ and equals $707/9>78$ at $K=3$.
\end{proof}

We can write then
\begin{align*}
	&2\mathscr{A}_0(\eta,K) R(\eta,K) + S(\eta,K)
	\\
	& = \sum_{n, \ell \ge 1} \frac{2 \pi}{K^2} \left[ \eta \sin \left( \frac{n \pi \eta}{K} \right) + \pi \left( \cot\left( \frac{\pi \eta}{2K} \right) \sin \left( \frac{n \pi \eta}{K} \right) -2n \right) + \pi n \left(  \cos \left( \frac{n \pi \eta}{K} \right)-1 \right) \right] \sfq^{n \ell}
	\\
	& = 2 \mathscr{A}_0(\eta,K) \widetilde{R}(\eta,K) +  \widetilde{S}(\eta,K) 
	\\
	& \qquad+ \sum_{ \substack{ n, \ell \ge 1 \\ (n , \ell) \neq (1,1) } } \frac{2 \pi}{K^2} \left[ \eta \sin \left( \frac{n \pi \eta}{K} \right) + \pi \left( \cot\left( \frac{\pi \eta}{2K} \right) \sin \left( \frac{n \pi \eta}{K} \right) -2n \right) + \pi n \left(  \cos \left( \frac{n \pi \eta}{K} \right)-1 \right) \right] \sfq^{n \ell}.
\end{align*}

\begin{lemma} \label{lem:bound derivative higher order}
	Let $\eta\in(0,6)$ and $K\in(\eta/2,3]$. Then, we have
	\begin{equation}
		\partial_K \left[ 2\mathscr{A}_0(\eta,K) (R(\eta,K)-\widetilde{R}(\eta,K))  + (S(\eta,K) -\widetilde{S}(\eta,K))\right] \ge  - \eta^2  \frac{ 100500 }{K^6} \sfq^2 .
	\end{equation}
\end{lemma}
\begin{proof}
For the derivative with respect to the variable $K$ of the general term in the above summation we have the following bound
\begin{align*}
	& \Bigg| \partial_K  \left[  \eta \sin \left( \frac{n \pi \eta}{K} \right) + \pi \left( \cot\left( \frac{\pi \eta}{2K} \right) \sin \left( \frac{n \pi \eta}{K} \right) -2n \right)  
	+ \pi n \left(  \cos \left( \frac{n \pi \eta}{K} \right)-1 \right) \right] \frac{2 \pi}{K^2} \sfq^{n \ell} \Bigg|
	\\
	& \le \eta^2 \sfq^{n \ell} \left[ \left( \frac{4 \pi}{K^3} + \frac{4 \pi^3}{K^4} n \ell \right) \left( \frac{n \pi }{K } + \frac{\pi^2 (2n^3+n) }{6 K^2} + \frac{n^2 \pi^2}{ 2 K^2}  \right) + \frac{2 \pi}{K^2} \left( \frac{n \pi}{K^2} + \frac{\pi^2 (2n^3+n) }{ 3 K^3} + \frac{n^2 \pi^2}{K^3} \right) \right]
	\\
	& \le \eta^2 \sfq^{n \ell} \frac{2\pi^2 n }{3K^6} \left[ 2 n^3 \pi ^3 \ell + n^2 \left(2 \pi  K^2+2 \pi  K+3 \pi ^3 \ell \right)+n \left(3 \pi  K^2+6 \pi ^2 K \ell +3 \pi  K+\pi ^3 \ell \right)
	\right.
	\\
	& \qquad \qquad \qquad \qquad \left.
	+9 K^2+\pi  K^2+\pi  K \right]
	\\
	& \le \eta^2 \sfq^{n \ell} \frac{2\pi^2 n }{3K^6} \left[ 2 n^3 \pi ^3 \ell + 7 n^2  \pi ^3 \ell  + 13 n \pi^3 \ell +5 \pi^3 \right]
	\\
	& \le \eta^2 \sfq^{n \ell} \frac{2\pi^5 n \ell }{3K^6} \left[ 2 n^3 + 7 n^2  + 13 n  +5 \right]
	\\
	& \le \eta^2 \sfq^{n \ell} \frac{ \pi^5 n^4 \ell }{K^6} \frac{2}{3} 27 . 
\end{align*}
The first bound was obtained combining the following more elementary estimates
\begin{align*}
	\left| \eta \sin \left( \frac{n \pi \eta}{K} \right) \right| &\le \eta^2 \frac{n \pi }{K}
\\
	\left| \cot \left( \frac{\pi \eta}{2K} \right) \sin\left( \frac{n \pi \eta}{K}  \right) -2 n \right| &\le \eta^2 \frac{ (2n^3 +n) \pi^2 }{6 K^2}
\\
	\left| \cos \left( \frac{n \pi \eta}{K} \right)-1 \right| &\le \eta^2 \frac{ n^2 \pi^2 }{2 K^2}
\\
	\left| \partial_K \left[ \eta \sin \left( \frac{n \pi \eta}{K} \right) \right] \right| &\le \eta^2 \frac{n\pi}{K^2}
\\
	\left| \partial_K \left[\cot \left( \frac{\pi \eta}{2K} \right) \sin\left( \frac{n \pi \eta}{K}  \right) -2 n \right] \right| &\le \eta^2 \frac{(2n^3+n) \pi^2}{3 K^3}
\\
	\left| \partial_K \left[ \cos \left( \frac{n \pi \eta}{K} \right)-1 \right] \right| &\le \eta^2 \frac{ n^2 \pi^2 }{ K^3}
\\
	\left| \partial_K \left[ \frac{2 \pi }{K} \sfq^{n \ell} \right] \right| &\le  \sfq^{n \ell} \left( \frac{4 \pi}{K^3} + \frac{4 \pi^3}{K^4} n \ell \right).
\end{align*}
Finally we have the estimate
\begin{align*}
	\sum_{ \substack{ n, \ell \ge 1 \\ (n , \ell) \neq (1,1) } } \sfq^{n \ell} n^4 \ell 
	& =\biggl( \sum_{n \ge 1} n^4 \frac{\sfq^n}{(1-\sfq^n)^2} \biggr)-\sfq
	\\
	& \le \frac{1}{(1-\sfq)^2} \biggl(\sum_{n \ge 1} n^4 \sfq^n \biggr)-\sfq
	\\
	& = \sfq^2 \frac{ \left(\sfq^6-7 \sfq^5+21 \sfq^4-35 \sfq^3+36 \sfq^2-10 \sfq+18\right)}{(1-\sfq)^7}
	\\
	& \le \sfq^2 \frac{ \left(\sfq^6 +21 \sfq^4+36 \sfq^2+18\right)}{(1-\sfq)^7}.
\end{align*}
The above estimate shows that
\begin{align*}
		&
		\partial_K \sum_{ \substack{ n, \ell \ge 1 \\ (n , \ell) \neq (1,1) } } \frac{2 \pi}{K^2} \left[ \eta \sin \left( \frac{n \pi \eta}{K} \right) + \pi \left( \cot\left( \frac{\pi \eta}{2K} \right) \sin \left( \frac{n \pi \eta}{K} \right) -2n \right) 
        + 
        \pi n \left(  \cos \left( \frac{n \pi \eta}{K} \right)-1 \right) \right] \sfq^{n \ell}
		\\
		& > - \eta^2 \sfq^2 \frac{18 \pi^5  }{K^6} \frac{ \left(\sfq^6 +21 \sfq^4+36 \sfq^2+18\right)}{(1-\sfq)^7}
		\\
		& > - \eta^2 \sfq^2 \frac{ 100500 }{K^6}
\end{align*}
where in the last inequality we used the fact that for $K\in(0,3]$ we have
\begin{equation}
	 18 \pi^5   \frac{ \left(\sfq^6 +21 \sfq^4+36 \sfq^2+18\right)}{(1-\sfq)^7} < 100457.
\end{equation}
\end{proof}

\begin{lemma} \label{lem:bound derivative R^2}
	Let $\eta\in(0,6)$ and $K\in(\eta/2,3]$. Then, we have
	\begin{equation}
		\left| \partial_K R(\eta,K)^2 \right| \le  \eta^2 \, \frac{62500}{K^6} \, \e^{-4 \pi^2/K}.
	\end{equation}
\end{lemma}

\begin{proof}
Define the constant
\begin{equation}
	\mathsf{c} = \frac{1}{1-\e^{-2 \pi}} = 1.00187...
\end{equation}
Notice that, for all $0<\eta<2K\le 6$ we have 
	\begin{equation}
		\left| \frac{1}{1-\e^{-2n \pi^2/K}} \sin\left( \frac{n \pi \eta}{K} \right) \right| < \mathsf{c} \frac{n \pi \eta}{K} .
	\end{equation}
	Then, we have
	\begin{equation}\label{eq:bound abs R}
		\begin{split}
			\left| R(\eta,K) \right| & \le \frac{2 \pi}{K} \sum_{n \ge 1} \e^{-2n \pi^2/K}  \left| \frac{1}{1-\e^{-2n \pi^2/K}} \sin\left( \frac{n \pi \eta}{K} \right) \right|
			\\& \le \frac{2 \pi^2 \eta}{K^2} \mathsf{c} \sum_{n \ge 1} n \e^{-2n \pi^2/K} 
			\\
			& =  \frac{2 \pi^2 \eta}{K^2} \mathsf{c} \frac{\e^{-2\pi^2/K}}{ (1-\e^{-2\pi^2/K})^2}
			\\
			& \le  \frac{2 \pi^2 \eta}{K^2} \mathsf{c}^3 \e^{-2\pi^2/K}.
		\end{split} 
	\end{equation}
	To produce a bound for the derivative of $R$ compute
	\begin{equation}
		\begin{split}
			& \partial_K \left[ \frac{2 \pi}{K}  \frac{\e^{-2n \pi^2/K}}{1-\e^{-2n \pi^2/K}} \sin\left( \frac{n \pi \eta}{K} \right) \right]
			\\
			& = -\frac{2 \pi}{K^2} \left[ \frac{n \pi \eta \cos \left( \frac{n \pi \eta}{K} \right)}{(1-\e^{-2 n \pi^2/K})K} + \frac{ \sin \left( \frac{n \pi \eta}{K} \right)}{(1-\e^{-2 n \pi^2/K})} - \frac{ 2 n \pi^2 \sin \left( \frac{n \pi \eta}{K} \right)}{(1-\e^{-2 n \pi^2/K})^2 K}  \right] \e^{-2 n \pi^2/K}.
		\end{split}
	\end{equation}
	Then, through basic estimates we obtain
	\begin{equation}
		\begin{split}
    			& \left| \partial_K \left[ \frac{2 \pi}{K}  \frac{\e^{-2n \pi^2/K}}{1-\e^{-2n \pi^2/K}} \sin\left( \frac{n \pi \eta}{K} \right) \right] \right|
			\\
			& \le \frac{2 \pi}{K^2} \left[ \frac{n \pi \eta}{K} \mathsf{c} + \frac{n \pi \eta}{K} \mathsf{c} + \frac{2 n^2 \pi^3\eta }{K^2} \mathsf{c}^2  \right]  \e^{-2 n \pi^2/K}.
			\\
			& \le \eta \frac{4 \pi^2}{K^3} \mathsf{c}^2 \left[ 1     + \frac{\pi^2 }{K}   \right] n^2 \e^{-2 n \pi^2/K}.
		\end{split}
	\end{equation}
	Summing over $n$, we obtain
	\begin{equation}\label{eq:bound abs derivative R}
		\begin{split}
			|\partial_K R(\eta,K)| & \le \eta \frac{4 \pi^2}{K^3} \mathsf{c}^2 \left[ 1  + \frac{\pi^2 }{K}   \right]   \sum_{n \ge 1} n^2 \e^{-2 n \pi^2/K}
			\\
			&
			= \eta \frac{4 \pi^2}{K^3} \mathsf{c}^2 \left[ 1  + \frac{\pi^2 }{K}   \right] \frac{1+ \e^{-2 \pi^2/K}}{(1-\e^{-2 \pi^2/K})^3} \e^{-2 \pi^2/K}
			\\
			&
			\le \eta \frac{8 \pi^2}{K^3} \mathsf{c}^5 \left[ 1  + \frac{\pi^2 }{K}   \right] \e^{-2 \pi^2/K}.
		\end{split}
	\end{equation}
	To complete the proof we combine the estimates \eqref{eq:bound abs R}, \eqref{eq:bound abs derivative R}, obtaining
	\begin{equation}
		\begin{split}
			\left| \partial_K R(\eta,K)^2 \right| &\le 2 \left( \frac{2 \pi^2 \eta}{K^2} \mathsf{c}^3 \right) \left( \eta \frac{8 \pi^2}{K^3} \mathsf{c}^5 \left[ 1  + \frac{\pi^2 }{K}   \right]  \right) \e^{-4 \pi^2/K} 
			\\
			& \le \eta^2 \frac{ 64 \pi^6 \mathsf{c}^8}{K^6} \e^{-4 \pi^2/K}
			\\
			& \le \eta^2 \, \frac{72000}{K^6} \, \e^{-4 \pi^2/K},
		\end{split}
	\end{equation}
    where in the last inequality we used the explicit evaluation $64 \pi^6 \mathsf{c}^8 = 62455.9... < 62500$.
\end{proof}

\subsection{Large $K$: expansion in hyperbolic functions}

Introducing the auxiliary function
\begin{equation}
	\mathsf{f}_\eta(x) =  \coth\left( x - \frac{\eta}{2} \right) - \coth\left( x + \frac{\eta}{2} \right)
\end{equation} 
and using the hyperbolic series expansions \eqref{eq:wZ}, \eqref{eq:wpK}, we have
\begin{equation}\label{eq:a and b}
	\mathscr{A}(\eta,K) = \frac{1}{2} \coth\left( \frac{\eta}{2} \right) - \frac{1}{2} \sum_{n=1}^{\infty}\mathsf{f}_\eta(nK)
\end{equation}
and
\begin{equation}
	\begin{split}
		&\mathscr{B} (\eta,K) 
		\\
		&= -\frac{5 + \cosh(\eta)}{24 \sinh (\frac{\eta}{2})^{2} } - \frac{1}{6} + \frac{1}{4} \sum_{n=1}^{\infty} \left\{ \frac{6}{\sinh(K n)^2} - \frac{1}{\sinh(K n-\frac{\eta}{2})^2} - \frac{1}{\sinh(K n+\frac{\eta}{2})^2} \right\}.
	\end{split}
\end{equation}
Using the above expansions we can write
\begin{equation}\label{eq: f2 series}
	\begin{split}
		&\mathscr{A}(\eta,K)^2+\mathscr{B}(\eta,K)
		\\ 
		& = \frac{1}{4}\coth\left( \frac{\eta}{2} \right)^2 -\frac{5 + \cosh(\eta)}{24 \sinh (\frac{\eta}{2})^{2} } - \frac{1}{6} +  \left( \frac{1}{2} \sum_{n \ge 1} \mathsf{f}_\eta(n K)  \right)^2
		\\
		& + \frac{1}{4} \sum_{n \ge 1} \left\{ \frac{6}{\sinh(K n)^2} - \frac{1}{\sinh(K n-\frac{\eta}{2})^2} - \frac{1}{\sinh(K n+\frac{\eta}{2})^2} - 2 \coth\left( \frac{\eta}{2} \right) \mathsf{f}_\eta(n K) \right\}.
	\end{split}
\end{equation}

\begin{lemma}
	The function $\mathsf{f}_\eta(x)$ is strictly positive and strictly decreasing for $x>\eta/2$.
\end{lemma}

\begin{proof}
	This is evident by computing the derivative $\mathsf{f}_\eta'(x)$.
\end{proof}

\begin{lemma} \label{lem: increasing term}
	The function
	\begin{equation}
		\begin{split}
			&\frac{6}{\sinh(x)^2} - \frac{1}{\sinh(x-\frac{\eta}{2})^2} - \frac{1}{\sinh(x+\frac{\eta}{2})^2} - 2 \coth\left( \frac{\eta}{2} \right) \mathsf{f}_\eta(x)
			\\
			& = \left[ \frac{2}{\sinh(x)^2} - \frac{1}{\sinh(x-\frac{\eta}{2})^2} - \frac{1}{\sinh(x+\frac{\eta}{2})^2} \right] + 2 \left[\frac{2}{\sinh(x)^2}-  \coth\left( \frac{\eta}{2} \right) \mathsf{f}_\eta(x) \right]
		\end{split}
	\end{equation}
	is strictly increasing for $x>\eta/2$. 
\end{lemma}

\begin{proof}
To see this we analyze separately the two terms in the second line. First observe that
\begin{equation}
	\frac{2}{\sinh(x)^2} - \frac{1}{\sinh(x-\frac{\eta}{2})^2} - \frac{1}{\sinh(x+\frac{\eta}{2})^2}
\end{equation}
is increasing in $x$ by virtue of the fact that the first derivative of the function $\frac{1}{\sinh(x)^2}$ is concave, which is straightforward to verify.
For the remaining term we evaluate its derivative as follows
\begin{align*}
	&\frac{\d}{\d x} \left[\frac{2}{\sinh(x)^2}-  \coth\left( \frac{\eta}{2} \right) \mathsf{f}_\eta(x) \right]
	\\
	& = \coth \left( \frac{\eta}{2} \right) \left[ \frac{1}{\sinh (x-\frac{\eta}{2})^2 } - \frac{1}{\sinh (x+\frac{\eta}{2})^2 } \right] - 4\frac{\coth(x)}{\sinh(x)^2}
	\\
	& = \lim_{\eta \to 0} \left\{ \coth \left( \frac{\eta}{2} \right) \left[ \frac{1}{\sinh (x-\frac{\eta}{2})^2 } - \frac{1}{\sinh (x+\frac{\eta}{2})^2 } \right] - 4\frac{\coth(x)}{\sinh(x)^2}\right\}
	\\ 
	& \quad + \int_0^{\eta} \frac{\d}{\d s} \left\{ \coth \left( \frac{s}{2} \right) \left[ \frac{1}{\sinh (x-\frac{s}{2})^2 } - \frac{1}{\sinh (x+\frac{s}{2})^2 } \right] - 4\frac{\coth(x)}{\sinh(x)^2} \right\} \d s
	\\
	& = \frac{1}{2}\int_0^{\eta} [2+ \cosh(2x)+\cosh(s)] \frac{\sinh(2x) \sinh(s)}{\sinh(x-\frac{s}{2})^3 \sinh(x+\frac{s}{2})^3}  \d s  >0,
\end{align*}
whenever $x>\eta/2>0$. Above, in the second equality we expressed the function of $\eta$ as the integral of its derivative over $(0,\eta)$ plus its value at $\eta=0$.
\end{proof}

Using the expression~\eqref{eq: f2 series} we can evaluate the derivative
\begin{equation}\label{eq: a^2 + b expansion K large}
	\begin{split}
		&\partial_K [\mathscr{A}(\eta,K)^2+\mathscr{B}(\eta,K)] 
		\\
		& = \frac{1}{2} \left( \sum_{n \ge 1} \mathsf{f}_\eta(nK) \right) \left(  \sum_{n \ge 1} n \mathsf{f}_\eta'(nK)  \right)
		\\
		&+ \frac{1}{4} \sum_{n \ge 1} \partial_K \left\{ \frac{6}{\sinh(K n)^2} - \frac{1}{\sinh(K n-\frac{\eta}{2})^2} - \frac{1}{\sinh(K n+\frac{\eta}{2})^2} - 2 \coth\left( \frac{\eta}{2} \right) \mathsf{f}_\eta(n K) \right\}.
	\end{split}
\end{equation}

\begin{proposition} \label{prop:increasing property K large}
	Let $\eta > 0$ and let $K>\max(\eta/2,3)$. Then, we have
	\begin{equation}
		 \partial_K \left[ \mathscr{A} (\eta,K)^2 + \mathscr{B} (\eta,K) \right] > 0. 
	\end{equation}
\end{proposition}

The proof of \Cref{prop:increasing property K large} relies on several preliminary bounds, which we collect in the next several Lemmas. In the remainder of the subsection we will use the constant
\begin{equation}
	\sfC = \frac{1}{1-\e^{-2n \pi^2/5}} = 1.01968...
\end{equation}

\begin{lemma} \label{lem:bound f K large}
	We have
	\begin{equation}\label{eq:bound f K large}
		\mathsf{f}_\eta(x) \le 4 \sfC \eta \e^{-x} \qquad \text{for all } \eta>0 \quad \text{and} \quad x>\max(6,\eta).
	\end{equation}
	As a result
	\begin{equation}\label{eq:bound sum f higher terms}
		\sum_{n\ge 2 }\mathsf{f}_\eta(n K) \le 4 \sfC^2 \eta \e^{- 2K}  \qquad \text{for all } \eta>0 \quad \text{and} \quad K>\max(3,\eta/2).
	\end{equation}
\end{lemma}
\begin{proof}
	For any $x>\max(\eta,6)$, we have
	\begin{equation}\label{eq: bound f preliminary}
		\begin{split}
			\mathsf{f}_\eta(x) &= 2 \frac{\sinh(\eta)}{\cosh(2x) - \cosh(\eta)}
			\\
			& \le  \sup_{x > \max(\eta,6)} \left\{  \frac{1}{1 - \frac{\cosh(\eta)}{\cosh(2x) }} \right\} \frac{2 \sinh(\eta)}{\cosh(2x)}
			\\
			& \le \sup_{x > \max(\eta,6)} \left\{  \frac{1}{1 - \frac{\cosh(\eta)}{\cosh(2x) }} \right\}  \left( 4 \sinh(\eta) \e^{-\eta} \right) \e^{-x},
		\end{split}
	\end{equation}
	where in the second inequality we used the basic bounds $\cosh(2x)^{-1} \le 2 \e^{-2x} \le 2 \e^{-\eta} \e^{-x}$. When $\eta > 6$, we have
	\begin{equation}
		\begin{split}
			\sup_{x > \max(\eta,6)} \left\{  \frac{1}{1 - \frac{\cosh(\eta)}{\cosh(2x) }} \right\} &
            \le  \frac{1}{1 - \frac{\cosh(\eta)}{\cosh(2\eta) }} 
            \le \frac{1}{1 - \frac{\cosh(6)}{\cosh(12) }}
            = 1.00248...
            < \sfC,
		\end{split}
	\end{equation}
	where we used the fact that the functions $x\to \left( 1- \frac{\cosh(\eta)}{\cosh(2x) } \right)^{-1}$ and $\eta\to \left( 1- \frac{\cosh(\eta)}{\cosh(2\eta) } \right)^{-1}$ are strictly decreasing for $x>\eta>6$. On the other hand we have
	\begin{equation}
		4 \sinh(\eta) \e^{-\eta} < 4 \lim_{\eta\to \infty}  4 \sinh(\eta) \e^{-\eta} = 2,
	\end{equation}
	since $\sinh(\eta) \e^{-\eta}$ is an increasing function. These bounds prove that
	\begin{equation}
		\mathsf{f}_\eta(x) \le 2 \sfC \e^{-x} \qquad \text{for} \quad 6<\eta<x.
	\end{equation}
	Consider now the case $\eta\in(0,6)$. Analyzing the right-hand side of \eqref{eq: bound f preliminary} we find, using similar considerations,
	\begin{equation}
			\sup_{x > \max(\eta,6)} \left\{  \frac{1}{1 - \frac{\cosh(\eta)}{\cosh(2x) }} \right\}
            = \frac{1}{1 - \frac{\cosh(\eta)}{\cosh(12) }}
            < \frac{1}{1 - \frac{\cosh(6)}{\cosh(12) }}
            < \sfC
	\end{equation}
	and 
	$
		4 \sinh(\eta) \e^{-\eta} \le 4 \eta.
	$
	This proves that
	\begin{equation}
		\mathsf{f}_\eta(x) \le 4 \sfC \eta \e^{-x} \qquad \text{for} \quad 0<\eta<6<x.
	\end{equation}
	Combining the two bounds obtained for the cases $\eta\in (0,6)$ and $\eta \ge 6$ we prove \eqref{eq:bound f K large}. To show \eqref{eq:bound sum f higher terms}, we simply use \eqref{eq:bound f K large} to each term of the sum using the fact that, whenever $K> \max(3,\eta/2)$, we have $2K > \max(6,\eta)$ obtaining
	\begin{equation}
		\begin{split}
			\sum_{n\ge 2 }\mathsf{f}_\eta(n K)
            &
            \le 4 \sfC \eta  \sum_{n\ge 2 } \e^{-nK} 
            \le  4 \sfC \eta \frac{\e^{- 2 K}}{1 - \e^{-K}}
            \le 4 \sfC^2 \eta \e^{- 2K}.
		\end{split}
	\end{equation}
\end{proof}

\begin{lemma} \label{lem:bound f' K large}
	We have
	\begin{equation}\label{eq: bound f' K large}
		-\mathsf{f}_\eta ' (x) \le 8 \sfC \eta \e^{-x} \qquad \text{for all} \quad \eta>0 \quad \text{and} \quad x > \max(6,\eta).
	\end{equation}
	As a result,
	\begin{equation}\label{eq:bound sum f' higher terms}
		-\sum_{n \ge 2} n \mathsf{f}_\eta(nK) \le  16 \sfC^3  \eta \e^{-2K} \qquad \text{for all} \quad \eta>0 \quad \text{and} \quad K > \max(3,\eta/2).
	\end{equation}
\end{lemma}
\begin{proof}
	We proceed in the same way as in \Cref{lem:bound f K large}. Through an explicit computation we obtain
	\begin{equation}
			-\mathsf{f}_\eta'(x) = \frac{4 \sinh(2x) \sinh(\eta)}{(\cosh(2x) - \cosh(\eta))^2}
			 =  \frac{\tanh(2x)}{\left( 1- \frac{\cosh(\eta)}{\cosh(2x)} \right)^2}  \frac{4 \sinh(\eta)}{\cosh(2x)}.
	\end{equation}
	Using the inequalities
	\begin{equation}
    \begin{aligned}
		&\left( 1- \frac{\cosh(\eta)}{\cosh(2x)} \right)^{-2} \le \left( 1- \frac{\cosh(6)}{\cosh(12)} \right)^{-2} = 1.00498... < \sfC,
	\\
		&\tanh(2x) \le 1,
	\quad\cosh(2x)^{-1} \le 2 \e^{-\eta} \e^{-x},
	\quad
		\sinh(\eta) \e^{-\eta} \le \eta,
        \end{aligned}
	\end{equation}
	we get \eqref{eq: bound f' K large}. To show \eqref{eq:bound sum f' higher terms} we use the inequality \eqref{eq: bound f' K large} obtaining
	\begin{align*}
		- \sum_{n \ge 2} n \mathsf{f}_\eta'(nK) &\le 8 \sfC \eta  \sum_{n \ge 2} n \e^{-nK}
        =  8 \sfC \eta \frac{2-\e^{-K}}{(\e^K-1)^2}
        \le 16 \sfC^3  \eta \e^{-2K}
	\end{align*}
    and the proof is complete.
\end{proof}

\begin{lemma} \label{lem:lower bound series f f'}
	Let $\eta > 0$ and let $K>\max(3,\eta/2)$. Then
	\begin{equation}
		\left( \sum_{n\ge 1} \mathsf{f}_\eta (n K) \right) \left( \sum_{n \ge 1} n \mathsf{f}_\eta ' (n K) \right) 
		\ge \left( \mathsf{f}_\eta (K) + 4 \sfC^2 \eta \e^{- 2K}  \right) \left( \mathsf{f}_\eta' (K) - 16 \sfC^3  \eta \e^{-2K}\right).
	\end{equation} 
\end{lemma}

\begin{proof}
	This is a straightforward application of \Cref{lem:bound f K large} and \Cref{lem:bound f' K large}.
\end{proof}

\begin{lemma} \label{lem:lower positive term}
	Let $\eta >0 $ and let $K>\eta/2$. Then
	\begin{equation}\label{eq:bound sum increasing terms}
		\begin{split}
			&\frac{1}{4} \sum_{n \ge 1} \partial_K \left\{ \frac{6}{\sinh(K n)^2} - \frac{1}{\sinh(K n-\frac{\eta}{2})^2} - \frac{1}{\sinh(K n+\frac{\eta}{2})^2} - 2 \coth\left( \frac{\eta}{2} \right) \mathsf{f}_\eta(n K) \right\}
			\\
			& \ge \frac{1}{2} \left( \frac{ \coth \left(K+\frac{\eta}{2}\right) -\coth \left(\frac{\eta}{2}\right)}{ \sinh^2\left(K+\frac{\eta}{2}\right) } + \frac{ \coth \left(K-\frac{\eta}{2}\right) + \coth \left(\frac{\eta}{2}\right) }{ \sinh^2\left(K-\frac{\eta}{2}\right) }  - \frac{ 6 \coth (K) }{ \sinh^2(K) } \right).
		\end{split}
	\end{equation} 
\end{lemma}
\begin{proof}
	We have shown in \Cref{lem: increasing term} that the right-hand side of \eqref{eq:bound sum increasing terms} is a sum of strictly positive terms and as a result it is bounded from below by the first term in the sum.
\end{proof}

\begin{proof}[Proof of \Cref{prop:increasing property K large}]
	From \eqref{eq: a^2 + b expansion K large} and \Cref{lem:lower bound series f f'}, \Cref{lem:lower positive term}, we have to show that for $\eta>0$ and $K>\max(\eta/2,3)$, the function
	\begin{equation}\label{eq:K large function to estimate}
		\begin{split}
			&\left( \mathsf{f}_\eta (K) + 4 \sfC^2 \eta \e^{- 2K}  \right) \left( \mathsf{f}_\eta' (K) - 16 \sfC^3  \eta \e^{-2K}\right) 
			\\
			& + \left( \frac{ \coth \left(K+\frac{\eta}{2}\right) -\coth \left(\frac{\eta}{2}\right)}{ \sinh^2\left(K+\frac{\eta}{2}\right) } + \frac{ \coth \left(K-\frac{\eta}{2}\right) + \coth \left(\frac{\eta}{2}\right) }{ \sinh^2\left(K-\frac{\eta}{2}\right) }  - \frac{ 6 \coth (K) }{ \sinh^2(K) } \right)
			\\
			& = \left[ \mathsf{f}_\eta(K)- \coth\left( \frac{\eta}{2} \right) \right]  \mathsf{f}_\eta'(K) + \frac{\coth(K+\frac{\eta}{2})}{\sinh(K+\frac{\eta}{2})^2}+ \frac{\coth(K-\frac{\eta}{2})}{\sinh(K-\frac{\eta}{2})^2} - 6 \frac{\coth(K)}{\sinh(K)^2}
			\\
			& \qquad + 4 \eta \sfC^2 \e^{-2K} \left[ \mathsf{f}_\eta' (K) - 4 \sfC \mathsf{f}_\eta (K)  \right] -  64 \eta^2 \sfC^5  \e^{-4K}
			\\
			& = \frac{\coth (K) \sinh \left(\frac{\eta}{2}\right)^2 }{\sinh(K)^2   \sinh \left(K+\frac{\eta}{2}\right)^2  \sinh \left(K-\frac{\eta}{2}\right)^2 } (2 \cosh (2 K)+\cosh (4 K)-3 \cosh (\eta))
			\\
			& \qquad + 4 \eta \sfC^2 \e^{-2K} \left[ \mathsf{f}_\eta' (K) - 4 \sfC \mathsf{f}_\eta (K)  \right] -  64 \eta^2 \sfC^5  \e^{-4K},
		\end{split}
	\end{equation}
	is strictly positive. We have
	\begin{align*}
			&\frac{\coth (K) \sinh \left(\frac{\eta}{2}\right)^2 }{\sinh(K)^2   \sinh \left(K+\frac{\eta}{2}\right)^2  \sinh \left(K-\frac{\eta}{2}\right)^2 } [2 \cosh (2 K)+\cosh (4 K)-3 \cosh (\eta)]
			\\
			& > 
			\frac{ \sinh \left(\frac{\eta}{2}\right)^2 }{\sinh(K)^2   \sinh \left(K+\frac{\eta}{2}\right)^2  \sinh \left(K-\frac{\eta}{2}\right)^2 } [\cosh (4 K)- \cosh (2K)]
			\\
			& > 
			\frac{ \sinh \left(\frac{\eta}{2}\right)^2 \e^{-2K+\eta} }{\sinh(K)^2   \sinh \left(K+\frac{\eta}{2}\right)^2 } [\cosh (4 K)- \cosh (2K)]  \left[ \frac{1}{\left(K-\frac{\eta}{2}\right)^2}+\frac{5}{3} \right]
			\\
			& = 
			\sinh \left(\frac{\eta}{2}\right)^2 8 \e^{-2K}  \frac{ (1+\e^{-8K} -\e^{-2K} -\e^{-6K} )  }{(1- \e^{-2K})^2 (1- \e^{-2K-\eta})^2 }   \left[ \frac{1}{\left(K-\frac{\eta}{2}\right)^2}+\frac{5}{3} \right]
			\\
			& > 
			\sinh \left(\frac{\eta}{2}\right)^2 8 \e^{-2K}  \frac{ (1+\e^{-8K} -\e^{-2K} -\e^{-6K} )  }{(1- \e^{-2K})^2 (1- \e^{-2K-\eta})^2 }   \left[ \frac{1}{\left(K-\frac{\eta}{2}\right)^2}+\frac{5}{3} \right]
			\\
			& > 
			\sinh \left(\frac{\eta}{2}\right)^2 8 \e^{-2K}  \left[ \frac{1}{\left(K-\frac{\eta}{2}\right)^2}+\frac{5}{3} \right],
	\end{align*}
	where we only used the definitions of hyperbolic functions and the elementary bounds
	\begin{equation}
		 \frac{1}{\sinh \left( K - \frac{\eta}{2} \right)^2 }> \e^{-2K+\eta} \left[ \frac{1}{\left(K-\frac{\eta}{2}\right)^2}+\frac{5}{3} \right],
	\end{equation}
	\begin{equation}
		\frac{ (1+\e^{-8K} -\e^{-2K} -\e^{-6K} )  }{(1- \e^{-2K})^2 (1- \e^{-2K-\eta})^2 } >1, \qquad \text{for } K>\max(\eta/2,3), \, \eta > 0.
	\end{equation}
	For the remaining terms in the right-hand side of \eqref{eq:K large function to estimate}, we have
	\begin{align*}
		&4 \eta \sfC^2 \e^{-2K} \left[ \mathsf{f}_\eta' (K) - 4 \sfC \mathsf{f}_\eta (K)  \right] -  64 \eta^2 \sfC^5 \e^{-4K}
		\\
		& = - \frac{32 \sfC^2 \e^{-4 K} \eta}{\left(-2 \e^{2 K} \cosh (\eta)+\e^{4 K}+1\right)^2}    \biggl[8 \sfC^3 \e^{4 K} \eta+2 \sfC^3 \e^{8 K} \eta-8 \sfC^3 \e^{2 K} \eta \cosh (\eta)
		\\
		& \qquad\qquad -8 \sfC^3 \e^{6 K} \eta \cosh (\eta)+4 \sfC^3 \e^{4 K} \eta \cosh (2 \eta)+2 \sfC^3 \eta+2 \sfC \e^{4 K} \sinh (\eta)
		\\
		& \qquad\qquad 
		+2 \sfC \e^{8 K} \sinh (\eta)-2 \sfC \e^{6 K} \sinh (2 \eta)-\e^{4 K} \sinh (\eta)+\e^{8 K} \sinh (\eta)\biggr]
		\\
		& > - \frac{32 \sfC^2 \e^{-4 K} \eta}{\left(-2 \e^{2 K} \cosh (\eta)+\e^{4 K}+1\right)^2}    \biggl[8 \sfC^3 \e^{4 K} \eta+2 \sfC^3 \e^{8 K} \eta  +4 \sfC^3 \e^{4 K} \eta \cosh (2 \eta)
		\\
		& \qquad\qquad\qquad\qquad  +2 \sfC^3 \eta+2 \sfC \e^{4 K} \sinh (\eta)
		+2 \sfC \e^{8 K} \sinh (\eta)  +\e^{8 K} \sinh (\eta)\biggr]
		\\
		& > - \frac{32 \sfC^5 \e^{-4 K} \eta}{\left(-2 \e^{2 K} \cosh (\eta)+\e^{4 K}+1\right)^2}    
		\\
		& \quad \times \left[ 2  \eta +  \e^{4 K} ( 4  \eta \cosh (2 \eta) + 2 \sinh (\eta) + 10 \eta) +3  \e^{8 K} \sinh (\eta)\right]
		\\
		& > - \frac{32 \sfC^5 \e^{-4 K} \eta  \sinh (\eta) }{\left(-2 \e^{2 K} \cosh (\eta)+\e^{4 K}+1\right)^2} \left[ 2   +  \e^{4 K} ( 4   \cosh (2 \eta) + 12 ) +3 \e^{8 K} \right]
		\\
		& = - \frac{32 \sfC^5 \e^{-4 K} \eta  \sinh (\eta) }{\left(\e^{-4 K}-2 \e^{-2 K} \cosh (\eta)+1\right)^2} \left[ 2 \e^{-8 K}  +  \e^{-4 K+2\eta} ( 2  + 2 \e^{-4\eta}    + 12 \e^{-2\eta} ) +3   \right]
		\\
		& > - \frac{ 192 \sfC^7 \e^{-4 K} \eta  \sinh (\eta) }{\left(\e^{-4 K}-2 \e^{-2 K} \cosh (\eta)+1\right)^2}
		\\
		& > -  192 \sfC^7 \e^{-4 K} \eta  \sinh (\eta) \left( \frac{1}{\left(K-\frac{\eta}{2}\right)^2}+1  \right)
		\\
		& > -  200 \e^{-4 K} \eta  \sinh (\eta) \left( \frac{1}{\left(K-\frac{\eta}{2}\right)^2}+1  \right) ,
	\end{align*}
	where we used
	\begin{align*}
		&2 \e^{-8 K}  +  \e^{-4 K+2\eta} ( 2  + 2 \e^{-4\eta}    + 12 \e^{-2\eta} ) +3    
		\\
		& < 2 \e^{-24}  +  2  + 2 \e^{-12}    + 12 \e^{-12}  +3 
		\\
		& < 6
	\end{align*}
	and
	\begin{align*}
		\frac{1 }{\left(\e^{-4 K}-2 \e^{-2 K} \cosh (\eta)+1\right)^2} &= \frac{1}{\left(1-\e^{\eta-2 K}\right)^2 \left(1-\e^{-2 K-\eta}\right)^2}  
		\\
		&
		< \sfC^2 \frac{1}{\left(1-\e^{\eta-2 K}\right)^2 } 
		\\
		& 
        <  \sfC^2 \left(   \frac{1}{\left(K-\frac{\eta}{2}\right)^2}+1 \right).
	\end{align*}
	The last inequality follows from the elementary bound $(1-\e^{-2x})^{-2}<1+x^{-2}$.
	Summing the right-hand sides of the above bounds we get
	\begin{align*}
		&\sinh \left(\frac{\eta}{2}\right)^2 8 \e^{-2K}  \left[ \frac{1}{\left(K-\frac{\eta}{2}\right)^2}+\frac{5}{3} \right] -  200 \e^{-4 K} \eta  \sinh (\eta) \left[ \frac{1}{\left(K-\frac{\eta}{2}\right)^2}+1  \right]
		\\
		& = \frac{8 \e^{-2K} \sinh \left(\frac{\eta}{2}\right)^2 }{\left(K-\frac{\eta}{2}\right)^2} \left[ 1  -  25 \e^{-2 K} \frac{\eta  \sinh (\eta)}{\sinh \left(\frac{\eta}{2}\right)^2} \right] 
        + \sinh \left(\frac{\eta}{2}\right)^2 40 \e^{-2K}\left[ \frac{1}{3} -  5 \e^{-2 K} \frac{\eta  \sinh (\eta)}{\sinh \left(\frac{\eta}{2}\right)^2} \right]
		\\
		& >  8 \e^{-2K} \sinh \left(\frac{\eta}{2}\right)^2  \left(  \frac{1}{\left(K-\frac{\eta}{2}\right)^2} \left[ 1  -  50 \e^{-2 K} (\eta + 2) \right] 
		+  5  \left[ \frac{1}{3} -  10 \e^{-2 K} (\eta + 2) \right] \right)
		\\
		& >  8 \e^{-2K} \sinh \left(\frac{\eta}{2}\right)^2  \left(  \frac{1}{\left(K-\frac{\eta}{2}\right)^2} + \frac{5}{3} \right) \left[ 1  -  50 \e^{-2 K} (\eta + 2) \right] ,
	\end{align*}
	where we used the elementary bound
	\begin{equation}
		\frac{\eta  \sinh (\eta)}{\sinh \left(\frac{\eta}{2}\right)^2} =  2 \eta \frac{1+\e^{-\eta}}{1-\e^{-\eta}}  \le 2\eta + 4.
	\end{equation}
	Finally we have
	\begin{equation}
		1  -  50 \e^{-2 K} (\eta + 2) > 0 \qquad \text{for } K>\max(\eta/2,3), \, \eta>0
	\end{equation}
	which can be seen splitting the cases $\eta\in(0,6]$ and $\eta>6$.
    When $\eta\in(0,6]$ we have
	\begin{equation}
		1  -  50 \e^{-2 K} (\eta + 2) > 1  -  50 \e^{-6} (\eta + 2) > 1  -  400 \e^{-6}  
        > 0.
	\end{equation}
	On the other hand, when $\eta>6$ we have
	\begin{equation}
		1  -  50 \e^{-2 K} (\eta + 2) > 1  -  50 \e^{-\eta} (\eta + 2) > 1  -  400 \e^{-6} > 0,
	\end{equation}
	because the function $\e^{-\eta} (\eta + 2)$ is strictly decreasing for $\eta>6$. This completes the proof.
\end{proof}

\newcommand{\etalchar}[1]{$^{#1}$}


\begin{thebibliography}{DLDMS19}

\bibitem[ABW23]{aggarwal_borodin_wheeler_tPNG}
A.~Aggarwal, A.~Borodin, and M.~Wheeler.
Deformed Polynuclear Growth in (1+1) Dimensions.
Int. Math. Res. Not. IMRN~2023, no.~7, 5728--5780.

\bibitem[ACG23]{aggarwal2023asep}
A.~Aggarwal, I.~Corwin, and P.~Ghosal.
The ASEP speed process.
Adv. Math.~{\bf 422} (2023), Paper No. 109004, 57 pp.

\bibitem[ACH24]{aggarwal2024scaling}
A.~Aggarwal, I.~Corwin, and M.~Hegde.
Scaling limit of the colored ASEP and stochastic six-vertex models.
Preprint arXiv:2403.01341.

\bibitem[AS77]{ablowitz1977asymptotic}
M.~Ablowitz and H.~Segur.
Asymptotic solutions of the Korteweg--de~Vries equation.
Studies in Appl. Math.~{\bf 57} (1976/77), no.~1, 13--44.

\bibitem[BB14]{BleherBothnerConstant}
P.~Bleher and T.~Bothner.
Calculation of the constant factor in the six-vertex model.
Ann. Inst. Henri Poincaré D~{\bf 1} (2014), no.~4, 363--427.

\bibitem[BB18]{bothner2018large}
T.~Bothner and R.~Buckingham.
Large deformations of the Tracy--Widom distribution I: Non-oscillatory asymptotics.
Comm. Math. Phys.~{\bf 359} (2018), no.~1, 223--263.

\bibitem[BB19]{betea_bouttier_periodic}
D.~Betea and J.~Bouttier.
The periodic Schur process and free fermions at finite temperature.
Math. Phys. Anal. Geom.~{\bf 22} (2019), no.~1, Paper No.~3, 47~pp.

\bibitem[BBD08]{baik2008asymptotics}
J.~Baik, R.~Buckingham, and J.~DiFranco.
Asymptotics of Tracy--Widom distributions and the total integral of a Painlevé II function.
Comm. Math. Phys.~{\bf 280} (2008), no.~2, 463--497.

\bibitem[BCMS25]{byun2025precise}
S.-S.~Byun, C.~Charlier, P.~Moreillon, and N.~Simm.
Precise large deviations in geometric last passage percolation.
Preprint arXiv:2510.17470.

\bibitem[BCT22]{bothner2022momenta}
T.~Bothner, M.~Cafasso, and S.~Tarricone.
Momenta spacing distributions in anharmonic oscillators and the higher order finite temperature Airy kernel.
Ann. Inst. Henri Poincaré Probab. Stat.~{\bf 58} (2022), no.~3, 1505--1546.

\bibitem[BDIK15]{bothner2015asymptotic}
T.~Bothner, P.~Deift, A.~Its, and I.~Krasovsky.
On the asymptotic behavior of a log gas in the bulk scaling limit in the presence of a varying external potential I.
Comm. Math. Phys.~{\bf 337} (2015), no.~3, 1397--1463.

\bibitem[BDJ99]{baik1999distribution}
J.~Baik, P.~Deift, and K.~Johansson.
On the distribution of the length of the longest increasing subsequence of random permutations.
J. Amer. Math. Soc.~{\bf 12} (1999), no.~4, 1119--1178.

\bibitem[BDJ00]{Baik1999second}
J.~Baik, P.~Deift, and K.~Johansson.
On the distribution of the length of the second row of a Young diagram under Plancherel measure.
Geom. Funct. Anal.~{\bf 10} (2000), no.~4, 702--731.

\bibitem[BEK97]{BoutetEgorovaKhruslov}
A.~Boutet de Monvel, I.~Egorova, and E.~Khruslov.
Soliton asymptotics of the Cauchy problem solution for the Toda lattice.
Inverse Problems~{\bf 13} (1997), no.~2, 223--237.

\bibitem[BF06]{bleher2006exact}
P.~Bleher and V.~Fokin.
Exact solution of the six-vertex model with domain wall boundary conditions. Disordered phase.
Comm. Math. Phys.~{\bf 268} (2006), no.~1, 223--284.

\bibitem[BG16]{borodin2016moments}
A.~Borodin and V.~Gorin.
Moments match between the KPZ equation and the Airy point process.
SIGMA Symmetry Integrability Geom. Methods Appl.~{\bf 12} (2016), Paper No.~102, 7~pp.

\bibitem[BIP19]{bothner2019analysis}
T.~Bothner, A.~Its, and A.~Prokhorov.
On the analysis of incomplete spectra in random matrix theory through an extension of the Jimbo--Miwa--Ueno differential.
Adv. Math.~{\bf 345} (2019), 483--551.

\bibitem[BK92]{BlochKodama}
A.~Bloch and Y.~Kodama.
Dispersive regularization of the Whitham equation for the Toda lattice.
SIAM J. Appl. Math.~{\bf 52} (1992), no.~4, 909--928.

\bibitem[BKMM07]{BKMM2003}
J.~Baik, T.~Kriecherbauer, K.~T.-R. McLaughlin, and P.~D. Miller.
Discrete orthogonal polynomials. Asymptotics and applications.
Ann. of Math. Stud., 164.
\textit{Princeton University Press, Princeton, NJ,} 2007.

\bibitem[BKS11]{BoutetKotlyarovShepelski11}
A.~Boutet de Monvel, V.~Kotlyarov, and D.~Shepelsky.
Focusing NLS equation: long-time dynamics of step-like initial data.
Int. Math. Res. Not. IMRN 2011, no. 7, 1613--1653.

\bibitem[BL09]{bleher2009exact_critical_line}
P.~Bleher and K.~Liechty.
Exact solution of the six-vertex model with domain wall boundary conditions. Critical line between ferroelectric and disordered phases.
J. Stat. Phys.~{\bf 134} (2009), no.~3, 463--485.

\bibitem[BL10]{bleher2010exact}
P.~Bleher and K.~Liechty.
Exact solution of the six-vertex model with domain wall boundary conditions: antiferroelectric phase.
Comm. Pure Appl. Math.~{\bf 63} (2010), no.~6, 779--829.

\bibitem[BL11]{bleher2011uniform}
P.~Bleher and K.~Liechty.
Uniform asymptotics for discrete orthogonal polynomials with respect to varying exponential weights on a regular infinite lattice.
Int. Math. Res. Not. IMRN~2011, no.~2, 342--386.

\bibitem[BL14]{bleher2013random}
P.~Bleher and K.~Liechty.
Random matrices and the six-vertex model.
CRM Monogr. Ser., 32.
\textit{American Mathematical Society, Providence, RI,} 2014.

\bibitem[BLM09]{bertola2009mesoscopic}
M.~Bertola, S. Y.~Lee, and M. Y.~Mo.
Mesoscopic colonization in a spectral band.
J. Phys. A~{\bf 42} (2009), no.~41, 415204, 17~pp.

\bibitem[BM19]{BertolaMinakov}
M.~Bertola and A.~Minakov.
Laguerre polynomials and transitional asymptotics of the modified Korteweg--de~Vries equation for step-like initial data.
Anal. Math. Phys.~{\bf 9} (2019), no.~4, 1761--1818.

\bibitem[BO00]{Borodin2000a}
A.~Borodin and G.~Olshanski.
Distributions on partitions, point processes, and the hypergeometric kernel.
Comm. Math. Phys.~{\bf 211} (2000), no.~2, 335--358.

\bibitem[BO01]{Borodin1999RSK}
A.~Borodin and G.~Olshanski.
$z$-measures on partitions, Robinson--Schensted--Knuth correspondence, and $\beta=2$ random matrix ensembles.
Random matrix models and their applications, 71--94.
Math. Sci. Res. Inst. Publ., 40.
\textit{Cambridge University Press, Cambridge,} 2001.

\bibitem[BO09]{Borodin2007}
A.~Borodin and G.~Olshanski.
Infinite-dimensional diffusions as limits of random walks on partitions.
Probab. Theory Related Fields~{\bf 144} (2009), no.~1-2, 281--318.

\bibitem[BO17]{BO2016_ASEP}
A.~Borodin and G.~Olshanski.
The ASEP and determinantal point processes.
Comm. Math. Phys.~{\bf 353} (2017), no.~2, 853--903.

\bibitem[BOO00]{Borodin2000b}
A.~Borodin, A.~Okounkov, and G.~Olshanski.
Asymptotics of Plancherel measures for symmetric groups.
J. Amer. Math. Soc.~{\bf 13} (2000), no.~3, 481--515.

\bibitem[Bor00]{borodin2000riemann}
A.~Borodin.
Riemann-Hilbert problem and the discrete Bessel kernel.
Internat. Math. Res. Notices 2000, no.~9, 467--494.

\bibitem[Bor03]{borodin2003discrete}
A.~Borodin.
Discrete gap probabilities and discrete Painlevé equations.
Duke Math. J.~{\bf 117} (2003), no.~3, 489--542.

\bibitem[Bor07]{borodin2007periodic}
A.~Borodin.
Periodic Schur process and cylindric partitions.
Duke Math. J.~{\bf 140} (2007), no.~3, 391--468.

\bibitem[Bor18]{borodin2016stochastic_MM}
A.~Borodin.
Stochastic higher spin six vertex model and Macdonald measures.
J. Math. Phys.~{\bf 59} (2018), no.~2, 023301, 17~pp.

\bibitem[BP24]{burenev2024thermodynamics}
I.~Burenev and A.~Pronko.
Thermodynamics of the five-vertex model with scalar-product boundary conditions.
Comm. Math. Phys.~{\bf 405} (2024), no.~6, Paper No.~148, 56~pp.

\bibitem[BSY25]{byun2025free}
S.S.~Byun, S.M.~Seo and M.~Yang.
Free energy expansions of a conditional GinUE and large deviations of the smallest eigenvalue of the LUE.
Comm. Pure and Appl. Math.~{\bf 78} (2025), no.~12, 2247~2304.

\bibitem[BT91]{basor1991fisher}
E.~Basor and C.~Tracy.
The Fisher--Hartwig conjecture and generalizations.
Current problems in statistical mechanics (Washington, DC, 1991).
Phys. A~{\bf 177} (1991), no. 1-3, 167--173.

\bibitem[Buf16]{bufetov2016rigidity}
A.~Bufetov.
Rigidity of determinantal point processes with the Airy, the Bessel and the gamma kernel.
Bull. Math. Sci.~{\bf 6} (2016), no.~1, 163--172.

\bibitem[Buf24]{bufetov2024expectation}
A.~Bufetov.
The expectation of a multiplicative functional under the sine-process.
Translation of Funktsional. Anal. i Prilozhen.~{\bf 58} (2024), no.~2, 23--33.
Funct. Anal. Appl.~{\bf 58} (2024), no.~2, 120--128.

\bibitem[BV07]{BuckinghamVenakides07}
R.~Buckingham and S.~Venakides.
Long-time asymptotics of the nonlinear Schrödinger equation shock problem.
Comm. Pure Appl. Math.~{\bf 60} (2007), no.~9, 1349--1414.

\bibitem[CC22]{cafasso_claeys_KPZ}
M.~Cafasso and T.~Claeys.
A Riemann--Hilbert approach to the lower tail of the Kardar--Parisi--Zhang equation.
Comm. Pure Appl. Math.~{\bf 75} (2022), no.~3, 493--540.

\bibitem[CCR21]{cafasso2021airy}
M.~Cafasso, T.~Claeys, and G.~Ruzza.
Airy kernel determinant solutions to the KdV equation and integro-differential Painlevé equations.
Comm. Math. Phys.~{\bf 386} (2021), no.~2, 1107--1153.

\bibitem[CCR22]{charlier2022uniform}
C.~Charlier, T.~Claeys, and G.~Ruzza.
Uniform tail asymptotics for Airy kernel determinant solutions to KdV and for the narrow wedge solution to KPZ.
J. Funct. Anal.~{\bf 283} (2022), no.~8, Paper No.~109608, 54~pp.

\bibitem[CFWW25]{charlier2025asymptotics}
C.~Charlier, B.~Fahs, C.~Webb and M.D.~Wong.
Asymptotics of Hankel determinants with a multi-cut regular potential and Fisher-Hartwig singularities.
Mem. Am. Math. Soc.~{\bf 310}, no.~7 (2025)

\bibitem[CG20]{lwtail}
I.~Corwin and P.~Ghosal.
Lower tail of the KPZ equation.
Duke Math. J.~{\bf 169} (2020), no.~7, 1329--1395.

\bibitem[CG23]{claeys2023determinantal}
T.~Claeys and G.~Glesner.
Determinantal point processes conditioned on randomly incomplete configurations.
Ann. Inst. Henri Poincaré Probab. Stat.~{\bf 59} (2023), no.~4, 2189--2219.

\bibitem[CGK{\etalchar{+}}18]{corwin2018coulomb}
I.~Corwin, P.~Ghosal, A.~Krajenbrink, P.~Le~Doussal, and L.-C. Tsai.
Coulomb-Gas Electrostatics Controls Large Fluctuations of the Kardar--Parisi--Zhang Equation.
Phys. Rev. Lett.~{\bf 121} (2018), no.~6, 060201, 6~pp.

\bibitem[CGRT24]{claeys2024janossy}
T.~Claeys, G.~Glesner, G.~Ruzza, and S.~Tarricone.
Jánossy densities and Darboux transformations for the stark and cylindrical KdV equations.
Comm. Math. Phys.~{\bf 405} (2024), no.~5, Paper No.~113, 56~pp.

\bibitem[Cla08]{claeys2008birth}
T.~Claeys.
Birth of a cut in unitary random matrix ensembles.
Int. Math. Res. Not. IMRN 2008, no. 6, Art. ID rnm166, 40 pp.

\bibitem[CM25]{claeys2024large}
T.~Claeys and J.~Mauersberger.
Large deviations for the log-gamma polymer.
J. Lond. Math. Soc. (2) {\bf 112}~(2025), no.~3, Paper No.~e70295, 44~pp.

\bibitem[CMP25]{colomo2025five}
F.~Colomo, M.~Mannatzu, and A.G.~Pronko.
The five-vertex model as a discrete log-gas.
Preprint arXiv:2512.23223 (2025).

\bibitem[CR23]{CafassoRuzza23}
M.~Cafasso and G.~Ruzza.
Integrable equations associated with the finite-temperature deformation of the discrete Bessel point process.
J. Lond. Math. Soc. (2)~{\bf 108} (2023), no.~1, 273--308.

\bibitem[CS25]{claeys2024deformations}
T.~Claeys and G.~Silva.
Deformations of biorthogonal ensembles and universality.
Electron. J. Probab.~{\bf 30} (2025), Paper No.~148, 35~pp.

\bibitem[CT24]{claeys2024integrable}
T.~Claeys and S.~Tarricone.
On the integrable structure of deformed sine kernel determinants.
Math. Phys. Anal. Geom.~{\bf 27} (2024), no.~1, Paper No.~3, 35~pp.

\bibitem[Dei99]{DeiftBook}
P.~Deift.
Orthogonal polynomials and random matrices: a Riemann-Hilbert approach.
Courant Lect. Notes Math., 3.
\textit{New York University, Courant Institute of Mathematical Sciences, New York; American Mathematical Society, Providence, RI,}~1999. 

\bibitem[DIK08]{Deift_Its_Krasovsky_Airy}
P.~Deift, A.~Its, and I.~Krasovsky.
Asymptotics of the Airy-kernel determinant.
Comm. Math. Phys.~{\bf 278} (2008), no.~3, 643--678.

\bibitem[DIKZ07]{deift2007widom}
P.~Deift, A.~Its, I.~Krasovsky, and X.~Zhou.
The Widom--Dyson constant for the gap probability in random matrix theory.
J. Comput. Appl. Math.~{\bf 202} (2007), no.~1, 26--47.

\bibitem[DK93]{douglas1993large}
M.~Douglas and V.~Kazakov.
Large $N$ phase transition in continuum QCD$_2$.
Phys. Lett. B~{\bf 319} (1993), no. 1-3, 219--230.

\bibitem[DKKZ96]{deift1996toda}
P.~Deift, S.~Kamvissis, T.~Kriecherbauer, and X.~Zhou.
The Toda rarefaction problem.
Comm. Pure Appl. Math.~{\bf 49} (1996), no.~1, 35--83.

\bibitem[DKM{\etalchar{+}}99]{DeiftEtAl_multicut}
P.~Deift, T.~Kriecherbauer, K.~T.-R. McLaughlin, S.~Venakides, and X.~Zhou.
Uniform asymptotics for polynomials orthogonal with respect to varying exponential weights and applications to universality questions in random matrix theory.
Comm. Pure Appl. Math.~{\bf 52} (1999), no.~11, 1335--1425.

\bibitem[DL23]{drillick2023strong}
H.~Drillick and Y.~Lin.
Strong law of large numbers for the stochastic six vertex model.
Electron. J. Probab.~{\bf 28} (2023), Paper No.~148, 21~pp.

\bibitem[DLDMS15]{dean2015universal}
D.~Dean, P.~Le~Doussal, S.~Majumdar, and G.~Schehr.
Universal ground-state properties of free fermions in a $d$-dimensional trap.
Europhys. Lett. EPL~{\bf 112} (2015), no.~6, Art.~60001, 5~pp.

\bibitem[DLDMS19]{dean2019noninteracting}
D.~Dean, P.~Le~Doussal, S.~Majumdar, and G.~Schehr.
Noninteracting fermions in a trap and random matrix theory.
J. Phys. A~{\bf 52} (2019), no.~14, 144006, 32~pp.

\bibitem[DLM25a]{dlm24}
S.~Das, Y.~Liao, and M.~Mucciconi.
Lower tail large deviations of the stochastic six vertex model.
Int. Math. Res. Not. IMRN~2025, no.~18, Paper No.~rnaf276, 38~pp.

\bibitem[DLM25b]{das2025large}
S.~Das, Y.~Liao, and M.~Mucciconi.
Large deviations for the $q$-deformed polynuclear growth.
Ann. Probab.~{\bf 53} (2025), no.~4, 1223--1286.

\bibitem[{\relax DLMF}]{DLMF}
\textit{NIST Digital Library of Mathematical Functions}.
\url{https://dlmf.nist.gov/}, Release 1.2.5 of 2025‑12‑15.
F.~W.~J.~Olver, A.~B.~Olde~Daalhuis, D.~W.~Lozier, B.~I.~Schneider,
R.~F.~Boisvert, C.~W.~Clark, B.~R.~Miller, B.~V.~Saunders,
H.~S.~Cohl, and M.~A.~McClain, eds.

\bibitem[DVZ94]{deift1994collisionless}
P.~Deift, S.~Venakides, and X.~Zhou.
The collisionless shock region for the long-time behavior of solutions of the KdV equation.
Comm. Pure Appl. Math.~{\bf 47} (1994), no.~2, 199--206.

\bibitem[DZ93]{deift1993steepest}
P.~Deift and X.~Zhou.
A steepest descent method for oscillatory Riemann-Hilbert problems. Asymptotics for the MKdV equation.
Ann. of Math. (2)~{\bf 137} (1993), no.~2, 295--368.

\bibitem[DZ99]{deuschel_zeitouni_1999}
J.-D.~Deuschel and O.~Zeitouni.
On increasing subsequences of I.I.D. samples.
Combin. Probab. Comput.~{\bf 8} (1999), no.~3, 247--263.

\bibitem[Ehr06]{ehrhardt2006dyson}
T.~Ehrhardt.
Dyson's constant in the asymptotics of the Fredholm determinant of the sine kernel.
Comm. Math. Phys.~{\bf 262} (2006), no.~2, 317--341.

\bibitem[Ehr10]{ehrhardt2010asymptotics}
T.~Ehrhardt.
The asymptotics of a Bessel-kernel determinant which arises in random matrix theory.
Adv. Math.~{\bf 225} (2010), no.~6, 3088--3133.

\bibitem[EMPT23]{egorova2020long}
I.~Egorova, J.~Michor, A.~Pryimak, and G.~Teschl.
Long-time asymptotics for Toda shock waves in the modulation region.
J. Math. Phys. Anal. Geom.~{\bf 19} (2023), no.~2, 396--442.

\bibitem[EMT18]{egorova2014long}
I.~Egorova, J.~Michor, and G.~Teschl.
Long-time asymptotics for the Toda shock problem: non-overlapping spectra.
J. Math. Phys. Anal. Geom.~{\bf 14} (2018), no.~4, 406--451.

\bibitem[EPT24]{EPTKdV}
I.~Egorova, M.~Piorkowski, and G.~Teschl.
Asymptotics of KdV shock waves via the Riemann-Hilbert approach.
  Indiana Univ. Math. J.~{\bf 73} (2024), 
   645-690.

\bibitem[Eyn06]{eynard2006universal}
B.~Eynard.
Universal distribution of random matrix eigenvalues near the `birth of a cut' transition.
J. Stat. Mech. Theory Exp.~2006, no.~7, P07005, 33~pp.

\bibitem[FMS11]{forrester2011non}
P.~Forrester, S.~Majumdar, and G.~Schehr.
Non-intersecting Brownian walkers and Yang-Mills theory on the sphere.
Nuclear Phys. B~{\bf 844} (2011), no.~3, 500--526.

\bibitem[GGJ{\etalchar{+}}23]{girotti2023soliton}
M.~Girotti, T.~Grava, R.~Jenkins, K.~T.-R. McLaughlin, and A.~Minakov.
Soliton versus the gas: Fredholm determinants, analysis, and the rapid oscillations behind the kinetic equation.
Comm. Pure Appl. Math.~{\bf 76} (2023), no.~11, 3233--3299.

\bibitem[GGJM21]{girotti2021rigorous}
M.~Girotti, T.~Grava, R.~Jenkins, and K.~T.-R. McLaughlin.
Rigorous asymptotics of a KdV soliton gas.
Comm. Math. Phys.~{\bf 384} (2021), no.~2, 733--784.

\bibitem[Gho15]{ghosh2015determinantal}
S.~Ghosh.
Determinantal processes and completeness of random exponentials: the critical case.
Probab. Theory Related Fields~{\bf 163} (2015), no.~3-4, 643--665.

\bibitem[GL25]{gorin2025boundary}
V.~Gorin and K.~Liechty.
Boundary statistics for the six-vertex model with DWBC.
Comm. Pure Appl. Math.~{\bf 78} (2025), no.~10, 1847--1948.

\bibitem[GM20]{GravaMinakov}
T.~Grava and A.~Minakov.
On the long-time asymptotic behavior of the modified Korteweg--de~Vries equation with step-like initial data.
SIAM J. Math. Anal.~{\bf 52} (2020), no.~6, 5892--5993.

\bibitem[GNW79]{GREENE1979}
C.~Greene, A.~Nijenhuis, and H.~Wilf.
A probabilistic proof of a formula for the number of Young tableaux of a given shape.
Adv. in Math.~{\bf 31} (1979), no.~1, 104--109.

\bibitem[GIK{\etalchar{+}}25]{GIK+25}
M.~Guest, A.~Its, M.~Kosmakov, K.~Miyahara, and R.~Odoi.
Connection formulae for the radial Toda equations I.
Nonlinearity~{\bf 38} (2025), 035015.

\bibitem[GP73]{gurevich1973decay}
A.~Gurevich and L.~Pitaevskii.
Decay of Initial Discontinuity in the Korteweg--de~Vries Equation.
Pis'ma Zh. \`{E}ksper. Teoret. Fiz.~{\bf 17} (1973), no.~5, 268--271; translation in
JETP Lett.~{\bf 17} (1973), no.~5, 193--195.

\bibitem[GP74]{gurevich1974nonstationary}
A.~Gurevich and L.~Pitaevskii.
Nonstationary structure of a collisionless shock wave.
Zh. \`{E}ksper. Teoret. Fiz.~{\bf 65} (1974), no.~2, 590--604; translation in
JETP Lett.~{\bf 38} (1974), no.~2, 291--297.

\bibitem[GS23]{ghosal2023universality}
P.~Ghosal and G.~Silva.
Universality for multiplicative statistics of Hermitian random matrices and the integro-differential Painlevé II equation.
Comm. Math. Phys.~{\bf 397} (2023), no.~3, 1237--1307.

\bibitem[GS25]{Ghosal_Silva_6VM}
P.~Ghosal and G.~Silva.
The Stochastic Six Vertex model and discrete Orthogonal Polynomial ensembles.
Preprint arXiv:2512.22544 (2025).

\bibitem[IIKS90]{its1990differential}
A.~Its, A.~Izergin, V.~Korepin, and N.~Slavnov.
Differential equations for quantum correlation functions.
Proceedings of the Conference on Yang-Baxter Equations, Conformal Invariance and Integrability in Statistical Mechanics and Field Theory
Internat. J. Modern Phys. B~{\bf 4} (1990), no.~5, 1003--1037.

\bibitem[IMS24]{IMS_matching}
T.~Imamura, M.~Mucciconi, and T.~Sasamoto.
Identity between restricted Cauchy sums for the $q$-Whittaker and skew Schur polynomials.
SIGMA Symmetry Integrability Geom. Methods Appl.~{\bf 20} (2024), Paper No.~064, 28~pp.

\bibitem[IS05]{imamura2005polynuclear}
T.~Imamura and T.~Sasamoto.
Polynuclear growth model with external source and random matrix model with deterministic source.
Phys. Rev. E~{\bf 71} (2005), no.~4, Paper No.~041606.

\bibitem[Its11]{ItsLargeN}
A.~Its.
Large $N$ asymptotics in random matrices: the Riemann--Hilbert approach.
\textit{Random matrices, random processes and integrable systems,} 351--413.
CRM Ser. Math. Phys.
\textit{Springer, New York,} 2011.

\bibitem[Joh01]{Johansso1999Plancherel}
K.~Johansson.
Discrete orthogonal polynomial ensembles and the Plancherel measure.
Ann. of Math. (2)~{\bf 153} (2001), no.~1, 259--296.

\bibitem[JR21]{johansson_rahman_multitime}
K.~Johansson and M.~Rahman.
Multitime distribution in discrete polynuclear growth.
Comm. Pure Appl. Math.~{\bf 74} (2021), no.~12, 2561--2627.

\bibitem[JR22]{johansson_rahman_inhomogeneous}
K.~Johansson and M.~Rahman.
On inhomogeneous polynuclear growth.
Ann. Probab.~{\bf 50} (2022), no.~2, 559--590.

\bibitem[Kam93]{kamvissis1993toda}
S.~Kamvissis.
On the Toda shock problem.
Phys. D~{\bf 65} (1993), no.~3, 242--266.

\bibitem[Kam21]{kamchatnov2021gurevich}
A.~Kamchatnov.
Gurevich--Pitaevskii problem and its development.
Phys.-Usp.~{\bf 64} (2021), no.~1, 48--82.

\bibitem[Khr76]{Khruslov76}
Asymptotic behavior of the solution of the Cauchy problem for the Korteweg--de~Vries equation with steplike initial data.
Math. USSR-Sb.~{\bf 28} (1976), no.~2, 229--248.


\bibitem[KK89]{KhruslovKotlyarov89}
V.~Kotlyarov and E.~Khruslov.
Asymptotic solitons of the modified Korteweg--de~Vries equation.
Inverse Problems~{\bf 5} (1989), no.~6, 1075--1088.

\bibitem[KM10]{KotlyarovMinakov2010}
V.~Kotlyarov and A.~Minakov.
Riemann--Hilbert problem to the modified Korteveg--de~Vries equation: long-time dynamics of the steplike initial data.
J. Math. Phys.~{\bf 51} (2010), no.~9, 093506, 31~pp.

\bibitem[Kra09]{krasovsky2009large}
I.~Krasovsky.
Large gap asymptotics for random matrices.
\textit{New Trends in Mathematical Physics: Selected Contributions of the 15th International Congress on Mathematical Physics,} 413--419. 
\textit{Springer,} 2009.

\bibitem[Kra21]{krajenbrink2020painleve}
A.~Krajenbrink.
From Painlev\'{e} to Zakharov--Shabat and beyond: Fredholm determinants and integro-differential hierarchies.
J. Phys. A~{\bf 54} (2021), no.~3, Paper No.~035001, 51~pp.


\bibitem[Lis11]{Lisovyy_Dyson}
O.~Lisovyy.
Dyson's constant for the hypergeometric kernel.
\textit{New trends in quantum integrable systems,} 243--267.
\textit{World Scientific Publishing Co. Pte. Ltd., Hackensack, NJ,} 2011

\bibitem[LM15]{LevyMaida}
T.~L{\'e}vy and M.~Ma{\"\i}da.
On the Douglas-Kazakov phase transition. Weighted potential theory under constraint for probabilists.
\textit{Modélisation Aléatoire et Statistique---Journées MAS 2014,} 89--121.
ESAIM Proc. Surveys, 51
\textit{EDP Sciences, Les Ulis,} 2015

\bibitem[LS77]{logan_shepp1977variational}
B.~Logan and L.~Shepp.
A variational problem for random Young tableaux.
Advances in Math.~{\bf 26} (1977), no.~2, 206--222.

\bibitem[LS25]{landon2023tail}
B.~Landon and P.~Sosoe.
Tail estimates for the stationary stochastic six-vertex model and ASEP.
Probab. Math. Phys.~{\bf 6} (2025), no.~4, 1327--1378.

\bibitem[LW16]{LiechtyWang2016}
K.~Liechty and D.~Wang.
Nonintersecting Brownian motions on the unit circle.
Ann. Probab.~{\bf 44} (2016), no.~2, 1134--1211.

\bibitem[Mich16]{Mich16}
J.~Michor. 
Wave phenomena of the Toda lattice with steplike initial data. 
Phys. Lett. A {\bf 380} 11-12 (2016): 1110-1116.

\bibitem[Mo08]{mo2008riemann}
M.~Mo.
The Riemann-Hilbert approach to double scaling limit of random matrix eigenvalues near the ``birth of a cut'' transition.
Int. Math. Res. Not. IMRN~2008, no.~13, Art.~ID~rnn042, 51~pp.

\bibitem[MQR25]{matetski2025polynuclear}
K.~Matetski, J.~Quastel, and D.~Remenik.
Polynuclear growth and the toda lattice.
J. Eur. Math. Soc. (JEMS), (2025) published online first.

\bibitem[Oko01]{okounkov2001infinite}
A.~Okounkov.
Infinite wedge and random partitions.
Selecta Math. (N.S.)~{\bf 7} (2001), no.~1, 57--81.

\bibitem[Oko05]{okounkov2006uses}
A.~Okounkov.
The uses of random partitions.
\textit{XIVth International Congress on Mathematical Physics,} 379--403.
\textit{World Scientific Publishing Co. Pte. Ltd., Hackensack, NJ,} 2005.


\bibitem[PS02]{Praehofer2002}
M.~Pr{\"a}hofer and H.~Spohn.
Prähofer, Michael(D-MUTU-ZM); Spohn, Herbert(D-MUTU-ZM)
Current fluctuations for the totally asymmetric simple exclusion process.
\textit{In and out of equilibrium (Mambucaba, 2000),} 185--204.
Progr. Probab., 51.
\textit{Birkhäuser Boston, Inc., Boston, MA,} 2002.

\bibitem[QR22]{quastel_remenik_2022_KP}
J.~Quastel and D.~Remenik.
KP governs random growth off a 1-dimensional substrate.
Forum Math. Pi~{\bf 10} (2022), Paper No.~e10, 26~pp.

\bibitem[Rom15]{romik2015surprising}
D.~Romik.
The surprising mathematics of longest increasing subsequences.
IMS Textb., 4.
\textit{Cambridge University Press, New York,} 2015.

\bibitem[Ruz25]{ruzza2025bessel}
G.~Ruzza.
Bessel kernel determinants and integrable equations.
Ann. Henri Poincar\'e~{\bf 26} (2025), no.~6, 2035--2068.

\bibitem[Sch61]{Schensted1961}
C.~Schensted.
Longest increasing and decreasing subsequences.
Canadian J. Math.~{\bf 13} (1961), 179--191.

\bibitem[Sep98]{seppalainen_98_increasing}
T.~Sepp{\"a}l{\"a}inen.
Large deviations for increasing sequences on the plane.
Probab. Theory Related Fields~{\bf 112} (1998), no.~2, 221--244.

\bibitem[SS90]{sagan1990robinson}
B.~Sagan and R.~Stanley.
Robinson--Schensted algorithms for skew tableaux.
J. Combin. Theory Ser. A~{\bf 55} (1990), no.~2, 161--193.


\bibitem[Ten15]{tenenbaum2015introduction}
G.~Tenenbaum.
Introduction to analytic and probabilistic number theory.
Third edition. Translated from the 2008 French edition by Patrick D. F. Ion.
Grad. Stud. Math., 163.
\textit{American Mathematical Society, Providence, RI,} 2015.

\bibitem[Tra91]{tracy1991asymptotics}
C.~Tracy.
Asymptotics of a $\tau$-function arising in the two-dimensional Ising model.
Comm. Math. Phys.~{\bf 142} (1991), no.~2, 297--311.

\bibitem[Tsa22]{tsai_lower_tail}
L.-C. Tsai.
Exact lower-tail large deviations of the KPZ equation.
Duke Math. J.~{\bf 171} (2022), no.~9, 1879--1922.

\bibitem[TW93]{tracy1993level}
C.~Tracy and H.~Widom.
Level-spacing distributions and the Airy kernel.
Phys. Lett. B~{\bf 305} (1993), no. 1-2, 115--118.

\bibitem[TW98]{tracywidomcyltoda1}
C.~Tracy and H.~Widom.
Asymptotics of a class of solutions to the cylindrical Toda equations.
Comm. Math. Phys.~{\bf 190} (1998), no.~3, 697--721.

\bibitem[TW99]{tracywidomcyltoda2}
C.~Tracy and H.~Widom.
Asymptotics of a class of Fredholm determinants.
\textit{Spectral problems in geometry and arithmetic (Iowa City, IA, 1997),}167–174.
Contemp. Math., 237
\textit{American Mathematical Society, Providence, RI,} 1999

\bibitem[TW09]{tracy2009asymptotics}
C.~Tracy and H.~Widom.
Asymptotics in ASEP with step initial condition.
Comm. Math. Phys.~{\bf 290} (2009), no.~1, 129--154.

\bibitem[VDO91]{venakides1991toda}
S.~Venakides, P.~Deift, and R.~Oba.
The Toda shock problem.
Comm. Pure Appl. Math.~{\bf 44} (1991), no.~8-9, 1171--1242.

\bibitem[VK77]{VershikKerov_LimShape1077}
A.~Vershik and S.~Kerov.
Asymptotics of the Plancherel measure of the symmetric group and the limiting form of Young tableaux.
Dokl. Akad. Nauk SSSR~{\bf 233} (1977), no.~6, 1024--1027; translation in
Soviet Math. Dokl.~{\bf 18} (1977), 527--531.

\bibitem[Whi74]{Whitham}
G.~Whitham.
Linear and nonlinear waves.
Pure Appl. Math.
\textit{Wiley-Interscience [John Wiley \& Sons], New York-London-Sydney,} 1974.

\bibitem[Wid97]{widomToda}
H.~Widom.
Some classes of solutions to the Toda lattice hierarchy.
Comm. Math. Phys.~{\bf 184} (1997), no.~3,
653--667.

\bibitem[Zho24]{zhong2024large}
C.~Zhong.
Large deviation principle for the Airy point process.
Preprint arXiv:2404.06006 (2024).

\bibitem[ZJ00]{zinn2000sixvertex}
P.~Zinn-Justin.
Six-vertex model with domain wall boundary conditions and one-matrix model.
Phys. Rev. E (3)~{\bf 62} (2000), no.~3, part~A, 3411–3418.

\end{thebibliography}
\end{document}